\documentclass[final,3p,times]{elsarticle}

\usepackage{amssymb}
\usepackage{amsthm}
\usepackage{graphicx}
\usepackage{dcolumn}
\usepackage{bm}
\usepackage{subfigure}
\usepackage{amssymb}
\usepackage{verbatim}
\usepackage{amscd}
\usepackage{amssymb}
\usepackage{amsmath, amsfonts}
\usepackage{setspace}
\usepackage[]{graphicx}        
\usepackage{amsthm}
\usepackage{natbib}

\biboptions{sort&compress} 

\theoremstyle{plain}
\newtheorem{theorem}{Theorem}
\theoremstyle{plain}
\newtheorem{lemma}{Lemma}
\theoremstyle{plain}
\newtheorem{sublemma}{Sublemma} 
\theoremstyle{plain}

\theoremstyle{remark}

\theoremstyle{conjecture}

\theoremstyle{observation}
\newtheorem{observation}{Observation}
\theoremstyle{definition}

\theoremstyle{corollary}

\theoremstyle{definition}

\theoremstyle{definition}
\newtheorem{definition}{Definition}
\theoremstyle{assumption}

\theoremstyle{definition}

\theoremstyle{problem}

\theoremstyle{fact}

\begin{document}

\begin{frontmatter}

\title{Classification of quantum phases and topology of logical operators in an exactly solved model of quantum codes}

\author{Beni Yoshida}
\address{Center for Theoretical Physics, Massachusetts Institute of Technology, Cambridge, Massachusetts 02139, USA}
\ead{rouge@mit.edu}

\begin{abstract}
Searches for possible new quantum phases and classifications of quantum phases have been central problems in physics. Yet, they are indeed challenging problems due to the computational difficulties in analyzing quantum many-body systems and the lack of a general framework for classifications. While frustration-free Hamiltonians, which appear as fixed point Hamiltonians of renormalization group transformations, may serve as representatives of quantum phases, it is still difficult to analyze and classify quantum phases of arbitrary frustration-free Hamiltonians exhaustively. Here, we address these problems by sharpening our considerations to a certain subclass of frustration-free Hamiltonians, called stabilizer Hamiltonians, which have been actively studied in quantum information science. We propose a model of frustration-free Hamiltonians which covers a large class of physically realistic stabilizer Hamiltonians, constrained to only three physical conditions; the locality of interaction terms, translation symmetries and scale symmetries, meaning that the number of ground states does not grow with the system size. We show that quantum phases arising in two-dimensional models can be classified exactly through certain quantum coding theoretical operators, called logical operators, by proving that two models with topologically distinct shapes of logical operators are always separated by quantum phase transitions.
\end{abstract}


\end{frontmatter}

\tableofcontents

\section{Introduction}\label{sec:introduction}

In condensed matter physics, studies on quantum phases in two-dimensional spin systems have been central in addressing the underlying mechanisms behind correlated spin systems~\cite{Sachdev_Text}. There have also been considerable interests in quantum phases in two-dimensional systems from quantum information science community~\cite{Nielsen_Chuang}, as topological phases of spin systems are useful in realizing quantum information theoretical ideas~\cite{Kitaev97, Dennis02, Kitaev03, Nayak08}. In this paper, we ask two fundamental questions concerning quantum phases arising in two-dimensional spin systems. 
\begin{itemize}
\item[(a)] What kinds of quantum phases are allowed to exist in two-dimensional spin systems? 
\item[(b)] How do we classify these quantum phases? What kinds of properties distinguish different quantum phases? 
\end{itemize}

Finding exactly solvable models is a key to searching for possible quantum phases. The AKLT state, which is the ground state of an exactly solvable spin $1$ Hamiltonian, provides useful clues about the ground state properties of quantum anti-ferromagnets~\cite{AKLT87, AKLT88}. The Toric code, the simplest exactly solvable model with topological order, gives insights on how topological phases emerge in correlated spin systems~\cite{Kitaev97, Kitaev03}. These Hamiltonians mentioned here are called \emph{frustration-free} Hamiltonians since ground states can be obtained by minimizing each term locally. While these exactly solvable frustration-free models may not be physical, involving more than two spins at one time or higher order interactions, it is relatively easy to find their ground states compared with frustrated Hamiltonians which we encounter in physically realistic systems. A key idea is to realize that these frustration-free Hamiltonians may appear as low energy effective theories for real spin systems~\cite{Kitaev06b}, and thus, may serve as a ``bureau of quantum phases'' with which one may approximately study the ground state properties of actual Hamiltonians. 

So far, a lot of frustration-free spin Hamiltonians with various physical properties have been discovered in a search for possible quantum phases arising in two-dimensional spin systems~\cite{Kitaev03, Levin05}. However, whether these models could exhaust all the possible quantum phases or not is far from obvious. A brute force approach might be to analyze all the possible frustration-free Hamiltonians and study their ground state properties. However, being frustration-free does not mean that Hamiltonians are exactly solvable since studying their ground state properties may be challenging since computations of order parameters to distinguish quantum phases may involve a large number of spins. Moreover, finding frustration-free Hamiltonians is much more difficult than finding a ground state of frustration-free Hamiltonians since it is computationally difficult to check whether a given Hamiltonian is frustration-free or not. For example, even when a Hamiltonian consists only of projectors as in the AKLT Hamiltonian and the Toric code, to determine whether the Hamiltonian is frustration-free or not is known to be the hardest problem in the computational complexity class QMA (Quantum Marlin-Arthur), a quantum analog of NP~\cite{Kitaev_Text, Bravyi08}. (So, it could be a difficult problem even for a quantum computer to solve !). 

As for a classification of quantum phases, two frustration-free Hamiltonians may be classified through the presence or the absence of a quantum phase transition (QPT). In quantum many-body systems, two quantum phases with different physical properties are separated by a QPT, which is a sudden non-analytic change of the ground state properties as a result of parameter changes in the Hamiltonian bridging two quantum phases~\cite{Sachdev_Text}. Two frustration-free Hamiltonians may be considered to belong to different quantum phases when a parameterized Hamiltonian connecting them undergoes a QPT as a parameter varies in the Hamiltonian~\cite{Hastings05, Bravyi10b, Chen10}. 

While studying the ground state properties of parameterized Hamiltonians is generally difficult, QPTs occurring in parameterized Hamiltonians have been studied through Renormalization Group (RG) transformations which are certain scale transformations acting on parameterized Hamiltonians~\cite{Fisher98} or parameterized ansatz of ground states~\cite{Verstraete04, Verstraete06, Vidal07, Gu08, Chen10}. The central idea in classifying quantum phases through RG transformations is to capture only the scale invariant physical properties of corresponding quantum phases by using fixed point Hamiltonians which are invariant under RG transformations. Since fixed point Hamiltonians are obtained by washing out short-range correlations, it is known that fixed point Hamiltonians in gapped quantum phases can be represented as frustration-free Hamiltonians where each term can be minimized locally without considering neighboring terms. 

However, the challenge in classifying quantum phases in correlated spin systems underlies behind the fact that the existence of a QPT depends on how Hamiltonians are varied from the one to the other. In particular, there are cases where two frustration-free Hamiltonians are connected through some parameterized Hamiltonian without a QPT, but are separated by a QPT for other parameterized Hamiltonian~\cite{Lieb61, Vidal03, Wolf06, Orus09, Skrovseth09}. To show that two frustration-free Hamiltonians belong to different quantum phases, one needs to analyze all the possible paths of parameterized Hamiltonians connecting them and see if they are always separated by QPTs or not. Thus, the need is to find some order parameter to classify quantum phases in a way independent of choices of parameterized Hamiltonians. 

Searches for possible quantum phases and their classifications lie at the heart of studies on quantum many-body systems. Currently, both problems however face challenges as illustrated in the above. In this paper, we make an attempt to overcome these difficulties concerning studies on quantum phases by combining theoretical tools developed in quantum information science and the notion of topology.

\textbf{Possible quantum phases and stabilizer codes:}
While analyses on arbitrary frustration-free Hamiltonians are indeed challenging, there exists a certain subclass of frustration-free Hamiltonians, called stabilizer Hamiltonians~\cite{Gottesman96, Nielsen_Chuang}, which play a crucial role in quantum information science. Stabilizer Hamiltonians serve as a natural architecture to realize quantum error-correcting codes in correlated many-body spin systems. The basic idea of stabilizer codes is to encode a qubit in strongly entangled and correlated states so that no small error can break the encoded qubit. A remarkable feature of stabilizer codes is the existence of system Hamiltonians which supports encoded qubits in degenerate ground states with a finite energy gap. Such system Hamiltonians, called stabilizer Hamiltonians, are frustration-free Hamiltonians which consist only of Pauli operators commuting with each other. Since protecting a qubit from decoherence is essential in realizing quantum information theoretical ideas such as fault-tolerant quantum computation~\cite{Nielsen_Chuang}, studies on stabilizer Hamiltonians have been central in quantum information science.

Stabilizer Hamiltonians are certainly limited compared with arbitrary frustration-free Hamiltonians, yet they may characterize a large number of quantum phases. In fact, it has been realized that many interesting spin systems may be described in the language of stabilizer codes~\cite{Hein04, Fattal04, Hamma05, Bacon06, Kitaev06b, Bravyi09, Beni10}, as the Toric code is one of the earliest examples of stabilizer codes discovered in a search for useful quantum codes~\cite{Kitaev97}. Also, a good news in analyzing quantum phases arising in stabilizer Hamiltonians is that there are many useful theoretical tools to analyze the ground state properties in stabilizer codes, which are originally developed in quantum coding theoretical contexts~\cite{Nielsen_Chuang, Hein04, Fattal04}, but are potentially useful for studying non-local correlations arising in many-body systems~\cite{Bravyi09, Beni10}. Thus, the analysis on quantum phases arising in stabilizer Hamiltonians and their classification will be the necessary first step in order to address quantum phases arising in arbitrary frustration-free Hamiltonians and search for all the possible two-dimensional quantum phases.

There are several useful theoretical tools in analyzing the ground state properties of stabilizer Hamiltonians. Basic analysis tools for the ground state properties of stabilizer Hamiltonians were developed for the cases where a ground state is unique~\cite{Hein04, Fattal04}. These tools were further generalized to stabilizer Hamiltonians with the ground state degeneracy with an aim to study global entanglement arising in many-body spin systems such as topological order~\cite{Beni10}. However, there is still an important gap between physically realistic correlated spin systems and stabilizer codes discussed in quantum information science community, as most stabilizer codes encode qubits in highly non-local Hamiltonians which are not physically realistic. 

Several authors have initiated studies on physically realistic stabilizer Hamiltonians by analyzing coding properties of stabilizer codes supported by local terms~\cite{Bravyi09, Kay08, Bravyi10}. However, physically realizable Hamiltonians must have some physical symmetries such as translation symmetries too if they are to be realized in some lattice of spins. Also, there is another important physical constraint which must be considered in analyzing quantum phases arising in stabilizer Hamiltonians. According to the underlying philosophy behind RG transformations, quantum phases must be characterized by non-local correlations which survive even at the thermodynamic limit (a limit where the system size becomes infinite) and are invariant under scale transformations~\cite{Fisher98}. Since quantum phases must be classified in a scale invariant way, stabilizer Hamiltonians need to have scale invariance too if they are to be used as candidates for possible quantum phases. 

In this paper, to search for possible quantum phases in two-dimensional spin systems, we begin by proposing a model of stabilizer Hamiltonians which have both physical realizability and scale invariance. In particular, the model is built on the following physical constraints.
\begin{itemize}
\item \textbf{Locality of interactions:}\ A system of qubits, defined on a two-dimensional lattice, is governed by a stabilizer Hamiltonian which consists only of local interactions.
\item \textbf{Translation symmetries:}\ The Hamiltonian is invariant under some finite translations of qubits.
\item \textbf{Scale symmetries:}\ The number of degenerate ground states of the Hamiltonian does not depend on the system size.
\end{itemize} 
We call stabilizer Hamiltonians which satisfy the above three physical constraints \emph{Stabilizer codes with Translation and Scale symmetries} (\textbf{STS models}).

\textbf{Classification of quantum phases and topology:}
The goal of this paper is to analyze and classify quantum phases arising in STS models completely. In the analysis and classification of quantum phases arising in STS models, quantum phases must be distinguished by some quantities or objects which are scale invariant. Here, it is useful to realize the similarity between the notion of topology and the classification of quantum phases by observing that geometric shapes of objects are scale invariant. In topology, geometric shapes of objects are classified in terms of smoothness and non-analyticity. Two objects are considered to be the same when they can be transformed each other continuously, while they are considered to be different when they can be transformed each other through only non-analytic changes of geometric shapes. In a similar fashion, quantum phases are classified in terms of continuous changes of ground states induced by changes of parameterized Hamiltonians. Two frustration-free Hamiltonians are considered to belong to the same quantum phase when their ground states can be connected continuously. Two frustration-free Hamiltonians are considered to belong to different quantum phases when they can be connected through only parameterized Hamiltonians which undergo QPTs with non-analytic changes of ground states. Thus, the notion of topology and a classification of quantum phases are fundamentally akin to each other. 

The notion of topology has been playing crucial roles in classifying physical properties of quantum many-body systems in various fields of physics~\cite{Nayak08, Volovik_Text}. Quantum field theory, equipped with the notion of topology, has been proposed in high energy physics where physical properties of systems depend on topological characteristics of the geometric manifolds where the systems are defined~\cite{Witten89}. As a physical realization of such topological theories in low-dimensional many-body systems, quantum hall systems defined on a geometric manifold were studied and the notion of topological order has been introduced in order to discuss and characterize topology-dependent non-local correlations~\cite{Wen90}. The notion of topology appears not only in real physical space, but also in the momentum space too where QPTs are triggered by topological changes in the spectrum~\cite{Lifshitz60}. Finally, when Hamiltonians have several parameters, the parameter space may have non-trivial topological structures which may result in interesting behaviors of geometric (Berry) phases~\cite{Wilczek_Text, Hamma06}. 

With this intimate connection between topology and quantum phases, we wish to classify quantum phases arising in STS models by introducing the notion of topology into quantum coding theoretical tools developed in studies of stabilizer Hamiltonians. What plays a central role in analyses on the ground state properties in stabilizer Hamiltonians are certain operators, called \emph{logical operators}~\cite{Nielsen_Chuang}, which are Pauli operators commuting with the system Hamiltonian, but act non-trivially inside the ground state space. A useful empirical knowledge commonly shared in quantum information science community is that ground states are highly entangled when logical operators are supported by a large number of spins for a fixed system size. In particular, it has been pointed out that geometric non-localities of logical operators may characterize global entanglement arising in ground states of stabilizer Hamiltonians, such as topological order~\cite{Bravyi09, Beni10}. 

In this paper, we classify quantum phases arising in STS models by introducing the notion of topology into geometric shapes of logical operators. In particular, we achieve the following two programs.
\begin{itemize}
\item \textbf{Exact solution of the model:} We solve the STS model exactly by identifying all the possible geometric shapes of logical operators. We establish the connection between the scale invariant ground state properties, such as topological order, and geometric shapes of logical operators completely.\footnote{Here, we claim that the model is exactly solved by finding logical operators since logical operators allow us to determine essential physical properties of the model as we shall see in the main discussion of this paper.}
\item \textbf{QPTs and topology in logical operators:} We show that STS models with different geometric shapes of logical operators are always separated by a QPT. We find that the existence of a QPT originates from changes of geometric shapes of topologically distinct logical operators. We show that topological structures of geometric shapes of logical operators can be used as order parameters to distinguish quantum phases arising in one and two-dimensional STS models.
\end{itemize} 

\textbf{Organization and contents:}
The paper is organized as follows. In Section~\ref{sec:review}, we begin by introducing basic notions used throughout this paper by giving introductions to QPTs and stabilizer codes. A definition of quantum phases and a criteria for their classification are given. Also, logical operators, which are central in the analyses on stabilizer Hamiltonians, are reviewed. In Section~\ref{sec:model}, we introduce the STS model by imposing physical constraints and scale invariance. Then, we develop a basic analysis tool, concerning a certain symmetry on logical operators, which arises as a result of physical symmetries and scale invariance we impose on stabilizer Hamiltonians. In Section~\ref{sec:1D} and Section~\ref{sec:2D}, we analyze quantum phases arising in one and two-dimensional STS model and show that different quantum phases are characterized by geometric shapes of logical operators. We find that the existence of a QPT originates from changes of geometric shapes of topologically distinct logical operators. We show that topological characteristics of geometric shapes of logical operators can be used as order parameters to distinguish quantum phases arising in one and two-dimensional STS models.

Here we make some comments on the contents of the paper. First, the paper is written in a way accessible to both condensed matter physicists and quantum information scientists, providing introductions to quantum coding theory and the notion of QPTs and topological order. In particular, the paper is presented in a self-consistent way. Second, while a mathematical rigor is a virtue, but so is the physical intuition. Thus, we have presented all the derivations and calculations of logical operators in appendices (from~\ref{sec:TE} to~\ref{sec:loop}) so that we can concentrate on physical discussions concerning quantum phases arising in STS models. Finally, while the locality of interaction terms and translation symmetries are commonly used as physical constrains in studies of many-body spin systems, the importance of scale symmetries may not be obvious. In~\ref{sec:translation}, we expand our analyses to stabilizer Hamiltonians without scale symmetries in order to study the role of scale symmetries. We point out the connection between translation symmetries of ground states and scale symmetries by analyzing the phenomena called weak breaking of translation symmetries~\cite{Kitaev06b}.

\section{Review: classification of quantum phases and stabilizer Hamiltonians}\label{sec:review} 

We begin with quick reviews of theoretical notions which are central to discussions in this paper. In Section~\ref{sec:review1}, we review how different quantum phases are classified through frustration-free Hamiltonians with some examples of quantum phase transitions (QPTs). Then, we state the criteria for classifying quantum phases clearly for the sake of later discussions. The notion of scale invariance and its role in a classification of quantum phases are reviewed and discussed. In Section~\ref{sec:review2}, we give a review of stabilizer codes which are quantum error-correcting codes realized by a certain subclass of frustration-free Hamiltonians. Some theoretical tools which become central in analyzing STS models are introduced here. 

Here, we make a comment on the notations used throughout this paper. A state of a qubit is described as a superposition of $|0\rangle$ and $|1\rangle$, and Pauli operators are represented by alphabets $X$, $Y$, and $Z$ in this paper. 

\subsection{Classification of quantum phases through frustration-free Hamiltonians}\label{sec:review1}

In a quantum many-body system at zero temperature, physical properties of ground states may change abruptly as a result of parameter changes in the Hamiltonian~\cite{Sachdev_Text}. In such a quantum phase transition (QPT), quantum phases may be characterized by frustration-free Hamiltonians. Here, we first review two parameterized Hamiltonians which undergo QPTs and discuss how quantum phases can be classified. Then, we state a criteria to classify quantum phases in this paper clearly.

The models of QPTs we review in this subsection are known to be exactly solvable, meaning that one can characterize the ground state properties and obtain the energy spectrum through analytical solutions of the Hamiltonians. The solutions to models may be found in the original literature~\cite{Lieb61} and some standard textbooks on condensed matter physics (for example, see~\cite{Sachdev_Text}). Solutions and discussions on these models with quantum information theoretical viewpoints may be found in~\cite{Vidal03, Latorre04}.

\textbf{An example of QPTs and frustration-free Hamiltonians:}
We consider the Ising model in a transverse field in one dimension with periodic boundary conditions:
\begin{align}
H(\epsilon) \ = - (1 - \epsilon)\sum_{j}Z^{(j)}Z^{(j+1)} - \epsilon \sum_{j}X^{(j)}
\end{align}
where $Z^{(j)}$ acts on $j$-th qubit (spin $1/2$ particle). The system consists of $N$ qubits (spin $1/2$ particles). At $\epsilon = 0$, the system Hamiltonian is
\begin{align}
H(0) \ = - \sum_{j}Z^{(j)}Z^{(j+1)}
\end{align}
and the ground states are $|0 \cdots 0 \rangle$ and $|1 \cdots 1 \rangle$. At $\epsilon = 1$, the system Hamiltonian is
\begin{align}
H(1) \ = - \sum_{j}X^{(j)}
\end{align}
and the ground state is $|+ \cdots + \rangle$ where $|+\rangle = \frac{1}{\sqrt{2}}(|0\rangle + |1 \rangle)$. It is known that physical properties of ground states drastically change around $\epsilon = 1/2$. When $0 \leq \epsilon < 1/2$, physical properties of a ground state are close to ones of $H(0)$. When $1/2 < \epsilon \leq 1$, physical properties of a ground state are close to ones of $H(1)$. This drastic change of physical properties at $\epsilon = 1/2$ is called a QPT. Two Hamiltonians $H(0)$ and $H(1)$ are separated by a QPT at $\epsilon = 1/2$, and may represent different quantum phases (Fig.~\ref{fig_RG}(a)). The existence of a QPT can be detected by considering the energy gap between a ground state and excited states. Let us consider how the energy gap changes around the transition point in this model. Then, at the transition point $\epsilon = 1/2$, the energy gap $\Delta(\epsilon)$ between a ground state and excited states becomes zero at the thermodynamic limit: $\Delta(1/2) \rightarrow 0$ for $N \rightarrow \infty$. Thus, the vanishing energy gap serves as an indicator of a QPT.


\begin{figure}[htb!]
\centering
\includegraphics[width=0.65\linewidth]{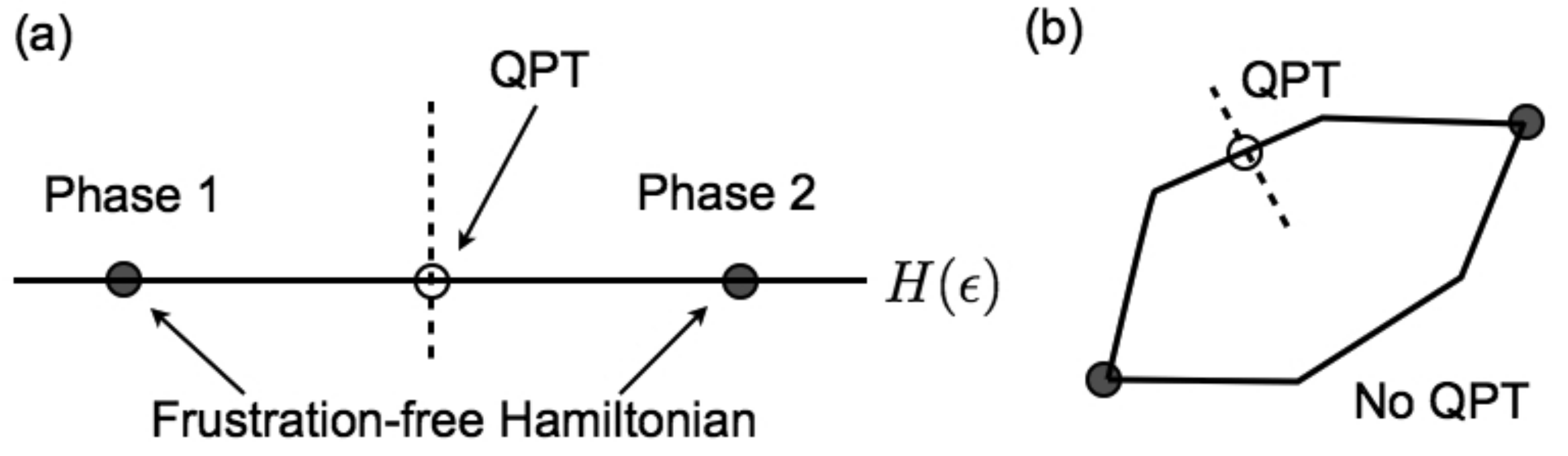}
\caption{(a) Two quantum phases separated by a QPT. Black circles represent frustration-free Hamiltonians and a blank circle represents a QPT. Different quantum phases are represented by frustration-free Hamiltonians. (b) A path dependence of the presence of a QPT. Two frustration-free Hamiltonians are considered to belong to the same quantum phase when there exists a parametrized Hamiltonian which connects them without closing the energy gap.
} 
\label{fig_RG}
\end{figure}

Since a ground state of the original parameterized Hamiltonian $H(\epsilon)$ cannot be obtained by minimizing each term in the Hamiltonians except at $\epsilon = 0, 1$, $H(\epsilon)$ is said to be \emph{frustrated}. In general, finding a ground state of a frustrated Hamiltonian is computationally difficult. However, Hamiltonians $H(0)$ and $H(1)$ are not frustrated, meaning that their ground states can be obtained by minimizing each term in the Hamiltonian independently. In this example, physical properties of frustration-free Hamiltonians $H(0)$ and $H(1)$ are easy to analyze. Then, physical properties of the original Hamiltonian $H(\epsilon)$ may be studied by adding perturbations to frustration-free Hamiltonians $H(0)$ and $H(1)$. Thus, $H(0)$ and $H(1)$ may serve as easily solvable reference models to characterize quantum phases (Fig.~\ref{fig_RG}(a)). 

\textbf{Path dependence of the presence of a QPT:}
Then, we are tempted to call two quantum phases different if they are separated by a QPT. However, there is a subtlety in this classification of quantum phases since the presence of a QPT may depend on how we change the Hamiltonian. Let us see an example:
\begin{align}
H(\epsilon) \ = - (1 - \epsilon)\sum_{j}Z^{(j)}Z^{(j+1)} - \epsilon \sum_{j} X^{(j)}X^{(j+1)}.
\end{align}
This model exhibit a QPT at $\epsilon = 1/2$ where the energy gap $\Delta(\epsilon)$ becomes zero. Then, one might hope to consider two quantum phases at $0 \leq \epsilon < 1/2$ and at $1/2 < \epsilon \leq 1$ different. However, there exists another parameterized Hamiltonian $H'(\epsilon)$ which connects $H(0)$ and $H(1)$ without closing an energy gap between ground states and excited states (Fig.~\ref{fig_RG}(b)):
\begin{align}
H'(\epsilon) \ = \ - \sum_{j}M(\epsilon)^{(j)}M(\epsilon)^{(j+1)}, \quad M(\epsilon)^{(j)} \ = \ (1-\epsilon) Z^{(j)} + \epsilon X^{(j)}
\end{align}
with $H'(0) = H(0)$ and $H'(1) = H(1)$. It is easy to see that the energy gap of the above Hamiltonian $H'(\epsilon)$ is independent of $\epsilon$, and remain finite. Also, since $H(0)$ can be obtained by rotating the axis of each qubit in $H(1)$, we may consider two Hamiltonians to be in the same quantum phase.

\textbf{Classification of quantum phases through frustration-free Hamiltonians:}
As we have seen in the above examples, the existence of a QPT depends on paths of parameterized Hamiltonians which connect two frustration-free Hamiltonians. A standard way to classify quantum phases is to consider two Hamiltonians belonging to the same quantum phase when they can be transformed each other without closing the energy gap through some path of a parameterized Hamiltonian~\cite{Hastings05, Bravyi10b, Chen10}. Here, following this spirit, we shall state the criteria to classify quantum phases arising in frustration-free Hamiltonians as follows (Fig.~\ref{fig_phase}).
\begin{itemize}
\item Two frustration-free Hamiltonians $H_{A}$ and $H_{B}$ belong to different quantum phases if and only if there is no parameterized Hamiltonian $H(\epsilon)$ which connects $H_{A}$ and $H_{B}$ without closing an energy gap or changing the number of ground states.
\end{itemize}
Here, by ``parameterized Hamiltonians'', we mean the following:
\begin{itemize}
\item Parameterized Hamiltonians $H(\epsilon)$ consist only of local terms which change continuously with some external parameter $\epsilon$, and amplitudes (norms) of local terms do not depend on the system size.
\end{itemize}
By ``without closing an energy gap'', we mean that
\begin{itemize}
\item The energy gap between degenerate ground states and a first excited state is finite even at the thermodynamic limit.
\end{itemize}
Thus, two different quantum phases are always separated by a QPT regardless of choices of parameterized Hamiltonians $H(\epsilon)$.

Here, we make some comments on the validity of the above classification of quantum phases. While the vanishing energy gap may be a signature of QPTs, the original definition of a QPT is a non-analytic change of the ground state properties. In fact, there are examples of parameterized Hamiltonians whose ground state energy changes continuously while the energy gap vanishes~\cite{Wolf06}. Thus, the connection between the energy gap and the non-analyticity in ground states is not completely established. However, we use the vanishing energy gap as a criteria for classifying quantum phases for simplicity of discussion in this paper. In this sense, the classification presented above relies on the assumption that QPTs are triggered by the vanishing energy gap or the change of the number of degenerate ground states~\cite{Sachdev_Text}. 


\begin{figure}[htb!]
\centering
\includegraphics[width=0.30\linewidth]{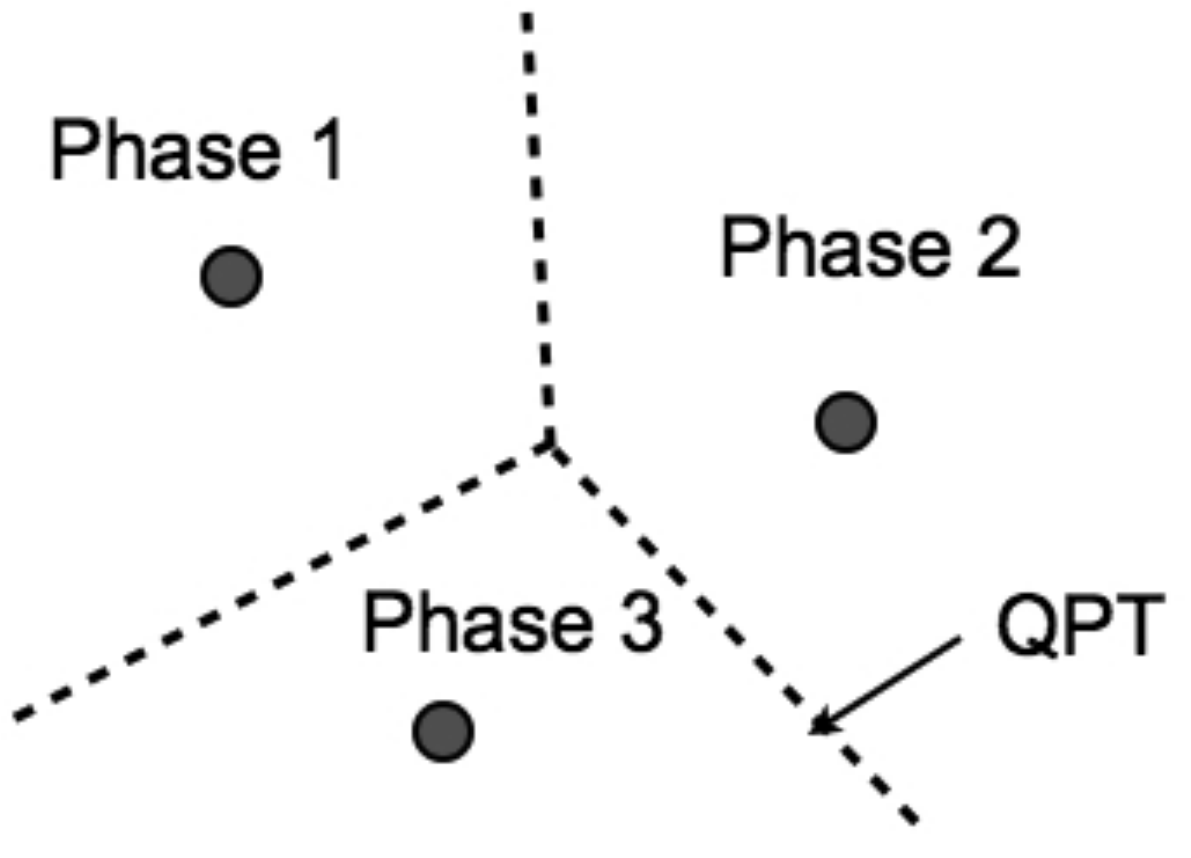}
\caption{A classification of quantum phases. Black dots represent frustration free Hamiltonians in corresponding quantum phases. Frustration-free Hamiltonians belonging to different quantum phases are always separated by a QPT.
} 
\label{fig_phase}
\end{figure}

\textbf{RG transformations and scale invariance:}
QPTs in parameterized Hamiltonians have been analyzed commonly through RG transformations~\cite{Fisher98}, which are certain scale transformations. Here, we review the basic idea of RG transformations and discuss the relation between RG transformations and frustration-free Hamiltonians in terms of scale invariance. Note that while we shall describe a general idea of RG transformations, we do not give any specific procedure of RG transformations since detailed reviews on various procedures of RG transformations are beyond the scope of this paper.

RG transformation is a way to capture only the non-local correlations from many-body systems by analyzing how the system properties changes under scale transformations. RG transformations usually consist of two elements, called coarse-graining and rescaling. Coarse-graining is a process to group several spins into a single particle with a larger Hilbert space. In this coarse-graining, a microscopic structure of original spins inside a larger spin (coarse-grained spin) is completely lost and one ends up with a system consisting only of larger spins. Rescaling is a process to replace a coarse-grained larger spin with a original small spin by eliminating some spin degrees of freedom inside each coarse-grained spin, which is interacting weakly with other coarse-grained spins. By repeating scale transformations consisting of coarse-graining and rescaling, only the non-local correlations extending all over the lattice may survive.

When quantum phases are classified through RG transformations, the key idea is the use of Hamiltonians which are invariant under scale transformations, called \emph{fixed point Hamiltonians}. By applying RG transformations to some parameterized Hamiltonian $H(\epsilon)$, one can systematically obtain Hamiltonians which are invariant under RG transformations. Since fixed point Hamiltonians are invariant under RG transformations, they do not have any length scale and can capture long-range physical properties which may survive even at the thermodynamic limit where the system size goes to the infinity. Then, fixed-point Hamiltonians help us to capture only the scale invariant properties of corresponding quantum phases.

Now, since fixed point Hamiltonians are obtained by washing out short-range correlations, it is known that fixed point Hamiltonians in gapped quantum phases can be represented as frustration-free Hamiltonians where each term can be minimized locally without considering neighboring terms. Then, one may take a slight liberty and say that a study of all the possible frustration-free Hamiltonians is almost equivalent to a study of all the possible fixed point Hamiltonians which may appear as a result of RG transformations. In fact, one of the physical constraints (scale symmetries) which we shall impose on frustration-free Hamiltonians comes from the observation that fixed point Hamiltonians are free from any length scale (see discussions in Section~\ref{sec:model2}).  

\textbf{Frustration-free Hamiltonians:}
Since our goal is to study quantum phases arising in frustration-free Hamiltonians, we now give their formal definition. 

When a Hamiltonian is said to be frustration-free, a ground state can be obtained by minimizing each local term in the Hamiltonian simultaneously. A useful feature of frustration-free Hamiltonians is that they can be represented as summations of projection operators. Consider a class of Hamiltonians which can be represented as summations of projectors:
\begin{align}
H = - \sum_{j} \hat{P}_{j} \qquad \mbox{for all} \ j,
\end{align}
where $\hat{P}_{j}$ are projectors. Such Hamiltonians are called frustration-free when projectors commute with each other ($[\hat{P}_{j}, \hat{P}_{j'}]=0$) and whose ground states satisfy 
\begin{align}
\hat{P}_{j}|\psi\rangle = |\psi\rangle.
\end{align}

A frustration-free Hamiltonian needs not to be a summation of projections. However, for any Hamiltonian whose ground states minimize each local term independently, one can represent the Hamiltonian as a summation of commuting projectors. For example, one may obtain $H(0)$ and $H(1)$ appeared in the discussions on the Ising model in a transverse field by setting $\hat{P}_{j} = \frac{1}{2}(I+Z^{(j)}Z^{(j+1)})$ and $\hat{P}_{j} = \frac{1}{2}(I+X^{(j)})$ up to some constant correction.

\subsection{Stabilizer codes and logical operators}\label{sec:review2}

While frustration-free Hamiltonians may be used as representatives of quantum phases, their analyses are generally difficult both analytically and computationally~\cite{Bravyi06}. Fortunately, for stabilizer Hamiltonians which are quantum error-correcting codes realized by a certain subclass of frustration-free Hamiltonians, there exist various theoretical tools to analyze their ground state properties. Of particular importance are certain operators, called \emph{logical operators}, which are central in studying coding properties of stabilizer codes. While properties of logical operators are discussed mainly in coding theoretical contexts, they become particularly useful in characterizing quantum phases arising in stabilizer Hamiltonians as we shall see in the main discussion of this paper. Here, we give a quick review of stabilizer codes and logical operators~\cite{Gottesman96}. We also review basic theoretical tools to analyze properties of entanglement through logical operators~\cite{Beni10}. Some notations concerning stabilizer Hamiltonians are also introduced here. At the end of this subsection, an example of stabilizer Hamiltonians is presented. 

\textbf{Stabilizer group:} 
A stabilizer code can be characterized by an Abelian subgroup of Pauli operators. Consider an $N$ qubit system and the Pauli operator group
\begin{align}
\mathcal{P}\ =\ \langle iI, X^{(1)},Z^{(1)},\dots , X^{(N)}, Z^{(N)} \rangle
\end{align}
which is generated by Pauli operators $X^{(j)}$ and $Z^{(j)}$ acting on single qubits labeled by $j$. A stabilizer code is defined with the \emph{stabilizer group}
\begin{align}
\mathcal{S}\ =\ \langle  S^{(1)}, S^{(2)} , \cdots , S^{(N-k)}  \rangle \ \subset \ \mathcal{P},
\end{align}
which is a self-adjoint Abelian subgroup of the Pauli operator group $\mathcal{P}$ without containing $-I$ or $iI$ where $k \geq 0$. $S^{(j)}$ represent independent generators for $\mathcal{S}$. Elements in the stabilizer group $\mathcal{S}$ are called \emph{stabilizers}. Qubits are encoded in a subspace $V_{\mathcal{S}}$ spanned by states $|\psi\rangle$ which satisfy
\begin{align}
S^{(j)}|\psi\rangle\ =\ |\psi\rangle \qquad \mbox{for all}\ j
\end{align}
where $V_{\mathcal{S}}$ is called the \emph{codeword space} of the stabilizer code. In total, $k$ qubits can be encoded in $V_{\mathcal{S}}$, and encoded qubits are called \emph{logical qubits}.

\textbf{System Hamiltonian:}
In the stabilizer formalism, there exists a system Hamiltonian which can support the encoded states as degenerate ground states with a finite energy gap (Fig.~\ref{fig_stabilizer}(a)).
If a Hamiltonian
\begin{align}
H\ =\ - \sum_{j} S^{(j)} \label{eq:stabilizer_Hamiltonian}
\end{align}
satisfies $\langle \{S^{(j)}, \forall j \} \rangle = \mathcal{S}$, the ground state space of the Hamiltonian is the same as the codeword space $V_{\mathcal{S}}$ since the energy of the system can be minimized for states satisfying $S^{(j)}|\psi\rangle = |\psi\rangle$. There are $2^{k}$ degenerate ground states where $k$ logical qubits can be encoded. We call Hamiltonians in Eq.~(\ref{eq:stabilizer_Hamiltonian}) \emph{stabilizer Hamiltonians}.

One may notice that stabilizer Hamiltonians are frustration-free up to some constant correction by setting $\hat{P}_{j} = \frac{1}{2}(I + S^{(j)})$. Then, a ground state of a stabilizer Hamiltonian satisfy $\hat{P}_{j}|\psi\rangle = |\psi\rangle$ since $S^{(j)}|\psi\rangle = |\psi\rangle$, and thus, stabilizer Hamiltonians are frustration-free. Note that frustration-free Hamiltonians appeared in examples of QPTs in Section~\ref{sec:review1} are stabilizer Hamiltonians. 

\begin{figure}[htb!]
\centering
\includegraphics[width=0.80\linewidth]{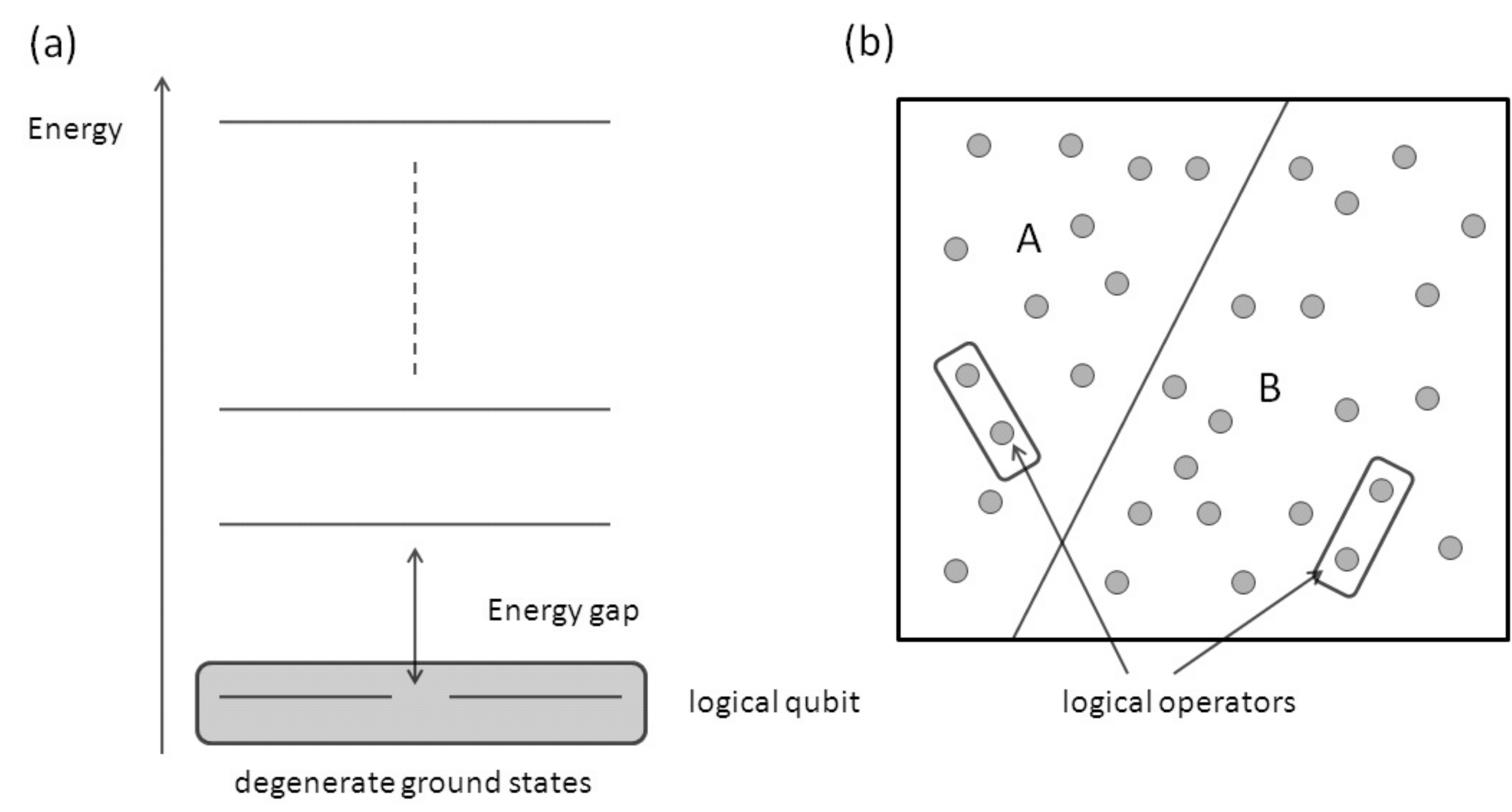}
\caption{(a) The gapped energy spectrum of a stabilizer Hamiltonian with degenerate ground states. (b) A bi-partition of a system into two complementary subsets of qubits $A$ and $B = \bar{A}$. Small dots represent qubits, and rectangles represent logical operators.} 
\label{fig_stabilizer}
\end{figure}

\textbf{Logical operators:}
\emph{Logical operators} are Pauli operators $\ell \in \mathcal{P}$ which satisfy
\begin{align}
[\ell , S^{(j)}]\ =\ 0\ \hspace{1ex} \mbox{for all}\quad j \qquad \mbox{and} \qquad \ell \ \not\in\ \langle iI, \mathcal{S} \rangle.
\end{align}
Two logical operators $\ell$ and $\ell'$ are said to be \emph{equivalent} if and only if $\ell$ and $\ell'$ act in a similar way inside the ground state space:
\begin{align}
\ell\ \sim\ \ell' \qquad &\Leftrightarrow \qquad \ell|\psi\rangle \ = \ \ell' |\psi\rangle \quad \mbox{for all} \quad |\psi\rangle \ \in \ V_{\mathcal{S}}\\
&\Leftrightarrow \qquad  \ell \ell'\ \in\ \langle \mathcal{S} \rangle
\end{align}
where ``$\sim$'' represents the equivalence between logical operators $\ell$ and $\ell'$. Given a set of logical operators $\textbf{L}$, all the logical operators inside $\textbf{L}$ are said to be \emph{independent} when
\begin{align}
\langle \{  \forall l \ \in \ \textbf{L} \} \rangle\ \cap\ \mathcal{S} \ =\ \emptyset.
\end{align}
Here, ``$\emptyset$'' represents a null set.

\textbf{Centralizer group:}
Logical operators can be found in the centralizer group which consists of all the Pauli operators commuting with stabilizers:
\begin{align}
\mathcal{C}\ =\ \Big\langle\ \Big\{\ U \in \mathcal{P}\ :\ [U,S^{(j)}] =0,\hspace{1ex} \mbox{for all}\ j\ \Big\}\ \Big\rangle.
\end{align}
The centralizer group can be represented in a \emph{canonical form}~\cite{Beni10}:
\begin{align}
\mathcal{C}\ =\ \left\langle
\begin{array}{cccccc}
 \ell_{1}, & \cdots , &  \ell_{k} , & S^{(1)} ,    &     \cdots , & S^{(N-k)}         \\
 r_{1}, & \cdots , &  r_{k} , &   &   &   
\end{array}
 \right\rangle 
\end{align}
where each operator represents independent generators for $\mathcal{C}$. Pairs of generators $\ell_{p}$ and $r_{p}$ anti-commute with each other while any other pair of generators commute with each other. For example, $\{\ell_{p},r_{p}\}=0$ while $[\ell_{p},r_{q}]=0$ for $p \not =q$. Pairs of anti-commuting generators $\ell_{p}$ and $r_{p}$ ($p=1, \cdots ,k$) are independent logical operators in a stabilizer code.

\textbf{Canonical set of logical operators:}
A set of logical operators with the following commutation relations is called a \emph{canonical set of logical operators}:
\begin{align} \label{eq:canonical_set}
\Pi(\mathcal{S})\ =\ \left\{
\begin{array}{ccc}
\ell_{1} ,& \cdots ,& \ell_{k} \\
r_{1} ,& \cdots ,& r_{k}
\end{array} \right\}.
\end{align}
Here, the commutation relations are represented in a way similar to a canonical representation where only the operators in the same column anti-commute with each other. Note that the choice of a canonical set of logical operators is not unique. With a canonical set of logical operators, the subspace $V_{\mathcal{S}}$ can be decomposed as a direct product of $k$ subsystems:
\begin{equation}
|\psi\rangle\ =\ \bigotimes_{p=1}^{k} \Big( \alpha_{p} |\tilde{0} \rangle_{p} + \beta_{p} |\tilde{1} \rangle_{p} \Big) 
\end{equation}
where $\ell_{p}$ and $r_{p}$ act non-trivially only on $|\tilde{0} \rangle_{p}$ and $|\tilde{1} \rangle_{p}$. One can choose the basis $|\tilde{0}\rangle_{p}$ and $|\tilde{a}\rangle_{p}$ such that $\ell_{p}$ and $r_{p}$ act like Pauli operators applied to a logical qubit represented by $|\tilde{0} \rangle_{p}$ and $|\tilde{1} \rangle_{p}$:
\begin{align}
\ell_{p}|\tilde{0}\rangle_{p} \ = \ |\tilde{1}\rangle_{p}, \quad \ell_{p}|\tilde{1}\rangle_{p} \ = \ |\tilde{0}\rangle_{p}, \quad  r_{p}|\tilde{0}\rangle_{p} \ = \ |\tilde{0}\rangle_{p}, \quad r_{p}|\tilde{1}\rangle_{p} \ = \ - |\tilde{1}\rangle_{p}. 
\end{align}

\textbf{Number of generators:}
The number of independent generators for a group of Pauli operators $\mathcal{O} \in \mathcal{P}$ is denoted as $G(\mathcal{O})$. Note that $G(\mathcal{O})$ does not count trivial generators such as $I$ and $iI$. Here, we give the numbers of generators for groups of Pauli operators which are central in discussions on stabilizer Hamiltonians:
\begin{align}
G(\mathcal{P}) \ = \ 2N, \quad G(\mathcal{S}) \ = \ N-k, \quad  G(\mathcal{C}) \ = \ N + k. 
\end{align}

\textbf{Code distance:}
The code distance measures the robustness of the stabilizer code:
\begin{align}
d\ =\ \min(w(U)) \qquad \mbox{where} \quad U\ \in\ \mathcal{C}\quad \mbox{and} \quad U\ \not\in\ \langle iI, \mathcal{S} \rangle.
\end{align}
Here, $w(U)$ is the number of non-trivial Pauli operators in $U$. The code distance corresponds to a minimal number of single Pauli errors necessary to destroy a encoded qubit. Roughly speaking, a large code distance $d$ for fixed $N$ means that the code can encode qubits securely, and ground states of a stabilizer Hamiltonian are highly entangled. 

\textbf{Projections:}
Properties of entanglement can be studied by considering a bi-partition of the entire system into two complementary subset of qubits $A$ and $B =\bar{A}$. In discussing the bi-partite entanglement, projections of Pauli operators are often used. The projection of a Pauli operator $U \in \mathcal{P}$ onto a subset of qubits $A$ is denoted as $U|_{A}$. This keeps only the non-trivial Pauli operators which are inside $A$ and truncates Pauli operators acting outside the subset $A$. The restriction of a group of Pauli operators $\mathcal{O}$ into some subset of qubits $A$ is defined as
\begin{align}\label{eq:restriction}
\mathcal{O}_{A}\ =\ \Big\langle\ \Big\{ U \in \mathcal{O}\ :\ U|_{\bar{A}} = I\ \Big\}\ \Big\rangle.
\end{align}
Here, $\bar{A}$ represents a complement of $A$. $\mathcal{O}_{A}$ contains all the operators in $\mathcal{O}$ which are defined inside $A$. In this paper, we use the restrictions of the stabilizer group $\mathcal{S}$, the centralizer group $\mathcal{C}$ and the Pauli operator group $\mathcal{P}$, which are denoted as $\mathcal{S}_{A}$, $\mathcal{C}_{A}$ and $\mathcal{P}_{A}$ respectively. 

\textbf{Bi-partition of the system:}
For a stabilizer code in a bi-partition into two complementary subsets of qubits (Fig.~\ref{fig_stabilizer}(b)), the following theorem holds~\cite{Beni10}. 

\begin{theorem}[Bi-partitioning theorem of logical operators]\label{theorem_partition}
For a stabilizer code with $k$ logical qubits, let the number of independent logical operators defined inside a subset of qubits $A$ be $g_{A}$. Then, for two complementary subsets of qubits $A$ and $B=\bar{A}$, the following formula holds:
\begin{align}
g_{A}+g_{B}\ =\ 2k.
\end{align}
\end{theorem}

Note that $A$ and $B$ may support the same logical operators if they have two equivalent representations supported inside $A$ and $B$ respectively. Then, $g_{A}$ and $g_{B}$ may count the same logical operators twice. This formula becomes particularly useful in determining sizes and geometric shapes of logical operators as we shall see in the main discussions of this paper. 

\textbf{An example of a stabilizer code:}
Here, we present an example of stabilizer codes to give some intuitions on definitions of logical operators and related groups of Pauli operators. The example we consider is called a five qubit code~\cite{Bennett96, Laflamme96}. Consider a system of five qubits which is governed by the following Hamiltonian:
\begin{align}
H \ = \ - \sum_{j=1}^{5} S^{(j)}, \qquad S^{(j)} \ = \ X^{(j)}Y^{(j+1)}Y^{(j+2)}X^{(j+3)}
\end{align}
where $j \in \mathbb{Z}_{5}$. Note that $S^{(j)}$ commute with each other, and this Hamiltonian is a stabilizer Hamiltonian. The stabilizer group is
\begin{align}
\mathcal{S} \ = \ \langle S^{(1)}, S^{(2)},S^{(3)},S^{(4)} \rangle
\end{align}
since $S^{(5)}$ is not independent from other stabilizers. This may be seen from the following equation: 
\begin{align}
S^{(1)}\ \times \ S^{(2)}\ \times \ S^{(3)}\ \times \ S^{(4)}\ \times \ S^{(5)}\ = \ I.
\end{align}
This stabilizer Hamiltonian has two degenerate ground states with $k = 1$ since $G(\mathcal{S})=4$, and logical operators are
\begin{align}
\ell \ = \ X^{(1)}Z^{(2)}X^{(3)}, \qquad r \ = \ Z^{(1)}Y^{(2)}Z^{(3)} 
\end{align}
where $\{ \ell , r \}=0$. Note that these are examples of logical operators, and there are many equivalent representations for logical operators. One may see that logical operators $\ell$ and $r$ commute with all the stabilizers $S^{(j)}$ through direct computations. The centralize group can be represented in the following way:
\begin{align}
\mathcal{C}\ =\ \left\langle
\begin{array}{ccccc}
 \ell, &  S^{(1)} , & S^{(2)} ,& S^{(3)}, & S^{(4)}         \\
 r, &  &   &   &      
\end{array}
 \right\rangle 
\end{align}
Here, we denote the subset of qubits which consists of the first, second and third qubits as $A$. Then, we have $g_{A} = 2$, $g_{B} = 0$ and $g_{A} + g_{B} =  2$
since logical operators $\ell$ and $r$ can be defined inside $A$ while there is no logical operator defined inside $B$. Since both $\ell$ and $r$ have non-trivial supports on three qubits, the code distance is $d = 3$. 

\textbf{Quantum phases and logical operators:}
Having introduced stabilizer Hamiltonians and logical operators, let us discuss the connection between quantum phases and logical operators in stabilizer Hamiltonian.

Quantum phases arising in correlated spin systems have been traditionally characterized by physical quantities, called order parameters, where different quantum phases are distinguished by different values of order parameters~\cite{Sachdev_Text}. Recently, it has been realized that entanglement, a non-local manifestation of the indistinguishability of quantum states, plays an essential role in characterizing quantum phases. In particular, various entanglement measures, such as entanglement entropies, have been used as order parameters and shown to be remarkably powerful in distinguishing quantum phases~\cite{Osborne02, Vidal03, Kitaev06, Levin06, Li08, Eisert10}. With this connection in hand, we hope to find some quantities or objects which may characterize entanglement in stabilizer Hamiltonians and serve as order parameters to distinguish quantum phases.

Now, a basic idea of stabilizer codes is to encode qubits in highly entangled ground states so that small errors cannot break encoded qubits. Since the robustness of the code can be measured by the number of supports in logical operators, the existence of large logical operators for a fixed system size implies the presence of strong entanglement in ground states of stabilizer Hamiltonians~\cite{Bravyi09, Beni10}. Then, geometric sizes of logical operators may be used as indicators of entanglement arising in ground states of stabilizer Hamiltonians. In the main discussion of this paper, we shall use geometric non-localities of logical operators as \emph{order parameters} of quantum phases arising in stabilizer Hamiltonians. 

\section{The model: Stabilizer codes with Translation and Scale symmetries (STS model)}\label{sec:model}

Now let us move on to the main problem of this paper, concerning possible quantum phases and their classifications in two-dimensional spin systems on a lattice. We wish to approach this problem by analyzing and classifying quantum phases arising in stabilizer Hamiltonians.

In quantum information science, stabilizer codes have been particularly important since protecting or encoding qubits is essential in realizing various quantum information theoretical ideas such as fault-tolerant quantum computation and quantum communication~\cite{Nielsen_Chuang}. Recently, stabilizer codes have been increasing their importance in condensed matter physics too, for it has been realized that many interesting correlated spin systems may be described in the language of stabilizer codes~\cite{Kitaev97, Kitaev03, Hein04, Fattal04, Hamma05, Bacon06, Bravyi09, Beni10}. The connection between correlated spin systems and stabilizer codes is potentially useful and powerful since stabilizer codes are described through Pauli operators which obey an operator algebra based on a finite group, and there is a possibility to analyze their properties analytically, or at least in a computationally efficient way~\cite{Fattal04, Beni10}. 

However, there still remains a huge gap between realistic correlated spin systems commonly discussed in condensed matter physics community and stabilizer codes commonly discussed in quantum information science community. For example, most stabilizer codes, discovered from quantum coding theoretical motivations, encode qubits in ground states of highly non-local Hamiltonians. However, in real physical systems, Hamiltonians must have only local interaction terms. While there are a few pioneering works which have analyzed stabilizer Hamiltonians with local interaction terms~\cite{Kay08, Bravyi09, Bravyi10}, physically realizable Hamiltonians must have some physical symmetries such as translation symmetries too. 

Stabilizer codes supported by local and translation symmetric Hamiltonians may seem to have sufficient physical realizability. However, there also exists another important physical constraint which must be satisfied by stabilizer Hamiltonians in order for them to be viewed as representatives of quantum phases. Let us recall that fixed point Hamiltonians, which are frustration-free Hamiltonians obtained after RG transformations, must have some \emph{scale invariant properties}~\cite{Fisher98}. Then, we wish to analyze stabilizer Hamiltonians whose physical properties do not depend on the system size. 

Here, we propose a model of frustration-free Hamiltonians which are supported by not only local but also physically symmetric Hamiltonians with scale invariance. In particular, our model is constrained to only the following three physical conditions:
\begin{itemize}
\item \textbf{Locality of interactions:}\ A system of qubits, defined on a two-dimensional lattice, is governed by a stabilizer Hamiltonian which consists only of local interactions.
\item \textbf{Translation symmetries:}\ The Hamiltonian is invariant under some finite translations of qubits.
\item \textbf{Scale symmetries:}\ The number of degenerate ground states of the Hamiltonian does not depend on the system size.
\end{itemize}
We call stabilizer Hamiltonians which satisfy the above three physical constraints \emph{Stabilizer codes with Translation and Scale symmetries} (\textbf{STS models}). 

The main goal of this paper is to analyze quantum phases arising in STS models. In this section, we give the definition of STS models precisely, and develop a basic analysis tool of STS models. In Section~\ref{sec:model1}, we construct stabilizer codes realized by local and translation symmetric Hamiltonians. In Section~\ref{sec:model2}, we discuss how the number of ground states changes as the system size varies, and impose scale symmetries on stabilizer Hamiltonians. In Section~\ref{sec:model3}, we give a formal definition of STS models. In Section~\ref{sec:model4}, we introduce a basic analysis tool for STS models, which we shall call the translation equivalence of logical operators. Readers who want to start with specific examples of STS models may jump to Section~\ref{sec:1D} and Section~\ref{sec:2D}, and return to this section later. 

\subsection{Translation symmetries and locality of interactions}\label{sec:model1}
 
Here, we construct stabilizer codes realized by local and translation symmetric Hamiltonians. Consider a system of qubits defined on a $D$-dimensional square lattice (hypercubic lattice) which consists of  $N = L_{1} \times \cdots \times L_{D}$ qubits where $L_{m}$ is the total number of qubits in the $\hat{m}$ direction for $m = 1 ,\cdots ,D$. We particularly consider stabilizer codes realized by system Hamiltonians with \emph{locally defined interaction terms} $S_{j}$:
\begin{align}
H = - \sum_{j} S_{j}
\end{align}
where each of $S_{j}$ is defined inside some geometrically localized regions. The meaning of ``locally defined'' will be clarified soon. The stabilizer group $\mathcal{S}$ is generated from each of the interaction terms $S_{j}$: $\mathcal{S}= \langle\{S_{j}, \forall j \}\rangle$. 

In addition, we assume that the system Hamiltonian possesses \emph{translation symmetries}:
\begin{align}
T_{m}(H)\ =\ H \qquad (m = 1, \cdots, D)
\end{align}
where $T_{m}$ represent translation operators which shift the positions of qubits by some positive integer $v_{m}$ in the $\hat{m}$ direction. For simplicity of discussion, we define:
\begin{align}
n_{m}\ \equiv\ L_{m}/v_{m} \qquad (m = 1, \cdots, D)
\end{align}
which are integer values so that translations by $T_{m}$ can cover the entire lattice in a commensurate way. One may consider a hypercube of $v_{1}\times \cdots \times v_{D}$ qubits as a unit cell of the system where $v_{m}$ are periodicities of the lattice in the $\hat{m}$ directions. We also assume that the system has \emph{periodic boundary conditions} for simplicity of discussion. Then, the entire space may be viewed as a square lattice (hypercubic lattice) defined on a $D$-dimensional torus:
\begin{align}
\mathbb{T}^{D}\ =\ \textbf{S}^{1} \times \cdots \times \textbf{S}^{1}
\end{align}
where $\textbf{S}^{1}$ is a circle. We may represent translation symmetries of the Hamiltonian in terms of the stabilizer group as follows:
\begin{align}
T_{m}(\mathcal{S})\ =\ \mathcal{S} \qquad (m = 1, \cdots, D).
\end{align}

When a system is invariant under some finite translations, it is more convenient to treat a hypercube of
\begin{align}
v \equiv v_{1} \times \cdots \times v_{D}
\end{align}
qubits (a unit cell) as a single entity. One may view these $v$ qubits as a single \textit{composite particle} (Fig.~\ref{fig_STS}). Then, $T_{m}$ is a translation operator with respect to composite particles, and \emph{the system Hamiltonian is invariant under unit translations of composite particles}. The integer $n_{m} = L_{m}/v_{m}$ is viewed as the number of composite particles in the $\hat{m}$ direction. 

\begin{figure}[htb!]
\centering
\includegraphics[width=0.80\linewidth]{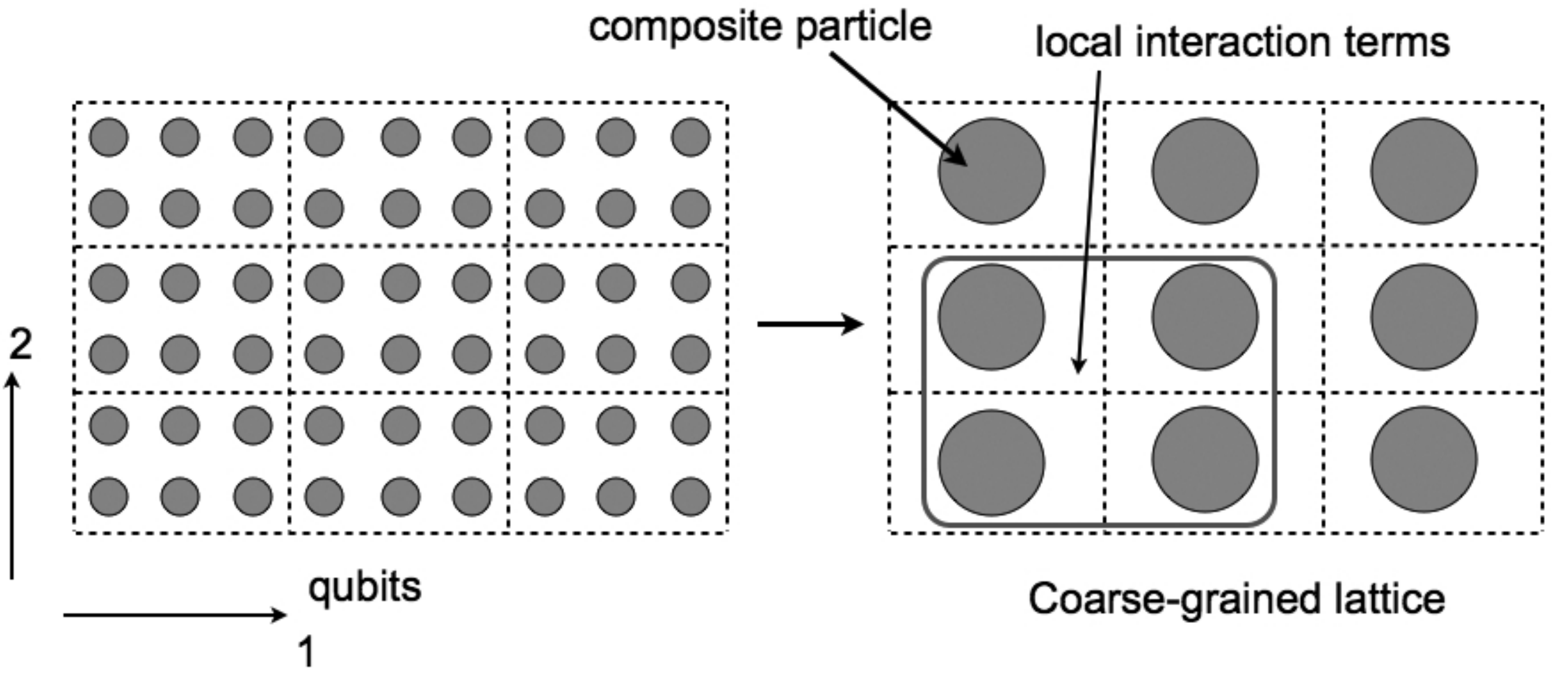}
\caption{Stabilizer codes supported by local and translation symmetric Hamiltonians. A two-dimensional example is shown. Dotted lines represent the periodicities of the Hamiltonian. In the above example, the Hamiltonian is invariant under translations of three qubits in the $\hat{1}$ direction and translations of two qubits in the $\hat{2}$ direction. 
A system of qubits is coarse-grained by considering a unit cell of $3 \times 2$ qubits as a composite particle so that the Hamiltonian is invariant under unit translations of composite particles. Interaction terms $S_{j}$ are defined locally inside a region with $2 \times 2$ composite particles.} 
\label{fig_STS}
\end{figure}

We may view this picture through composite particles as a coarse-graining of a system of qubits. When a system of qubits is coarse-grained through composite particles, a microscopic structure of qubits comprising each composite particle is lost, and one can neglect the effect of unitary transformations acting on qubits inside each composite particle. To make the difference between a picture through qubits and a picture through composite particles clear, we call a square lattice with composite particles a \emph{coarse-grained lattice}.  From now on, \emph{the entire system is considered as a square lattice of $n_{1}\times \cdots \times n_{D}$ composite particles}. Thus, we treat each of composite particles as a single particle whose inner state corresponds to a subspace $(\mathbb{C}^{2})^{\otimes v}$. 

Having introduced a coarse-grained picture through composite particles, let us formally define the meaning of ``locally defined'' interaction terms. Each of interaction terms $S_{j}$ is defined inside some geometrically localized region of composite particles. Then, we can find a finite integer $c \ll n_{m}$ such that all the interaction terms $S_{j}$ are defined inside translations of a hypercubic region with $c \times \cdots \times c$ composite particles. For simplicity of discussion, we assume that all the interaction terms $S_{j}$ can be defined inside translations of a hypercubic region with $2 \times \cdots \times 2$ composite particles by assuming $c = 2$ (Fig.~\ref{fig_STS}).\footnote{One might wonder if our constraint on the locality of interaction terms $S_{j}$ is too strict since it rules out stabilizer codes where $S_{j}$ are defined inside geometrically localized regions, but not inside regions with $2 \times \cdots \times 2$ composite particles. However, by increasing the size of composite particles, one can make all the interaction terms $S_{j}$ defined inside some regions with $2 \times \cdots \times 2$ composite particles. Of course, such newly defined composite particles may leave $n_{m} = L_{m}/v_{m}$ a non-integer value. However, it may be possible to analyze the original stabilizer code with a non-integer $n_{m}$ through the analysis of the newly defined stabilizer code with an integer $n_{m}$.}

\textbf{Translation symmetries and locality of interactions:} Here, we define stabilizer codes supported by local and translation symmetric Hamiltonian in the following way:
\begin{itemize}
\item The entire system is a square lattice (hypercubic lattice) of composite particles which consist of $v$ qubits. 
\item The system Hamiltonian $H = - \sum S_{j}$ is invariant under unit translations of composite particles: $T_{m}(H) = H$ for $m = 1,\cdots , D$.
\item Interaction terms $S_{j}$ are defined locally inside a hypercube of $2 \times \cdots \times 2$ composite particles.
\end{itemize}

Though we have constructed stabilizer Hamiltonians on square lattices, our discussion in this paper can be readily applied to stabilizer Hamiltonians defined on any general lattices. 

\subsection{Scale symmetries}\label{sec:model2}

Now, we impose another physical symmetry, which we shall call \emph{scale symmetries}. An important feature of fixed point Hamiltonians is that they do not have any length scale since they are invariant under RG transformations. In order for stabilizer Hamiltonians to be used as reference models for quantum phases, they must not have scale dependence too. So far, we have discussed the cases where the system size $\vec{n} \equiv (n_{1}, \cdots, n_{D})$ is fixed. Here, we consider changes of the number of composite particles $n_{m}$ while keeping interaction terms $S_{j}$ the same. In particular, we impose scale symmetries on translation symmetric stabilizer Hamiltonians by requiring that the number of degenerate ground states, or the number of logical qubits, does not depend on the system size $\vec{n}$. 

Let us begin by representing the stabilizer group $\mathcal{S}$ for different system sizes $\vec{n}$. For this purpose, we label each composite particle according to its position on a coarse-grained lattice. Each composite particle is denoted as $P_{\vec{r}}$ where $\vec{r}=(r_{1},\cdots, r_{D})$ represents the relative position of a composite particle with $1 \leq r_{m} \leq n_{m}$ (see Fig.~\ref{fig_label}). Now, let us denote a hypercubic region with $2 \times \cdots \times 2$ composite particles as $P_{local}$ which includes composite particles $P_{r_{1},\cdots , r_{D}}$ with $r_{m} = 1$ or $2$ (Fig.~\ref{fig_label}). We define the group of stabilizers which is generated from all the interaction terms $S_{j}$ defined inside $P_{local}$ as $\pi$:
\begin{align}
\pi\ =\ \Big\langle\ \Big\{\ S_{j} : S_{j}|_{\bar{P}_{local}}\ =\ I\  \Big\}\ \Big\rangle.
\end{align}
Here, $\bar{P}_{local}$ denotes a complement of $P_{local}$. We call this group of stabilizers the \emph{local stabilizer group}. Then, the stabilizer group $\mathcal{S}$ for the system size $\vec{n}$ is generated from translations of the local stabilizer group $\pi$ as follows:
\begin{align}\label{eq:stabilizer}
\mathcal{S}\ =\ \mathcal{S}_{\vec{n}}\ \equiv\ \Big\langle\ \Big\{\ T_{1}^{x_{1}} \cdots T_{D}^{x_{D}} (\pi)\ : 
 \ x_{m} = 1 , \cdots, n_{m} \quad \mbox{for all} \ m \ \in \ \mathbb{Z}_{D} \ \Big\}\ \Big\rangle.
\end{align}
Here, $T_{m}^{x_{m}}$ represents a translation of $x_{m}$ composite particles in the $\hat{m}$ direction, and $T_{m}(\pi)$ represents translations of elements in the local stabilizer group $\pi$.

\begin{figure}[htb!]
\centering
\includegraphics[width=0.35\linewidth]{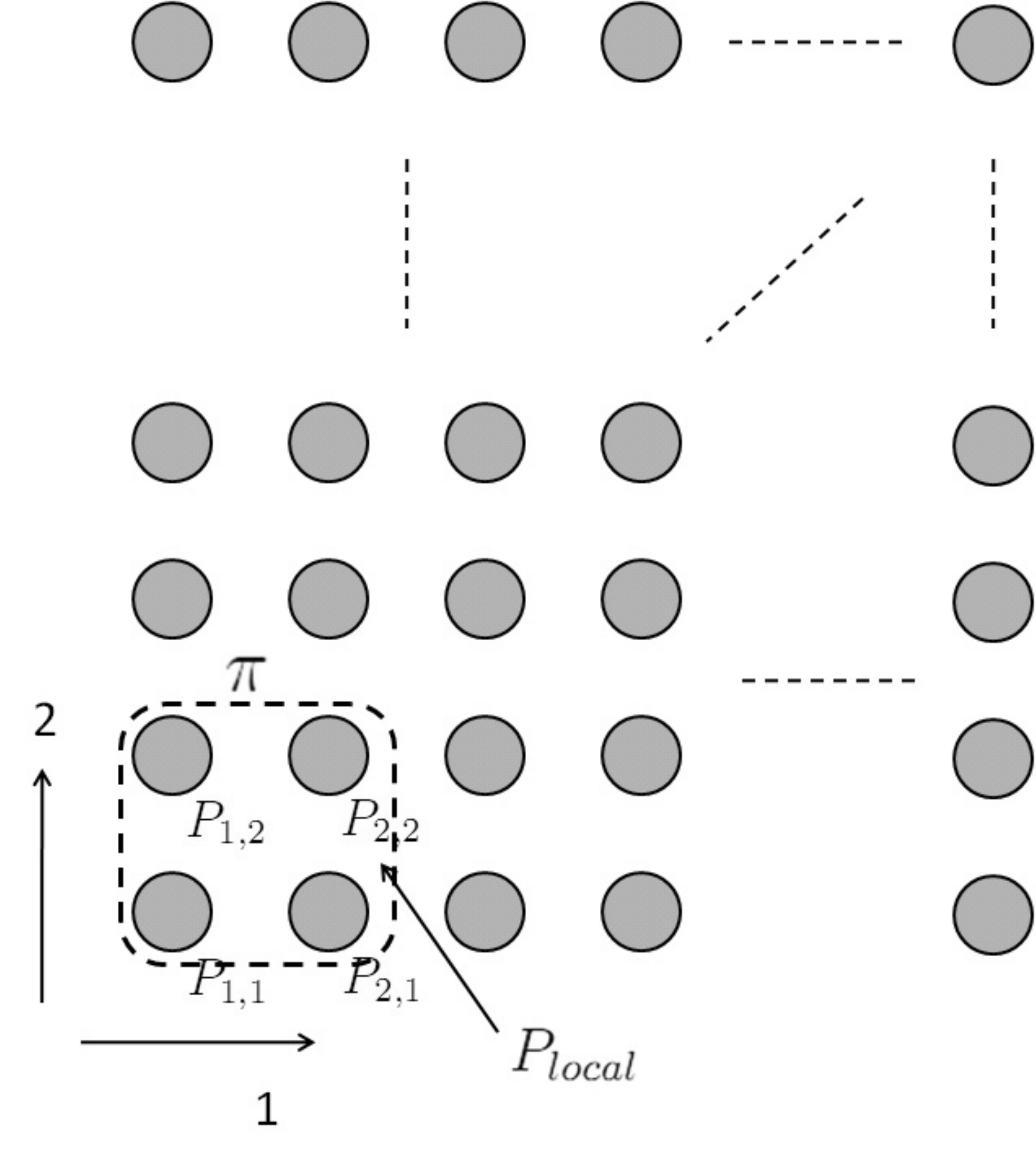}
\caption{Labeling of composite particles and the local stabilizer group $\pi$ defined inside $P_{local}$.
} 
\label{fig_label}
\end{figure}

We may choose interaction terms $S_{j}$ defined inside $P_{local}$ so that they generate the local stabilizer group $\pi$. For example, let us pick up generators for $\pi = \langle S_{1}^{local}, S_{2}^{local}, \cdots \rangle$. Then, the stabilizer code can be supported by a Hamiltonian which consists of translations of $S_{j}^{local}$:
\begin{align}
H\ =\ -  \sum_{x_{1},\cdots,x_{D}}T_{1}^{x_{1}}\cdots T_{1}^{x_{D}}(\sum_{j} S_{j}^{local}).
\end{align}
We note that some of interaction terms $S_{j}^{local}$ may be defined inside regions smaller than a region $P_{local}$.

Now, we denote the number of logical qubits $k$ for a stabilizer code for the system size $\vec{n}$ as $k_{\vec{n}}$. More precisely, if we denote the number of independent generators for $\mathcal{S}_{\vec{n}}$ as $G(\mathcal{S}_{\vec{n}})$, we have $k_{\vec{n}} \equiv N - G(\mathcal{S}_{\vec{n}})$ where $N = v \prod_{m}n_{m}$. Here, $v$ is the number of qubits inside each composite particle. Then, a stabilizer code has \emph{scale symmetries} when the number of logical qubits $k_{\vec{n}}$ does not depend on the system size $\vec{n}$:
\begin{align}\label{eq:scale_symmetries}
k_{\vec{n}}\ =\ k \qquad \mbox{for all}\ \vec{n}.
\end{align}
Thus, the number of logical qubits is always $k$ regardless of the system size $\vec{n}$. Here, we emphasize that in a system with scale symmetries, the number of logical qubits $k$ remains constant under not only global scale transformations: $\vec{n} \rightarrow c \vec{n}$ where $c$ is some positive integer, but also any changes of $n_{m}$.

\begin{figure}[htb!]
\centering
\includegraphics[width=0.8\linewidth]{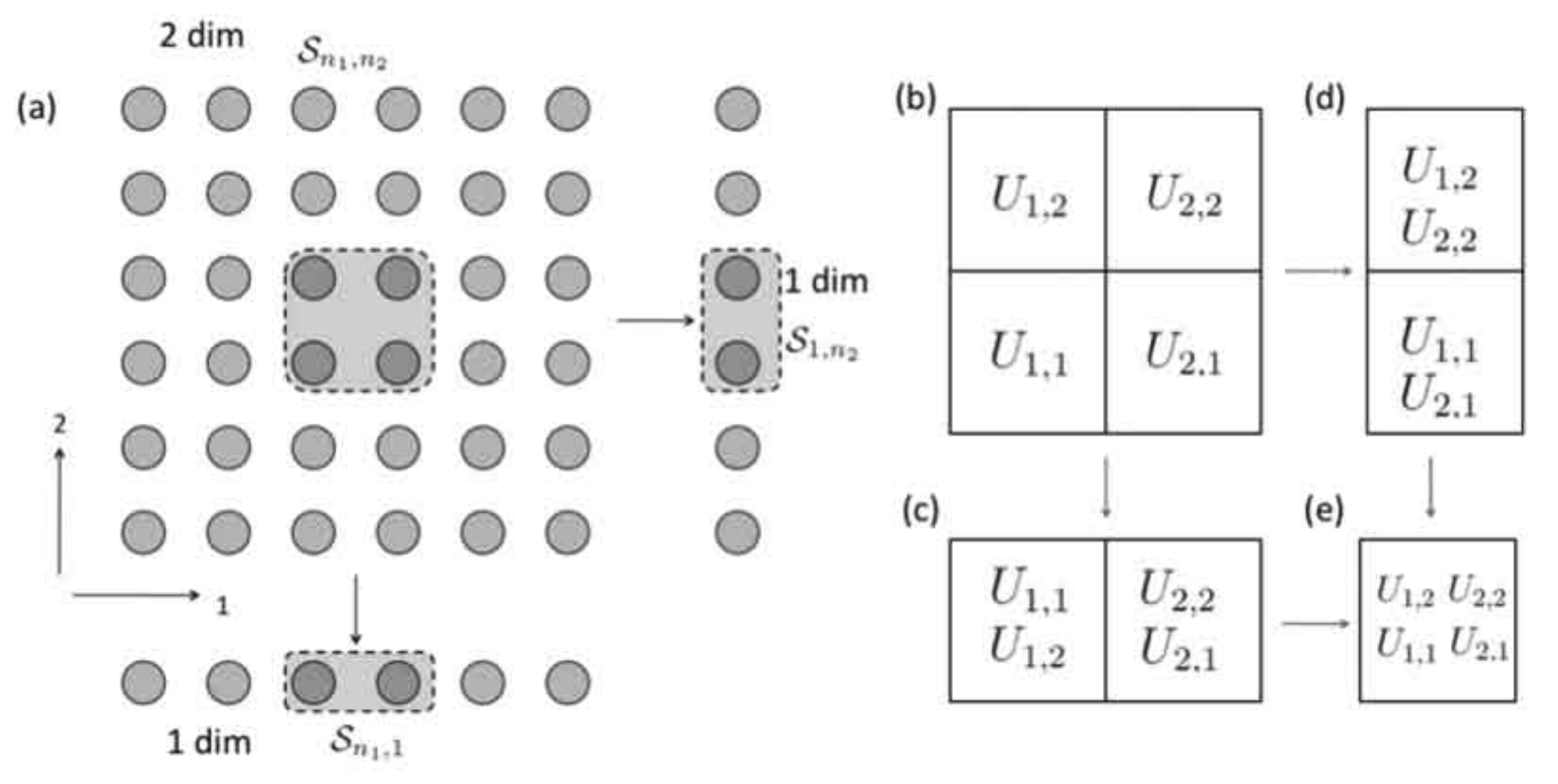}
\caption{(a) Effectively lowered dimensions. An example of a two-dimensional system with $n_{1} \times n_{2}$ composite particles is shown. For the cases where $n_{1}=1$ or $n_{2}=1$, the spatial dimension is effectively lowered to one. Interaction terms are defined inside some regions with $1 \times 2$ and $2\times 1$ composite particles respectively. (b)-(e) Examples of folded stabilizers in two dimensions. 
(b) $n_{1},n_{2}\not=1$ (a original stabilizer). (c)  $n_{1}\not=1$ and $n_{2}=1$. (d) $n_{1}=1$ and $n_{2}\not=1$. (e) $n_{1}=n_{2}=1$.
} 
\label{fig_dimension_reduction}
\end{figure}

However, one may notice that the above constraint on scale symmetries is not well-defined for the case with $n_{m}=1$ for some $m$ since the local stabilizer group $\pi$ is defined inside $P_{local}$ which is a hypercubic region with $2^{D} = 2 \times \cdots \times 2$ composite particles (Fig.~\ref{fig_dimension_reduction}(a)). This difficulty could prevent us from studying stabilizer codes at small $\vec{n}$ to help us to study the properties of stabilizer codes at large $\vec{n}$. However, due to the commuting properties of interaction terms $S_{j}^{local}$ in stabilizer Hamiltonians, one can define the local stabilizer group $\pi$ properly even for the cases where $n_{m}=1$ for some $m$. In order to show how this problem can be fixed, we consider a one-dimensional case ($D=1$) first. Let us recall that $P_{local}$ consists of two composite particles $P_{1}$ and $P_{2}$ in one dimension. We can represent interaction terms $S_{j}^{local} \in \pi$ as
\begin{align}
S_{j}^{local}\ =\ U_{j} T_{1} (V_{j})   \qquad (n_{1}\ \not=\ 1)
\end{align}
where $U_{j}$ and $V_{j}$ are Pauli operators acting on a composite particle $P_{1}$. Then, $U_{j}$ acts on a composite particle $P_{1}$ while $T_{1} (V_{j})$ acts on a composite particle $P_{2}$. Here, since the translation of $S_{j}^{local}$ must commute with $S_{j'}^{local}$, we have $[U_{j},V_{j'}] = 0$ for all $j$ and $j'$. Also, we have either $[U_{j},U_{j'}]=[V_{j},V_{j'}]=0$ or $\{U_{j},U_{j'}\}=\{V_{j},V_{j'}\}=0$ since $[S_{j}^{local},S_{j'}^{local}]=0$. Now, for the case with $n_{1}=1$, we define $S_{j}^{local}$ in the following way:
\begin{align}
S_{j}^{local}\ \equiv\ U_{j}V_{j}    \qquad ( n_{1}\ =\ 1)
\end{align}
by \emph{folding} $S_{j}^{local}$ at the boundary. Then, one can see that these folded $S_{j}^{local}$ commute with each other through some calculations. Thus, by this definition of $S_{j}^{local}$, the local stabilizer group $\pi$ is well-defined for the case with $n_{1}=1$. 

In a way similar to this, we can define interaction terms $S_{j}^{local} \in \pi$ in higher-dimensions so that $\pi$ is well-defined for any $\vec{n}$ by folding $S_{j}^{local}$ at the boundary in the $\hat{m}$ direction where $n_{m}=1$. For example, in two dimensions, let us consider the following interaction term (see Fig.~\ref{fig_dimension_reduction}(b))
\begin{align}
S^{local}\ &=\ U_{1,1}T_{1}(U_{2,1})T_{2}(U_{1,2})T_{1}T_{2}(U_{2,2}) \qquad n_{1}\ >\ 1\ \mbox{and}\ n_{2}\ >\ 1 
\end{align}
where $U_{1,1}$, $U_{1,2}$, $U_{2,1}$ and $U_{2,2}$ are Pauli operators acting on a composite particle $P_{1,1}$. Then, for the cases with $n_{m}=1$ for some $m$, the interaction term is defined as follows (see Fig.~\ref{fig_dimension_reduction}(c)(d)(e)):
\begin{align}
&S^{local}\ =\ U_{1,1}U_{2,1}T_{2}(U_{1,2}U_{2,2})\qquad &n_{1}\ =\ 1\ \mbox{and}\ n_{2}\ >\ 1& \\ 
&S^{local}\ =\ U_{1,1}U_{1,2}T_{1}(U_{2,1}U_{2,2}) \qquad &n_{1}\ >\ 1\ \mbox{and}\ n_{2}\ =\ 1& \\
&S^{local}\ =\ U_{1,1}U_{2,1}U_{1,2}U_{2,2}        \qquad \qquad &n_{1}\ =\ 1\ \mbox{and}\ n_{2}\ =\ 1& .
\end{align}

With this definition of interaction terms, one can discuss stabilizer codes for any $\vec{n}$ including the cases where $n_{m}=1$ for some $m$. As shown in Fig.~\ref{fig_dimension_reduction}(a), by setting $n_{m}=1$ for some $m$, the spatial dimension of the system can be effectively lowered. For example, a two-dimensional stabilizer code with $n_{1}=1$ can be treated as a one-dimensional system.\footnote{When the system possesses symmetries, it is sometimes possible to lower spatial dimensions of the system effectively as a result of symmetries. See~\cite{Nussinov09} for general treatments on dimensional reductions in topologically ordered systems.} In such a case, the local stabilizer group, which is originally defined inside a region with $2 \times 2$ composite particles, is now defined inside a region with $1 \times 2$ composite particles since interaction terms are folded in the $\hat{1}$ direction.\\

\textbf{Scale symmetries:} Here, we summarize constraints resulting from scale symmetries:
\begin{itemize}
\item The number of logical qubits $k$ does not depend on the system size $\vec{n}$ and remains constant:
\begin{align}
k_{\vec{n}} \ = \ k \qquad \mbox{for all} \ \vec{n}.
\end{align}
\end{itemize}

Compared with the importance of translation symmetries, the importance of scale symmetries may be less obvious. In fact, there exist many examples of stabilizer Hamiltonians without scale symmetries. However, most examples without scale symmetries are trivial, as seen in the following example:
\begin{align}
H \ = \ - \sum_{i,j} Z^{(i,j)}Z^{(i,j+1)}
\end{align}
where $Z^{(i,j)}$ acts on a qubit labeled by $(i,j)$ in a two-dimensional square lattice with periodic boundary conditions. The model is an array of one-dimensional classical ferromagnets which have interactions only in the vertical direction (the $\hat{2}$ direction). The model has $2^{L_{1}}$ ground states where $L_{1}$ is the number of qubits in the $\hat{1}$ direction. This model should be discussed as a one-dimensional system rather than a two-dimensional system since each chain of ferromagnets is completely decoupled. Thus, while we have imposed scale symmetries to limit our considerations to stabilizer Hamiltonians without length scales, scale symmetries help us to focus on physically interesting examples by excluding trivial stabilizer Hamiltonians such as the above example. 

There are also many stabilizer codes whose $k_{\vec{n}}$ remain finite ($k_{\vec{n}}\leq k$) for all $\vec{n}$, but do not remain constant. However, one can reduce such stabilizer Hamiltonians to STS models by increasing the size of composite particles so that stabilizer Hamiltonians possess scale symmetries with respect to newly defined composite particles. This restoration of scale symmetries is further discussed in~\ref{sec:translation}.

\subsection{Formal definition of STS model}\label{sec:model3}

Based on these discussions, let us formally define STS models.

\begin{definition}
A stabilizer code defined with the stabilizer group $\mathcal{S}_{\vec{n}}$ where $\vec{n}=(n_{1},\cdots , n_{D})$ is called a \textit{Stabilizer code with Translation and Scale symmetries} (STS model) if and only if there exists a local stabilizer group $\pi$ which satisfies the following conditions.
\begin{itemize}
\item All the elements in $\pi$ are defined inside a hypercubic region $P_{local}$ with $2^{D}= 2 \times \cdots \times 2$ composite particles.
\item The stabilizer group $\mathcal{S}_{\vec{n}}$ is generated from translations of the local stabilizer group $\pi$:
\begin{align}
\mathcal{S}_{\vec{n}}\ =\ \Big\langle\ \Big\{\ T_{1}^{x_{1}} \cdots T_{D}^{x_{D}} (\pi)\ : 
 \ x_{m} = 1 , \cdots, n_{m} \quad \mbox{for all} \ m\ \in\ \mathbb{Z}_{D},\Big\}\ \Big\rangle
\end{align}
\item The number of logical qubits $k_{\vec{n}}$ is independent of the system size $\vec{n}$:
\begin{align}
k_{\vec{n}}\ =\ k \qquad \mbox{for all}\ \vec{n}.
\end{align}
\end{itemize} 
\end{definition} 

\subsection{Translation equivalence of logical operators}\label{sec:model4}

We have introduced a model of frustration-free Hamiltonians, an STS model, which may cover a large class of physically realizable stabilizer Hamiltonians. Now, let us develop a basic analysis tool for analyzing quantum phases in STS models, which we shall call the translation equivalence of logical operators.

The analysis on logical operators is central in studying properties of entanglement arising in ground states of stabilizer Hamiltonians~\cite{Gottesman96, Bravyi09, Beni10}. Here,  we wish to know the effect of translation and scale symmetries on logical operators. Translation symmetries of Hamiltonians are particularly useful in analyzing properties of logical operators. For example, a translation of a logical operator $\ell$ is also a logical operator due to translation symmetries:
\begin{align}
\ell \ \in \ \textbf{L}_{\vec{n}} \ \Rightarrow \ T_{m}(\ell)\ \in \ \textbf{L}_{\vec{n}}
\end{align}
where $\textbf{L}_{\vec{n}}$ is a set of all the logical operators. This observation may help us to determine possible forms of logical operators in STS models. 

However, the relation between the original logical operator $\ell$ and the translated logical operator $T_{m}(\ell)$ is not immediately clear. In particular, while $\ell$ and $T_{m}(\ell)$ have similar forms which can be transformed each other just by translations, $\ell$ and $T_{m}(\ell)$ might be not equivalent in general. For example, consider an array of one-dimensional classical ferromagnets discussed in Section~\ref{sec:model2}. Then, we notice that unit translations of logical operators might not be equivalent to the original logical operators when logical operators are translated in the $\hat{1}$ direction. Thus, the need is to establish the relation between $\ell$ and its translation $T_{m}(\ell)$.

In this subsection, we show that $\ell$ is always equivalent to its own translation $T_{m}(\ell)$ as a result of scale symmetries in STS models. In particular, the following theorem holds.

\begin{theorem}[Translation equivalence of logical operators]\label{theorem_main}
For each and every logical operator $\ell$ in an STS model, a unit translation of $\ell$ with respect to composite particles in any direction is always equivalent to the original logical operator $\ell$:
\begin{align}
T_{m}(\ell)\ \sim\ \ell, \qquad \forall \ell\ \in\ \textbf{L}_{\vec{n}}  \qquad (m\ =\ 1, \cdots, D) 
\end{align}
where $\textbf{L}_{\vec{n}}$ is a set of all the logical operators for an STS model defined with the stabilizer group $\mathcal{S}_{\vec{n}}$.
\end{theorem}

We call this property of logical operators under translations the \textbf{translation equivalence of logical operators}. 

\begin{figure}[htb!]
\centering
\includegraphics[width=0.35\linewidth]{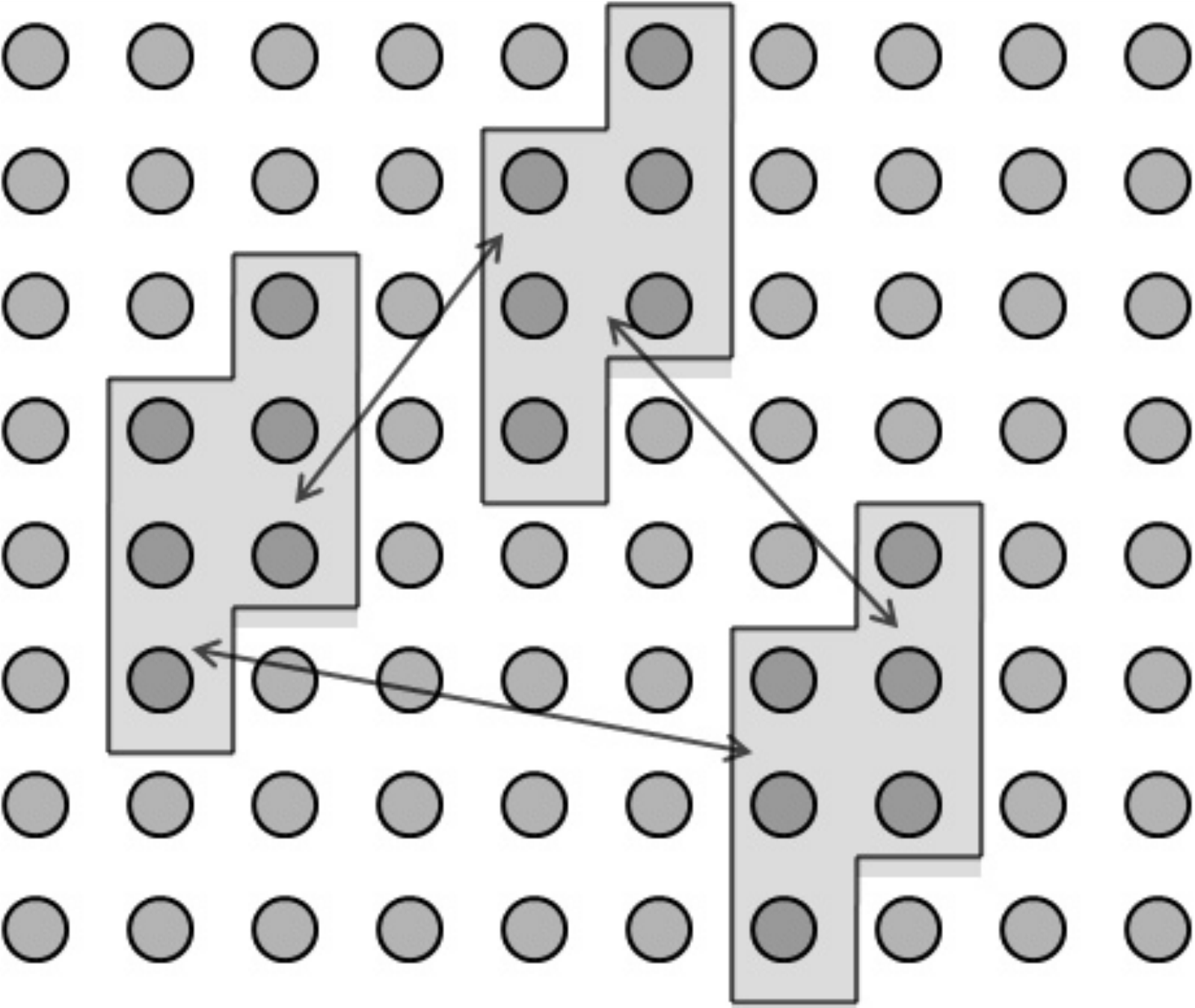}
\caption{The translation equivalence of logical operators. Logical operators remain equivalent under any translations. Rectangles surrounding composite particles represent logical operators defined inside corresponding regions. Two-sided arrows mean that two connected logical operators are equivalent.
} 
\label{fig_translation_equivalence}
\end{figure}

The proof of Theorem~\ref{theorem_main} is presented in~\ref{sec:TE}. Here, we only show the existence of finite translations which keep logical operators equivalent for some sufficiently large $\vec{n}$. In particular, we show the existence of a finite integer $a_{m}$ which satisfies $T_{m}^{a_{m}}(\ell)\sim \ell$ for all the logical operators $\ell$ for some fixed $\vec{n}$. 

Let us consider translations of logical operators in the $\hat{m}$ direction. Due to the translation symmetries of the system Hamiltonian, translations of a given logical operator $\ell$ are also logical operators. Let us label each of translated logical operators as $\ell(j) \equiv T_{m}^{j}(\ell)$:
\begin{align}
\ell(0)\ \rightarrow\ \ell(1)\ \rightarrow\  \cdots\ \rightarrow\ \ell(n_{m}-1)\ \rightarrow\ \ell(0)
\end{align}
where ``$\rightarrow$'' represents translations in the $\hat{m}$ direction. Here, due to the periodic boundary conditions, we have the same logical operator $\ell(n_{m})\equiv\ell(0)$ after translating $n_{m}$ times. In principle, $\ell(j)$ can be independent each other in the absence of scale symmetries. However, since there are at most $2^{2k}$ different logical operators in this stabilizer Hamiltonian, there always exists a finite integers $b_{m} \leq 2^{2k}$ such that $\ell(b_{m}) \sim \ell(0)$ where $n_{m}$ is a multiple of $b_{m}$. By repeating this argument for all the $2k$ independent logical operators, one can find a finite integer $a_{m}$ such that $\ell \sim T_{m}^{a_{m}}(\ell)$ for all the logical operators $\ell$. 

The above argument shows that there exists some finite integer $a_{m}$ where $T_{m}^{a_{m}}(\ell)\sim \ell$ for STS models with sufficiently large $n_{m}$. However, one can obtain a much more general and stronger result than this, as summarized in Theorem~\ref{theorem_main}. In other words, $a_{m}=1$ for any $n_{m}$ in the above argument. The proof of Theorem~\ref{theorem_main} relies on the fact that the number of logical qubits $k$, or the number of logical operators, is not only small, but also independent of the system size $\vec{n}$.

The translation equivalence of logical operators is a key in analyzing quantum phases arising in STS models. In particular, since any translations of logical operators are equivalent to the original logical operators, in characterizing the ground state properties, the positions of logical operators are not important and only the geometric shapes of logical operators become essential as we shall see in the next section. The notion of topology arises also due to the translation equivalence of logical operators. 

\section{One-dimensional STS model: quantum phases and logical operators}\label{sec:1D}

Now, we start the analysis on quantum phases in STS models. In this section, we discuss quantum phases arising in one-dimensional STS models.

Quantum phases arising in one-dimensional spin systems are relatively easy to analyze compared with higher-dimensional spin systems. In particular, there are several powerful numerical algorithms to analyze the ground state properties of parameterized Hamiltonians which connect frustration-free Hamiltonians. For example, the density matrix renormalization group (DMRG) approach, whose basic idea comes from RG transformations, provides efficient numerical algorithms to compute various physical quantities in one-dimensional systems~\cite{White92, Schollwock05}. Also, recent developments on the matrix product state formalism solidified the usefulness of DMRG approaches and expanded its applicabilities~\cite{Verstraete04, Verstraete06}. In fact, it has been proven that any gapped spin systems on one-dimensional lattices can be efficiently simulated~\cite{Hastings06}. Then, one might hope that a problem of finding and classifying quantum phases arising in one-dimensional STS models is not a difficult one.

However, a classification of quantum phases is much more challenging than an analysis on a single parameterized Hamiltonian through a numerical simulation. In particular, since the existence of a QPT depends on paths of parameterized Hamiltonians, one needs to analyze all the possible parameterized Hamiltonians connecting two STS models to see if they belong to different quantum phases or not. Thus, the need is to find order parameters to classify quantum phases in a path-independent way. 

A key idea behind RG transformations in characterizing quantum phases is the use of fixed point Hamiltonians which are invariant under RG transformations. These fixed point Hamiltonians capture only the scale invariant ground state properties of the corresponding quantum phases. Then, following the spirit of the analyses on quantum phases through RG transformations, we wish to classify quantum phases arising in STS models through some quantity or object which is scale invariant. 

A useful observation in addressing quantum phases arising in STS models is to realize that geometric shapes of logical operators do not change under scale transformations. Then, we hope that geometric shapes of logical operators may be used as ``order parameters'' to distinguish different quantum phases. In this section, we show that different quantum phases in one-dimensional STS models are completely characterized by geometric shapes of logical operators. In particular, we show that two STS models belong to the same quantum phase if and only if geometric shapes of logical operators are the same. Thus, it is shown that two STS models with the same geometric shapes of logical operators can be connected without closing the energy gap, while two STS models with different geometric shapes of logical operators are always separated by a QPT.

In Section~\ref{sec:1D1}, we start our discussion by analyzing three specific examples of one-dimensional STSs. We discuss the role of geometric shapes of logical operators and their relation to the scale invariant ground state properties such as global entanglement. In Section~\ref{sec:1D2}, we extend our analysis to arbitrary one-dimensional STS models and give possible forms of logical operators. The ground state properties of STS models are also analyzed through geometric shapes of logical operators. Finally, in Section~\ref{sec:1D3}, we discuss the relation between logical operators and quantum phases, and show that quantum phases arising in one-dimensional STS models can be completely classified through geometric shapes of logical operators.

\subsection{Role of logical operators: concrete examples}\label{sec:1D1}

In this subsection, we analyze three specific examples of one-dimensional STS models in order to discuss the role of logical operators in classifying quantum phases. In particular, we discuss how the scale invariant ground state properties can be studied through geometric shapes of logical operators in these examples. 

In Section~\ref{sec:1D11}, we begin by analyzing a classical ferromagnet, which is the simplest example of a one-dimensional STS. In Section~\ref{sec:1D12}, we discuss an example without degenerate ground states or logical operators ($k=0$) to compare physical properties in the presence and absence of logical operators. In Section~\ref{sec:1D13}, we analyze another example of a non-trivial one-dimensional STS which can be obtained by extending the five qubit code to a one-dimensional spin chain. 

\subsubsection{Classical ferromagnet as a quantum code}\label{sec:1D11}

A classical ferromagnet, the simplest model of interacting spins, can be seen as a stabilizer code. Though the model has been completely analyzed a century ago, we utilize the simplicity of the model to give a concise, but insightful demonstration of the analysis on entanglement in ground states through geometric shapes of logical operators. 

Let us consider the following Hamiltonian (Fig.~\ref{fig_1Dex1}):
\begin{align}
H \ = \ - \sum_{j} Z^{(j)} Z^{(j+1)}
\end{align}
where $Z^{(j)}$ represents the Pauli operator $Z$ acts on a $j$-th qubit. The total number of qubits is $N$ and the system has periodic boundary conditions. The ground states of a classical ferromagnet are $|0\cdots 0\rangle$ and $|1 \cdots 1\rangle$.

\begin{figure}[htb!]
\centering
\includegraphics[width=0.70\linewidth]{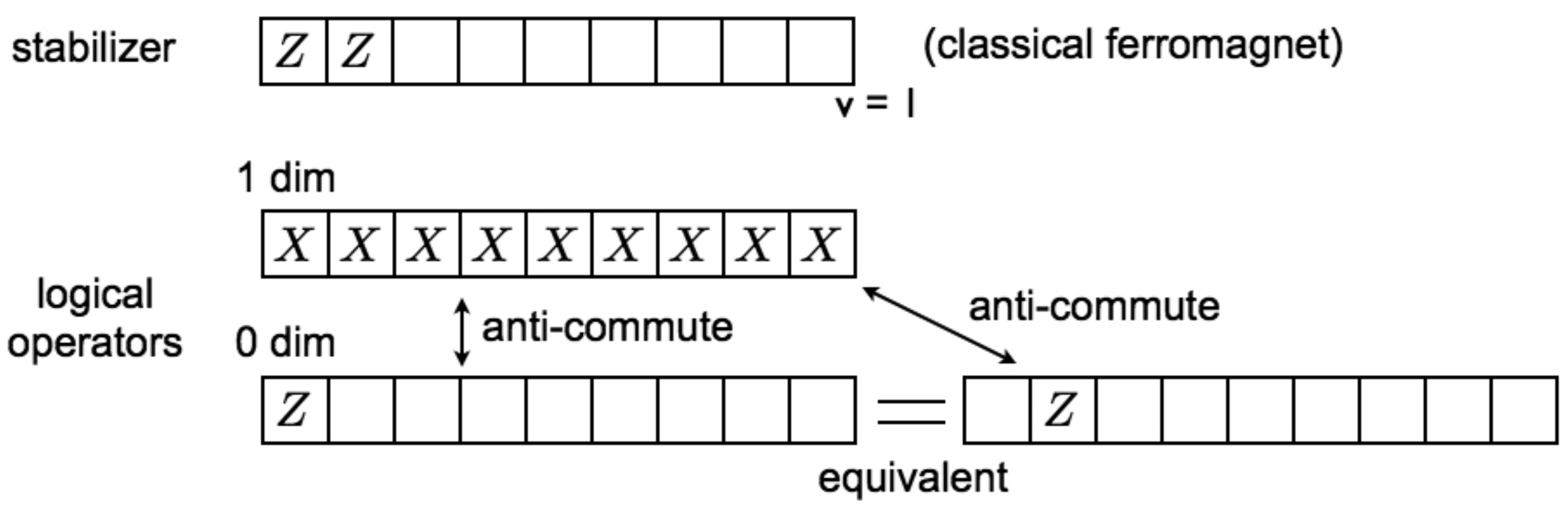}
\caption{A classical ferromagnet.
} 
\label{fig_1Dex1}
\end{figure}

The classical ferromagnet can be viewed as a stabilizer code with a translation symmetry since interaction terms $Z^{(j)}Z^{(j+1)}$ commute with each other. The stabilizer group is 
\begin{align}
\mathcal{S}_{N} \ = \ \langle Z^{(1)}Z^{(2)} , Z^{(2)}Z^{(3)}, \cdots , Z^{(N)}Z^{(1)} \rangle. 
\end{align}
This stabilizer code is an STS model since it satisfies a scale symmetry. While there are $N$ qubits and $N$ stabilizers in this system, one stabilizer is redundant since 
\begin{align}
Z^{(1)}Z^{(2)} \times Z^{(2)}Z^{(3)} \times \cdots \times Z^{(N)}Z^{(1)} \ = \ I.
\end{align}
Thus, the stabilizer group has only $N-1$ independent generators, and $k =1$ for any $N$ with $G(\mathcal{S_{N}})= N-1$. Here, $G(\mathcal{S}_{N})$ is the number of independent generators for the stabilizer group $\mathcal{S}_{N}$. 

Logical operators of a classical ferromagnet are
\begin{align}
\ell \ = \ Z^{(1)}, \qquad r \ = \ X^{(1)}X^{(2)}\cdots X^{(N)}, \qquad \{ \ell, r \} \ = \ 0.
\end{align}
One can see that both of logical operators satisfy the translation equivalence of logical operators since $T_{1}(r)=r$ and $\ell T_{1}(\ell) = Z^{(1)}Z^{(2)} \in \mathcal{S}_{N}$. According to geometric shapes of logical operators, we may call $\ell$ a \emph{zero-dimensional logical operator} and $r$ a \emph{one-dimensional logical operator} (Fig.~\ref{fig_1Dex1}). Logical operators characterize transformations between degenerate ground states:
\begin{align}
&\ell |0\cdots 0\rangle \ = \ |0\cdots 0\rangle, \quad \ell |1 \cdots 1\rangle \ = \ - |1 \cdots 1\rangle \\
&r |0\cdots 0\rangle \ = \ |1\cdots 1\rangle, \quad r |1 \cdots 1\rangle \ = \ |0 \cdots 0\rangle.
\end{align}

With the above classification of logical operators, we naturally think that a one-dimensional logical operator, non-locally defined all over the lattice, may characterize the existence of some global entanglement in ground states. Here, the ground state space is spanned by two orthogonal basis $|0\cdots 0\rangle$ and $|1 \cdots 1\rangle$. Then, in principle, at zero temperature, a classical ferromagnet can support a GHZ state: $|\mbox{GHZ}\rangle \ = \ \frac{1}{\sqrt{2}}(|0\cdots 0\rangle + |1\cdots 1\rangle)$.
This GHZ state is ``stabilized'' by the one-dimensional logical operator $r$ since $r |\mbox{GHZ}\rangle = |\mbox{GHZ}\rangle$. Now, we discuss entanglement in a GHZ state in a relation with the one-dimensional logical operator $r$. A GHZ state is a globally entangled state since it is a superposition of two globally separated ground states $|0 \cdots 0 \rangle$ and $|1 \cdots 1 \rangle  = X^{(1)}\cdots X^{(N)} |0 \cdots 0 \rangle$. One can also quantify the global entanglement of a GHZ state by computing a mutual information $E(A:B)$~\cite{Nielsen_Chuang}, which is a measure of entanglement between two subsets of qubits $A$ and $B$:
\begin{align}
E(A:B) \ = \ E(A) + E(B) - E(A\cup B)
\end{align}
where $E(A)$, $E(B)$ and $E(A\cup B)$ are entanglement entropies for regions $A$, $B$ and $A \cup B$. Then, for a GHZ state, which is stabilized by the one-dimensional logical operator $r$, we always have $E(A:B)=1$ for any pairs of disjoint regions $A$ and $B$ (Fig.~\ref{fig_1Dex1}). Thus, global entanglement arising in a GHZ state is scale invariant since $E(A:B)$ does not depend on the distance between $A$ and $B$. 

\subsubsection{Cluster state: a model without logical operators}\label{sec:1D12}

Next, let us discuss an example which does not have degenerate ground states or logical operators ($k=0$). The example we discuss is called a cluster state, possessing short-range entanglement between neighboring qubits. Its two-dimensional generalization is particularly useful as a resource to realize quantum information theoretical ideas such as measurement based quantum computation~\cite{Raussendorf03}. Though a cluster state has strong entanglement between neighboring qubits or composite particles, this short-range entanglement does not survive over the entire lattice, and is not a scale invariant property. In fact, such short-range entanglement can be removed by applying local unitary transformations on neighboring composite particles as we shall see below. 

The Hamiltonian for a one-dimensional cluster state is the following (Fig.~\ref{fig_1Dex2}):
\begin{align}
H \ = \ - \sum_{j} Z^{(j-1)} X^{(j)} Z^{(j+1)}.
\end{align}
This model is a stabilizer Hamiltonian since interaction terms $Z^{(j-1)} X^{(j)} Z^{(j+1)}$ commute with each other. The model does not have degenerate ground states since all the interaction terms $Z^{(j-1)} X^{(j)} Z^{(j+1)}$ are independent. The model has a unique ground state, called a cluster state, which satisfy the following conditions: $Z^{(j-1)} X^{(j)} Z^{(j+1)} |\psi\rangle \ = \ |\psi\rangle$.

A cluster state has short-range entanglement between neighboring qubits. While it is possible to see the existence of such entanglement by writing down a cluster state as a superposition of product states explicitly, a representation would look ugly and give us few physical insights. 

A useful property of stabilizer Hamiltonians is that entanglement entropies of ground states can be computed easily through the stabilizer group $\mathcal{S}$. Here, we begin by reviewing computations of entanglement entropies for ground states of stabilizer Hamiltonians. When a stabilizer Hamiltonian has only a single ground state, the ground state can be represented in terms of the stabilizer group as follows:
\begin{align}
|\psi\rangle \langle \psi | \ = \ \frac{1}{2^{N}} \prod_{S_{j} \in \mathcal{S}}(I + S_{j})
\end{align} 
where $N$ is the total number of qubits, and $S_{j}$ are independent generators for the stabilizer group $\mathcal{S}$. Then, the density matrix of a ground state of a stabilizer Hamiltonian can be represented in the following way:
\begin{align}
\hat{\rho}_{R} \ = \ \frac{1}{2^{G(\mathcal{S}_{R})}} \prod_{S_{j} \in \mathcal{S}_{R}}(I + S_{j})
\end{align}
where $S_{j}$ are independent generators for a restriction of the stabilizer group into $R$ which is denoted as $\mathcal{S}_{R}$. $G(\mathcal{S}_{R})$ is the number of independent generators for $\mathcal{S}_{R}$. Since the entanglement entropy is defined as
\begin{align}
E_{R} \ \equiv \ \mbox{Tr}\left[ \hat{\rho}_{R} \log \hat{\rho}_{R} \right] ,
\end{align}
it can be represented in terms of the restriction of stabilizer group in the following way~\cite{Beni10}:
\begin{align}
E_{R} \ = \ V_{R} - G(\mathcal{S}_{R}), \label{eq:entropy_stabilizer}
\end{align}
where $V_{R}$ is the number of qubits inside $R$. 

Now that we have a formula to compute entanglement entropies, let us compute the entanglement entropy for a region $A_{1}$ which consists of only one qubit. Then, we have $E_{A_{1}} =  1$ since $V_{A_{1}} = 1$ and $G(\mathcal{S}_{A_{1}})=0$. Next, let us compute the entanglement entropy for a region $A_{j}$ which consists of $j$ consecutive qubits. Then, we have 
$E_{A_{1}} = 2$ for $1  \leq  j \leq  N-1$
since $V_{A_{j}} = j$ and $G(\mathcal{S}_{A_{j}})=j-2$. This indicates the existence of short-range entanglement between neighboring qubits. 

However, a cluster state does not have global entanglement. This can be seen from the fact that $E(A:B)=0$ for any disjoint regions $A$ and $B$. Thus, entanglement arising in a cluster state is short-range and is not scale invariant. As we shall soon see, this is a direct consequence of the absence of one-dimensional logical operators. 

\begin{figure}[htb!]
\centering
\includegraphics[width=0.45\linewidth]{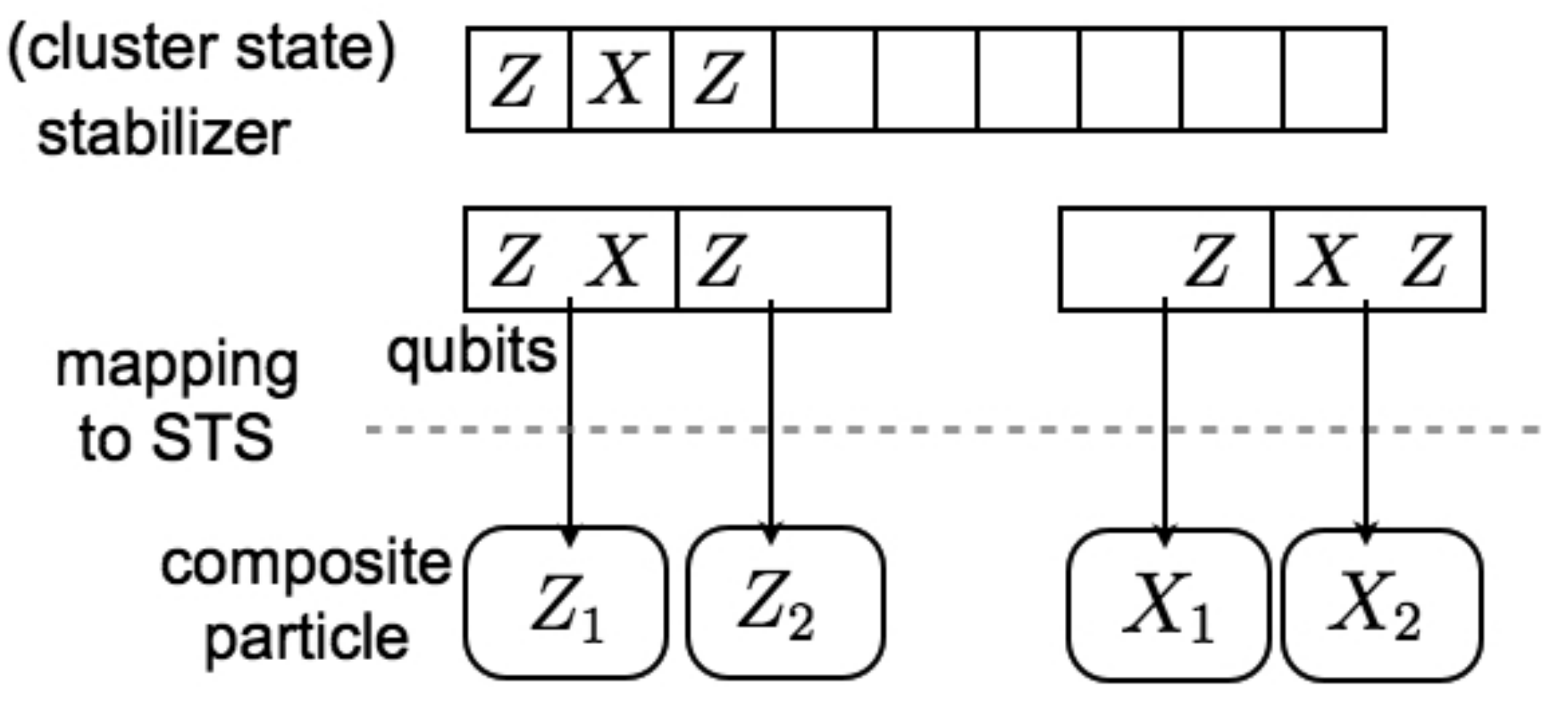}
\caption{A cluster state. A mapping to an STS model is shown.
} 
\label{fig_1Dex2}
\end{figure}

\textbf{Reduction to an STS model:}
This model can discussed in the framework of an STS model by considering two consecutive qubits as a single composite particle ($v=2$) when the total number of qubits $N$ is even. Let $Z_{p}^{(j)}$ and $X_{p}^{(j)}$ be Pauli operators acting on a $p$-th qubit ($p =1,2$) inside a $j$-th composite particle ($j=1,\cdots,N/2$) particle  with the following commutation relation:
\begin{align}
\left\{
\begin{array}{ccc}
Z_{1}^{(j)} ,& Z_{2}^{(j)} \\
X_{1}^{(j)} ,& X_{2}^{(j)}
\end{array}\right\}.
\end{align}
Then, after applying some appropriate unitary transformations on qubits inside each composite particle, one can represent the Hamiltonian for a cluster state in the following way:
\begin{align}
H \ = \ - \sum_{j} Z_{1}^{(j)}Z_{2}^{(j+1)} - \sum_{j}  X_{1}^{(j)}X_{2}^{(j+1)}. \label{eq:cluster_U}
\end{align}

\textbf{Disentangling short-range entanglement:}
Even after a coarse-graining, the model still possesses short-range entanglement between composite particles. However, there is no global entanglement since $E(A:B)=0$ for any disjoint regions $A$ and $B$ of composite particles. One may also see why there is no global entanglement by considering how stabilizers $Z_{1}^{(j)}Z_{2}^{(j+1)}$ and $X_{1}^{(j)}X_{2}^{(j+1)}$ act on qubits inside each composite particle. As shown in Fig.~\ref{fig_1Dex2_disentangle}, stabilizers connect only the first qubits in $j$-th composite particles and the second qubits in $j+1$-th composite particles. Thus, there is no entanglement which survives over the lattice, and entanglement is established only between neighboring composite particles.

As long as global entanglement is concerned, a cluster state is similar to a product state since both states have $E(A:B)=0$ for any disjoint regions $A$ and $B$. In fact, one can disentangle neighboring entanglement between qubits, or between composite particles by applying local unitary transformations and reduce a cluster state to a product state. Here, we first look at how short-range entanglement can be disentangled in a system of qubits. Then, let us discuss how neighboring entanglement can be disentangled in a coarse-grained lattice.

\begin{figure}[htb!]
\centering
\includegraphics[width=0.50\linewidth]{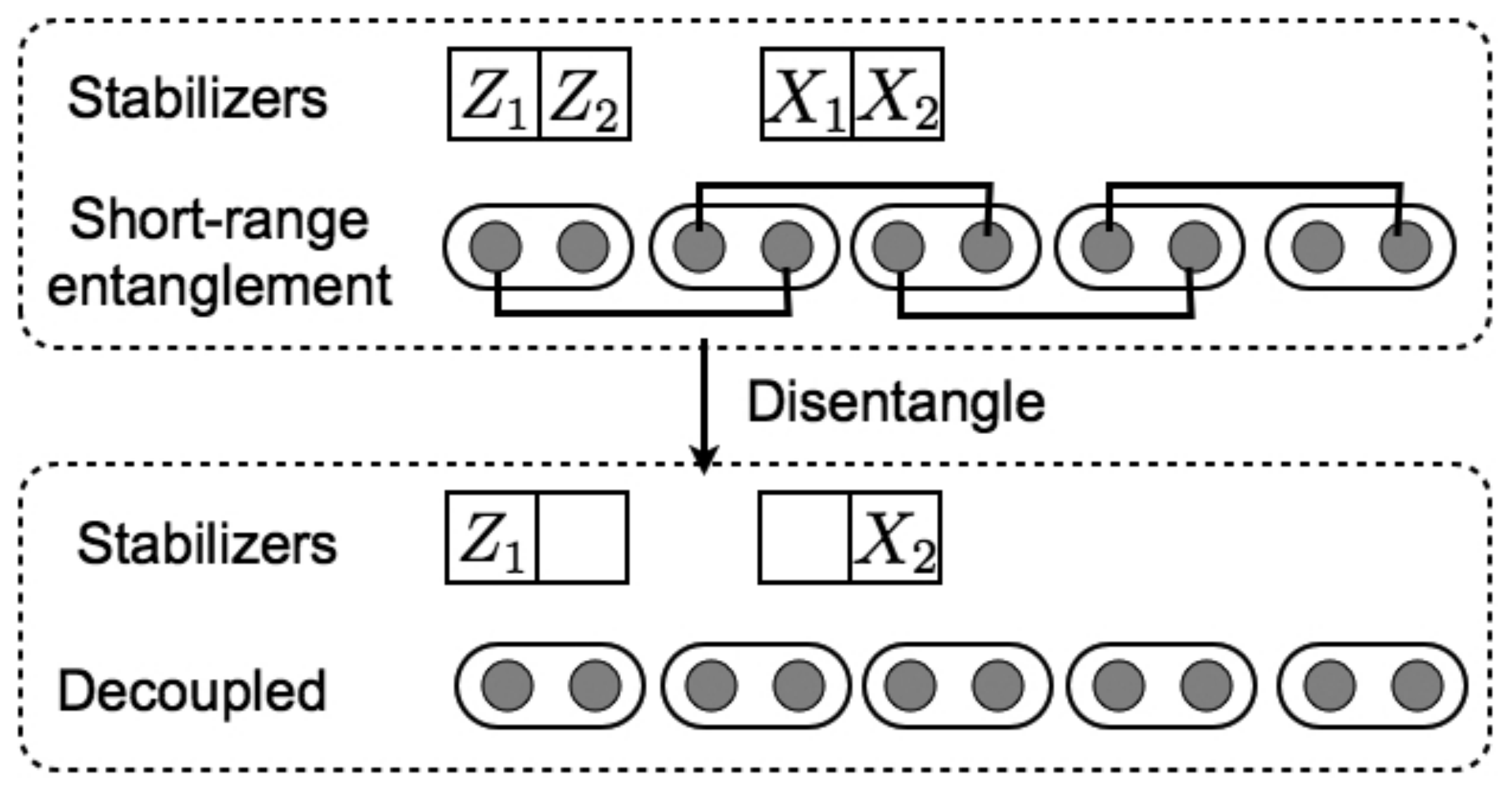}
\caption{A disentangling operation in a cluster state described through composite particles. Short-range entanglement between composite particles can be removed by some local unitary transformations acting on neighboring composite particles. 
} 
\label{fig_1Dex2_disentangle}
\end{figure}

The Hamiltonian for a cluster state can be obtained by applying control-$Z$ operations on the Hamiltonian for a product state $H  =  - \sum_{j} X^{(j)}$. Here, the control-$Z$ operation acting on two qubits is defined as follows:
\begin{align}
\mbox{CZ}\ : \ |00\rangle \ \rightarrow \ |00\rangle, \quad |01\rangle \ \rightarrow \ |01\rangle, \quad |10\rangle \ \rightarrow \ |10\rangle, \quad |11\rangle \ \rightarrow \ -|11\rangle
\end{align}
where CZ represents the control-$Z$ operation.
The effect of the control-$Z$ operation on two qubits can be represented in terms of Pauli operators in the following way:
\begin{align}
\mbox{CZ} \
\left\{
\begin{array}{ccc}
Z_{1} ,& Z_{2} \\
X_{1} ,& X_{2}
\end{array}\right\} \ \mbox{CZ}^{-1} \ = \ 
\left\{
\begin{array}{ccc}
Z_{1} ,& Z_{2} \\
X_{1}Z_{2} ,& Z_{1}X_{2}
\end{array}\right\}
\end{align}
where the operation transforms $X_{1}$ into $X_{1}Z_{2}$ and $X_{2}$ into $Z_{1}X_{2}$. Here, $Z_{1}$ and $X_{1}$ act on the first qubit, and $Z_{2}$ and $X_{2}$ act on the second qubit. $\mbox{CZ}^{-1}$ represents the inverse of CZ.

Now, we describe the reduction from a cluster state to a product state by representing control-Z operations in terms of composite particles. In particular, through some observations, we notice that the following local unitary transformations $U^{(j)}$ acting on $j$-th and $j+1$-th composite particles can disentangle neighboring entanglement between composite particles in Eq.~(\ref{eq:cluster_U}):
\begin{align}
U^{(j)} \  
\ \left\{
\begin{array}{cc}
Z_{1}^{(j)}Z_{2}^{(j+1)} , & X_{1}^{(j)}X_{2}^{(j+1)} \\
X_{1}^{(j)}              , & Z_{2}^{(j+1)}
\end{array}\right\} \ (U^{(j)})^{-1}
\ = \ \left\{
\begin{array}{cc}
Z_{1}^{(j)}    , & X_{2}^{(j+1)} \\
X_{1}^{(j)}    , & Z_{2}^{(j+1)}
\end{array} \right\}
\end{align}
where Pauli operators are transformed each other through $U^{(j)}$. Here, we note that $U^{(j)}$ commute with each other: $[U^{(j)},U^{(j')}]=0$. Then, by applying a unitary transformation $U \ = \ \prod_{j} U^{(j)}$, one can transform the original Hamiltonian in Eq.~(\ref{eq:cluster_U}) to
\begin{align}
UHU^{-1} \ = \ - \sum_{j} Z_{1}^{(j)} - \sum_{j} X_{2}^{(j)},
\end{align}
and short-range entanglement between neighboring composite particles can be washed out (Fig.~\ref{fig_1Dex2_disentangle}). Thus, the model can be reduced to a Hamiltonian which supports a ``product state of composite particles''. 

These observations imply that short-range entanglement arising in a cluster state is not scale invariant. According to the basic principle of a classification of quantum phases through RG transformations, such short-range entanglement is not relevant for characterizations of quantum phases. At the end of this section, we shall confirm this expectation rigorously for quantum phases arising in STS models. 

\subsubsection{Extended five qubit code: reduction to a classical ferromagnet}\label{sec:1D13}

We discuss our final example of one-dimensional STS models. We consider an example which can be obtained by generalizing the five qubit code to a one-dimensional spin chain. By introducing composite particles, one can discuss the model in the framework of STS models. We see that geometric shapes of logical operators are the same as those in a classical ferromagnet, and physical properties are similar. Then, we show that the model can be reduced to a ``classical ferromagnet of composite particles'' by disentangling short-range entanglement between neighboring composite particles. 

As we have seen in Section~\ref{sec:review2}, the five qubit code is constructed on five qubits, which is defined with the following stabilizer generator: $X^{(1)}Y^{(2)}Y^{(3)}X^{(4)}$
and its translations. Inspired by this five qubit code, we consider the following Hamiltonian defined with $N$ qubits:
\begin{align}
H \ = \ - \sum_{j=1}^{N} X^{(j)} Y^{(j+1)}Y^{(j+2)}X^{(j+3)}.
\end{align}
One can easily see that this extended five qubit code is also a stabilizer code since all the interaction terms $X^{(j)} Y^{(j+1)}Y^{(j+2)}X^{(j+3)}$ commute with each other. The code satisfies a scale symmetry since
\begin{align}
\prod_{j} X^{(j)} Y^{(j+1)}Y^{(j+2)}X^{(j+3)} \ = \ I
\end{align}
and $k =1$ for any $N$. 

\begin{figure}[htb!]
\centering
\includegraphics[width=0.70\linewidth]{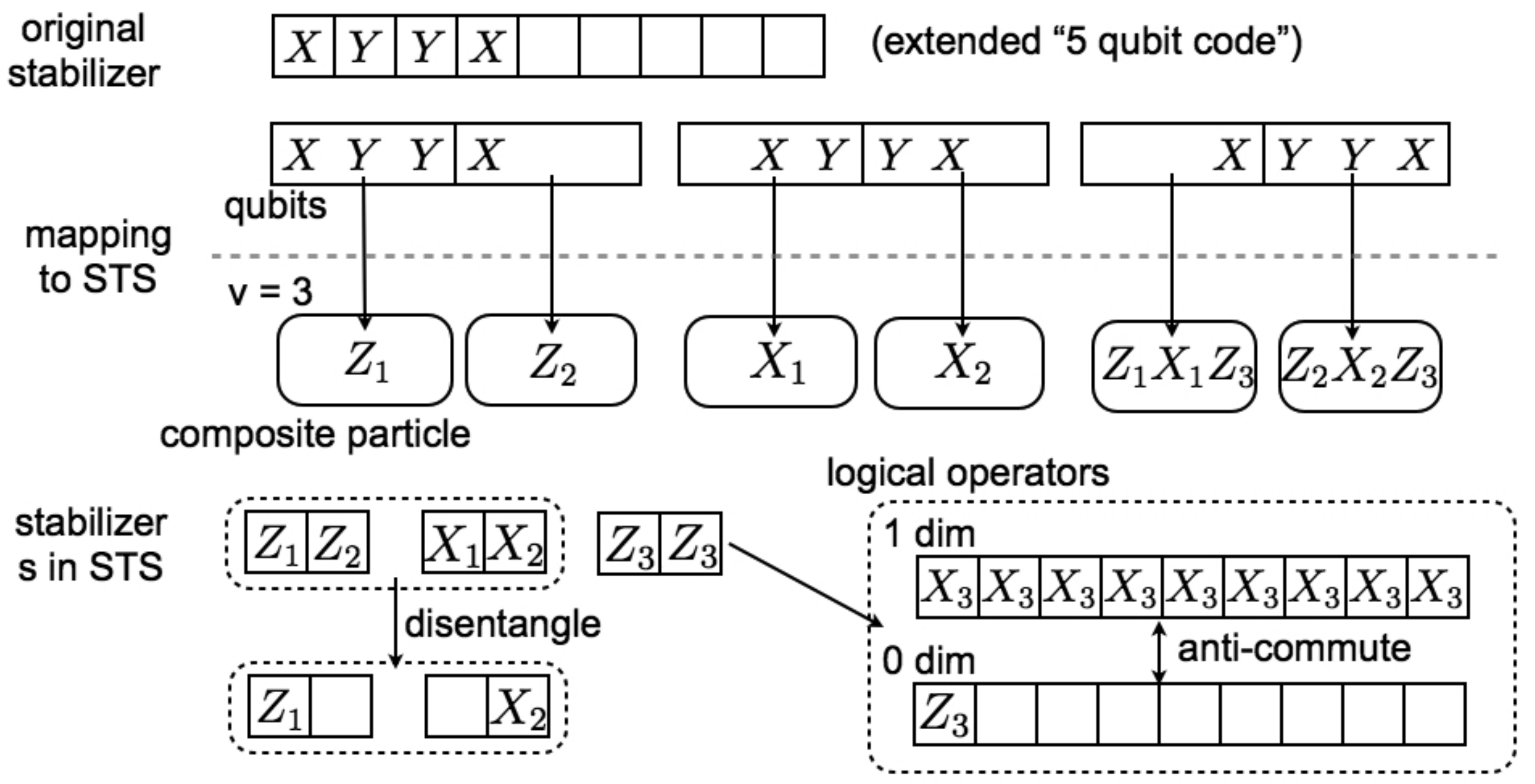}
\caption{The extended five qubit code.
} 
\label{fig_1Dex3}
\end{figure}

This model can be reduced to an STS model by considering three consecutive qubits as a single composite particle ($v=3$). Let $Z_{p}^{(j)}$ and $X_{p}^{(j)}$ be Pauli operators acting on a $p$-th qubits ($p = 1, \cdots ,3$) inside a $j$-th composite particle ($j = 1, \cdots ,n$) with the following commutation relations:
\begin{align}
\left\{
\begin{array}{ccc}
Z_{1}^{(j)} ,& Z_{2}^{(j)} ,& Z_{3}^{(j)} \\
X_{1}^{(j)} ,& X_{2}^{(j)} ,& X_{3}^{(j)}
\end{array}\right\}.
\end{align}
Then, after applying some appropriate local unitary transformations on qubits inside each composite particle, one can represent each of stabilizers in the following way:
\begin{align}
 Z_{1}^{(j)}Z_{2}^{(j+1)}, \qquad  X_{1}^{(j)}X_{2}^{(j+1)}, \qquad  Z_{3}^{(j)}Z_{3}^{(j+1)}.
\end{align}

With this notation through composite particles, one can easily write down logical operators:
\begin{align}
\ell \ = \ Z_{3}^{(1)}, \qquad r \ = \ X_{3}^{(1)}X_{3}^{(2)}\cdots X_{3}^{(N)}, \qquad \{ \ell, r \} \ = \ 0.
\end{align}
We notice that geometric shapes of logical operators are the same as those of a classical ferromagnet. Then, we expect that physical properties of the model are similar to those of a classical ferromagnet. In fact, one can discuss properties of entanglement in a way similar to a classical ferromagnet. The existence of a one-dimensional logical operator $r$ implies the presence of a global entanglement. Indeed, for a ground state satisfying $r|\psi\rangle = |\psi\rangle$, we have $E(A:B)=1$ for any disjoint regions $A$ and $B$, which can be computed easily by using a formula in Eq.~(\ref{eq:entropy_stabilizer}). Here, we choose two degenerate ground states of the model in the following way:
\begin{align}
\ell |\psi_{0} \rangle \ = \ |\psi_{0}\rangle, \qquad \ell|\psi_{1} \rangle \ = \ - |\psi_{1}\rangle, \qquad r |\psi_{0} \rangle \ = \ |\psi_{1}\rangle, \qquad r|\psi_{1} \rangle \ = \ |\psi_{0}\rangle.
\end{align}
Then, we notice that $|\psi\rangle$ has a \emph{GHZ-like entanglement} since $|\psi\rangle = \frac{1}{\sqrt{2}} (|\psi_{0}\rangle + |\psi_{1}\rangle)$ and $|\psi\rangle$ is a superposition of two globally separated ground states.

As long as the scale invariant ground state properties are concerned, an extended five qubit code is similar to a classical ferromagnet. In fact, this model can be reduced to a classical ferromagnet by applying unitary transformations acting on neighboring composite particles in a way similar to the reduction of a cluster state:
\begin{align}
U^{(j)} \ : \ Z_{1}^{(j)}Z_{2}^{(j+1)} \ \rightarrow Z_{1}^{(j)}, \qquad X_{1}^{(j)}X_{2}^{(j+1)} \ \rightarrow \ X_{2}^{(j+1)}.
\end{align}
Note that $[U^{(j)},U^{(j')}]=0$. Now, after this unitary transformations, only the third qubits in each composite particle are correlated through stabilizers $Z_{3}^{(j)}Z_{3}^{(j+1)}$ while the first and second qubits in each composite particle are decoupled due to the existence of $Z_{1}^{(j)}$ and $X_{2}^{(j+1)}$. Thus, the extended five-qubit code is equivalent to a classical ferromagnet, when decoupled qubits are removed. Later, we shall show that this model and a classical ferromagnet belong to the same quantum phase.

\textbf{Connection with RG algorithms based on the tensor product state formalism:}
We have seen that an extended five qubit code can be reduced to a classical ferromagnet by applying disentangling operations on neighboring composite particles. Here, it may be worth noting that these disentangling operations are closely related to a ``disentangler'' which is a key element in a novel RG algorithm based on the tensor product state formalism~\cite{Vidal07, Chen10}.

In the previous section, we have coarse-grained the system of qubits by introducing composite particles where microscopic properties of each qubit inside a newly defined composite particle are completely lost. In a coarse-grained lattice, one may consider two STS models as the same when they can be transformed each other through unitary transformations acting on qubits inside each composite particle. For example, Hamiltonians $H = - \sum_{j} Z^{(j)}Z^{(j+1)}$ and $H = - \sum_{j} X^{(j)}X^{(j+1)}$ may be considered to be the same, as discussed in Section~\ref{sec:review1}.

However, coarse-graining is not sufficient to characterize the scale invariant ground state properties since there still remain some short-range correlations which cannot be washed out by coarse-graining, as we have seen in analyses on a cluster state and an extended five qubit code. Such short-range entanglement must be also removed in classifying quantum phases, according to the basic philosophy of the classification of quantum phases through RG transformations.

Recently, a remarkable idea of removing short-range correlations has been proposed in a search for efficient RG algorithms based on the tensor product state formalism. The key idea is the use of a disentangling operation, called a disentangler, through local unitary transformations to remove these short-range entanglement between neighboring particles~\cite{Vidal07}. It has been demonstrated that RG transformations which combine coarse-graining and disentangling operations work in a remarkably efficient way in analyzing quantum phases arising in two-dimensional strongly correlated spin systems.

A disentangling operation used in the analysis on a cluster state is essentially similar to a disentangler used in this RG algorithm. In particular, a central idea behind these disentangling operation is that short-range entanglement is irrelevant to characterizations of quantum phases. This observation will be rigorously confirmed for quantum phases arising in STS model at the end of this section. 

Based on observations obtained in the analyses on a cluster state and an extended five qubit code, we wish to consider two STS models as the same when they can be transformed each other by local unitary transformations. By local unitary transformations, we include the following two elements:
\begin{itemize}
\item \textbf{Unitary operations on composite particles:} Unitary transformations acting on qubits inside each composite particle.
\item \textbf{Disentangling operations:} Unitary transformations acting on neighboring composite particles.
\end{itemize}
Here, we note that this is consistent with our original approach to distinguish quantum phases through geometric shapes of logical operators since geometric shapes of logical operators are invariant under local unitary transformations. Later, we shall see that one-dimensional STS models connected through local unitary transformations belong to the same quantum phase.\footnote{In general, local unitary transformations are used in much broader sense. In particular, any unitary evolutions induced by time evolution of local Hamiltonians for a finite time may be called local unitary transformations~\cite{Chen10}.}

\subsection{Logical operators in one-dimensional STS models}\label{sec:1D2}

We have seen that geometric shapes of logical operators are central in analyzing the scale invariant ground state properties of STS models through several examples. In this subsection, the analysis is extended to arbitrary one-dimensional STS models. We obtain a canonical set of logical operators (all the independent logical operators) of arbitrary STS models, and see that logical operators are either zero-dimensional or one-dimensional, as in a classical ferromagnet. We discuss how the ground state properties, such as global entanglement, in STS models can be studied by geometric shapes of logical operators, and show that any one-dimensional STS models can be reduced to multiple copies of a classical ferromagnet, or a product state, by disentangling neighboring entanglement between composite particles through local unitary transformations. 

\begin{figure}[htb!]
\centering
\includegraphics[width=0.7\linewidth]{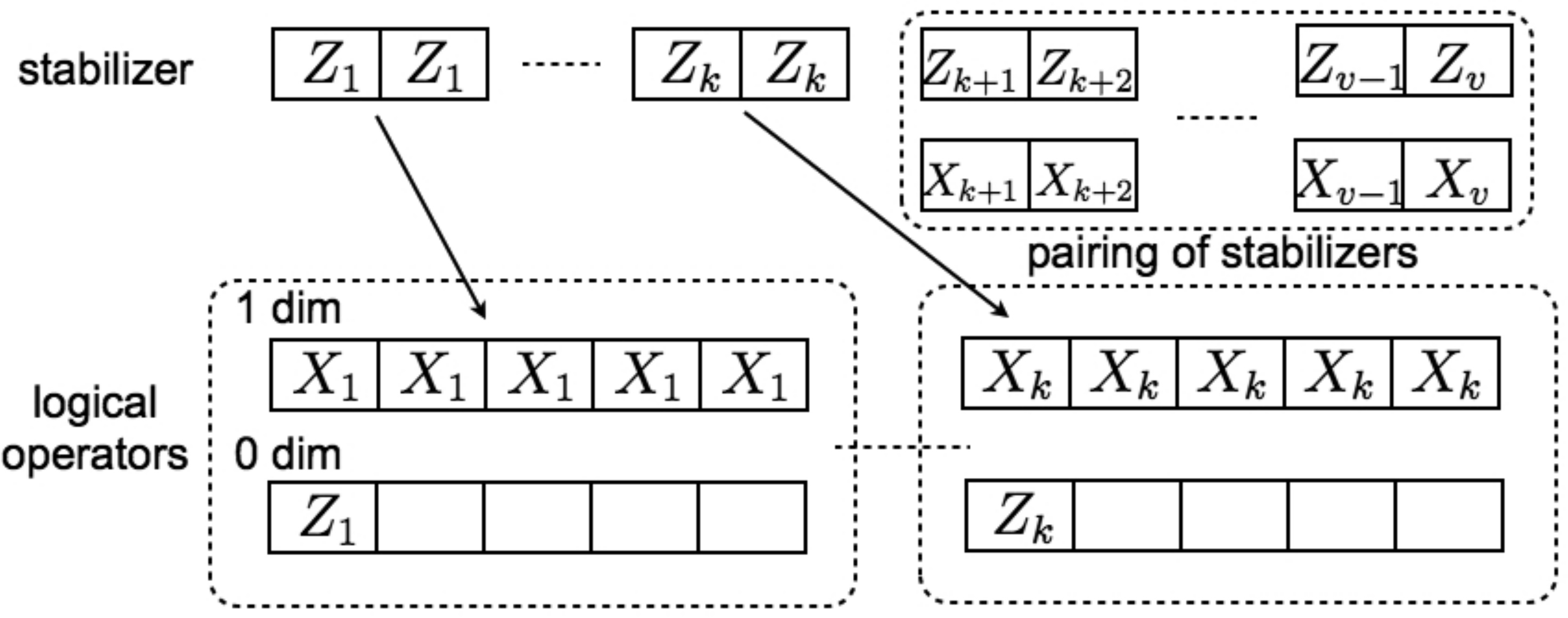}
\caption{Stabilizers and logical operators in arbitrary STS models.
} 
\label{fig_1Dex4}
\end{figure}

Let us consider an STS model defined with composite particles which consist of $v$ qubits. For simplicity of discussion, we neglect stabilizers acting on single composite particles since decoupled qubits are to be removed. Then, stabilizers in one-dimensional STS models can be represented in the following way (Fig.~\ref{fig_1Dex4}), as shown in~\ref{sec:1D_proof}:
\begin{align}
&\mbox{Ferromagnetic part:} \qquad \quad \ \ \ Z_{1}^{(j)} Z_{1}^{(j+1)}, \ \cdots, \ Z_{k}^{(j)} Z_{k}^{(j+1)}\\ 
&\mbox{Short-range entanglement:} \quad
Z_{k+1}^{(j)}Z_{k+2}^{(j+1)} , \ X_{k+1}^{(j)}X_{k+2}^{(j+1)}, \ \cdots, \ Z_{v-1}^{(j)}Z_{v}^{(j+1)} , \ X_{v-1}^{(j)}X_{v}^{(j+1)}
\end{align}
where $v-k$ is an even integer. Here, stabilizers $Z_{1}^{(j)} Z_{1}^{(j+1)}, \cdots,  Z_{k}^{(j)} Z_{k}^{(j+1)}$ create ferromagnet-like correlations, while $Z_{k+1}^{(j)}Z_{k+2}^{(j+1)}, X_{k+1}^{(j)}X_{k+2}^{(j+1)}, \cdots, Z_{v-1}^{(j)}Z_{v}^{(j+1)}, X_{v-1}^{(j)}X_{v}^{(j+1)}$ create short-range entanglement between neighboring composite particles as in a cluster state. Then, logical operators are
\begin{align}
\Pi(\mathcal{S}_{n}) \ = \ \left\{
\begin{array}{ccc}
\ell_{1} ,& \cdots ,& \ell_{k} \\ 
r_{1}    ,& \cdots ,& r_{k}
\end{array}\right\}
\end{align}
where
\begin{align}
\ell_{p} \ = \ Z_{p}^{(1)}, \qquad r_{j} \ = \ X_{p}^{(1)}X_{p}^{(2)}\cdots X_{p}^{(N)},  \qquad \mbox{for}\ \ p \ = \ 1 , \cdots , k.
\end{align}
Thus, geometric shapes of logical operators in one-dimensional STS models are either zero-dimensional or one-dimensional. Zero-dimensional logical operators and one-dimensional logical operators always form anti-commuting pairs (Fig.~\ref{fig_1Dex41}). 

\begin{figure}[htb!]
\centering
\includegraphics[width=0.80\linewidth]{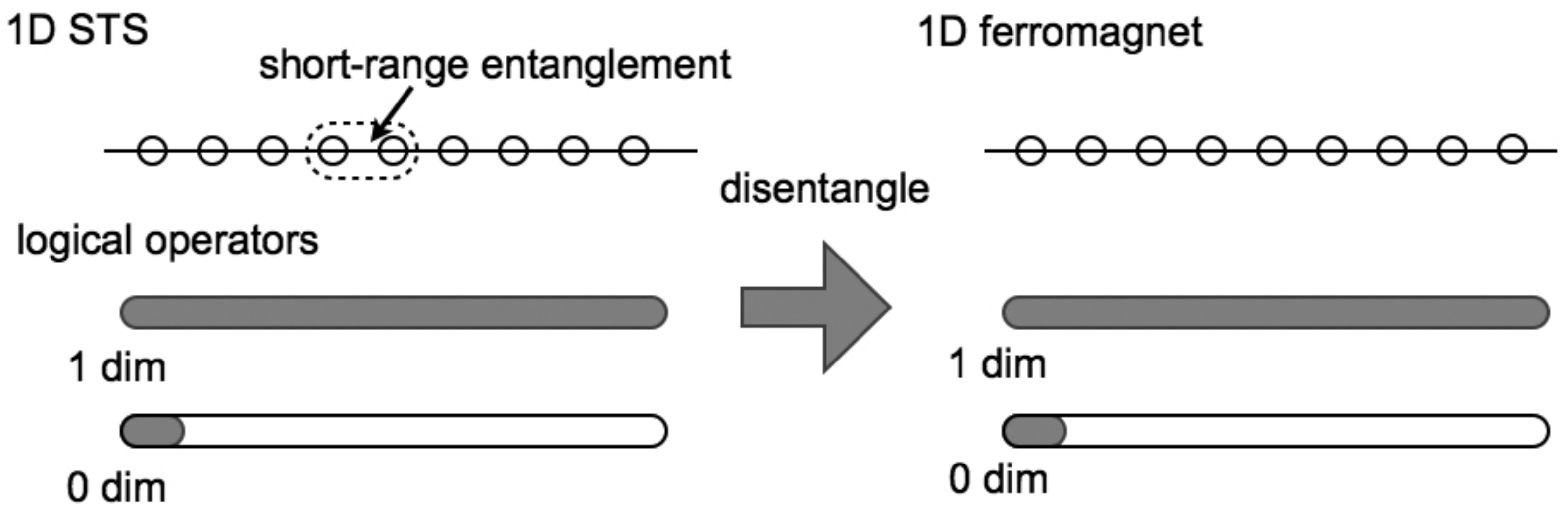}
\caption{Reduction of one-dimensional STS models to classical ferromagnets by disentangling short-range entanglement between neighboring composite particles. Geometric shapes of logical operators characterize scale invariant properties of STS models.
} 
\label{fig_1Dex41}
\end{figure} 

We notice that geometric shapes of logical operator are the same as those in a classical ferromagnet. Then, we expect that the scale invariant ground state properties in one-dimensional STS models can be discussed in a way similar to discussions on a classical ferromagnet. This expectation turns out to be true. In fact, due to the existence of one-dimensional logical operators, ground states of STS models can have GHZ-like entanglement, as in a classical ferromagnet. For example, when the number of logical qubits is $k$, there exists a ground state which has a mutual information $E(A:B)=k$ for any disjoint regions of composite particles $A$ and $B$. Note that this can be confirmed easily by using a formula in Eq.~(\ref{eq:entropy_stabilizer}).

Even more, one can show that all the one-dimensional STS models can be reduced to classical ferromagnets by applying disentangling operations between neighboring composite particles. Let us focus on a pair of stabilizers $Z_{k+1}^{(j)}Z_{k+2}^{(j+1)}$ and $X_{k+1}^{(j)}X_{k+2}^{(j+1)}$. Then, these two stabilizers can be transformed into $Z_{k+1}^{(j)}$ and $X_{k+2}^{(j+1)}$ by removing short-range entanglement between neighboring composite particles, as in the discussion of a cluster state. One can repeat the similar arguments for all the other pairs of stabilizers. Thus, one-dimensional STS models with $k$ logical qubits can be reduced to the following STS model with $v=k$ through local unitary transformations:
\begin{align}
H \ = \ - \sum_{j=1}^{n}\sum_{p=1}^{k} Z_{p}^{(j)}Z_{p}^{(j+1)}.
\end{align}

This model consists of ``multiple copies'' of a classical ferromagnet embedded in a non-interacting way. The $p$-th classical ferromagnet connects $p$-th qubits inside each composite particle with $Z_{p}^{(j)}Z_{p}^{(j+1)}$, and different classical ferromagnets embedded in the model are decoupled from each other. Therefore, when two STS models have logical operators with the same geometric shapes, they have similar global entanglement in ground states and can be reduced to an STS model in the above form. 

Here, we mention the importance of geometric shapes of logical operators as indicators of global entanglement in ground states. Since geometric shapes of logical operators are invariant under local unitary transformations, they can capture the scale invariant ground state properties such as global entanglement as seen in a GHZ state (Fig~\ref{fig_1Dex41}). We shall soon see that geometric shapes of logical operators distinguish quantum phases in one-dimensional STS models completely, serving as order parameters. 

\subsection{Quantum phases and geometric shapes of logical operators}\label{sec:1D3}

We have seen that the ground state properties, such as global entanglement, in STS models are similar when geometric shapes of their logical operators are the same. Here, our hope is to use geometric shapes of logical operators as order parameters to distinguish quantum phases. However, there is still an important gap between geometric shapes of logical operators and the notion of quantum phases. In particular, it is not clear if different geometric shapes of logical operators lead to different quantum phases separated by a QPT. Also, it is not clear if the same geometric shapes of logical operator imply that two STS models belong to the same quantum phase or not.

Here, we establish the relation between quantum phases and geometric shapes of logical operators. In this subsection, we show that two STS models belong to different quantum phases if and only if geometric shapes of logical operators are different. In particular, we show that two STS models with the same geometric shapes of logical operators can be connected without closing the energy gap, while two STS models with different geometric shapes of logical operators are always separated by a QPT.

\textbf{Cases with the same number of logical operators:} Let us begin by showing that two STS models belong to the same quantum phase when they have the same geometric shapes of logical operators. 

Consider two STS models $H_{A}$ and $H_{B}$ with the same number of logical qubits $k$, or the same number of anti-commuting pairs of logical operators. For simplicity of discussion, let us assume that the number of qubits inside each composite particle is $v$ for both $H_{A}$ and $H_{B}$ ($v \geq k$). We also assume that both $H_{A}$ and $H_{B}$ have the same system size. We denote the projection operators onto the ground state spaces of $H_{A}$ and $H_{B}$ as $\hat{P}_{A}$ and $\hat{P}_{B}$. From the discussion in the previous subsection, there always exist a local unitary transformation $U$ which transform $\hat{P}_{A}$ into $\hat{P}_{B}$:
\begin{align}
U\hat{P}_{A}U^{-1} \ = \ \hat{P}_{B}.
\end{align}
Then, one can show that there exists a parameterized Hamiltonian which connects $H_{A}$ and $H_{B}$ without closing the energy gap. 

Here, we give a sketch of a proof for the existence of such a parameterized Hamiltonian $H(\epsilon)$. From the discussion in the previous section, unitary transformations $U$ which connects $\hat{P}_{A}$ and $\hat{P}_{B}$ can be decomposed into two parts:
\begin{align}
U \ = \ U_{composite} \times U_{disentangle} \qquad \mbox{where} \quad [U_{composite}, U_{disentangle}] \ = \ 0.
\end{align}
The first part $U_{composite}$ represents unitary transformations acting on qubits inside each composite particle, while the second part $U_{disentangle}$ represents disentangling operations acting on neighboring composite particles. Let us begin with the case where $U_{disentangle}=I$. Then, by gradually inducing unitary transformations on each composite particle, one can transform $H_{A}$ into $UH_{A}U^{-1}$ without closing the energy gap or changing the number of ground states. Now, since the projection into the ground state space of $UH_{A}U^{-1}$ is the same as $\hat{P}_{B}$, one can change $UH_{A}U^{\dagger}$ to $H_{B}$ without closing the energy gap. Thus, $H_{A}$ and $H_{B}$ belong to the same quantum phase. Next, let us consider the case where $U_{composite}=I$. Let us recall a transformation from a product state Hamiltonian $H_{A} = - \sum_{j}X^{(j)}$ to a cluster state Hamiltonian $H_{B}= - \sum_{j}Z^{(j-1)}X^{(j)}Z^{(j+1)}$. Since the control-$Z$ operation acting on $j$-th and $j+1$-th qubits can be written as
\begin{align}
\mbox{CZ} \ = \ \exp \left[ i \frac{\pi}{4} (Z^{(j)} - I)(Z^{(j+1)} - I) \right],
\end{align}
by gradually inducing the control-$Z$ operation such that 
\begin{align}
\mbox{CZ}(\epsilon) = \exp \left[ i \frac{\epsilon\pi}{4} (Z^{(j)} - I)(Z^{(j+1)} - I) \right] 
\end{align}
from $\epsilon = 0$ to $\epsilon = 1$, one can transform a product state Hamiltonian to a cluster state Hamiltonian without changing the energy gap by setting $H(\epsilon) = \mbox{CZ}(\epsilon)H_{A}(\mbox{CZ}(\epsilon))^{-1}$. By generalizing this idea, one can show that $H_{A}$ and $H_{B}$ can be transformed each other without closing the energy gap when the number of logical operators is the same. It is straightforward to extend these discussions to the cases where $U_{composite}\not= I$ and $U_{disentangle}\not= I$. Also, when the numbers of qubits inside each composite particles are different for $H_{A}$ and $H_{B}$ (say, $v_{A}$ and $v_{B}$ with $v_{A}\not = v_{B}$), we coarse-grain by grouping $v$ qubits into a composite particle where $v$ is the least common multiple of $v_{A}$ and $v_{B}$. Then, one can repeat the similar argument and show that $H_{A}$ and $H_{B}$ belong to the same quantum phase. This completes the sketch of the proof.

\textbf{Cases with the different numbers of logical operators:}
Next, let us show that quantum phases represented by two STS models are different when geometric shapes of logical operators are different. Since zero-dimensional and one-dimensional logical operators always form anti-commuting pairs, different geometric shapes of logical operators imply different numbers of degenerate ground states. Then, when connecting two STS models through a parameterized Hamiltonian, the energy gap must close at some point as excited states need to be ground states in the course of a parameter change. Thus, two STS models with different $k$ belong to different quantum phases which are separated by a QPT. 

Here, we would like to make some comments on phase transitions between two STS models $H_{A}$ and $H_{B}$ which commute with each other, but have different numbers of logical operators. In such cases, phase transitions may occur exactly at $\epsilon = 0$ or $\epsilon = 1$, instead of some intermediate point between $\epsilon = 0$ and $\epsilon = 1$. Let us look at an example. Consider the Ising model in a parallel field:
\begin{align}
H(\epsilon) \ &= - (1 - \epsilon)\sum_{j}Z^{(j)}Z^{(j+1)} - \epsilon \sum_{j}Z^{(j)} \\
H_{A}       \ &=  H(0) \ = \ - \sum_{j}Z^{(j)}Z^{(j+1)}, \qquad H_{B} \ = \ H(1)  \ = - \sum_{j}Z^{(j)}. 
\end{align}
Since $H_{A}$ and $H_{B}$ belong to different quantum phases with different numbers of logical operators, we expect that the model undergoes some phase transition. This parameterized Hamiltonian is exactly solvable for any $\epsilon$ since all the terms in $H(\epsilon)$ commute with each other for any $\epsilon$. Though this model connects $H_{A}$ and $H_{B}$ from $\epsilon = 0$ to $\epsilon =1$, the phase transition occurs at $\epsilon = 0$. To see this more clearly, let us extend our analysis to the cases where $\epsilon < 0$ too. At $\epsilon = 0$, the model has degenerate ground states $|0 \cdots 0\rangle$ and $|1 \cdots 1\rangle$. At $\epsilon > 0$, the model has a single ground state $|0 \cdots 0\rangle$. At $\epsilon < 0$, the model has a single ground state $|1 \cdots 1\rangle$. Since the ground state properties change drastically at $\epsilon = 0$, the model undergoes a phase transition at $\epsilon = 0$. In this model $H_{A}$ and $H_{B}$ may not be ``separated'' by a phase transition since $H_{A}$ lies at the transition point. However, we consider that $H_{A}$ and $H_{B}$ belong to different quantum phases since the ground state degeneracy is lifted for $\epsilon \not= 0$.

\begin{figure}[htb!]
\centering
\includegraphics[width=0.70\linewidth]{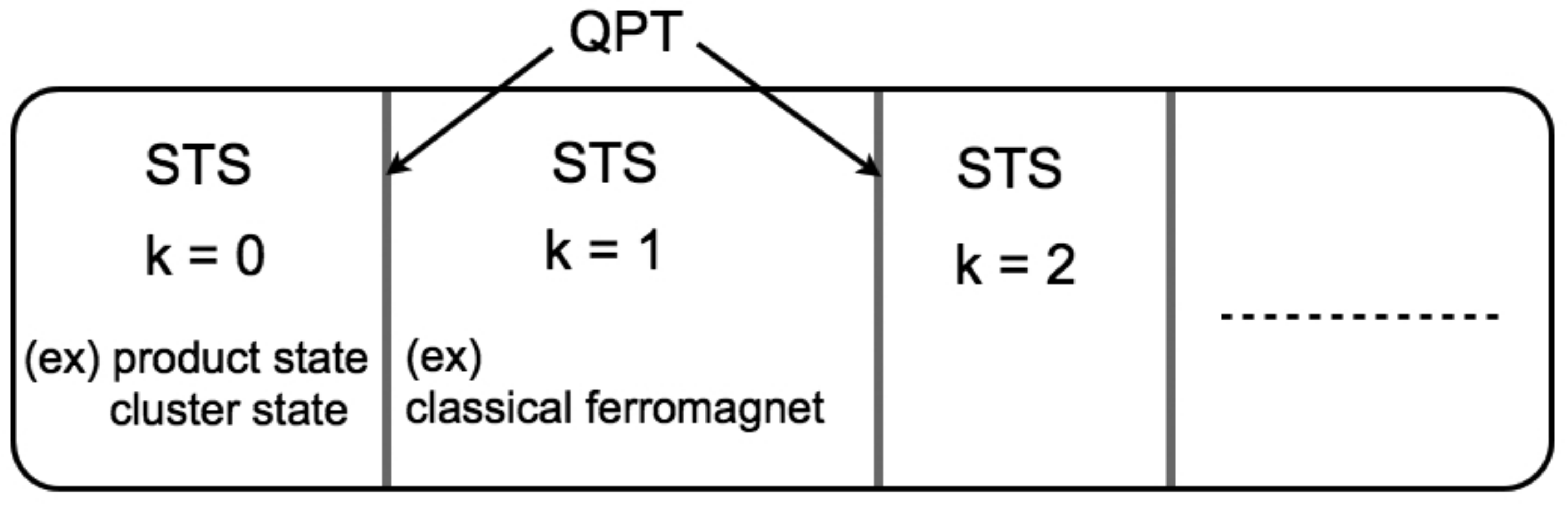}
\caption{A ``phase diagram'' of one-dimensional STS models. Different quantum phases can be characterized by geometric shapes of logical operators (the number of logical operators). Different quantum phases are separated by QPTs.  
} 
\label{fig_1D_phase}
\end{figure}

\textbf{Summary and applications:} We summarize the main result of this section (Fig.~\ref{fig_1D_phase}).

\begin{itemize}
\item Quantum phases in one-dimensional STS models can be characterized by geometric shapes of logical operators. Quantum phases represented by two STS models are different if and only if the numbers of logical operators, or the numbers of logical qubits $k$, are different.
\end{itemize}

Here, we mention the importance of the translation equivalence (Theorem~\ref{theorem_main}) in this classification of quantum phases in one-dimensional STS models. The underlying reason why we can specify quantum phases only through geometric shapes of logical operators is due to the translation equivalence of logical operators. In particular, as a result of the translation equivalence of logical operators, any translations of a logical operators are equivalent to the original logical operators. Then, one can distinguish quantum phases only through geometric shapes without considering the positions where logical operators are defined. 

Now, let us look at some examples. Within the framework of STS models, one can classify frustration-free Hamiltonians appeared in Section~\ref{sec:review} in the following way:
\begin{align}
H_{A} \ = \ - \sum_{j}Z^{(j)}Z^{(j+1)} \ \sim \  H_{A}' \ = \ - \sum_{j}X^{(j)}X^{(j+1)} \ \not \sim \   H_{B} \ = \ - \sum_{j}X^{(j)}
\end{align}
where $H_{A}$ and $H_{A'}$ belong to the same quantum phase while $H_{B}$ belong to a different quantum phase. This classification is consistent with the fact that the Ising model in a transverse field,
\begin{align}
H(\epsilon) \ = - (1 - \epsilon)\sum_{j}Z^{(j)}Z^{(j+1)} - \epsilon \sum_{j}X^{(j)}, \label{eq:Ising}
\end{align}
undergoes a QPT. Also, $H_{A} = - \sum_{j}Z^{(j)}Z^{(j+1)}$ and $H_{A}'  =  - \sum_{j}X^{(j)}X^{(j+1)}$ belong to the same quantum phase since they can be transformed each other through single qubit rotations without closing the energy gap. In combining the models discussed in Section~\ref{sec:1D1}, we have the following classification:
\begin{align}
&H_{A} \ = \ - \sum_{j}Z^{(j)}Z^{(j+1)} \ \sim \  H_{A}' \ = \ - \sum_{j}X^{(j)}X^{(j+1)} \ \sim \ H_{A}'' \ = \ - \sum_{j}X^{(j)}Y^{(j+1)}Y^{(j+2)}X^{(j+3)} \notag \\
&\not \sim \   H_{B} \ = \ - \sum_{j}X^{(j)} \ \sim \ H_{B}' \ \ = \ - \sum_{j}Z^{(j-1)}X^{(j)}Z^{(j+1)}.
\end{align}

While our discussions show only the existence of QPTs between STS models, analyses on QPTs in specific parameterized Hamiltonians are also important problems. However, detailed theoretical analyses on each parameterized Hamiltonian are beyond the scope of this paper. We note that QPTs occurring in some parameterized Hamiltonians connecting the above STS models are studied in previous works~\cite{Wolf06, Skrovseth09}.

\textbf{Connection with the symmetry-breaking theory:}
We have seen that changes of geometric shapes of logical operators may characterize different quantum phases. Here, we make comments on the similarity between our approach through geometric shapes of logical operators and the Landau's symmetry breaking theory on QPTs. 

The key idea of the Landau's symmetry breaking theory is to explain QPTs through changes of symmetries. Let us consider an example, the Ising model in a transverse field, described in Eq.~(\ref{eq:Ising}). The parameterized Hamiltonian has a symmetry with respect to the following global operator $r$:
\begin{align}
r \ = \ X^{(1)}\cdots X^{(N)}, \qquad rHr^{-1} \ = \ H.
\end{align}
Let us see how the symmetry with respect to $r$ changes as $\epsilon$ varies. For $\epsilon = 1$, the ground state is
$|\psi_{0}\rangle = | + \cdots + \rangle$.
Then, the symmetry with respect to $r$ is said to be spontaneously broken since the expectation value of $r$ is one and fixed to a single value. On the other hand, for $\epsilon = 0$, the ground states are
$|\psi_{0}\rangle = | 0 \cdots 0 \rangle$ and $|\psi_{1}\rangle = | 1 \cdots 1 \rangle$.
Then, the symmetry with respect to $r$ is not broken since the expectation value of $r$ inside the ground state space are not fixed to a single value. Thus, the symmetry with respect to a global operator $r$ changes during a QPT, which implies the existence of a QPT in this model.

Now, let us recall that logical operators are Pauli operators which commute with the system Hamiltonian, and 
$\ell H \ell^{-1} = H$
where $H$ is a stabilizer Hamiltonian. Then, logical operators $\ell$ characterize symmetries of stabilizer Hamiltonians. Here, one may notice that $r$ is a one-dimensional logical operator of a classical ferromagnet $H_{A} \equiv H(0)$. Since $r$ is a logical operator, $r$ has an anti-commuting pair $\ell = Z^{(1)}$. Then, the expectation value of $r$ is not fixed inside the ground state space of $H_{A}$, and the symmetry with respect to $r$ is not broken for $H_{A}$. However, this symmetry is broken for $H_{B}$ since $H_{B}$ does not have any logical operator. Thus, our approach in classifying quantum phases through logical operators is consistent with the Landau's symmetry breaking theory on QPTs. 

\section{Two-dimensional STS model: topological quantum phases and geometric shapes of logical operators}\label{sec:2D}

In one-dimensional STS models, different quantum phases are completely characterized by geometric shapes of logical operators. However, the only possible geometric shapes are anti-commuting pairs of zero-dimensional and one-dimensional logical operators. As a result, one-dimensional quantum phases are classified through only the number of logical operators without actually considering geometric shapes of logical operators. 

Unlike one-dimensional systems, two-dimensional systems have rich varieties of quantum phases, which may be seen from rich varieties of possible geometric shapes of logical operators in two dimensions. Here, we wish to search for possible quantum phases in two-dimensional STS models by finding logical operators and classify different quantum phases through geometric shapes of logical operators.

\textbf{Topological order and logical operators:}
To begin with, in order to discuss quantum phases through geometric shapes of logical operators, we need to study the ground state properties of STS models and their relations to logical operators. What makes two-dimensional systems strikingly distinct from one-dimensional systems is the possibility of the existence of topological order. Topological order is a highly non-local quantum correlation which is stable against any local and small perturbations. This peculiar correlation is known to be ``topological'', meaning that it is a scale invariant physical property which depends on topological structures of geometric manifolds where systems are defined~\cite{Wen90}. As we shall soon see, some of two-dimensional STS models, including the Toric code and the topological color code~\cite{Bombin06}, have topological order. Since topological order is a scale invariant characterization of non-local correlations arising in quantum many-body systems, in classifying quantum phases with topological order (topological phases), systems with ``different kinds of topological order'' must be appropriately distinguished. Here, we wish to distinguish topological phases in two-dimensional STS models through geometric shapes of logical operators. 

The main puzzle in analyzing topological order is the fact that topological order in correlated spin systems is currently characterized in various methods, such as the existence of anyonic excitations~\cite{Laughlin83, Moore91, Kitaev97, Kitaev03, Nayak08}, topological entanglement entropy~\cite{Kitaev06, Levin06} and the absence of local order parameters~\cite{Wen90, Bravyi06, Bravyi10b}. While these characterizations manifestly captures topological properties of non-local correlations, they must be some aspects of what is called ``topological order''. Though the connection between anyonic excitations and topological entanglement entropy has been established through the framework of topological quantum field theory~\cite{Kitaev06}, topological entanglement entropy cannot characterize anyonic excitations completely~\cite{Flammia09}. Therefore, if we hope to use geometric shapes of logical operators as order parameters for systems with topological order, all these existing characterizations of topological order must be explained through geometric shapes of logical operators in an unified way. 

We start our analyses on quantum phases in two-dimensional STS models by characterizing topological order arising in STS models completely through geometric shapes of logical operators. In Section~\ref{sec:2D1}, we study three examples of two-dimensional STS models, and analyze the connection between topological order and geometric shapes of logical operators. In particular, we approach three different characterizations of topological order, the existence of anyonic excitations, topological entanglement entropy and the absence of local order parameters, through geometric shapes of logical operators. In Section~\ref{sec:2D2}, we extend our analysis to arbitrary two-dimensional STS models and describe possible geometric shapes of logical operators. Based on geometric shapes of logical operators, topological order arising in two-dimensional STS models are discussed. Then, we show that geometric shapes of logical operators can distinguish topological properties of two-dimensional STS models appropriately.   

\textbf{Quantum phases and logical operators:}
Then, we shall move on to the analyses on quantum phases in two-dimensional STS models and their classifications. Since a topological STS model and a non-topological STS model have totally different physical properties, one may naturally think that they belong to different quantum phases. This expectation is also consistent with the basic spirit of RG transformations since topological order is a characterization of scale invariant non-local correlations. In fact, there are some analytical and numerical evidences for the existence of QPTs between topological phases and non-topological phases on specific models of parameterized Hamiltonians~\cite{Trebst07, Hamma08, Abasto08, J_Vidal09b, J_Vidal09}. However, the connection between a QPT and geometric shapes of logical operators has not been completely established yet. While a change in the number of ground states certainly leads to a QPT as we have seen in the discussion on one-dimensional STS models, it is not clear whether the change of geometric shapes of logical operators implies the existence of a QPT or not. 

A useful insight in analyzing quantum phases in two-dimensional STS models is to realize the resemblance between a mathematical notion of topology and a classification of quantum phases. In classifying geometric shapes of objects by using the notion of topology, two objects are considered to be the same when they can be transformed each other through continuous deformation, while they are considered to be different when they can be transformed each other through only non-analytic changes of geometric shapes. In a similar fashion, when quantum phases are classified, two frustration-free Hamiltonians are considered to belong to the same quantum phase when their ground states can be connected continuously, while two Hamiltonians are considered to be different when they can be connected only through parameterized Hamiltonians which undergo QPTs. 

The underlying ideas behind topology and a classification of quantum phases are fundamentally akin to each other, and the notion of topology has been adopted to classifications of quantum phases in various systems~\cite{Lifshitz60, Witten89, Wen90, Volovik_Text}. In STS models, geometric shapes of logical operators are central in characterizing the scale invariant physical properties such as topological order. In particular, it will be shown that the ground state properties drastically changes as a result of different geometric shapes of logical operators in Section~\ref{sec:2D1} and Section~\ref{sec:2D2}. Therefore, in classifying quantum phases arising in STS models, we introduce the notion of topology into geometric shapes of logical operators.  

In Section~\ref{sec:2D3}, we show that parameterized Hamiltonians connecting two STS models with different geometric shapes of logical operators are always separated by a QPT by proving that the energy gap must close at some point. Then, we show that different quantum phases in two-dimensional STS models are characterized by geometric shapes of logical operators, and thus, topological characteristics of logical operators can serve as order parameters in distinguishing quantum phases arising in two-dimensional STS models. 

\subsection{Role of logical operators: concrete examples}\label{sec:2D1}

In this subsection, we analyze geometric shapes of logical operators in three specific examples of two-dimensional STS models. Then, we analyze topological order arising in these models through geometric shapes of logical operators.

In Section~\ref{sec:2D11}, we begin by analyzing a two-dimensional classical ferromagnet, which is a model without topological order, and show that it has an anti-commuting pair of zero-dimensional and two-dimensional logical operators. In Section~\ref{sec:2D12}, we discuss the Toric code, an STS model with topological order, and see that it has anti-commuting pairs of one-dimensional logical operators. We demonstrate that various characterizations of topological order can be explained through geometric shapes of logical operators by reviewing the notion of topological order through analyses on the Toric code. In Section~\ref{sec:2D13}, we give another example of an STS model with anti-commuting pairs of one-dimensional logical operators and compare topological order arising in the model with topological order of the Toric code. 

\subsubsection{Two-dimensional classical ferromagnet}\label{sec:2D11}

We begin by presenting geometric shapes of logical operators in a two-dimensional classical ferromagnet. A two-dimensional classical ferromagnet is described by the following Hamiltonian:
\begin{align}
H \ = \ - \sum_{i,j} Z^{(i,j)} Z^{(i,j+1)} - \sum_{i,j} Z^{(i,j)}Z^{(i+1,j)}
\end{align}
where $Z^{(i,j)}$ represents the Pauli operator $Z$ acts on a qubit labeled by a vector $(i,j)$. The total number of qubits is $N=L_{1}\times L_{2}$, and the system has periodic boundary conditions. The classical ferromagnet can be viewed as a stabilizer code with translation symmetries. The stabilizer group is 
$\mathcal{S}_{L_{1},L_{2}}  =  \langle \{ Z^{(i,j)} Z^{(i,j+1)}, Z^{(i,j)}Z^{(i+1,j)} \}_{\forall i,j} \rangle$.
This stabilizer code is an STS model with $v=1$, satisfying scale symmetries since $k =1$ for any $L_{1}$ and $L_{2}$ with $G(\mathcal{S}_{L_{1},L_{2}} )= L_{1}L_{2}-1$. Logical operators are
\begin{align}
\ell \ = \ Z^{(1,1)} \ = \ 
      \begin{bmatrix}
          I      , & I      , & \cdots ,& I     \\
       \vdots  & \vdots  & \vdots & \vdots \\
      I   , & I  , & \cdots ,& I       \\
       Z      , & I      , & \cdots ,& I     
      \end{bmatrix}, \qquad 
r \ = \ \prod_{i,j}X^{(i,j)}  \ = \ 
      \begin{bmatrix}
       X   , & X      , & \cdots ,& X     \\
        X   , & X      , & \cdots ,& X     \\
       \vdots  & \vdots  & \vdots & \vdots \\
       X   , & X  , & \cdots ,& X       
      \end{bmatrix}.
\end{align}
One can see that both of logical operators satisfy the translation equivalence of logical operators. According to geometric shapes of logical operators, we may call $\ell$ a \emph{zero-dimensional logical operator} and $r$ a \emph{two-dimensional logical operator}.

As in a one-dimensional classical ferromagnet, a two-dimensional classical ferromagnet can support a GHZ state at zero temperature: $|\mbox{GHZ}\rangle = \frac{1}{\sqrt{2}}(|0\cdots 0\rangle + |1 \cdots 1\rangle)$. 
A GHZ state is stabilized by a two-dimensional logical operator $r$ since $r|\mbox{GHZ}\rangle = |\mbox{GHZ}\rangle$. A GHZ state is a globally entangled state since it is a superposition of two globally separated ground states $|0 \cdots 0 \rangle$ and $|1 \cdots 1 \rangle$. Properties of entanglement may be characterized in a way similar to a one-dimensional classical ferromagnet by using a mutual information $E(A:B)$ as we have discussed in Section~\ref{sec:1D1}. Note that the model does not have topological order.

\subsubsection{The Toric code as an STS model and topological order}\label{sec:2D12}

Next, let us discuss the Toric code~\cite{Kitaev97, Kitaev03}, which is a stabilizer code with topological order. We first show that the Toric code can be reduced to an STS model by introducing composite particles. Then, we describe geometric shapes of logical operators and show that one-dimensional logical operators form anti-commuting pairs. Finally, we demonstrate that topological order arising in the Toric code can be characterized by geometric shapes of logical operators. In particular, we discuss topological entanglement entropy~\cite{Hamma05}, anyonic excitations~\cite{Kitaev97} and the absence of local order parameters~\cite{Bravyi06, Bravyi10b} through logical operators. 

In the Toric code, qubits are defined on edges of a square lattice. The Hamiltonian is (Fig.~\ref{fig_2Dex2})
\begin{align}
H \ = \ - \sum_{s} \mathcal{A}_{s} - \sum_{p} \mathcal{B}_{p}, 
\qquad
\mathcal{A}_{s} \ = \ \prod_{\textbf{j} \in s} X^{(\textbf{j})}, \qquad \mathcal{B}_{p} \ = \ \prod_{\textbf{j} \in p} Z^{(\textbf{j})}.
\end{align}
Here, $s$ represent ``stars'' and $p$ represent ``plaquettes'' (Fig.~\ref{fig_2Dex2}), and $\textbf{j}$ represents the position of qubits. One may see that these interaction terms commute with each other, and the Hamiltonian is a stabilizer Hamiltonian. The Toric code has $k=2$, as seen from the following equations:
\begin{align}
\prod_{s} \mathcal{A}_{s} \ = \ I, \qquad \prod_{p} \mathcal{B}_{p} \ = \ I.
\end{align}
With some observation, one may see that there does not exist any other set of stabilizers whose product becomes an identity $I$. 

\begin{figure}[htb!]
\centering
\includegraphics[width=0.70\linewidth]{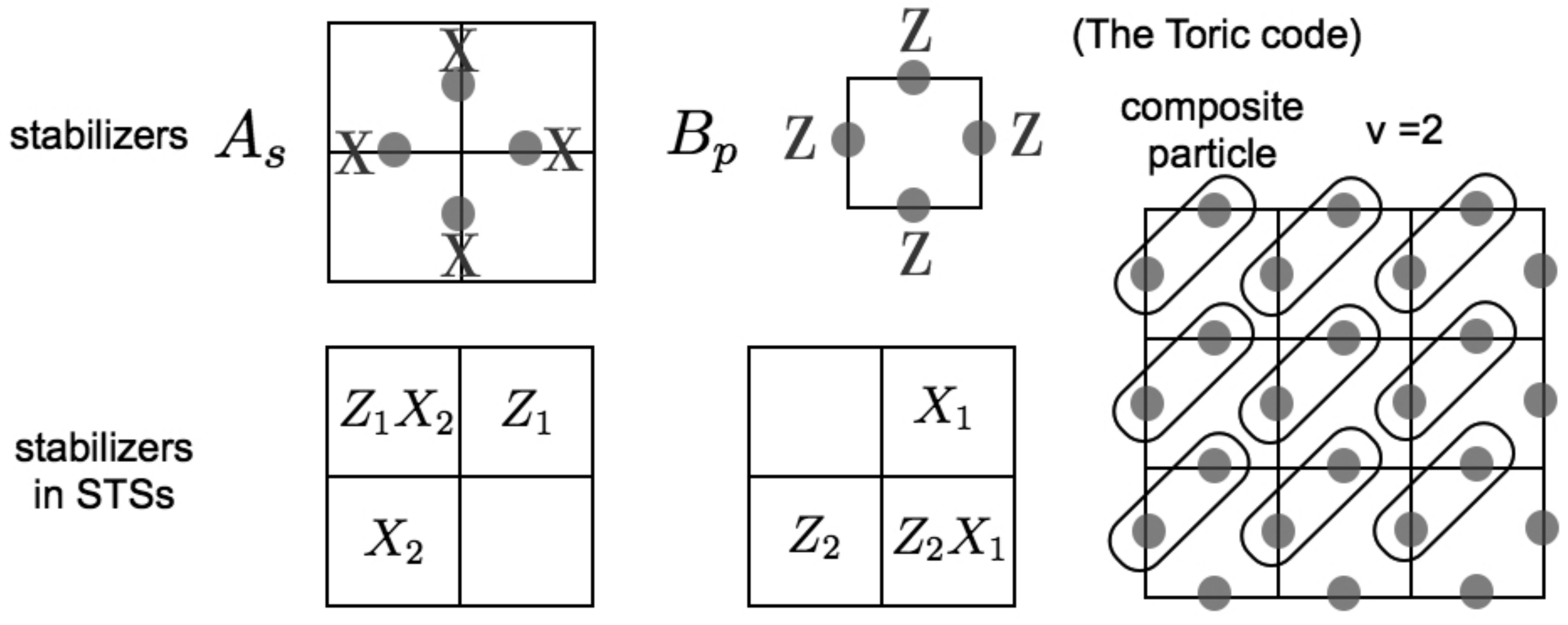}
\caption{Reduction of the Toric code to an STS model.
} 
\label{fig_2Dex2}
\end{figure}

\textbf{Reduction to an STS model:}
The Toric code can be reduced to an STS model by grouping two qubits into a composite particle ($v =2$) as shown in Fig.~\ref{fig_2Dex2}. By applying some appropriate unitary transformations on qubits inside each composite particle, interaction terms $\mathcal{A}_{s}$ and $\mathcal{B}_{p}$ can be represented in the following way (Fig.~\ref{fig_2Dex2}):
\begin{align}
&\mathcal{A}^{(i,j)} \ = \
Z_{1}^{(i,j)}X_{2}^{(i,j)} Z_{1}^{(i+1,j)}X_{2}^{(i,j+1)} \ = \
\left[
\begin{array}{cc}
X_{2}      ,&  I \\
Z_{1}X_{2} ,& Z_{1} 
\end{array}\right]^{(i,j)}\\
&\mathcal{B}^{(i,j)} \ = \ 
X_{1}^{(i+1,j)}Z_{2}^{(i,j+1)}Z_{2}^{(i+1,j+1)}X_{1}^{(i+1,j+1)} \ = \
\left[
\begin{array}{cc}
Z_{2} ,& X_{1}Z_{2} \\
I     ,& X_{1} 
\end{array}\right]^{(i,j)}.
\end{align}
Here, $Z_{1}^{(i,j)}$ represents a Pauli operator $Z_{1}$ acting on a composite particle labeled by $(i,j)$. $Z_{1}$, $Z_{2}$, $X_{1}$ and $X_{2}$ are single Pauli operators acting on a single composite particle. $2 \times 2$ matrices represent stabilizers graphically. One can see that this model satisfies scale symmetries by noticing 
$\prod_{i,j} \mathcal{A}^{(i,j)} = I$ and $\prod_{i,j} \mathcal{B}^{(i,j)} = I$.

Now, the two pairs of anti-commuting logical operators in the Toric code can be described in the following way:
\begin{align}
&\ell_{1}\ =\ \prod_{j} Z_{1}^{(1,j)} 
         \ = \ 
      \begin{bmatrix}
       Z_{1}      , & I      , & \cdots ,& I     \\
       \vdots  & \vdots  & \vdots & \vdots \\
       Z_{1}      , & I      , & \cdots ,& I     \\
      Z_{1}   , & I  , & \cdots ,& I       
      \end{bmatrix}, \qquad \quad \ \
\ell_{2}\ =\ \prod_{j} Z_{2}^{(1,j)} 
         \ = \ 
      \begin{bmatrix}
       Z_{2}      , & I      , & \cdots ,& I     \\
       \vdots  & \vdots  & \vdots & \vdots \\
       Z_{2}      , & I      , & \cdots ,& I     \\
      Z_{2}   , & I  , & \cdots ,& I       
      \end{bmatrix} \notag \\
&r_{1}\ =\ \prod_{i} X_{1}^{(i,1)} 
         \ = \ 
      \begin{bmatrix}
       I      , & I      , & \cdots ,& I     \\
       \vdots  & \vdots  & \vdots & \vdots \\
       I      , & I      , & \cdots ,& I     \\
      X_{1}   , & X_{1}  , & \cdots ,& X_{1}       
      \end{bmatrix}, \qquad
r_{2}\ =\ \prod_{i} X_{2}^{(i,1)} 
         \ = \ 
      \begin{bmatrix}
       I      , & I      , & \cdots ,& I     \\
       \vdots  & \vdots  & \vdots & \vdots \\
       I      , & I      , & \cdots ,& I     \\
       X_{2}   , & X_{2}  , & \cdots ,& X_{2}       
      \end{bmatrix}. \notag
\end{align}
Here, $n_{1} \times n_{2}$ matrices represent logical operators graphically where $n_{1}$ and $n_{2}$ are the numbers of composite particles in the $\hat{1}$ and $\hat{2}$ directions. According to geometric shapes of these logical operators, we may call them \emph{one-dimensional logical operators}.  

One can easily see that the translation equivalence of logical operators holds from the following equations:
\begin{align}
\prod_{j}\mathcal{A}^{(1,j)} \ = \ \ell_{1}T_{1}(\ell), \quad \prod_{i}\mathcal{A}^{(i,1)} \ = \ r_{2}T_{2}(r_{2}), \quad \prod_{j}\mathcal{B}^{(1,j)} \ = \ \ell_{2}T_{1}(\ell_{2}), \quad \prod_{i}\mathcal{B}^{(i,1)} \ = \ r_{1}T_{2}(r_{1}). \notag
\end{align}
One can also see that only the pairs of mutually orthogonal logical operators, whose intersections are zero-dimensional, may anti-commute with each other. Thus, pairs of anti-commuting logical operators always have zero-dimensional intersections. 

\textbf{Topological order and geometric shapes of logical operators:}
Having reduced the Toric code to an STS model, let us discuss topological order arising in the Toric code through one-dimensional logical operators. Currently, characterizations of topological order are ambiguous since there are several methods to identify the existence of topological order in ground states of correlated spin systems. In order to distinguish quantum phases with different kinds of topological order appropriately through geometric shapes of logical operators, it is a necessary first step to demonstrate that all the existing characterizations of topological order may be explained through logical operators. Here, we review these characterizations by analyzing topological entanglement entropy, anyonic excitations and the absence of local order parameters in the Toric code through geometric shapes of logical operators.

\textbf{(1) Topological entanglement entropy:}
Ground states of the Toric code have some global entanglement. This global entanglement can be quantified by an entanglement measure, called topological entanglement entropy~\cite{Kitaev06, Levin06}. Here, we follow the definition presented in~\cite{Levin06}. Consider four regions $A$, $B$, $C$ and $D$ described in Fig.~\ref{fig_topo_ent}. Then, topological entanglement entropy is defined as follows:
\begin{align}
S_{topo} \ = \ E_{B} + E_{C} - E_{A} - E_{D}
\end{align}
at the limit where regions $A$, $B$, $C$ and $D$ become infinitely large, and at the thermodynamic limit. Here, $E_{R}$ represents the entanglement entropy defined for a region of qubits $R$. It is known that the topological entanglement entropy $S_{topo}$ becomes non-zero when a ground state has topological order, while $S_{topo}$ becomes zero when a ground state does not have topological order. For simplicity of discussion, we set the width of $A$, $B$, $C$ and $D$ to be unity (see Fig~\ref{fig_topo_ent}). 

\begin{figure}[htb!]
\centering
\includegraphics[width=0.35\linewidth]{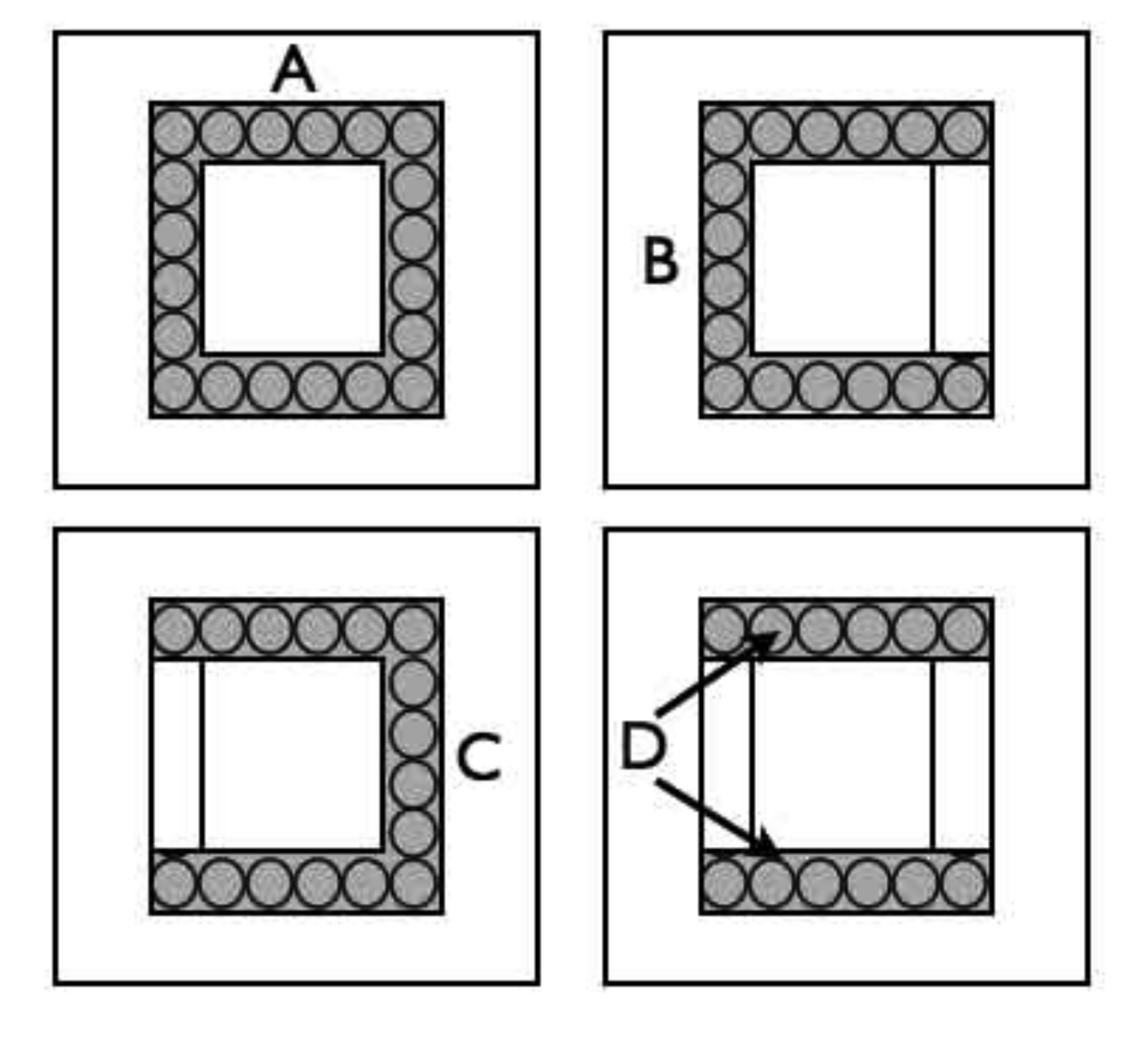}
\caption{Regions used in computing topological entanglement entropy. Each circle represents composite particles. The width of these regions is taken to be unity for simplicity of discussion.} 
\label{fig_topo_ent}
\end{figure}

Usually, the computation of topological entanglement entropy is challenging since we need to estimate the quantity at the thermodynamic limit. However, this quantity can be easily computed for ground states of stabilizer Hamiltonians by considering the numbers of generators for restrictions of the stabilizer groups into $A$, $B$, $C$ and $D$. Let us recall that the entanglement entropy can be written as follows: $E_{R} = V_{R} - G(\mathcal{S}_{R})$
when there is only a single ground state in the stabilizer Hamiltonian. Here, $V_{R}$ is the total number of qubits inside $R$ and $G(\mathcal{S}_{R})$ is the number of independent generators for the restriction of the stabilizer group, denoted as $\mathcal{S}_{R}$. This formula also holds for the cases where the Hamiltonian has degenerate ground states if there is no logical operator defined inside $R$. In the Toric code, there is no logical operator defined inside $A$, $B$, $C$, or $D$ as we shall show later in the discussion on the absence of local order parameters. Thus, we obtain a concise formula for the topological entanglement entropy:
\begin{align}
S_{topo} \ = \ G(\mathcal{S}_{A}) + G(\mathcal{S}_{D}) - G(\mathcal{S}_{B}) - G(\mathcal{S}_{C})
\end{align}
since $V_{A} + V_{D} = V_{B} + V_{C}$. 

Now, one can compute the topological entanglement entropy just by counting the number of stabilizers defined inside each region. However, most generators appear multiple times and do not contribute to the summation. For example, if there is a stabilizer $S_{j}$ defined inside $D$, $S_{j}$ can be also defined inside $A$, $B$ and $C$. Then, such $S_{j}$ does not contribute to the topological entanglement entropy. If there is a stabilizer $S_{j}$ is defined inside $B$ (but not inside $D$), $S_{j}$ can be also defined inside $A$. Such $S_{j}$ does not contribute to the topological entanglement entropy. Then, through some observations, one may notice that only the stabilizers defined inside $A$, but not defined either inside $B$, $C$ or $D$, contribute to the topological entanglement entropy. Thus, $S_{topo}$ is equal to the number of stabilizers which can be defined only inside $A$

Such stabilizers can be constructed from stabilizers $\mathcal{A}^{(i,j)}$ and $\mathcal{B}^{(i,j)}$ in the following way:
\begin{align}
\mathcal{A}(x,y) \ &= \ \prod_{i=1}^{x}\prod_{j=1}^{y} \mathcal{A}^{(i,j)}  \ = \ 
\begin{bmatrix}
 X_{2} ,& X_{2} ,& X_{2} ,& \cdots ,& X_{2} ,& X_{2} ,&    \\
   Z_{1}   &       &       &        &       &       & Z_{1}     \\
   Z_{1}   &       &       &        &       &       & Z_{1}     \\
  \vdots   &       &       &        &       &       & \vdots     \\       
   Z_{1}   &       &      &        &       &       & Z_{1}     \\       
   Z_{1}   &       &       &        &       &       & Z_{1}     \\       
     Z_{1}X_{2} ,& X_{2} ,& X_{2} ,& \cdots ,& X_{2} ,& X_{2} ,& Z_{1}    
\end{bmatrix} \label{eq:A_loop} \\
\mathcal{B}(x,y) \ &= \ \prod_{i=1}^{x}\prod_{j=1}^{y} \mathcal{B}^{(i,j)} \ = \
\begin{bmatrix}
   Z_{2}  ,& X_{1} ,& X_{1} ,& \cdots ,& X_{1} ,& X_{1} ,&     Z_{2}X_{1} \\
   Z_{2}   &       &       &        &       &       & Z_{2}     \\
   Z_{2}   &       &       &        &       &       & Z_{2}     \\
  \vdots   &       &       &        &       &       & \vdots     \\       
   Z_{2}   &       &      &        &       &       & Z_{2}     \\       
   Z_{2}   &       &       &        &       &       & Z_{2}     \\       
              & X_{1} ,& X_{1} ,& \cdots ,& X_{1} ,& X_{1} ,& X_{1}     
\end{bmatrix}.
\end{align}
Here, $x + 1 \times y + 1$ matrices represent stabilizers graphically. We have left the entries for identity operators blank for brevity of notation. Then, we notice that $S_{topo} = 2$ for the Toric code as proven in~\cite{Hamma05}. Though we have considered cases where $A$, $B$, $C$ and $D$ have the width of unity, it is immediate to generalize the above argument for the cases where the width is larger than one. Thus, the Toric code has non-zero topological entanglement entropy, which implies the existence of topological order. The connection between geometric shapes of logical operators and non-zero topological entanglement entropy will be clarified in the discussion of anyonic excitations, presented soon in the below. 

An important observation is that the topological entanglement entropy does not depend on the system size, and represents non-local correlations which are scale invariant. Thus, in classifying quantum phases, systems with different topological entanglement entropies must be distinguished appropriately.

\textbf{(2) Anyonic excitations:}
Topologically ordered systems are known to have a finite energy gap between degenerate ground states and excited states, which does not vanish even at the thermodynamic limit. Due to this finite energy gap, Hamiltonians may support quasiparticle excitations, which are called anyonic excitations. These anyonic excitations are localized (involving only a finite number of qubits inside some localized regions), and can be considered as particles, called anyons~\cite{Wilczek82, Wilczek82b}. When a system has a topological order, these anyons obey unusual quantum statistics. Here, we characterize properties of anyons under braiding in terms of geometric shapes of logical operators.

Anyonic excitations arising in the Toric code can be completely characterized by one-dimensional logical operators. 
A pair of anyons can be created by applying $Z_{1}^{(1,1)}$ to a ground state of the Toric code $|\psi\rangle$. Since $Z_{1}^{(1,1)}$ anti-commutes with $\mathcal{B}^{(n_{1},n_{2})}$ and $\mathcal{B}^{(n_{1},1)}$, $Z_{1}^{(1,1)}$ creates a pair of excitations by flipping eigenvalues of $\mathcal{B}^{(n_{1},n_{2})}$ and $\mathcal{B}^{(n_{1},1)}$ from $1$ to $-1$ (Fig.~\ref{fig_anyon}(a)). These excitations can be viewed as quasiparticles, in other words, anyons. Anyons can propagate around the torus in the $\hat{2}$ direction, and their propagations can be characterized by a logical operator $\ell_{1}$ and its segment (Fig.~\ref{fig_anyon}(a)):
\begin{align}
\ell_{1}\ = \  
      \begin{bmatrix}
       Z_{1}      , & I      , & \cdots ,& I     \\
       \vdots  & \vdots  & \vdots & \vdots \\
       Z_{1}      , & I      , & \cdots ,& I     \\
      Z_{1}   , & I  , & \cdots ,& I       
\end{bmatrix}, \qquad \ell_{1}(y_{1},y_{2}) \ \equiv \ \prod_{j=y_{1}}^{y_{2}}Z_{1}^{(1,j)}
\end{align}
where $1 \leq y_{1} \leq y_{2} \leq n_{2}$. Then, one notices that $\ell_{1}(y_{1},y_{2})$ creates excitations by flipping eigenvalues of $\mathcal{B}^{(n_{1},y_{1}-1)}$ and $\mathcal{B}^{(n_{1},y_{2})}$, as shown in Fig.~\ref{fig_anyon}(a). A pair of anyons can be also annihilated by applying $\ell_{1}(1,n_{2}) = \ell_{1}$. Thus, $\ell_{1}$ can characterize the creation of anyons and their propagations around the torus in the $\hat{2}$ direction, and their annihilation. 

Here, we call these anyons associated with stabilizers $\mathcal{B}^{(i,j)}$ ``b-type'' anyons. Though a logical operators $\ell_{1}$ can characterize propagations of b-type anyons only in the $\hat{2}$ direction, their propagations in the $\hat{1}$ direction can be characterized by a logical operator $r_{2}$. A pair of b-type anyons can be also created by applying $X_{2}^{(1,1)}$ to a ground state of the Toric code $|\psi\rangle$ since $X_{2}^{(1,1)}$ anti-commutes with $\mathcal{B}^{(n_{1},n_{2})}$ and $\mathcal{B}^{(1,n_{2})}$ (Fig.~\ref{fig_anyon}(a)). Their propagations in the $\hat{1}$ direction can be characterized by a logical operator $r_{2}$ and its segment (Fig.~\ref{fig_anyon}(a)):
\begin{align}
r_{2}\ = \  
      \begin{bmatrix}
      I      , & I      , & \cdots ,& I     \\
       \vdots  & \vdots  & \vdots & \vdots \\
      I      , & I      , & \cdots ,& I     \\
      X_{2}   , & X_{2}  , & \cdots ,& X_{2}       
\end{bmatrix}, \qquad r_{2}(x_{1},x_{2}) \ \equiv \ \prod_{i=x_{1}}^{x_{2}}X_{2}^{(i,1)}
\end{align}
 
A similar discussion holds for other one-dimensional logical operators $\ell_{2}$ and $r_{1}$ (Fig.~\ref{fig_anyon}(c)(d)). These logical operators characterize anyons associated with stabilizers $\mathcal{A}^{(i,j)}$. We call these anyons ``a-type'' anyons. Thus, we notice that there are two different kinds of anyons which can travel around the torus because of one-dimensional logical operators.

\begin{figure}[htb!]
\centering
\includegraphics[width=0.60\linewidth]{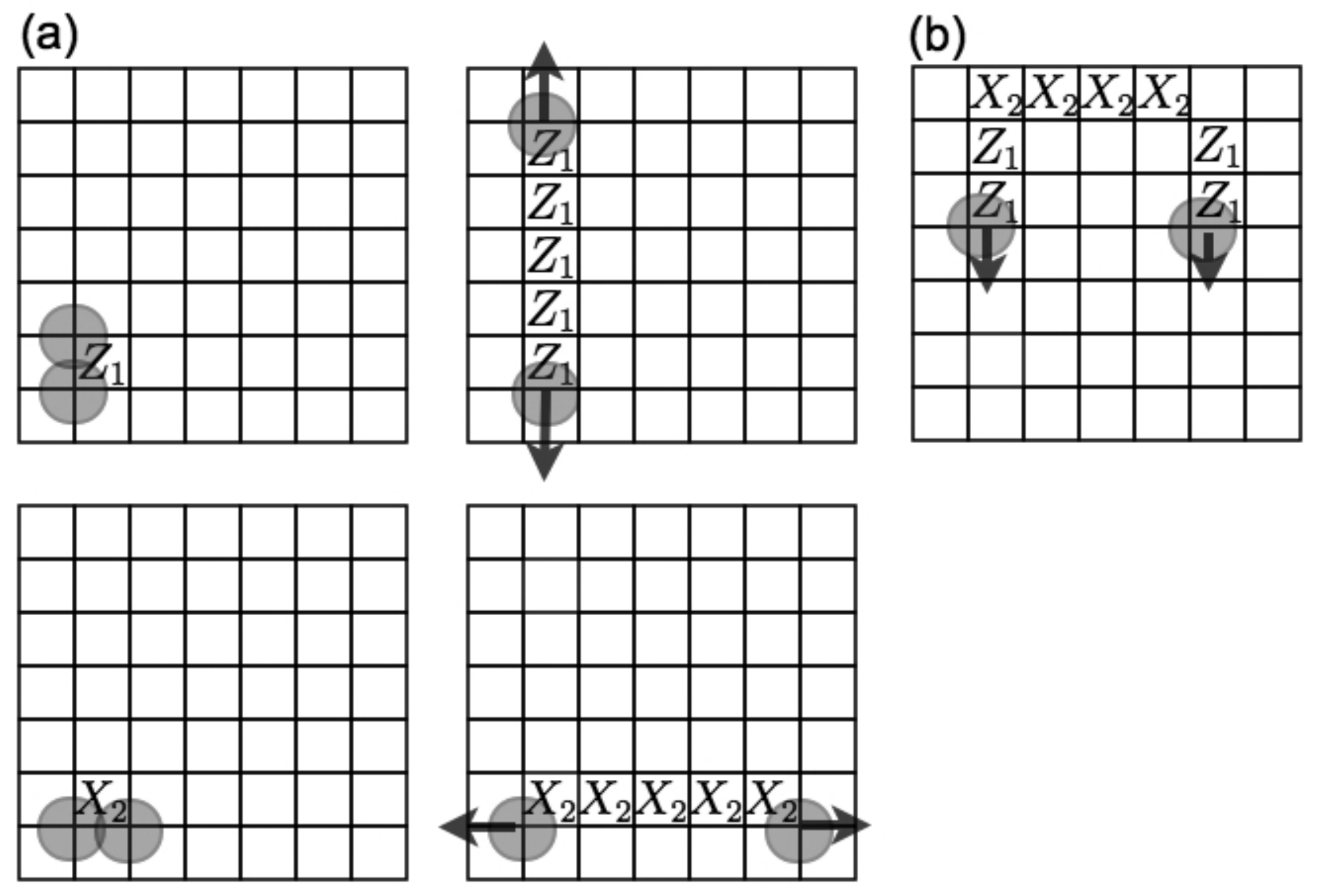}
\caption{(a) Propagations of anyons characterized by a segment of a one-dimensional logical operator $\ell_{1}$ and $r_{2}$. Shaded circles represent b-type anyons which are associated with stabilizers $\mathcal{B}^{(i,j)}$. (b) A segment of $\mathcal{A}(x,y)$ and propagations of b-type anyons.}
\label{fig_anyon}
\end{figure}

Though we have described propagations of anyons only in the vertical or horizontal direction in the torus, anyons can propagate around the torus freely and their propagations are not restricted to the vertical or horizontal direction. In order to discuss propagations which do not wind around the torus, stabilizers $\mathcal{A}(x,y)$ and $\mathcal{B}(x,y)$ become essential. Let us recall that $\mathcal{A}(x,y)$ and $\mathcal{B}(x,y)$ have geometric shapes like a loop. We consider the following segment of $\mathcal{A}(x,y)$:
\begin{align}
\begin{bmatrix}
 X_{2} ,& X_{2} ,& X_{2} ,& \cdots ,& X_{2} ,& X_{2} ,&    \\
   Z_{1}   &       &       &        &       &       & Z_{1}     \\
   Z_{1}   &       &       &        &       &       & Z_{1}      
\end{bmatrix}.
\end{align}
Then, we notice that this segment also creates a pair of b-type anyons at the endpoints of the segment (Fig.~\ref{fig_anyon}(b)). With some observations, we also notice that any segment of $\mathcal{A}(x,y)$ can create a pair of b-type anyons at the endpoints. Thus, $\mathcal{A}(x,y)$ can characterize the propagation of b-type anyons around the loop. A similar discussion also holds for $\mathcal{B}(x,y)$ where $\mathcal{B}(x,y)$ characterizes the propagation of a-type anyons around the loop.  

By combining $\mathcal{A}(x,y)$, $\mathcal{B}(x,y)$ and one-dimensional logical operators, one can characterize arbitrary propagations of anyons in the Torus. Here, let us analyze the connection between loop-like stabilizers $\mathcal{A}(x,y)$, $\mathcal{B}(x,y)$ and one-dimensional logical operators $\ell_{1}$, $\ell_{2}$, $r_{1}$ and $r_{2}$. Here, let us recall the form of $\mathcal{A}(x,y)$ described in Eq.~(\ref{eq:A_loop}).
Then, we notice that $\mathcal{A}(x,y)$ includes a segment of $\ell_{1}$ inside a vertex in the $\hat{2}$ direction and a segment of $r_{2}$ inside a vertex in the $\hat{1}$ direction. Thus, one can smoothly connect propagations of b-type anyons by attaching $\mathcal{A}(x,y)$ to $\ell_{1}$ and $r_{2}$. A similar observation holds for a-type anyons too.

Finally, let us show that braiding of a-type and b-type anyons create a non-trivial phase. One can braid an ``a-type'' anyon around a ``b-type'' anyons by using propagations characterized by $\mathcal{A}(x,y)$ and $\mathcal{B}(x,y)$ (Fig.~\ref{fig_anyon2}). Through the braiding, one obtains an extra phase factor ``$-1$'' as a result of anti-commutations at the intersection of propagation paths characterized by $\mathcal{A}(x,y)$ and $\mathcal{B}(x,y)$. This extra phase factor is a direct consequence of anti-commutations between one-dimensional logical operators. Thus, one-dimensional logical operators and loop-like stabilizers $\mathcal{A}(x,y)$ and $\mathcal{B}(x,y)$ characterize the existence of anyonic excitations with non-trivial braiding property. 

\begin{figure}[htb!]
\centering
\includegraphics[width=0.25\linewidth]{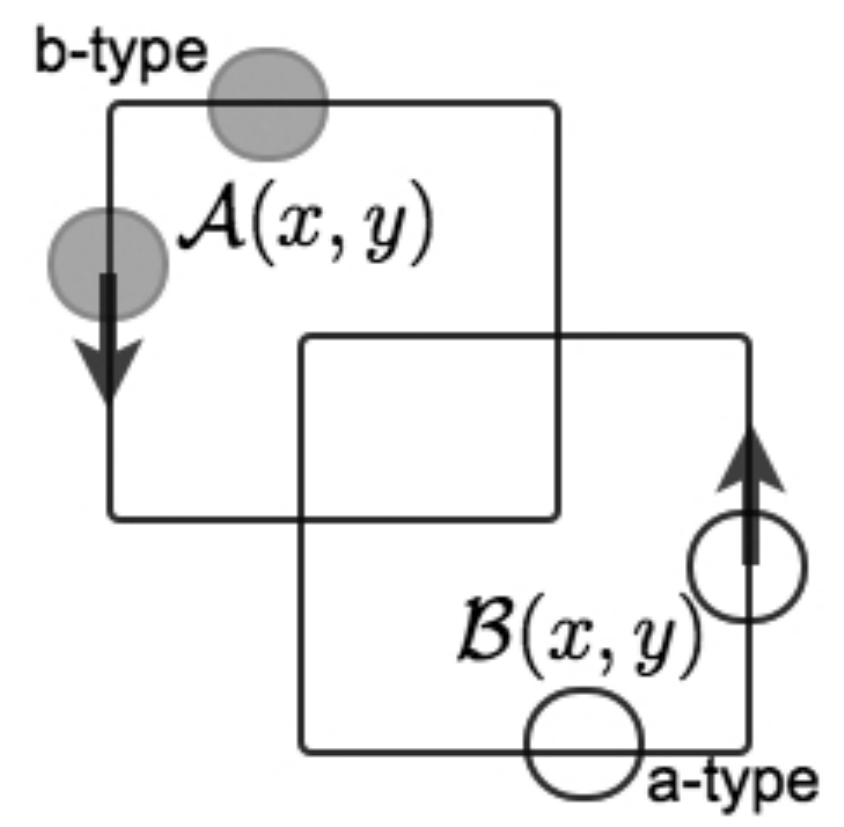}
\caption{Braiding between $a$-type and $b$-type anyons. Their propagations are characterized by $\mathcal{A}(x,y)$ and $\mathcal{B}(x,y)$.
} 
\label{fig_anyon2}
\end{figure}

Having analyzed anyons through one-dimensional logical operators, one can see the connection between anyonic excitations and topological entanglement entropy clearly. In the computation of topological entanglement entropy, we have seen that only the loop-like stabilizers $\mathcal{A}(x,y)$ and $\mathcal{B}(x,y)$ contribute to $S_{topo}$. This implies that the \emph{topological entanglement entropy counts the number of different types of anyons which may propagate around the loop $A$}. This observation is consistent with the fact that topological entanglement entropy is related to the total quantum dimension of anyons in the following way:
\begin{align}
S_{topo} \ = \ 2 \log D
\end{align}
where $D$ is the total quantum dimension of anyonic excitations. For the Toric code, we have $D=2$ since there are two different anyons.\footnote{One may say that anyonic excitations are local objects and topological order could be characterized locally. However, topological order refers to non-local correlations arising in ground states of topologically ordered systems at zero temperature, and anyonic excitations are signatures of the non-local ground state properties. While anyonic excitations are localized objects, their braiding properties still retain scale invariance since the non-trivial phase after the braiding does not depend on specific paths where anyons propagated.}

\textbf{(3) Local indistinguishability:}
An important feature of topologically ordered ground states is that there is no local order parameter which can distinguish topological phases. We have already seen that topological entanglement entropy is a global quantity which is computed at the limit where regions $A$, $B$, $C$ and $D$ become infinitely large and at the thermodynamic limit (Fig.~\ref{fig_topo_ent}). In general, globally defined quantities or objects are necessary in characterizing global entanglement arising in topologically ordered ground states. 

The characterization of topological order we discuss here is closely related to the fact that there is no local order parameter to characterize topological phases	. In a topologically ordered system, it is known that, for mutually orthogonal ground states $|\psi_{0}\rangle$ and $|\psi_{1}\rangle$ in a topologically ordered system, the following equations hold:
\begin{align}
\langle \psi_{0} | \hat{O} | \psi_{0} \rangle \ = \ \langle \psi_{1} | \hat{O} | \psi_{1} \rangle, \qquad \langle \psi_{0} | \hat{O} | \psi_{1} \rangle \ = \ 0 
\end{align}
for any locally defined physical observable $\hat{O}$ at the thermodynamic limit~\cite{Bravyi06, Bravyi10b}. Therefore, one cannot distinguish two different topologically ordered ground states through locally defined physical observables. We may call these conditions \emph{local indistinguishability conditions}. 


While this characterization of topological order is rather abstract, it is mathematically convenient in formulating topological order and is used commonly in rigorous treatments of topologically ordered systems. Here, we confirm that ground states of the Toric code satisfy the above indistinguishability conditions. A useful observation is that there is no local physical observable $\hat{O}$ which can change a ground state $|\psi_{0} \rangle$ to $|\psi_{1} \rangle$. In the language of quantum codes, this implies that there is no logical operator $\ell$ which can be defined locally. Then, we naturally expect that it is sufficient to show that there is no locally defined logical operators. 

The Toric code has two pairs of one-dimensional logical operators which are globally defined. Then, one might think that there is no local logical operator. However, these four logical operators are some particular representations, and  there exist many equivalent representations of the same logical operators since applications of stabilizers keep logical operators equivalent. Thus, one needs to show that one-dimensional logical operators in the Toric code do not have any equivalent representations which are defined locally. 

This can be easily proven by a theoretical tool developed for studies of a bi-partite entanglement in stabilizer codes (Theorem~\ref{theorem_partition} introduced in Section~\ref{sec:review2}). In particular, for a given bi-partition of a system of qubits into two complementary subsets $A$ and $B = \bar{A}$, we have
\begin{align}
g_{A} + g_{B} \ = \ 2k
\end{align}
where $g_{A}$ and $g_{B}$ are the numbers of independent logical operators which can be defined inside $A$ and $B$ respectively. Here, we take $A$ as the largest zero-dimensional region which consists of $n_{1}-1 \times n_{2}-1$ composite particles (Fig.~\ref{fig_distinguishability}). Since the entire system consists of $n_{1} \times n_{2}$ composite particles, $B$ is a union of one-dimensional regions which extend in the $\hat{1}$ and $\hat{2}$ directions. Now, since the Toric code has $k = 2$, we have $g_{A} + g_{B} = 4$. Then, since all the one-dimensional logical operators can be defined inside $B$, we have $g_{B}=4$, and $g_{A}= 0$. Thus, there is no logical operator defined inside a zero-dimensional region $A$. 

\begin{figure}[htb!]
\centering
\includegraphics[width=0.25\linewidth]{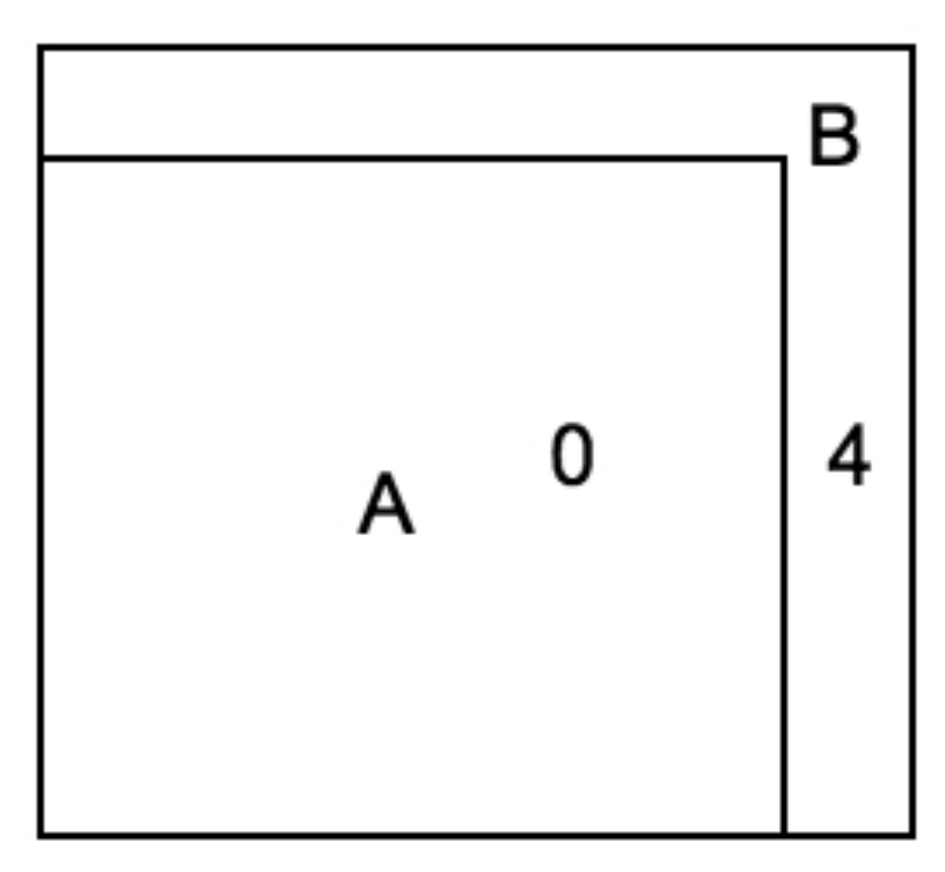}
\caption{A bi-partition into a zero-dimensional region $A$ and its complement $B$. Numbers shown indicate $g_{A}$ and $g_{B}$ which are the numbers of independent logical operators defined inside $A$ and $B$ respectively.
} 
\label{fig_distinguishability}
\end{figure}

Now, consider a physical observable $\hat{O}$ defined inside some localized region $R$. Since there is no logical operators inside $R$, the density matrices of all the ground states of the Toric code are the same: $\hat{\rho}_{R} = \frac{1}{2^{G(\mathcal{S}_{R})}} \prod_{S_{j}\in \mathcal{S}_{R}}(I+S_{j})$
where $S_{j}$ are independent generators for $\mathcal{S}_{R}$. Then, one can compute an expectation value of a local physical observable $\hat{O}$ defined inside $R$ in the following way:
\begin{align}
\langle \psi_{0} | \hat{O} | \psi_{0} \rangle \ = \ \langle \psi_{1} | \hat{O} | \psi_{1} \rangle \ = \ \mbox{Tr}\left[ \hat{O}\hat{\rho}_{R} \right].
\end{align}
One can also easily show that $\langle \psi_{0} | \hat{O} | \psi_{1} \rangle = 0$ by using the fact that there is no logical operator defined inside $R$. Thus, the Toric code satisfy the local indistinguishability conditions. 

Here, we emphasize that one-dimensional logical operators are indeed ``one-dimensional'' since they cannot be defined inside zero-dimensional regions. Therefore, the dimensions we have assigned to logical operators have topologically invariant meanings which are commonly shared by sets of all the equivalent logical operators~\cite{Beni10}. 

\textbf{A GHZ state, revisited:}
Here, we make some comments on the difference between a GHZ state and the Toric code. While both a GHZ state and a ground state of the Toric code have global entanglement, a GHZ state is not topologically ordered. This may be seen by the fact that a GHZ state does not satisfy the conditions \textbf{(1)} and \textbf{(3)}. While a classical ferromagnet can support localized excitations, these excitations cannot propagate around the lattice since the excitation energy increases as excitations move. As for the indistinguishability condition, a zero-dimensional logical operator $Z^{(1,1)}$ connects two ground states $\frac{1}{\sqrt{2}}(|0 \cdots 0 \rangle +|1\cdots1 \rangle )$ and $\frac{1}{\sqrt{2}}(|0\cdots 0 \rangle -|1\cdots1\rangle )$. Though a classical ferromagnet does not satisfy conditions \textbf{(1)} and \textbf{(3)}, a GHZ state has a non-zero topological entanglement entropy since $S_{topo}=1$. This is because topological entanglement entropy characterizes properties of a single ground state while characterizations of topological order need considerations on all the ground states.

\subsubsection{Another model with topological order}\label{sec:2D13}

In the Toric code, the existence of one-dimensional logical operators is responsible for topological order. Then, we expect that one-dimensional logical operators are central in characterizing topological phases. In order to confirm this expectation, let us give another STS model which has one-dimensional logical operators and discuss its physical properties.

We consider a system of $L_{1} \times L_{2}$ qubits governed by the following Hamiltonian:
\begin{align}
H \ &= \ - \sum_{i,j} S^{(i,j)} 
\end{align}
where
\begin{align}
S^{(i,j)} \  &= \   X^{(i-1,j)}X^{(i,j)}X^{(i+1,j)}Z^{(i,j-1)}Z^{(i,j+1)} \notag \\
     &= \   
\begin{bmatrix}
      & Z &  \\
X      & X & X \\
       & Z & 
\end{bmatrix}^{(i,j)}.
\end{align}
Here, blank entries represent identity operators $I$. One can see that interaction terms $S^{(i,j)}$ commute with each other. In a strict sense, this system is not an STS model since the number of logical qubits change according to $L_{1}$ and $L_{2}$. In particular, with some calculations, we have
\begin{align}
&k\ =\ 1 \qquad \mbox{for} \qquad L_{1} \ \not= \ 0 \ (\mbox{mod}\ 3), \quad L_{2} \ \not= \ 0 \ (\mbox{mod}\ 2)\\
&k\ =\ 2 \qquad \mbox{for} \qquad L_{1} \ = \ 0 \ (\mbox{mod}\ 3), \quad L_{2} \ \not= \ 0 \ (\mbox{mod}\ 2)\\
&k\ =\ 2 \qquad \mbox{for} \qquad L_{1} \ \not= \ 0 \ (\mbox{mod}\ 3), \quad L_{2} \ = \ 0 \ (\mbox{mod}\ 2)\\
&k\ =\ 4 \qquad \mbox{for} \qquad L_{1} \ = \ 0 \ (\mbox{mod}\ 3), \quad L_{2} \ = \ 0 \ (\mbox{mod}\ 2).
\end{align}
However, the model can be treated as an STS model when we view $3 \times 2$ qubits as a composite particle. Let us define $n_{1} = L_{1}/3$ and $n_{2} = L_{2}/2$ by choosing $L_{1}$ to be a multiple of $3$ and $L_{2}$ to be a multiple of $2$. Then, the system possesses scale symmetries since $k = 4$ for any $n_{1}$ and $n_{2}$. 

Now, let us represent stabilizers in this model in terms of composite particles. By applying appropriate unitary transformations on qubits inside each composite particle, we can represent stabilizers in the following way:
\begin{equation}
\begin{split}
&S_{1}^{(i,j)} \ = \
\begin{bmatrix}
Z_{1}        ,&  I      \\
Z_{1} X_{2}  ,& X_{2}  
\end{bmatrix}^{(i,j)}, \
S_{2}^{(i,j)} \ = \
\begin{bmatrix}
X_{1}        ,& X_{1}Z_{2}    \\
   I          ,& Z_{2}  
\end{bmatrix}^{(i,j)}, \
S_{3}^{(i,j)} \ = \
\begin{bmatrix}
Z_{3}        ,& I    \\
Z_{3}X_{4}   ,& X_{4}  
\end{bmatrix}^{(i,j)} \\
&S_{4}^{(i,j)} \ = \
\begin{bmatrix}
X_{3}        ,& X_{3}Z_{4}  \notag  \\ 
 I            ,& Z_{4}  
\end{bmatrix}^{(i,j)}, \ \
S_{5}^{(i,j)} \ = \
\begin{bmatrix}
X_{5}        ,& I    \\
X_{6}        ,& I   
\end{bmatrix}^{(i,j)}, \ 
S_{6}^{(i,j)} \ = \
\begin{bmatrix}
Z_{5}       ,& I    \\
Z_{6}        ,& I
\end{bmatrix}^{(i,j)}
\end{split}
\end{equation}
Then, logical operators can be easily found by looking at stabilizers $S_{1}^{(i,j)}$, $S_{2}^{(i,j)}$, $S_{3}^{(i,j)}$ and $S_{4}^{(i,j)}$ since they have forms similar to $\mathcal{A}^{(i,j)}$ and $\mathcal{B}^{(i,j)}$ in the Toric code. Logical operators are
\begin{align}\Pi(\mathcal{S}_{n_{1},n_{2}}) \ = \ \left\langle
\begin{array}{cccc}
\ell_{1}    ,& \ell_{2} ,& \ell_{3} ,& \ell_{4}   \\
r_{1}       ,& r_{2}    ,& r_{3}    ,& r_{4}
\end{array} \right\rangle
\end{align}
where
\begin{equation}
\begin{split}
\ell_{1} \ = \prod_{i} X_{1}^{(i,1)}, \quad \ell_{2} \ = \prod_{i} X_{2}^{(i,1)}, \quad \ell_{3} \ = \prod_{i} X_{3}^{(i,1)}, \quad \ell_{4} \ = \prod_{i} X_{4}^{(i,1)}\\
r_{1} \ = \prod_{j} Z_{1}^{(1,j)}, \quad r_{2} \ = \prod_{j} Z_{2}^{(1,j)}, \quad r_{3} \ = \prod_{j} Z_{3}^{(1,j)}, \quad r_{4} \ = \prod_{j} Z_{4}^{(1,j)}.
\end{split}
\end{equation}
Thus, the model has four pairs of anti-commuting one-dimensional logical operators.

Since the model has logical operators whose geometric shapes are similar to those in the Toric code, we naturally expect that the model has similar physical properties as those in the Toric code. Let us confirm this expectation by checking the conditions \textbf{(1)}, \textbf{(2)} and \textbf{(3)} one by one. First, topological entanglement entropy can be easily computed since stabilizers which can be defined only inside a loop-like region $A$ are:
\begin{equation}
\begin{split}
&S_{1}(x,y) \ \equiv \ \prod_{i=1}^{x}\prod_{j=1}^{y} S_{1}^{(i,j)}, \quad S_{2}(x,y) \ \equiv \ \prod_{i=1}^{x}\prod_{j=1}^{y} S_{2}^{(i,j)} \\ 
&S_{3}(x,y) \ \equiv \ \prod_{i=1}^{x}\prod_{j=1}^{y} S_{3}^{(i,j)}, \quad S_{4}(x,y) \ \equiv \ \prod_{i=1}^{x}\prod_{j=1}^{y} S_{4}^{(i,j)}.
\end{split}
\end{equation}
As a result of the existence of these loop-like stabilizers, this model has $S_{topo}=4$. Second, anyonic excitations arising in this model can be characterized by one-dimensional logical operators. For example, $\ell_{1}$ and $r_{2}$ characterize propagations of anyons associated with stabilizers $S_{2}^{(i,j)}$ in the $\hat{2}$ and $\hat{1}$ directions respectively. Also, $S_{1}(x,y)$ characterize propagations of these anyons around a loop-like region. In total, there are four kinds of different anyons, which may be called a-type b-type c-type and d-type anyons, which are characterized by ``$\ell_{1}$ and $r_{2}$'', ``$\ell_{2}$ and $r_{1}$'', ``$\ell_{3}$ and $r_{4}$'' and ``$\ell_{4}$ and $r_{3}$''. A braiding between a-type and b-type anyons and a braiding betwee c-type and d-type anyons create non-trivial phases $-1$. Other braidings do not generate any phase. Finally, from a similar reasoning used in the discussion on the Toric code, it can be shown that ground states of this model satisfy local indistinguishability conditions.

\begin{figure}[htb!]
\centering
\includegraphics[width=0.80\linewidth]{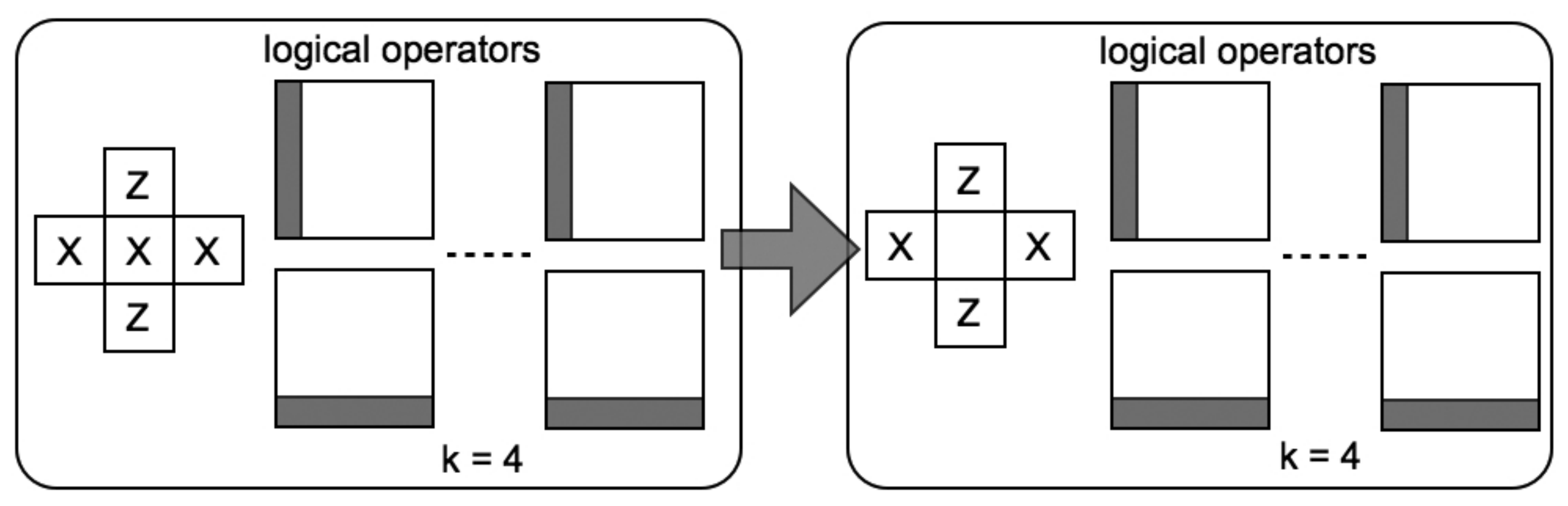}
\caption{Reduction to two copies of the Toric codes. Both models have one-dimensional logical operators with $k=4$.
} 
\label{fig_2Dex3}
\end{figure}

Now, since topological properties of ground states are similar to those of ground states in the Toric code, we naturally expect that this model and the Toric code may belong to the same quantum phase. In fact, the model can be reduced to two copies of the Toric code by applying local unitary transformations. From the forms of stabilizers $S_{1}^{(i,j)}, \cdots , S_{6}^{(i,j)}$, one may expect that stabilizers $S_{1}^{(i,j)}, \cdots , S_{4}^{(i,j)}$ are responsible for the existence of topological order. In particular, it can be seen that $5$th and $6$th qubits inside each composite particle are decoupled from other qubits, and are not relevant to topological order. In fact, one can remove $S_{5}^{(i,j)}$ and $S_{6}^{(i,j)}$ by applying disentangling operations between neighboring composite particles in a way similar to the reduction of a cluster state to a product state used in the discussion of one-dimensional STS models. $S_{1}^{(i,j)}, \cdots , S_{4}^{(i,j)}$ look the same as interaction terms $\mathcal{A}^{(i,j)}$ and $\mathcal{B}^{(i,j)}$ in the Toric code. Also, it can be seen that $1$st and $2$nd qubits are decoupled from $3$rd and $4$th qubits. Thus, one may see that two copies of the Toric code are embedded in this model in a non-interacting way. 

Here, in order to see the reduction of the model to the toric code in a more physically intuitive way, we give an actual model which includes two Toric codes:
\begin{align}\label{eq:two_Toric}
H' \ = \ - \sum_{i,j} 
\begin{bmatrix}
      & Z &  \\
X      &  & X \\
       & Z & 
\end{bmatrix}^{(i,j)}
\end{align}
with $N = L_{1}\times L_{2}$ qubits where $L_{1}$ and $L_{2}$ are even integers. Here, qubits at $(i,j)$ with $i + j = 0$ (mod $2$) is decoupled from qubits at $(i,j)$ with $i + j = 1$ (mod $2$). If we pick up qubits only at $(i,j)$ with $i + j = 0$ (mod $2$) and rotate the lattice by 45 degree, we have the following Hamiltonian:
\begin{align}
H'' \ = \ - \sum_{i,j} 
\begin{bmatrix}
Z     & X  \\
X     & Z  
\end{bmatrix}^{(i,j)}.
\end{align}
Through some observations, we may notice that this model is equivalent to the Toric code. Thus, the Hamiltonian in Eq.~(\ref{eq:two_Toric}) includes two copies of the Toric code in a non-interacting way. 

In summarizing observations obtained so far, the following two Hamiltonians are equivalent;
\begin{align}
H \ = \ - \sum_{i,j} 
\begin{bmatrix}
      & Z &  \\
X      & X  & X \\
       & Z & 
\end{bmatrix}^{(i,j)} \ \sim \
H' \ = \ - \sum_{i,j} 
\begin{bmatrix}
      & Z &  \\
X      &  & X \\
       & Z & 
\end{bmatrix}^{(i,j)}.
\end{align}
Since these two Hamiltonians can be transformed each other through local unitary transformations, they belong to the same quantum phase. This can be shown through a similar discussion used in the classification of quantum phases in one-dimensional STS models. Here, we note that local unitary transformations do not change the geometric shapes of logical operators, and thus, logical operators indeed characterize the scale invariant physical properties. 

\subsection{Logical operators in two-dimensional STS models and topological order}\label{sec:2D2}

We have analyzed possible geometric shapes of logical operators in two-dimensional STS models through several examples, and found that logical operators form anti-commuting pairs in the following way:
\begin{align}
&\mbox{Classical Ferromagnet ($k \ = \ 1$)} &\mbox{$0$-dim}\ - \ \mbox{$2$-dim} \notag \\
&\mbox{The Toric code ($k \ = \ 2$)}        &\mbox{$1$-dim}\ - \ \mbox{$1$-dim} \notag
\end{align}
We have also seen that one-dimensional logical operators may be responsible for the existence of topological order. 

Here, we wish to establish the relation between physical properties of arbitrary two-dimensional STS models and geometric shapes of logical operators. However, it is generally difficult to characterize physical properties in higher-dimensional systems ($D > 1$) both analytically and computationally. In particular, since there are many possibilities for geometric shapes of logical operators, one might think that finding logical operators are much more difficult in two-dimensional STS models than those in one-dimensional STS models.

Fortunately, it is possible to determine geometric shapes of logical operators in two-dimensional STS models completely. Moreover, we can characterize topological order arising in two-dimensional STS models completely from geometric shapes of logical operators. In this subsection, we describe possible geometric shapes of logical operators. In particular, we show that possible shapes are ``anti-commuting pairs of zero-dimensional and two-dimensional logical operators'' and ``anti-commuting pairs of one-dimensional logical operators''. Based on these logical operators, we discuss physical properties of arbitrary STS models. While we shall concentrate on presenting geometric shapes of logical operators and discussing topological order arising in two-dimensional STS models, all the derivations of logical operators are given in~\ref{sec:proof}.

In Section~\ref{sec:2D21}, we present possible geometric shapes of logical operators. In Section~\ref{sec:2D22}, we discuss topological order arising in two-dimensional STS models through geometric shapes of logical operators. 

\subsubsection{Geometric shapes of logical operators in two-dimensional STS models}\label{sec:2D21}

We present possible geometric shapes of logical operators. In two-dimensional STS models, possible geometric shapes are ``anti-commuting pairs of zero-dimensional and two-dimensional logical operators'' and ``anti-commuting pairs of one-dimensional logical operators'' (Fig.~\ref{fig_theorem_2dim}). Here, we describe the above statement more precisely. Consider a two-dimensional STS model with $n_{1}\times n_{2}$ composite particles, $k$ logical qubits and $v$ qubits inside each composite particle, which is defined with the stabilizer group $\mathcal{S}_{n_{1},n_{2}}$. A region with $x \times y$ composite particles is denoted as $U_{x,y}$, which includes composite particles $P_{r_{1},r_{2}}$ with $1\leq r_{1} \leq x$ and $1 \leq r_{2} \leq y$. Then, for $n_{1},n_{2} >2v$, there exist a canonical set of logical operators
\begin{align}
\Pi(\mathcal{S}_{n_{1},n_{2}})\ =\ \left\{
\begin{array}{cccccc}
\ell_{1} ,& \cdots ,& \ell_{k} \\
r_{1} ,& \cdots ,& r_{k}  
\end{array}
\right\}
\end{align}
which satisfy the following conditions after some appropriate local unitary transformations on qubits inside each composite particle:
\begin{itemize}
\item \textbf{Zero-dimensional and two-dimensional logical operators:} For $p = 1 , \cdots, k_{0}$ ($k_{0}\leq k$), $\ell_{p}$ are logical operators defined inside a region $U_{2v,1}$, and $r_{p}$ are logical operators defined all over the lattice in a periodic way: 
\begin{align}
r_{p}\ =\ \prod_{i,j} X_{p}^{(i,j)} 
     \  =\ 
      \begin{bmatrix}
      X_{p}  , & X_{p}  , & \cdots ,& X_{p}\\
      \vdots   & \vdots   & \vdots  &  \vdots \\
      X_{p}  , & X_{p}  , & \cdots ,& X_{p}\\
      X_{p}  , & X_{p}  , & \cdots ,& X_{p}
      \end{bmatrix}. \notag
      \end{align}
\item \textbf{One-dimensional logical operators:} For $p = k_{0}+1 , \cdots, k$, $\ell_{p}$ are logical operators defined inside a region with $n_{1}\times 1$ composite particles and $r_{p}$ are logical operators defined inside a region with $1\times n_{2}$ composite particles:
\begin{align}
\ell_{p}\ =\ \prod_{j} Z_{p}^{(1,j)} 
         \ = \ 
      \begin{bmatrix}
       Z_{p}      , & I      , & \cdots ,& I     \\
       \vdots  & \vdots  & \vdots & \vdots \\
       Z_{p}      , & I      , & \cdots ,& I     \\
       Z_{p}   , & I  , & \cdots ,& I       
      \end{bmatrix}, \quad
r_{p}\ =\ \prod_{i} X_{p}^{(i,1)} 
         \ = \ 
      \begin{bmatrix}
       I      , & I      , & \cdots ,& I     \\
       \vdots  & \vdots  & \vdots & \vdots \\
       I      , & I      , & \cdots ,& I     \\
       X_{p}   , & X_{p}  , & \cdots ,& X_{p}       
      \end{bmatrix}.\notag 
      \end{align}
\end{itemize}

\begin{figure}[htb!]
\centering
\includegraphics[width=0.7\linewidth]{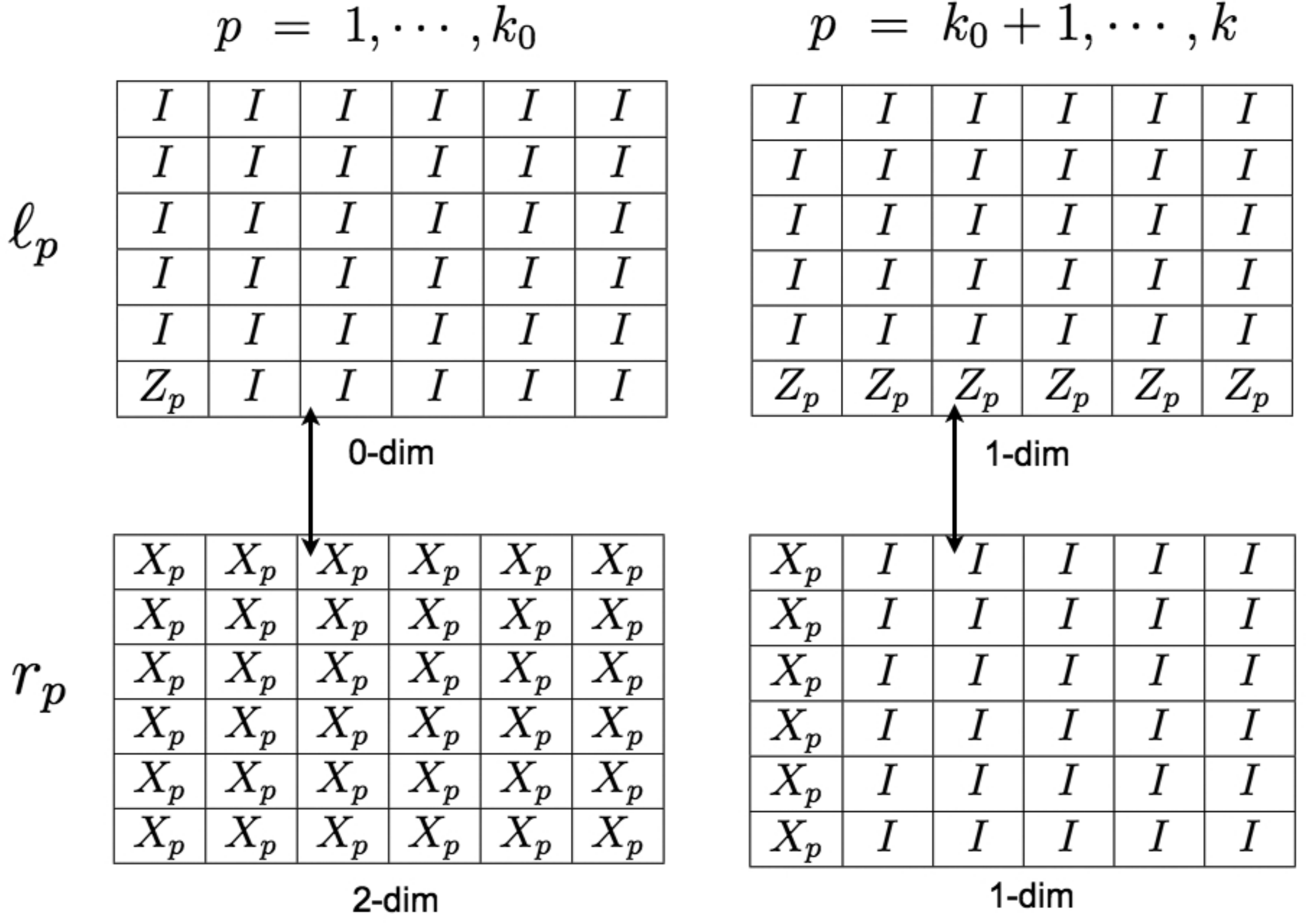}
\caption{A canonical set of logical operators in a two-dimensional STS model. Two-sided arrows represent anti-commutations. Note that zero-dimensional logical operators $\ell_{p}$ ($p = 1 , \cdots, k_{0}$) are represented inside a single composite particle instead of a region with $2v \times 1$ composite particles ($U_{2v,1}$) for ease of graphical representations. Dimensions are assigned to logical operators according to their geometric shapes. 
} 
\label{fig_theorem_2dim}
\end{figure}

Computations of this canonical set of logical operators are presented in~\ref{sec:proof}. Logical operators are graphically shown in Fig.~\ref{fig_theorem_2dim}. Here, we have assigned dimensions to logical operators according to their geometric shapes. Note that zero-dimensional logical operators are ``not completely zero-dimensional'' since they are defined inside a region with $2v \times 1$ composite particles which is denoted as $U_{2v,1}$. 
For ease of graphical representations, we draw logical operators $\ell_{p}$ ($p = 1 , \cdots, k_{0}$) inside a region $U_{2v,1}$ by assuming that they can be actually supported by a single composite particle $P_{1,1}$ (or a region $U_{1,1}$) in Fig.~\ref{fig_theorem_2dim}. 

Here, we notice that the sum of dimensions of logical operators which form anti-commuting pairs is always two. Also, intersections of anti-commuting logical operators are always zero-dimensional. We call this constraint on the dimensions of logical operators the \emph{dimensional duality of logical operators} in STS models. This duality also holds for one-dimensional STS models. We summarize this observation in the following way.\footnote{We note that there are several models of stabilizer Hamiltonians constructed in higher-dimensional systems which satisfy this dimensional duality of logical operators. For example, generalizations of the Toric code defined on a $D$-dimensional torus ($D>2$) have anti-commuting pairs of $m$-dimensional and $(D-m)$-dimensional logical operators~\cite{Dennis02, Takeda04}. However, there exists a three-dimensional stabilizer code without scale symmetries which may have a pair of anti-commuting two-dimensional logical operators~\cite{Bravyi10c}.}

\begin{observation}[Dimensional duality of logical operators] 
In $D$-dimensional STS models ($D = 1,2$), $m$-dimensional and $(D-m)$-dimensional logical operators form anti-commuting pairs.
\end{observation}

\subsubsection{Topological order in STS models}\label{sec:2D22}

From the discussions in the previous subsection, one may naturally expect that an STS model has topological order when there exist one-dimensional logical operators. For example, one-dimensional logical operators may characterize propagations of anyons in the vertical and horizontal directions. However, whether there exist anyons which can propagate around the torus is not immediately clear. Also, computations of the topological entanglement entropy are not possible through only logical operators. 

What are central in characterizing topological order in the Toric code are loop-like stabilizers $\mathcal{A}(x,y)$ and $\mathcal{B}(x,y)$ which are generated from stabilizers $\mathcal{A}^{(i,j)}$ and $\mathcal{B}^{(i,j)}$. Then, we hope to find such stabilizers and loop-like stabilizers generated from them. In fact, there exists such pairs of stabilizers which can characterize propagation of anyons around a loop. See~\ref{sec:loop} for derivations and detailed discussions. Here, based on these stabilizers, we discuss topological order arising in two-dimensional STS models. 

Let us begin by the case with $k_{1} = 2$. Then, after some appropriate unitary transformations on qubits inside each composite particle, there exist the following pairs of stabilizers:
\begin{align}
S_{1}^{(i,j)} \ = \ 
\left[
\begin{array}{cc}
P_{1}Z_{1}X_{2} ,& P_{1}Z_{1} \\
P_{1}X_{2}      ,& P_{1} 
\end{array}\right]^{(i,j)}, \qquad
S_{2}^{(i,j)} \ = \ 
\left[
\begin{array}{cc}
P_{2}     ,& P_{2}X_{1} \\
P_{2}Z_{2} ,& P_{2}Z_{2}X_{1} 
\end{array}\right]^{(i,j)}
\end{align}
where $P_{1}$ and $P_{2}$ are some Pauli operators acting on a single composite particle. Here, one may notice that $S_{1}^{(i,j)} = \mathcal{A}^{(i,j)}$ and $S_{2}^{(i,j)} = \mathcal{B}^{(i,j)}$ if $P_{1}=P_{2}=I$. 

Logical operators can be represented in the following way:
\begin{equation}
\begin{split}
\ell_{1}\ =\ \prod_{j} Z_{1}^{(1,j)} \quad
\ell_{2}\ =\ \prod_{j} Z_{2}^{(1,j)} \quad
r_{1}\ =\ \prod_{i} X_{1}^{(i,1)} \quad
r_{2}\ =\ \prod_{i} X_{2}^{(i,1)}. 
\end{split}
\end{equation}
From $S_{1}^{(i,j)}$ and $S_{2}^{(i,j)}$, one can construct the following loop-like stabilizers:
\begin{align}
&S_{1}(x,y) \ \equiv \ \prod_{i=1}^{x}\prod_{j=1}^{y} S_{1}^{(i,j)}, \quad S_{2}(x,y) \ \equiv \ \prod_{i=1}^{x}\prod_{j=1}^{y} S_{2}^{(i,j)}.
\end{align}
Note that contributions from $P_{1}$ and $P_{2}$ cancel out inside the loop.

Next, let us consider the case with arbitrary $k_{1}$ where $k_{1}$ is the number of pairs of one-dimensional logical operators. For simplicity, let us consider the case where $k_{1}=k$ and $k_{0}=0$. One can prove that $k_{1}$ must be even, and there exist following stabilizers:
\begin{align}
&S_{1}^{(i,j)} \ = \ 
\left[
\begin{array}{cc}
P_{1}Z_{1}X_{2} ,& P_{1}Z_{1} \\
P_{1}X_{2}      ,& P_{1} 
\end{array}\right]^{(i,j)}, \ \cdots, \ 
S_{k-1}^{(i,j)} \ = \ 
\left[
\begin{array}{cc}
P_{k-1}Z_{k-1}X_{k} ,& P_{k-1}Z_{k-1} \\
P_{k-1}X_{k}      ,& P_{k-1} 
\end{array}\right]^{(i,j)} \notag \\
&S_{2}^{(i,j)} \ = \ 
\left[
\begin{array}{cc}
P_{2}     ,& P_{2}X_{1} \\
P_{2}Z_{2} ,& P_{2}Z_{2}X_{1} 
\end{array}\right]^{(i,j)}, \  \cdots, \
S_{k}^{(i,j)} \ = \ 
\left[
\begin{array}{cc}
P_{k}     ,& P_{k}X_{k-1} \\
P_{k}Z_{k} ,& P_{k}Z_{k}X_{k-1} 
\end{array}\right]^{(i,j)} \notag
\end{align}
where $P_{1},\cdots, P_{k}$ are some Pauli operators acting on a single composite particle. One may see that stabilizers form pairs just like a pair of $\mathcal{A}^{(i,j)}$ and $\mathcal{B}^{(i,j)}$ in the Toric code. Based on these stabilizers, one can find loop-like stabilizers too.

With these observations, we notice that there are $k_{1}$ different types of anyonic excitations which are characterized by pairs of ``$S_{1}^{(i,j)}$ and $S_{2}^{(i,j)}$'', ``$S_{3}^{(i,j)}$ and $S_{4}^{(i,j)}$'' and so on. These anyons described by each pair of stabilizers are not interacting with each other. Thus, the braiding rule between anyons is described as a direct product of $k_{1}/2$ braiding rules of the Toric codes. Topological entanglement entropy is $k_{1}$, which can be directly computed from the number of loop-like stabilizers.\footnote{In order to prove this rigorously, one needs to show that there is no other stabilizer which can be defined inside a loop-like region. This proof is not trivial, but can be obtained by using discussions similar to proofs of lemma~\ref{lemma_l} which appear in~\ref{sec:proof1}.}

Finally, let us see whether STS models satisfy local indistinguishability conditions or not. Here, we take the largest zero-dimensional region $R = U_{n_{1}-1,n_{2}-1}$ which consists of $n_{1}-1 \times n_{2}-1$ composite particles. Then, we notice that 
$g_{R} = k_{0}$  
since $g_{\bar{R}} = k_{0} + 2 k_{1}$ where $k_{0}+k_{1}=k$. Thus, when $k_{0}=0$, the model satisfies local indistinguishability conditions.

From these discussions, we may conclude that topological properties of two-dimensional STS models can be completely characterized from geometric shapes of logical operators. Thus, geometric shapes of logical operators may work as order parameters to distinguish topological phases with different topological order arising in STS models. 

\subsection{Quantum phases in two-dimensional STS models}\label{sec:2D3}

We have seen that different geometric shapes of logical operators lead to completely different ground state properties in STS models. Finally, let us classify quantum phases arising in two-dimensional STS models.

Since a topologically ordered system and a non-topologically ordered system have ground states with completely different physical properties, we naturally expect that topological and non-topological STS models belong to different quantum phases. However, it is not immediately clear whether the emergence and loss of topological order lead to the existence of a QPT or not. In particular, while changes in the number of ground states (the number of logical operators) must lead to a QPT, it is still not clear whether there is a QPT or not between two systems with the same number of ground states. 

Here, we show that geometric shapes of logical operators can be used as ``order parameters'' to distinguish quantum phases, including topological phases, arising in two-dimensional STS models by proving that any parameterized Hamiltonians connecting two STS models with different geometric shapes of logical operators are always separated by QPTs. We also show that the existence of a QPT originates from the non-analyticity of transforming logical operators with topologically distinct shapes each other since there is no continuous deformation (diffeomorphism) between them at the thermodynamic limit.

In Section~\ref{sec:2D31}, we show that two STS models with topologically different logical operators are always separated by a QPT. 
In Section~\ref{sec:2D32}, we expand our discussions on topological structures of logical operators to sets of equivalent logical operators. 

\subsubsection{The vanishing energy gap and topological QPT}\label{sec:2D31}

Let us begin by clarifying the problem we address. Consider two STS models $H_{A}$ and $H_{B}$ with and without topological order. When the number of logical operators in $H_{A}$ and $H_{B}$ are different, there must always be a QPT between $H_{A}$ and $H_{B}$. Therefore, we consider the cases where $H_{A}$ and $H_{B}$ have the same number of logical operators, but different geometric shapes of logical operators. In particular, for simplicity of discussion, we consider the case where $H_{A}$ has two pairs of anti-commuting zero-dimensional and two-dimensional logical operators ($k_{0}=2$, $k_{1}=0$ and $k=2$) while $H_{B}$ has two pairs of anti-commuting one-dimensional logical operators ($k_{0}=0$, $k_{1}=2$ and $k=2$). One can generalize this discussion to the cases with arbitrary numbers of logical operators easily.

Our goal is to show that $H_{A}$ and $H_{B}$ belong to different quantum phases. For this purpose, we suppose the opposite; there exists a parameterized Hamiltonian $H(\epsilon)$ which connects $H_{A}$ and $H_{B}$ without closing the energy gap or changing the number of ground states:
\begin{align}
H(0) \ = \ H_{A}, \qquad H(1) \ = \ H_{B}, \qquad \Delta(\epsilon) \ \geq \ c \quad \mbox{for all} \ \epsilon \label{eq:assumption}
\end{align}
where $\Delta(\epsilon)$ is the energy gap and $c$ is some constant which does not depend on the system size. 

The most striking difference between $H_{A}$ and $H_{B}$ is the presence and the absence of topological order which can be seen by different geometric shapes of their logical operators. This drastic change of geometric shapes may underlie a non-analytic change of the ground state properties and a vanishing energy gap. Then, one might hope to show the existence of a QPT by looking at how geometric shapes of logical operators change with a parameter change. However, unlike $H_{A}$ and $H_{B}$, one cannot define logical operators inside the Pauli group for $H(\epsilon)$ at $0 < \epsilon < 1$ since the Hamiltonian $H(\epsilon)$ is frustrated in general. 

Despite $H(\epsilon)$ is frustrated and not a stabilizer Hamiltonian in general, it is possible to define operators which act like logical operators. Let us denote projection operators onto the ground state space of $H(\epsilon)$ as $\hat{P}(\epsilon)$. Since the number of degenerate ground states does not change with $\epsilon$, one can find some unitary transformation which satisfy the following condition:
\begin{align}
\hat{P}(\epsilon) \ = \ U(\epsilon)\hat{P}_{A}U(\epsilon)^{-1}, \qquad \hat{P}_{A} \ \equiv \ \hat{P}(0), \quad \hat{P}_{B} \ \equiv \ \hat{P}(1) \label{eq:adiabatic}
\end{align}
where $\hat{P}_{A}$ and $\hat{P}_{B}$ represent projections onto the ground state spaces of $H_{A}$ and $H_{B}$ respectively. Here, we note that such a unitary transformation is not uniquely determined and has many degrees of freedom. The unitary transformation $U(\epsilon)$ may characterize the evolution of ground states with respect to $\epsilon$. Let us pick up some ground state $|\psi(0)\rangle$ of $H_{A}$ at $\epsilon=0$. Then, the following state $|\psi(\epsilon)\rangle =  U(\epsilon) |\psi(0)\rangle$
is a ground state of $H(\epsilon)$. Here, we denote an anti-commuting pair of zero-dimensional and two-dimensional logical operators in $H_{A}$ as $\ell_{A}$ and $r_{A}$. Then, the following operators act like logical operators inside the ground state space of $H(\epsilon)$, transforming degenerate ground states among them:
\begin{align}
\ell(\epsilon) \ = \ U(\epsilon)\ell_{A}U(\epsilon)^{-1}, \qquad r(\epsilon) \ = \ U(\epsilon)r_{A}U(\epsilon)^{-1}
\end{align}
since they act non-trivially inside the ground state space:
\begin{align}
\{ \ell(\epsilon), r(\epsilon) \} \ = \ 0, \qquad \ell(\epsilon)^{2} \ = \ r(\epsilon)^{2} \ = \ I.
\end{align}
Note that $\ell(0) = \ell$ and $r(0) = r$. We call $\ell(\epsilon)$ and $r(\epsilon)$ \emph{analytically continued logical operators}. Note that $\ell(\epsilon)$ and $r(\epsilon)$ may not commute with the parameterized Hamiltonian $H(\epsilon)$.

Now, let us discuss how analytically continued logical operators $\ell(\epsilon)$ change with respect to $\epsilon$. Here, we first give an intuitive explanation why the assumption of no QPT leads to a contradiction. Then, we add some mathematical rigor to our intuition. 

\textbf{Intuitive approach based on topology in logical operators:}
At the beginning of the discussion, we have assumed that there is no QPT between $\epsilon = 0$ and $\epsilon = 1$. Then, we naturally hope that the geometric shapes of $\ell(\epsilon)$ change smoothly with respect to $\epsilon$. A unitary transformation $U(\epsilon)$ must not have any non-analyticity with respect to $\epsilon$ even at the thermodynamic limit since a ground state $|\psi(\epsilon)\rangle$ must change continuously with $\epsilon$. Then, $\ell(\epsilon)$ also changes continuously with respect to $\epsilon$ without non-analyticity. At $\epsilon = 0$, a zero-dimensional logical operator $\ell(0) \equiv \ell_{A}$ is defined inside some localized region (Fig.~\ref{fig_2DQPT}). For $\epsilon$ sufficiently close to zero, $\ell(0)$ and $\ell(\epsilon)$ are similar, and $\ell(\epsilon)$ may be approximated as: 
\begin{align}
\ell(\epsilon) \approx \ell(0) + \epsilon \ell(0)'
\end{align}
where $\ell(0)'$ represents the derivative of $\ell(\epsilon)$ with respect to $\epsilon$ at $\epsilon = 0$. Then, the geometric shape of $\ell(\epsilon)$ should be similar to the one of $\ell(0)$ (Fig.~\ref{fig_2DQPT}). 

Now, let us analyze geometric shapes of $\ell(\epsilon)$ at $\epsilon = 1$. At $\epsilon = 1$, the analytically continued logical operators $\ell(1)$ and $r(1)$ act non-trivially inside the ground state space of $H_{B}$. Let us recall that there are no logical operators inside any zero-dimensional region $R$ in an STS model $H_{B}$: $g_{R}=0$ where $R$ is a region consisting of $(n_{1}-1) \times (n_{2}-1)$ composite particles. As a result, ground states of $H_{B}$ satisfy the local indistinguishability conditions, and there is no local operator which can transform ground states each other. As a result, neither $\ell(1)$ or $r(1)$ can be defined locally. 

\begin{figure}[htb!]
\centering
\includegraphics[width=0.85\linewidth]{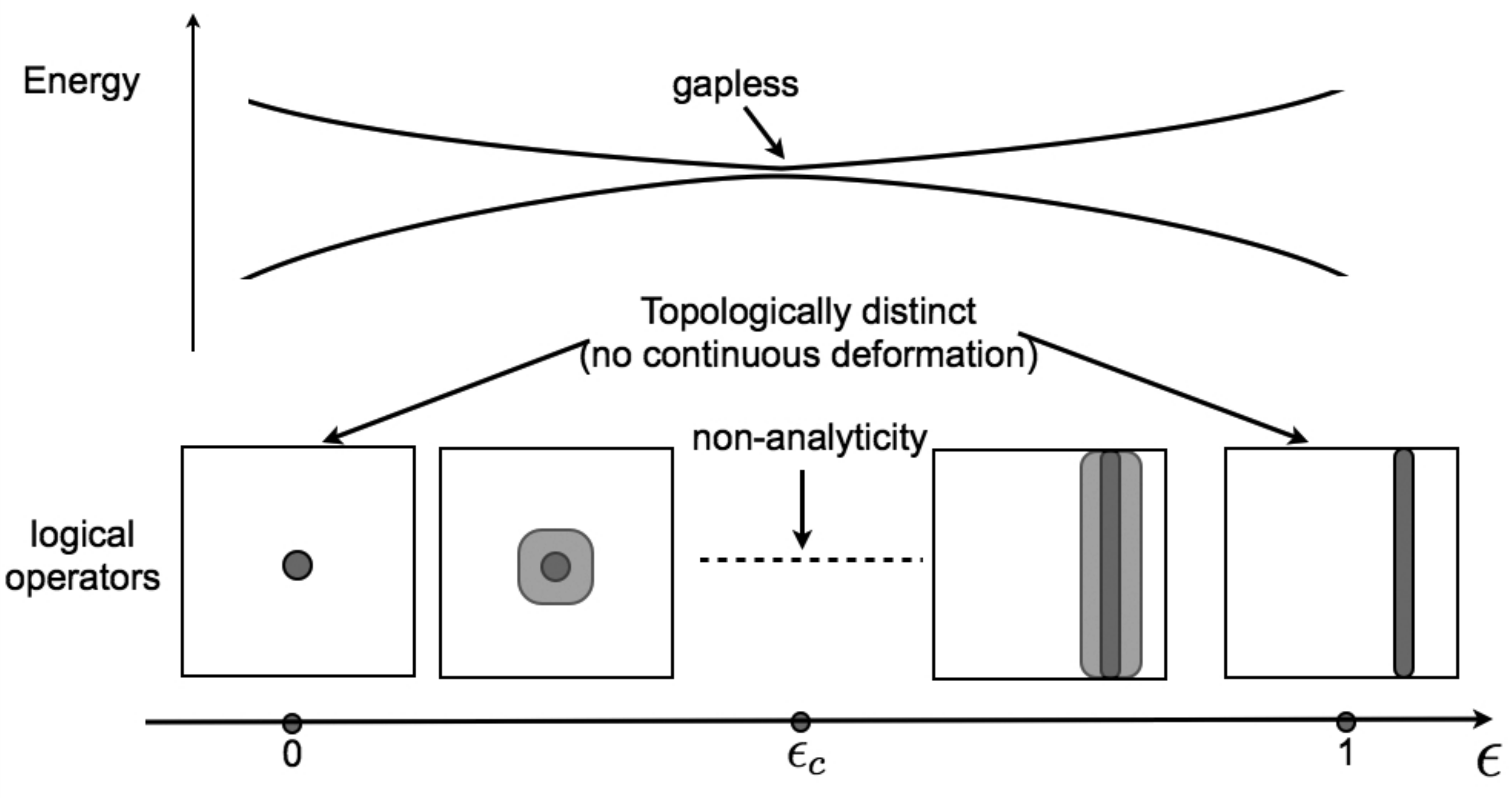}
\caption{Geometric shapes of an analytically continued logical operator, the vanishing energy gap, and the existence of a QPT. Lightly shaded regions represent the ``shapes'' of analytically continued logical operators.} 
\label{fig_2DQPT}
\end{figure}

While an analytically continued logical operator $\ell(1)$ acts like a logical operator of $H_{B}$, it may not be an actual logical operators of $H_{B}$ since it is not a Pauli operator in general. Here, for simplicity of discussion, we assume that $\ell(1)$ happens to be a one-dimensional logical operator $\ell_{B}$ in $H_{B}$ which is a Pauli operator. Then, geometric shapes of $\ell(\epsilon)$ changes from a zero-dimensional object $\ell_{A}$ to a one-dimensional object $\ell_{B}$ as we vary $\epsilon$ from $\epsilon=0$ to $\epsilon=1$. However, it is impossible to change geometric shapes of a zero-dimensional object to a one-dimensional object continuously at the thermodynamic limit since a continuous deformation between them is topologically prohibited. This contradicts with our original assumption that there is no QPT in $H(\epsilon)$. Thus, there must be a QPT between $H_{A}$ and $H_{B}$, and we may conclude that $H_{A}$ and $H_{B}$ belong to different quantum phases. Though we have assumed that $\ell(1)$ happens to be a one-dimensional logical operator $\ell_{B}$ in $H_{B}$, the observation that $\ell(1)$ cannot be defined locally is sufficient to reach the same conclusion.

\textbf{A more rigorous approach based on adiabatic continuation:}
While the above analysis implies that the topological distinction between geometric shapes of logical operators in $H_{A}$ and $H_{B}$ lead to the existence of a QPT, the discussion may lack in a mathematical rigor. In particular, we have assumed that geometric shapes of $\ell(0)$ and $\ell(\epsilon)$ for small $\epsilon$ are close. However, it is not clear if this assumption is correct or not. Also, while we have used the expression ``geometric shape'' naively to describe geometric properties of analytically continued logical operators $\ell(\epsilon)$, we have not stated clearly the definition of geometric shapes of analytically continued logical operators. 

Below, we shall clear these ambiguities by borrowing theoretical techniques developed in studies of adiabatic continuation~\cite{Hastings05, Bravyi06, Osborne07, Bravyi10b}. Adiabatic continuation is an idea of studying the ground state properties of Hamiltonians belonging to the same quantum phase by finding a gapped parameterized Hamiltonian connecting them.  Here, for simplicity of discussion, we begin with the cases where there is always a single ground state in parameterized Hamiltonians. If two Hamiltonians $H$ and $H'$ are in the same quantum phase, one can always find gapped parameterized Hamiltonian $H(\epsilon)$ connecting them due to the definition of quantum phases. Then, according to the adiabatic theorem~\cite{Kato50, Messiah}, by varying the parameterized Hamiltonian $H(\epsilon)$ sufficiently slowly, one can transfer a ground state of $H$ (denoted as $|\psi\rangle$) to a ground state of $H'$ (denoted as $|\psi ' \rangle$). A key idea is to realize that a unitary operator $U$ required to transform $|\psi\rangle$ to $|\psi '\rangle$ ($|\psi'\rangle = U |\psi\rangle$) can be extracted from the parameterized Hamiltonians $H(\epsilon)$. In particular, theoretical techniques on adiabatic continuation provide systematic formulas to obtain such a unitary transformation $U$ from a gapped parameterized Hamiltonian $H(\epsilon)$. While we have started our discussions with the cases where there is only a single ground state, adiabatic continuation also works when there is the ground state degeneracy and the number of degenerate ground states does not change. Even more, it works when the ground state degeneracy is slightly broken. (See references for applicabilities of adiabatic continuation~\cite{Hastings05, Bravyi06, Osborne07, Bravyi10b}).

Now, we interpret our previous analyses on the parameterized Hamiltonians in Eq.~(\ref{eq:assumption}) through adiabatic continuation. Let us recall that we have supposed that there exists a gapped parameterized Hamiltonian $H(\epsilon)$ which connect $H_{A}$ and $H_{B}$, and the number of ground stated does not change as we vary $\epsilon$. Then, one can represent the unitary transformation $U(\epsilon)$ defined in Eq.~(\ref{eq:adiabatic}) through the parameterized Hamiltonian $H(\epsilon)$. Now, we discuss properties of analytically continued logical operator $\ell(\epsilon) = U(\epsilon)\ell_{A} U(\epsilon)^{-1}$. A remarkable discovery from recent developments on studies of adiabatic continuation is the fact that one can  approximate analytically continued logical operators $\ell(\epsilon)$ with some operator $\tilde{\ell}(\epsilon)$ whose geometric shape is close to the one of $\ell_{A}$ as long as a parameterized Hamiltonian $H(\epsilon)$ remains gapped at the thermodynamic limit~\cite{Hastings05, Bravyi10b}. In particular, there exists an approximation $\tilde{\ell}(\epsilon)$ defined inside some localized region with $\xi(\epsilon) \times \xi(\epsilon)$ composite particles:
\begin{align}
\ell(\epsilon) \ \sim \ \tilde{\ell}(\epsilon)
\end{align} 
where $\xi(\epsilon)$ depends on the energy gap $\Delta(\epsilon)$, but remains finite even at the thermodynamic limit. Here, by approximation, we mean that $\ell(\epsilon)$ and $\tilde{\ell}(\epsilon)$ act in a similar way inside the ground state space. In particular, if $\ell(\epsilon)$ transforms two degenerate ground states $|\psi(\epsilon)\rangle$ and $|\psi'(\epsilon)\rangle$ of $H(\epsilon)$ in the following way:
\begin{align}
\ell(\epsilon) |\psi(\epsilon)\rangle \ = \  |\psi(\epsilon)'\rangle, \qquad
\langle \psi'(\epsilon) |\tilde{\ell}(\epsilon) |\psi(\epsilon)\rangle \ \sim \ O(1).
\end{align}
Here, the inner product is some constant which does not depend on the system size. 

Since $\tilde{\ell}(\epsilon)$ approximates $\ell(\epsilon)$, one may consider a localized region with $\xi(\epsilon) \times \xi(\epsilon)$ composite particles as the geometric shape of $\ell(\epsilon)$. Now, at $\epsilon = 1$, we have an approximation $\tilde{\ell}(1)$ which is defined inside some finite region with  $\xi(1)  \times \xi(1)$ composite particles. Then, we have
\begin{align}
\langle \psi'(1) |\tilde{\ell}(1) |\psi(1)\rangle \ \sim \ O(1)
\end{align}
for two orthogonal ground states $|\psi(1)\rangle$ and $|\psi'(1)\rangle$ of $H(1)\equiv H_{B}$. However, from the indistinguishability condition, there is no local physical observable which can transform ground states each other in $H_{B}$. Thus, $\xi(1)$ must be an infinite number at the thermodynamic limit, which leads to a contradiction. Therefore, there must be a QPT between $H_{A}$ and $H_{B}$ since the energy gap vanishes or the number of degenerate ground states changes at some point. 

\textbf{QPTs and topology in logical operators:}
We have seen that non-analytic changes of geometric shapes of logical operators lead to the existence of the vanishing energy gap. This observation may be better understood by the use of a mathematical language developed in studies of shapes of objects. Topology is a study of classifying objects in terms of smoothness and non-analyticity. If two objects can be transformed each other through continuous deformations (diffeomorphism), they are considered to be the same. Roughly speaking, diffeomorphism is a one-to-one mapping between two geometric manifolds where both the map itself and its inverse are differentiable. The notion of topology has been particularly useful in characterizing physical properties which may survive even at the thermodynamic limit, mainly in the context of field theory where physical properties are to be discussed at the continuum limit. 

In a strict sense, the notion of topology cannot be introduced in discussions of quantum phases arising in lattice systems since spins on lattices are discretely distributed and geometric shapes of logical operators are not smoothly determined. However, by considering a system of qubits at the thermodynamic limit, one can smoothen geometric shapes of logical operators effectively. Then, the fact that there is no continuous unitary transformation $U(\epsilon)$ transforming $\ell(0)$ into $\ell(1)$ is a direct consequence of the fact that there is no diffeomorphism which map a topologically trivial object (a zero-dimensional point) to a topologically non-trivial object (a one-dimensional loop winding around the torus). One cannot cut the winding of a one-dimensional logical operator without introducing some non-analyticity. 



\textbf{Summary and application:} We have shown that $H_{A}$ and $H_{B}$ belong to different quantum phases when geometric shapes of their logical operators are different. Though discussions have been limited to the cases where $H_{A}$ has two pairs of anti-commuting logical operators and $H_{B}$ has two pairs of anti-commuting logical operators, our analysis can be readily generalized to arbitrary two-dimensional STS models. 

Let us consider the case when $H_{A}$ and $H_{B}$ have various numbers of pairs of anti-commuting logical operators. In particular, let us denote the numbers of logical qubits as $k^{(A)}$ and $k^{(B)}$, and the numbers of zero-dimensional logical operators $k^{(A)}_{0}$ and $k^{(B)}_{0}$ with $k^{(A)}_{1} = k^{(A)} - k^{(A)}_{0}$ and $k^{(B)}_{1} = k^{(B)} - k^{(B)}_{0}$:
\begin{center}
\begin{tabular}{ccc}
          &  (1 dim)-(1 dim) & (0 dim)-(2 dim)  \\ \hline
$H_{A}$   &  $k^{(A)}_{1}$  & $k^{(A)}_{0}$ \\ 
$H_{B}$   &  $k^{(B)}_{1}$  & $k^{(B)}_{0}$ 
\end{tabular}
\end{center}
Then, two STS models belong to different quantum phases when $k^{(A)}_{1} \not= k^{(B)}_{1}$ or $k^{(A)}_{0}\not=k^{(B)}_{0}$. 

Here, we make some comments on the cases where two STS models have the same geometric shapes of logical operators: $k^{(A)}_{0} = k^{(B)}_{0}$ and $k^{(A)}_{1} = k^{(B)}_{1}$. Since topological properties of STS models, including topological entanglement entropy and anyonic excitations, can be completely characterized by one-dimensional logical operators, it is natural to expect that $H_{A}$ and $H_{B}$ belong to the same quantum phases. In fact, in Section~\ref{sec:2D13}, we introduced the model which has four pairs of anti-commuting logical operators and showed that the model is equivalent to a model where two copies of the Toric code are embedded in a non-interacting way. In this case, two STS models can be connected by some parameterized Hamiltonian without closing an energy gap since two models can be transformed each other only through local unitary transformations acting on neighboring composite particles. 

At this moment, we do not have a mathematical proof to the statement that $H_{A}$ and $H_{B}$ belong to the same quantum phases when geometric shapes of their logical operators are the same. However, we believe this statement is reasonable since any existing characterization of topological order cannot distinguish two STS models with logical operators whose geometric shapes are the same. Therefore, we conjecture that two STS models belong to the same quantum phases if and only if geometric shapes of their logical operators are the same: $k^{(A)}_{0} = k^{(B)}_{0}$ and $k^{(A)}_{1} = k^{(B)}_{1}$.

\begin{figure}[htb!]
\centering
\includegraphics[width=0.7\linewidth]{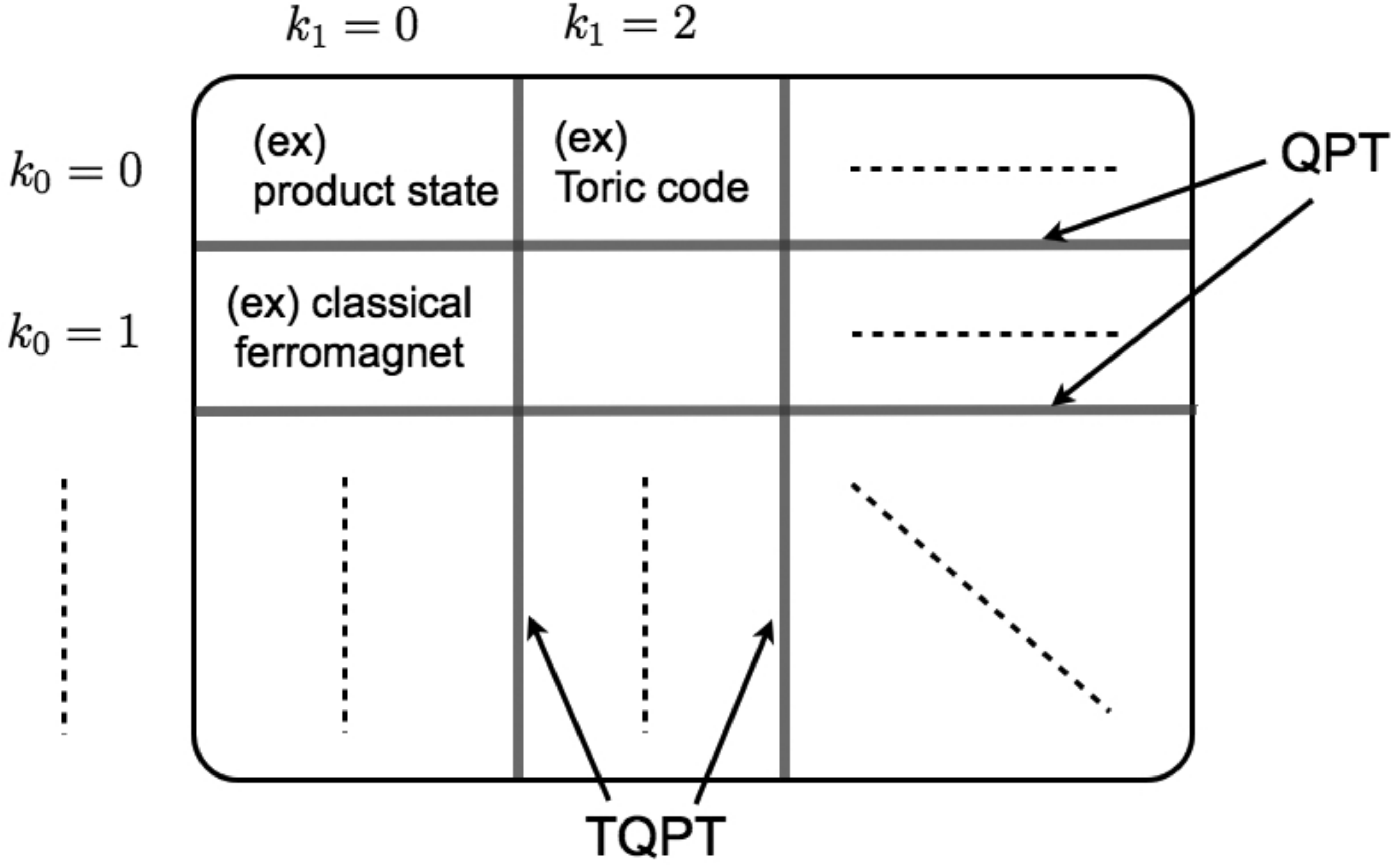}
\caption{A ``phase diagram'' of two-dimensional STS models.
} 
\label{fig_2D_phase}
\end{figure}

Here, we summarize the main result of this section (Fig.~\ref{fig_2D_phase}).

\begin{itemize}
\item Different quantum phases in two-dimensional STS models can be characterized by not only the number of logical operators, but also geometric shapes of logical operators. There must be a QPT when there is a topologically prohibited change in geometric shapes of logical operators between two STS models. Geometric shapes of logical operators can be used as order parameters to distinguish quantum phases arising in two-dimensional STS models.
\end{itemize}

Note that a QPT with a change of the numbers of one-dimensional logical operators involves the emergence or the loss of topological order. Such QPTs are called topological QPT (TQPT). In Fig.~\ref{fig_2D_phase}, a QPT which crosses the boundary in the lateral direction is a topological QPT, while a QPT which crosses the boundary in the vertical direction is a non-topological QPT.

Now, let us look at some examples. As a result of our analysis on geometric shapes of logical operators, one can classify quantum phases arising in STS models in the following way:
\begin{align}
&H_{A} \ = \sum_{i,j} \begin{bmatrix}
Z     & X  \\
X     & Z  
\end{bmatrix}^{(i,j)}    &(k_{0} \ = \ 0, \ k_{1} \ = \ 2) \notag \\
\not \sim \ &H_{B} \ = \ - \sum_{i,j} X^{(i,j)} \ \sim \ H_{B}' \ = \sum_{i,j} \begin{bmatrix}
      & Z &   \\
Z     & X & Z \\
      & Z
\end{bmatrix}^{(i,j)} &(k_{0} \ = \ 0, \ k_{1} \ = \ 0) \notag \\
\not \sim \ &H_{C} \ = \ - \sum_{i,j} Z^{(i,j)}Z^{(i,j+1)} - \sum_{i,j} Z^{(i,j)}Z^{(i+1,j)} &(k_{0} \ = \ 1, \ k_{1} \ = \ 0) \notag \\
\not \sim \ &H_{D} \ = \ - \sum_{i,j} 
\begin{bmatrix}
      & Z &  \\
X      & X  & X \\
       & Z & 
\end{bmatrix}^{(i,j)} \ \sim \
H_{D}' \ = \ - \sum_{i,j} 
\begin{bmatrix}
      & Z &  \\
X      &  & X \\
       & Z & 
\end{bmatrix}^{(i,j)} &(k_{0} \ = \ 0, \ k_{1} \ = \ 4) \notag .
\end{align}
Here, $H_{A}$ is the Toric code, $H_{B}$ is a Hamiltonian supporting a single product state, $H_{B}'$ is a Hamiltonian supporting a two-dimensional cluster state, $H_{C}$ is a classical ferromagnet and $H_{D}$ and $H_{D}'$ are Hamiltonians discussed in Section~\ref{sec:2D1}. $H_{A}$ and $H_{A}'$ have two pairs of anti-commuting one-dimensional logical operators, $H_{B}$ and $H_{B}'$ have no logical operators, $H_{C}$ has a pair of zero-dimensional and two-dimensional logical operators and $H_{D}$ and $H_{D}'$ have four pairs of anti-commuting one-dimensional logical operators.

\subsubsection{Topological deformation of logical operators}\label{sec:2D32}

While we have seen that topological structures of logical operators play crucial roles in the classification of quantum phases, our discussions have been limited to some specific representations of logical operators. However, logical operators have many equivalent representations and their geometric shapes may not be uniquely determined. We conclude the discussion of this paper by expanding our analyses on topological structures of logical operators to a set of equivalent logical operators. In particular, we show that the notion of continuous deformation naturally arises in geometric shapes of a set of equivalent logical operators in two-dimensional STS models. We note that a similar analysis on geometric shapes of a set of equivalent logical operators in the Toric code was presented in~\cite{Beni10}.

We begin by analyzing geometric shapes of two-dimensional logical operators $r_{p}$ ($1 \leq p \leq k_{0}$). Let us recall that $r_{p}$ anti-commutes with a zero-dimensional logical operator $\ell_{p}$. Here, the translations of $\ell_{p}$ are equivalent to $\ell_{p}$ due to the translation equivalence of logical operators (Theorem~\ref{theorem_main}). Then, we notice that all the logical operators which are equivalent to $r_{p}$ must have supports on every composite particle since $r_{p}$ needs to have some overlaps with all the translations of $\ell_{p}$ in order for $r_{p}$ to anti-commute with $\ell_{p}$. With this observation, we conclude that two-dimensional logical operator $r_{p}$ can be defined only inside a two-dimensional region, with supports on all the composite particles.

Next, let us analyze one-dimensional logical operators $\ell_{p}$ and $r_{p}$ ($k_{0}+1 \leq p \leq k$). In particular, we show that one-dimensional logical operators cannot be defined inside any zero-dimensional region $U_{a,b}$ with $a < n_{1}$ and $b < n_{2}$. In order to show this, it is sufficient to prove that one-dimensional logical operators cannot be defined inside a region with $n_{1}-1 \times n_{2}-1$ composite particles. Suppose that $\ell_{p}$ can be defined inside a region with $n_{1}-1 \times n_{2}-1$ composite particles. Then, since $r_{p}$ can be defined inside a region with $1 \times n_{2}$ composite particles, there exists a translation of $r_{p}$ which does not have an overlap with $\ell_{p}$. Then, due to the translation equivalence of logical operators, $\ell_{p}$ and $r_{p}$ commute with each other. However, this contradicts with the fact that $\ell_{p}$ and $r_{p}$ anti-commute with each other. Through similar discussions, we can show that any one-dimensional logical operators cannot be defined inside zero-dimensional regions.

So far, we have seen that dimensions of logical operators have topologically invariant meanings where $m$-dimensional logical operators cannot be defined inside $(m-1)$-dimensional regions ($m = 1,2$). Now, let us discuss the changes of geometric shapes of logical operators. For this purpose, let us fix some notations. We call a region with $n_{1}\times 1$ composite particles $Q(1)$ and a region with $1\times n_{2}$ composite particles $Q(2)$. 

A useful observation regarding the topological structures of logical operators can be obtained by considering the number of independent logical operators $g_{A}$ which can be defined inside a region $A$. Let us consider the case where we have two regions $A$ and $A'$ where $A$ is larger than $A'$, meaning that $A$ includes all the composite particles inside $A'$. Then, if $g_{A}=g_{A'}$, one can deform all the logical operators defined inside $A$ into $A'$, since logical operators defined inside $A$ must have equivalent representations supported inside $A'$. With this observation in mind, the following theorem is central in characterizing topological structures of a set of equivalent logical operators.

\begin{theorem}[Topological deformation of logical operators]\label{lemma_number}
For the numbers of independent logical operators in two-dimensional STSs, we have the following equations:
\begin{align}
g_{\bar{P}_{1,1}}\ = g_{Q(1)\cup Q(2)} \quad g_{\overline{Q(1)\cup Q(2)}}\ = g_{P_{1,1}} \quad g_{Q(1)}\ =\ g_{\bar{Q}(1)} \ =  k \ \quad g_{Q(2)}\ =\ g_{\bar{Q}(2)} \ = \ k.
\end{align}
\end{theorem}

\begin{figure}[htb!]
\centering
\includegraphics[width=0.80\linewidth]{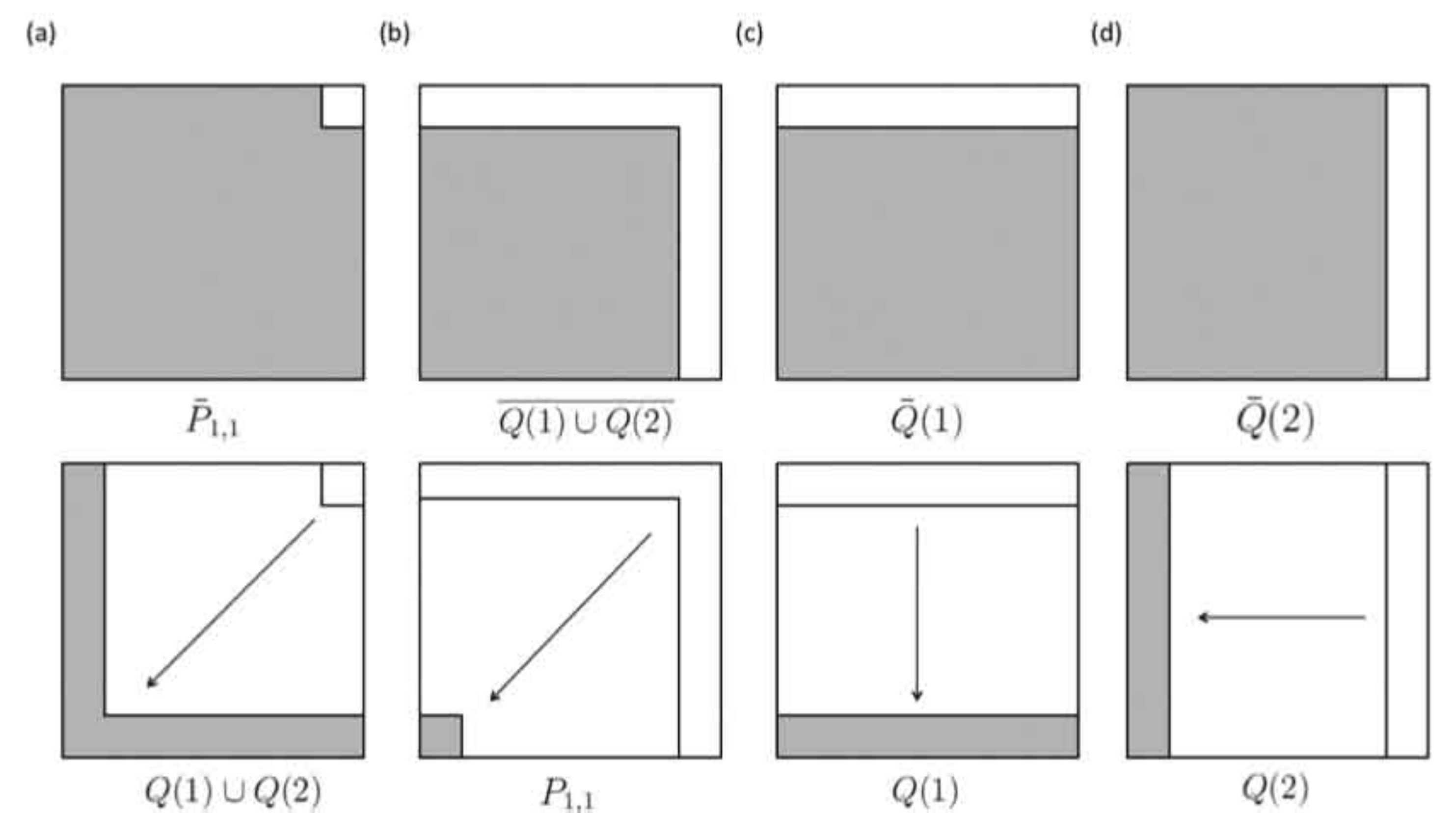}
\caption{(a) A deformation of logical operators from $\bar{P}_{1,1}$ to $Q(1)\cup Q(2)$. (b) $\overline{Q(1)\cup Q(2)}$ and $P_{1,1}$. (c) $\bar{Q}(1)$ and $Q(1)$. (d) $\bar{Q}(2)$ and $Q(2)$.
} 
\label{fig_deformation}
\end{figure}

Here, let us interpret the meaning of the theorem. Let us begin by analyzing $g_{\bar{P}_{1,1}} = g_{Q(1)\cup Q(2)}$ (Fig.~\ref{fig_deformation}(a)). $\bar{P}_{1,1}$ and $Q(1)\cup Q(2)$ are the regions with windings in the $\hat{1}$ and $\hat{2}$ directions, and one can deform a logical operator as long as we do not break the windings in geometric shapes of a logical operator. Next, let us analyze $g_{\overline{Q(1)\cup Q(2)}} = g_{P_{1,1}}$ (Fig.~\ref{fig_deformation}(b)). Both $\overline{Q(1)\cup Q(2)}$ and $P_{1,1}$ are one-dimensional regions without any winding, and one can deform a logical operator defined inside $\overline{Q(1)\cup Q(2)}$ until it becomes a point $P_{1,1}$. Finally, let us analyze $g_{Q(1)} = g_{\bar{Q}(1)}$ (Fig.~\ref{fig_deformation}(c)). Both $Q(1)$ and $\bar{Q}(1)$ have a winding in the $\hat{1}$ direction, but do not have a winding in the $\hat{2}$ direction, and one can deform a logical operator as long as we do not break the winding in the $\hat{1}$ direction. A similar discussion holds for $g_{Q(2)} = g_{\bar{Q}(2)}$ (Fig.~\ref{fig_deformation}(d)). The discussions so far can be summarized in the following way.

\begin{itemize}
\item One can deform the geometric shapes of logical operators continuously while they remain equivalent in the following ways:
\begin{align}
\bar{P}_{1,1} \ \rightarrow Q(1)\cup Q(2) \qquad \overline{Q(1)\cup Q(2)}\ \rightarrow \ P_{1,1} \qquad \bar{Q}(1)\ \rightarrow \ Q(1) \qquad \bar{Q}(2) \ \rightarrow \ Q(2) \notag
\end{align}
where, in ``$R \rightarrow R'$'', if a logical operator $\ell$ defined inside $R$ is given, then there exists a logical operator $\ell'$ defined inside $R'$.
\end{itemize}

\textbf{Proof of theorem~\ref{lemma_number}:} One can easily prove four equations through a simple counting of the number of logical operators with Theorem~\ref{theorem_partition}. Let us start with the proof of $g_{\bar{P}_{1,1}}\ = g_{Q(1)\cup Q(2)}$. Since two-dimensional logical operators cannot be defined inside $\bar{P}_{1,1}$, we have $g_{\bar{P}_{1,1}}=2k_{0} + k_{1}$. Also, since all the one-dimensional and zero-dimensional logical operators are defined inside $Q(1)\cup Q(2)$, we have $g_{Q(1) \cup Q(2)}=2k_{0} + k_{1}$, and thus $g_{\bar{P}_{1,1}}\ = g_{Q(1)\cup Q(2)}$. Now, since $g_{A} + g_{B} = 2k$ for $B = \bar{A}$ for any region $A$, $g_{\bar{P}_{1,1}}\ = g_{Q(1)\cup Q(2)}$ leads to $g_{\overline{Q(1)\cup Q(2)}}\ = g_{P_{1,1}}$. Also, by setting $A = Q(1)$, we have  $g_{Q(1)}\ =\ g_{\bar{Q}(1)} \ =  k$. This completes the proof.\hfill \qedsymbol \\ 

As a result of theorem~\ref{lemma_number}, we notice that there exist topologically invariant properties which are commonly shared by a set of equivalent logical operators. The theorem essentially states that one can deform a geometric shape of a given logical operators freely while keeping it equivalent as long as we do not change its topological structure. We call this property of logical operators in two-dimensional STS models \textbf{topological deformation of logical operators}.\footnote{While our discussions have been limited to two-dimensional stabilizer Hamiltonians, logical operators in several stabilizer Hamiltonians constructed in higher-dimensional systems also have the similar property. For example, generalizations of the Toric code defined on $D$-dimensional torus ($D>2$) have $m$-dimensional and $D-m$-dimensional logical operators~\cite{Dennis02, Takeda04}, and they obey the topological deformation of logical operators.} Therefore, the notion of topology arises naturally in geometric shapes of logical operators in STS models through topological deformation of logical operators.

\section{Summary, open questions and outlook}\label{sec:open_question}

In this paper, we have discussed possible quantum phases and their classifications in two-dimensional spin systems on a lattice. Our results may be summarized in the following three points. 

\begin{itemize}
\item \textbf{STS model:} We have proposed a model of frustration-free Hamiltonians which covers a large class of physically realistic stabilizer Hamiltonians which are constrained to only three physical conditions; the locality of interactions, translation symmetries and scale symmetries.

\item \textbf{Exact solution of the model:} As a key to the analyses on physical properties of STS models, we have found possible forms of logical operators and their geometric shapes completely. Then, we have characterized topological order arising in the model by geometric shapes of logical operators.

\item \textbf{Quantum phase transitions and logical operators:} We have shown that different quantum phases in STS models can be characterized by geometric shapes of logical operators. We have shown that the existence of a QPT results from non-analytic changes of geometric shapes of logical operators. 
\end{itemize}

Since possible geometric shapes of logical operators in two-dimensional STS models are pairs of zero-dimensional and two-dimensional logical operators as in a classical ferromagnet, or pairs of one-dimensional logical operators as in the Toric code, we conjecture that the model can be reduced to a system where classical ferromagnets and the Toric code are embedded in a non-interacting way. However, at this moment, we do not have a proof for it, and would like to leave it as an open question for the future.

Our construction of STS models and classification of quantum phases based on geometric shapes of logical operators can likely be broadened in many ways. 

\textbf{Qudit stabilizer codes and non-abelian quantum phases}: 
While we have studied only the systems of qubits (spin $1/2$ particles), systems of qudits (particles with larger spins) can provide rich varieties of quantum phases. There exists an extension of stabilizer codes to systems of qudits by the use of some generalization of Pauli operators~\cite{Gottesman01}. It is possible to construct STS models based on this qudit stabilizer formalism. Then, our analyses on STS models may be readily generalized to STS models constructed on systems of qudits. 

Recently, topologically ordered systems with anyonic excitations whose braiding rule is characterized by a non-Abelian group have been gathering significant attentions since such systems with non-Abelian topological order may be used as a resource for realizing fault-tolerant quantum computation~\cite{Kitaev97}. Since the stabilizer formalism is based on a set of Pauli operators which commute with each other, anyonic excitations supported by stabilizer Hamiltonians obey a braiding rule characterized by an Abelian group, as in the Toric code. However, the Toric code has natural extensions, constructed based on not Abelian groups, but non-Abelian groups~\cite{Kitaev03}. These extensions of the Toric code, called the quantum double model, support non-Abelian anyons. In the quantum double model, there are operators which are analogous to logical operators in stabilizer codes. Thus, our analyses and classifications of quantum phases through geometric shapes of logical operators may also be applied to non-abelian topological phases too. 

Also, while we have limited our considerations to a QPT between a non-topologically ordered system and a topologically ordered system, one can consider a QPT between different topological phases. It may be also possible to capture such a QPT in terms of geometric shapes of logical operators. 

\textbf{Higher-dimensional STS models:}
We have analyzed two-dimensional STS models as a class of stabilizer codes which build on physically reasonable systems. However, correlated spin systems in higher-dimensional systems also gather significant attentions in quantum information science community, concerning two important open questions in quantum coding theory. 

One of the ultimate goals in quantum coding theory is to create a self-correcting quantum memory~\cite{Bravyi09}. A self-correcting quantum memory is an idealistic memory which would correct errors by itself due to the large energy barrier separating degenerate ground states. If such a memory could exist, it will be a perfect quantum information storage device which may be used commercially in the future. Though there have been significant progresses in hypothetical constructions of a self-correcting quantum memory in a four-dimensional space~\cite{Dennis02, Pastawski09}, convincing proposals for its realization in three-dimensions have not appeared yet. It is shown that a self-correcting stabilizer codes cannot exist in two dimensions~\cite{Bravyi09}. Now, the feasibility of a self-correcting memory in three-dimensions is one of the most important and interesting open questions in quantum coding theory~\cite{Bacon06, Hamma09, Bombin09, Pastawski09}.

The feasibility of a self-correcting quantum memory is closely related to another important open question concerning the upper bound on the code distance of local stabilizer codes. It is commonly believed that the code distance of local stabilizer codes with $N$ qubits is upper bounded by $O(\sqrt{N})$ polynomially in the $N \rightarrow \infty$ limit. \footnote{We note that there exists a local stabilizer code whose code distance scales as $O(\sqrt{N}\log N)$~\cite{Freedman02}. Also, there is an example of three-dimensional stabilizer Hamiltonians without scale symmetries whose code distance may scale as $O(L^{2})$ for some specific choices of the system sizes~\cite{Bravyi10c}. However, no example of stabilizer codes has been found whose code distance surpass $O(\sqrt{N})$ ``polynomially'' regardless of the system size.} However, recently it have been shown that the code distance of local stabilizer codes is upper bounded by $O(L^{D-1})$ in $D$ dimensions where $L$ is a linear length of the stabilizer code \cite{Bravyi09}. Though this work has opened possibilities for the existence of a local stabilizer code whose code distance may exceed $O(\sqrt{N})$ polynomially, this bound is proven to be tight only for $D = 1,2$. The upper bound on the code distance of local stabilizer codes is also one of the important open questions in quantum coding theory.

Unfortunately, our analyses in this paper cannot give answers to these open questions since our discussions are limited to two dimensions. However, STS models may serve as a class of quantum codes which are physically realizable, and their analyses may provide partial answers toward these open questions~\cite{Beni11}.

\textbf{Scale symmetries and weak breaking of translation symmetries:}
The final problem we address is rather conceptual, but may have some fundamental importance. While the locality of interactions and translation symmetries are important physical constraints, it may not be obvious why scale symmetries are also important. Here, we mention the importance of scale symmetries in a relation with translation symmetries of ground states.

Topologically ordered systems are known to have degenerate ground states, and in analyzing topologically ordered systems, we wish to study properties of not only a single ground state, but the entire ground state space. In analyzing degenerate ground states of topologically ordered systems, a main challenge is the fact that translation symmetries may be broken. For example, consider a translation symmetric Hamiltonian which is invariant under unit translations of qubits. One might hope that this Hamiltonian would give rise to ground states which are also invariant under unit translations of qubits. However, this naive expectation is generally true only when there is a single ground state. There exist examples where degenerate ground states are invariant only under translations of several qubits, while the Hamiltonian is invariant under unit translations of qubits. This bizarre breaking of translation symmetries is observed commonly in topologically ordered systems, and is sometimes called a weak breaking of translation symmetries~\cite{Kitaev06b}. When translation symmetries are weakly broken, there will exist a ground state which differs from its own translation. Then, naturally, we wish to further coarse-grain the system so that all the ground states are invariant under unit translations of coarse-grained particles. 

Now, a naturally arising question is whether translation symmetries of STS models are broken or not. We have coarse-grained the system of qubits by introducing composite particles so that Hamiltonians are invariant under unit translations of composite particles. However, it is not clear whether ground states of STS models are also invariant under unit translations of composite particles or not. If translation symmetries of ground states are broken weakly, STS models must be further coarse-grained so that ground states restore translation symmetries with larger composite particles. 

Somewhat surprisingly, in STS models, all the degenerate ground states are always invariant under unit translation of composite particles, and translation symmetries are not broken weakly while STS models may have topological order. We shall give a proof for this statement in~\ref{sec:translation}. A key for the existence of translation symmetries of ground states is the existence of scale symmetries in stabilizer Hamiltonians. In fact, it is shown in~\ref{sec:translation} that for stabilizer Hamiltonians with translation symmetries, the existence of scale symmetries is the necessary and sufficient condition for translation symmetries of ground states. Then, scale symmetries may be the sufficient and necessary condition for the presence of translation symmetries in ground states in arbitrary gapped correlated spin systems. 

Currently, the importance of scale symmetries is not well appreciated, while the locality of interaction terms and translation symmetries of ground states are well respected as important constraints for physical realizability of system Hamiltonians. However, it will be interesting to see the roles of scale symmetries in various quantum many-body systems. In particular, the use of scale symmetries may enable us to have some deeper insights on the underlying mechanisms behind correlated spin systems, as we have succeeded in introducing the notion of topology into the discussions of logical operators thanks to scale symmetries. Possible applications of scale symmetries may include studies of computational complexities of Hamiltonian problems~\cite{Aharonov09}, RG algorithms based on the tensor product state representations~\cite{Vidal07, Gu08} and analyses on coding properties of quantum codes beyond stabilizer codes such as subsystem codes~\cite{Shor95, Bacon06, Aliferis07}.

\appendix

\section{Proof of the translation equivalence of logical operators (Theorem~\ref{theorem_main})}\label{sec:TE}

In this appendix, we give a proof of the translation equivalence of logical operators (Theorem~\ref{theorem_main}). In~\ref{sec:TE1}, we give a proof by using a certain lemma concerning properties of logical operators in STS models (lemma~\ref{lemma_1dim}). In~\ref{sec:TE2}, we give a proof of this lemma to complete the proof of Theorem~\ref{theorem_main}. In~\ref{sec:TE3}, an extension of this lemma is presented for the sake of later discussions in finding logical operators in two-dimensional STS models.

\subsection{Proof of the translation equivalence of logical operators}\label{sec:TE1}

We give a proof of the translation equivalence of logical operators.

\textbf{Translation equivalence of logical operators (Theorem~\ref{theorem_main}):}
For each and every logical operator $\ell$ in an STS model, a unit translation of $\ell$ with respect to composite particles in any direction is always equivalent to the original logical operator $\ell$:
\begin{align}
T_{m}(\ell)\ \sim\ \ell, \qquad \forall \ell\ \in\ \textbf{L}_{\vec{n}}  \qquad (m\ =\ 1, \cdots, D) 
\end{align}
where $\textbf{L}_{\vec{n}}$ is a set of all the logical operators for an STS model defined with the stabilizer group $\mathcal{S}_{\vec{n}}$.

\textbf{Proof of Theorem~\ref{theorem_main}:} It is sufficient to prove the theorem for translations in the $\hat{1}$ direction: $T_{1}(\ell)\sim \ell$. Here, we define the following $(D-1)$-dimensional regions (hyperplanes) (see Fig.~\ref{fig_lemma_1dim}):
\begin{align}
R(x)\  =\  \Big\{ \ P_{ r_{1}, \cdots, r_{m} } \ :\ r_{1}\ =\ x \quad \mbox{and} \quad \ 1\ \leq\ r_{m}\ \leq\ n_{m}\ \mbox{for}\ m\ \not= \ 1   \  \Big\}.
\end{align}
In total, there are $n_{1}$ of $(D-1)$-dimensional regions $R(x)$ for $x = 1, \cdots, n_{1}$. Then, the following lemma holds.

\begin{lemma}\label{lemma_1dim}
In an STS model defined with the stabilizer group $\mathcal{S}_{\vec{n}}$, there exists a canonical set of logical operators 
\begin{align}
\Pi(\mathcal{S}_{\vec{n}})\ =\ \left\{ 
\begin{array}{ccc}
\ell_{1} ,& \cdots ,& \ell_{k} \\
r_{1} ,& \cdots ,& r_{k}
\end{array}\right\}
\end{align}
such that $\ell_{p}$ are defined inside a $(D-1)$-dimensional region $R(1)$ ($p = 1, \cdots ,k$) while $r_{p}$ are defined over the entire lattice in a periodic way in the $\hat{1}$ direction:
\begin{align}
T_{1}(r_{p})\ =\ r_{p} \qquad (p\ =\ 1, \cdots ,k).
\end{align}
\end{lemma}

The meaning of this lemma becomes clear when all the logical operators are graphically shown as in Fig.~\ref{fig_lemma_1dim}. Here, $r_{p}$ are represented with Pauli operators $r_{p}'$ acting on $R(1)$ in the following way:
\begin{align}
r_{p} \ = \ \prod_{x}T_{1}^{x}(r_{p}').
\end{align}
With this canonical set of logical operators, one can prove that all the logical operators remain equivalent under unit translations in the $\hat{1}$ direction. In fact, one can immediately see that logical operators $r_{p}$ satisfy the translation equivalence since they are periodic: $T_{1}(r_{p})=r_{p}$. One can also prove that $\ell_{p} \sim T_{1}(\ell_{p})$ by seeing that $\ell_{p}T_{1}(\ell_{p})$ commutes with all the logical operators listed in a canonical set $\Pi(\mathcal{S}_{\vec{n}})$ given in lemma~\ref{lemma_1dim}, and thus, $\ell_{p}T_{1}(\ell_{p})$ must be a stabilizer. By considering the fact that any logical operator can be represented as a product of $\ell_{1},\cdots, \ell_{k}$, $r_{1},\cdots, r_{k}$ and stabilizers in $\mathcal{S}_{\vec{n}}$, one obtains
\begin{align}
T_{1}(\ell)\ \sim\ \ell, \qquad \forall \ell\ \in\ \textbf{L}_{\vec{n}}.
\end{align}
This completes the proof of Theorem~\ref{theorem_main}.

\begin{figure}[htb!]
\centering
\includegraphics[width=0.40\linewidth]{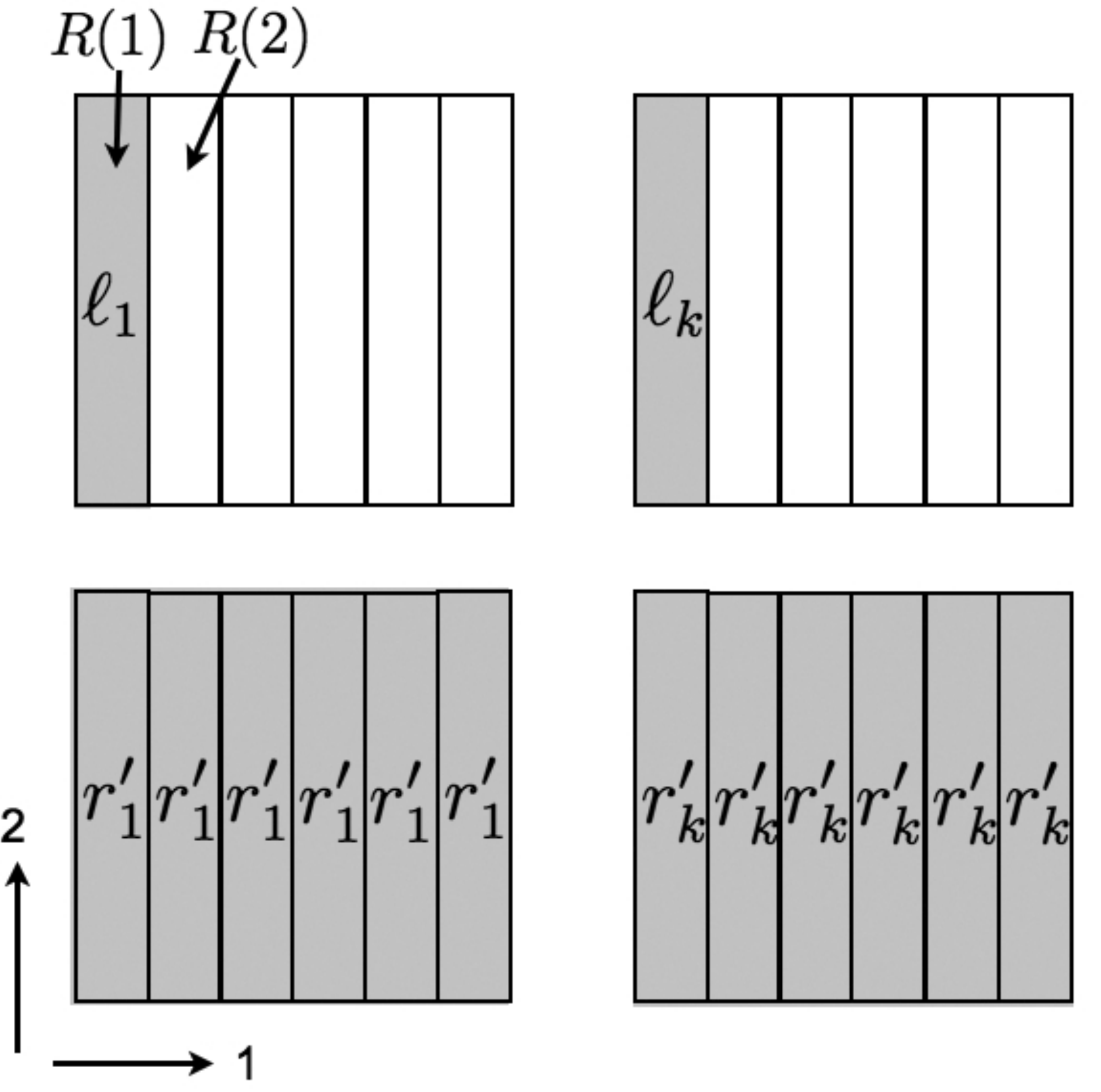}
\caption{A canonical set of logical operators. A two-dimensional example is illustrated ($D=2$). The entire lattice is separated into $(D -1)$-dimensional regions $R(x)$. (Here, $R(x)$ are one-dimensional regions for $D=2$). 
} 
\label{fig_lemma_1dim}
\end{figure}

\subsection{Proof of lemma~\ref{lemma_1dim}}\label{sec:TE2}

Now, we give a proof of lemma~\ref{lemma_1dim}. We begin our proof by proving the following lemma.

\begin{lemma}\label{lemma_independent}
Consider a stabilizer code defined with the stabilizer group $\mathcal{S}$, the centralizer group $\mathcal{C}$ and $k$ logical qubits. If there exist centralizer operators $\ell_{p}, r_{p} \in \mathcal{C}$ ($p=1, \cdots , a$) with $a \leq k$ which satisfy the following commutation relations
\begin{align}
\left\{ 
\begin{array}{ccc}
\ell_{1} ,& \cdots ,& \ell_{a} \\
r_{1} ,& \cdots ,& r_{a}
\end{array}\right\},
\end{align}
$\ell_{p}$ and $r_{p}$ are independent logical operators.
\end{lemma}

In other words, if we find pairs of anti-commuting centralizer operators represented in a canonical form, they are guaranteed to be independent logical operators in the stabilizer code. The proof of lemma~\ref{lemma_independent} is immediate by seeing that any product of operators taken among $\ell_{p}$ and $r_{p}$ ($p=1, \cdots , a$) is not a stabilizer. For example, if $\ell_{1}$ is included in the product, $r_{1}$ anti-commutes with the product, and the product is not a stabilizer. Thus, these $2a$ logical operators $\ell_{p}$ and $r_{p}$ are independent. 

Now, we present a proof of lemma \ref{lemma_1dim}. The proof utilizes the fact that the number of logical qubits does not depend on the system size, and is equal to $k$ even when $n_{1}=1$.

\begin{figure}[htb!]
\centering
\includegraphics[width=0.55\linewidth]{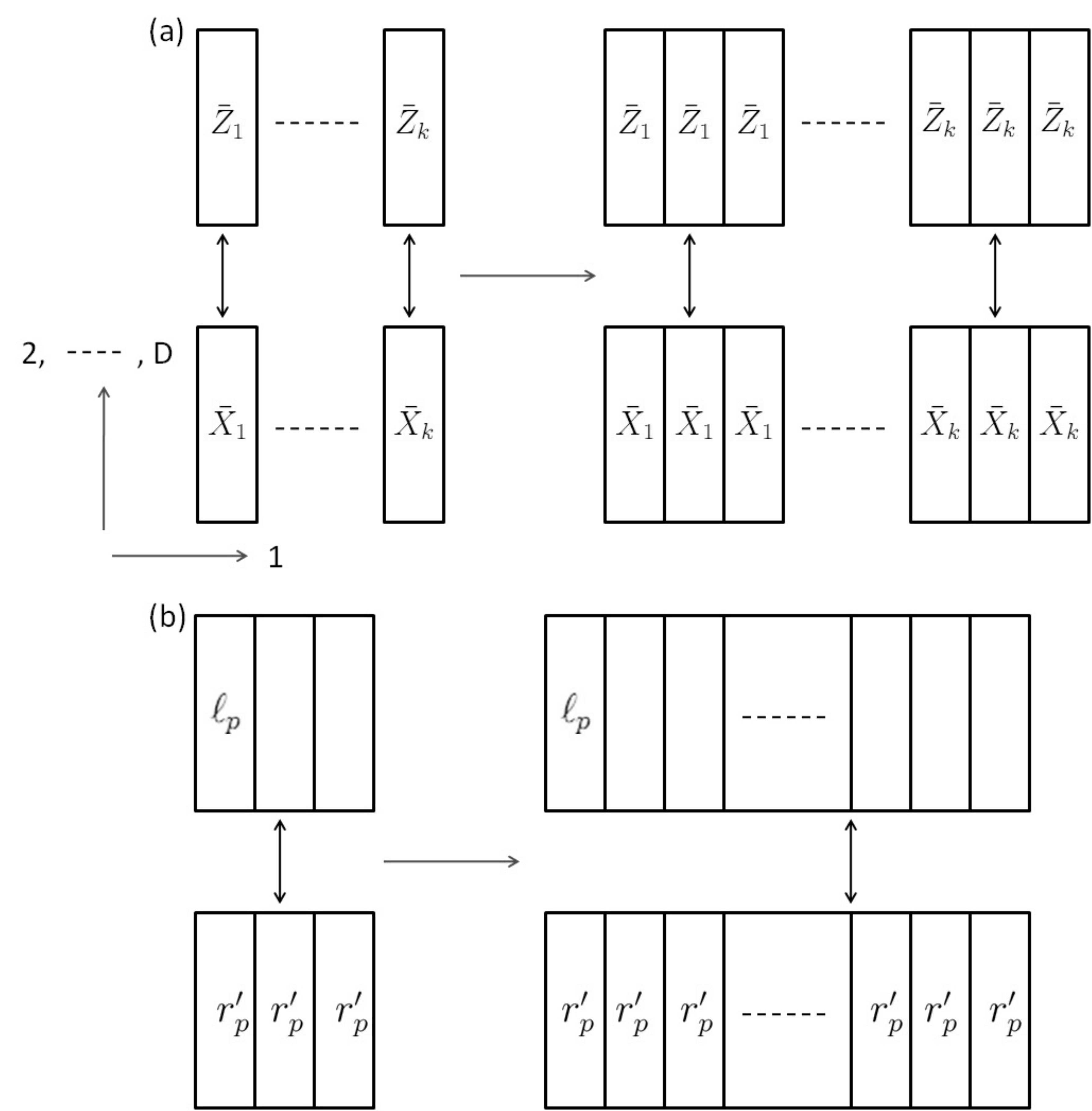}
\caption{(a) Canonical sets of logical operators for the cases with $n_{1}=1$ and $n_{1}=3$. A canonical set for $n_{1}=3$ can be constructed from a canonical set for $n_{1}=1$. (b) A canonical set of logical operators. A canonical set for arbitrary $n_{1}$ can be constructed from a canonical set for $n_{1}=3$.
} 
\label{fig_proof_aid1}
\end{figure}

\textbf{Proof of lemma~\ref{lemma_1dim}:}
Our proof starts with the cases where $n_{m}$ are fixed for $m \geq 2$. In the proof, we frequently change the value of $n_{1}$ while fixing $n_{m}$ for $m \geq 2$. We represent the stabilizer group for fixed $n_{m}$ for $m \geq 2$ as $\mathcal{S}'_{n_{1}} \equiv \mathcal{S}_{(n_{1},n_{2},\cdots,n_{D})}$. 

First, let us represent a canonical set of logical operators for the case with $n_{1}=1$ as follows:
\begin{align}
\Pi(\mathcal{S}_{1}')\ =\ \left\{
\begin{array}{ccc}
\bar{Z}_{1} ,& \cdots ,& \bar{Z}_{k} \\
\bar{X}_{1} ,& \cdots ,& \bar{X}_{k}
\end{array}\right\}
\end{align}
where $\bar{Z}_{p}$ and $\bar{X}_{p}$ are Pauli operators defined inside $R(1)$ for $p = 1, \cdots, k$. (Note that $R(1)$ represents the entire lattice since $n_{1}=1$). Next, let us consider the case with $n_{1}=3$. Then, we notice that following operators form a canonical set of logical operators (see Fig.~\ref{fig_proof_aid1}(a)):
\begin{align}
\Pi(\mathcal{S}_{3}')\ =\ \left\{
\begin{array}{ccc}
\bar{Z}_{1}T_{1}(\bar{Z}_{1})T_{1}^{2}(\bar{Z}_{1}) ,& \cdots ,& \bar{Z}_{k}T_{1}(\bar{Z}_{k})T_{1}^{2}(\bar{Z}_{k}) \\
\bar{X}_{1}T_{1}(\bar{X}_{1})T_{1}^{2}(\bar{X}_{1}) ,& \cdots ,& \bar{X}_{k}T_{1}(\bar{X}_{k})T_{1}^{2}(\bar{X}_{k})
\end{array} \right\}.
\end{align}
It is immediate to see that these $2k$ operators commute with all the stabilizers by considering the folding of stabilizer generators described in Section~\ref{sec:model2}. Since there are only $2k$ independent logical operators, a logical operator in this STS model is equivalent to some product of these $2k$ logical operators. Thus, for the case with $n_{1}=3$, any logical operator $\ell$ has a representation which is periodic in the $\hat{1}$ direction: $T_{1}(\ell)=\ell$. (This argument proves the translation equivalence of logical operators for the cases with odd $n_{1}$).

Let us continue to discuss the cases with $n_{1}=3$. We consider a bi-partition of the entire system into $R(1)$ and $R(2)\cup R(3)$. Then, from Theorem~\ref{theorem_partition}, we have
\begin{align}\label{eq:proof_lemma1_1}
g_{R(1)} + g_{R(2)\cup R(3)}\ =\ 2k.
\end{align}
Here, we show that all the logical operators defined inside $R(1)$ commute with each other by supposing that there exists a pair of anti-commuting logical operators $\ell'$ and $r'$ defined inside $R(1)$: $\{\ell', r \}=0$. Now, let us increase $n_{1}$. Then, $\ell'$ and $r'$ are also logical operators for the cases with $n_{1} > 3$. Since translations of $\ell'$ and $r'$ are also logical operators, we have $2n_{1}$ logical operators $T_{1}^{x}(\ell')$ and $T_{1}^{x}(r')$ with the following commutation relations
\begin{align}\left\{
\begin{array}{ccc}
T_{1}^{1}(\ell') ,& \cdots ,& T_{1}^{n_{1}}(\ell') \\
T_{1}^{1}(r')    ,& \cdots ,& T_{1}^{n_{1}}(r')  
\end{array}
\right\}
\end{align}
for $x = 1 , \cdots ,n_{1}$. From lemma~\ref{lemma_independent}, these $2n_{1}$ logical operators are independent, and we have $k_{\vec{n}} \geq 2n_{1}$. However, this contradicts with the fact that the number of logical qubits $k_{\vec{n}}$ does not depend on $\vec{n}$. Thus, all the logical operators defined inside $R(1)$ must commute with each other. 

A similar discussion holds for logical operators defined inside $R(2)\cup R(3)$. If there exists a pair of anti-commuting logical operators defined inside $R(2)\cup R(3)$, one can create a large number of independent logical operators for the cases with large $n_{1}$. Thus, all the logical operators defined inside $R(2)\cup R(3)$ commute with each other. 

Since there are at most $k$ independent logical operators which commute with each other, we have $g_{R(1)} \leq k$ and $g_{R(2)\cup R(3)} \leq k$. Then, from Eq.~(\ref{eq:proof_lemma1_1}), we have 
\begin{align}
g_{R(1)}\ =\ g_{R(2)\cup R(3)}\ =\ k.
\end{align}

Now, let us represent the canonical set of logical operators for $n_{1}=3$ as follows:
\begin{align}
\Pi(\mathcal{S}_{3}')\ =\ \left\{
\begin{array}{ccc}
\ell_{1} ,& \cdots ,& \ell_{k} \\
r_{1}    ,& \cdots ,& r_{k}
\end{array}\right\}
\end{align}
where $\ell_{p}$ is a logical operator defined inside $R(1)$ for $p = 1, \cdots ,k$. Since all the logical operators for $n_{1}=3$ have representations which are periodic in the $\hat{1}$ direction, one can choose $r_{p}$ such that 
\begin{align}
T_{1}(r_{p})\ =\ r_{p} \qquad (p\ =\ 1,\cdots,k).
\end{align}
Thus, we have proven lemma~\ref{lemma_1dim} for the case with $n_{1}=3$  (Fig.~\ref{fig_proof_aid1}(b)). 

Here, we represent $r_{p}$ as $r_{p}=r_{p}'T_{1}(r_{p}')T_{1}^{2}(r_{p}')$ for $n_{1}=3$ where $r_{p}'$ is a Pauli operator acting on $R(1)$. Then, by modifying the definition of $r_{p}$ a little, one can obtain a canonical set of logical operators for any $n_{1}$. Here, we redefine $r_{p}$ for arbitrary $n_{1}$ as  (Fig.~\ref{fig_proof_aid1}(b)) 
\begin{align}
r_{p}\ =\ \prod_{x=1}^{n_{1}} T_{1}^{x}(r_{p}').
\end{align}
We note that this new definition of $r_{p}$ includes the previous definition of $r_{p}$ originally given only for the case with $n_{1}=3$. Then, we obtain the canonical set of logical operators for arbitrary $n_{1}$:
\begin{align}
\Pi(\mathcal{S}_{n_{1}}')\ =\ \left\{ 
\begin{array}{ccc}
\ell_{1} ,& \cdots ,& \ell_{k} \\
r_{1} ,& \cdots ,& r_{k}
\end{array}\right\}.
\end{align}
This canonical set of logical operators has the form described in lemma~\ref{lemma_1dim}. Though we started our discussion for fixed $n_{m}$ for $m \geq 2$, one can apply the same discussion for any $\vec{n}$. This completes the proof of lemma~\ref{lemma_1dim}.

\subsection{Extension of lemma~\ref{lemma_1dim}}\label{sec:TE3}

One can also extend lemma~\ref{lemma_1dim} and obtain the following lemma, which will be useful in the later discussion in computing logical operators in two-dimensional STS models. In the below, $Q(1)$ and $Q(2)$ denote regions of $n_{1}\times 1$ and $1 \times n_{2}$ composite particles.

\begin{lemma}\label{lemma_1dim_strong}
\begin{enumerate}
\item Let the number of independent logical operators defined inside a subset of qubits $A$ be $g_{A}$. For a two-dimensional STS model defined with the stabilizer group $\mathcal{S}_{n_{1},n_{2}}$ and $k$ logical qubits, all the logical operators defined inside $\overline{Q(1)}$ can be defined inside $Q(1)$ when $n_{2} \geq 2$:
\begin{align}
g_{ \overline{Q(1)} } \ =\ g_{Q(1)}\ =\ k.
\end{align}
\item Consider an arbitrary set of independent logical operators $\{ \ell_{1}, \cdots, \ell_{k} \}$ which are defined inside $Q(1)$. Then, there exists a canonical set of logical operators $\Pi( \mathcal{S}_{n_{1},n_{2}} )$ which includes $\{ \ell_{1}, \cdots, \ell_{k}\}$: $\{ \ell_{1}, \cdots, \ell_{k}\} \subset \Pi(\mathcal{S}_{n_{1},n_{2}})$ in such a way that
\begin{align}
\Pi( \mathcal{S}_{n_{1},n_{2}} )\ =\ \left\{
\begin{array}{ccc}
\ell_{1} ,& \cdots  ,& \ell_{k} \\
r_{1}    ,& \cdots  ,& r_{k}
\end{array}\right\}
\end{align}
and $r_{1}, \cdots , r_{k}$ are periodic: 
\begin{align}
T_{2}(r_{p})\ =\ r_{p}.
\end{align}
\end{enumerate}
\end{lemma}

The lemma also holds when we interchange the $\hat{1}$ direction and the $\hat{2}$ direction. Here, let us emphasize the difference between the statement of this lemma and the statement of lemma~\ref{lemma_1dim}. In lemma~\ref{lemma_1dim}, we have shown that there exists ``some'' canonical set of logical operators where $\ell_{p}$ is defined inside $Q(1)$ while $r_{p}$ is periodic: $T_{2}(r_{p})=r_{p}$. However, this lemma ensures that for any given set of $k$ independent logical operators $\ell_{p}$ defined inside $Q(2)$, we can always find $r_{p}$ such that $r_{p}$ are periodic: $T_{2}(r_{p})=r_{p}$. We skip the proof of lemma~\ref{lemma_1dim_strong} since the proof can be obtained with a little modification to the proof of lemma~\ref{lemma_1dim}.

\section{Stabilizers in one-dimensional STS models}\label{sec:1D_proof}

In this appendix, we obtain stabilizers for one-dimensional STS models described in Section~\ref{sec:1D}. Consider a one-dimensional STS model with $n_{1}$ composite particles. We list all the independent stabilizers defined inside two consecutive composite particles $P_{1}$ and $P_{2}$ as
\begin{align}\label{eq:stabilizer_OOG}
S_{j}\ =\ \left[ \alpha_{j}, \beta_{j} \right].
\end{align}
Here, $\left[ \alpha_{j}, \beta_{j} \right]$ represents the form of the stabilizer $S_{j}$ graphically, and should not be confused with a commutator. We compute the overlapping operator group for a composite particle $P_{1}$~\cite{Beni10}, defined in the following way:
\begin{align}
\mathcal{O}^{P_{1}}\ \equiv\ \Big\langle\ \Big\{ \ U|_{P_{1}}\ :\ U\ \in\ \mathcal{S}_{n_{1}}\ \Big\} \ \Big\rangle \                     =\ \Big\langle\ \Big\{ \ \alpha_{j}, \beta_{j}, \forall j \ \Big\} \ \Big\rangle.
\end{align}
Let us represent the overlapping operator group in a canonical form. Due to translation symmetries of the system, we have
\begin{align}
[\alpha_{j}, \beta_{j'}]\ =\ 0  \qquad \mbox{for all}\ j,j'.
\end{align}
Here, by $[\alpha_{j}, \beta_{j'}] = 0$, we mean that $\alpha_{j}$ and $\beta_{j'}$ commute with each other. Then, in finding a canonical representation of $\mathcal{O}^{P_{1}}$, one can consider the contributions from $\alpha_{j}$ and the contributions from $\beta_{j}$ separately. Here, we consider two groups of operators $\mathcal{O}(\alpha)$ and $\mathcal{O}(\beta)$
\begin{align}
\mathcal{O}(\alpha) \ =\ \langle \{ \alpha_{j}, \forall j \} \rangle, \qquad \mathcal{O}(\beta)\ =\ \langle \{ \beta_{j}, \forall j \}. \rangle
\end{align}

We have described $S_{j}$ as in Eq.~(\ref{eq:stabilizer_OOG}). However, the choice of $S_{j}$ is not unique, and one can choose any set of $S_{j}$ as long as 
$\mathcal{S}_{P_{1}\cup P_{2}}  =  \langle S_{j} \ \forall j \rangle$.
Then, one can choose $S_{j}$ such that the canonical representations of $\mathcal{O}(\alpha)$ have the following form:
\begin{align}
\mathcal{O}(\alpha)\ =\ \left\langle
\begin{array}{ccccccc}
   \alpha_{1} , &     \cdots , & \alpha_{t}     , & \ell_{1},     &  \cdots , & \ell_{k'}   , & \mathcal{S}_{P_{1}}  \\
   \alpha_{t+1}  , &     \cdots , & \alpha_{2t}      , &     &   & &  
\end{array}
 \right\rangle
\end{align}
where
\begin{align}
S_{2t + j} \ = \ \left[\ell_{j},\ell_{j} \right] \qquad \mbox{for} \quad j \ = \ 1, \cdots, k',
\end{align}
and $t$ is some integer. For a derivation, see~\cite{Beni10}. Here, $\mathcal{S}_{P_{1}}$ represents all the independent generators for the restriction of $\mathcal{S}$ into $P_{1}$, and $\ell_{j}$ are logical operators defined inside $P_{1}$. $S_{2t+j}$ are stabilizers due to the translation equivalence of logical operators. We note that $k' \leq k$ since there are at most $k$ independent logical operators which commute with each other. Then, one can represent $\mathcal{O}(\beta)$ in the following form:
\begin{align}
\mathcal{O}(\beta)\ =\ \left\langle
\begin{array}{ccccccc}
   \beta_{1} , &     \cdots , & \beta_{t}     , & \ell_{1},     &  \cdots , & \ell_{k'}   ,& \mathcal{S}_{P_{1}}   \\
   \beta_{t+1}  , &     \cdots , & \beta_{2t}      , &      &    & & 
\end{array}
 \right\rangle.
\end{align}
One can easily verify that $\beta_{j}$ satisfies the above commutation relations. Having represented $\mathcal{O}(\alpha)$ and $\mathcal{O}(\beta)$ in canonical forms, it is immediate to represent $\mathcal{O}^{P_{1}}$ in a canonical form:
\begin{align}
\mathcal{O}^{P_{1}}\ =\ \left\langle
\begin{array}{cccccccccc}
   \alpha_{1} , &     \cdots , & \alpha_{t}     , &  \beta_{1} , &     \cdots , & \beta_{t}        ,& \ell_{1} ,& \cdots ,& \ell_{k'} ,& \mathcal{S}_{P_{1}} \\
   \alpha_{t+1}  , &     \cdots , & \alpha_{2t}      , & \beta_{t+1}  , &     \cdots , & \beta_{2t} ,&    &  &  & 
\end{array}
 \right\rangle.
\end{align}

Now, let us prove that $k' = k$. Here, for simplicity of discussion, let us assume that there is no stabilizer in $\mathcal{S}_{P_{1}}$. Then, one can represent the number of independent generators for the stabilizer group $\mathcal{S}_{n_{1}}$ by using the above canonical representations. First, stabilizers $S_{j}$ for $j = 1,\cdots,2t$ and their translations are always independent. Second, stabilizers $S_{j}$ $j = 2t + 1,\cdots, k'$ and their translations are not independent since
\begin{align}
\prod_{x} T_{1}^{x}(S_{j}) \ = \ I \qquad (j \ = \ 2t + 1,\cdots, 2t + k').
\end{align}
Then, we have $G(\mathcal{S}_{n_{1}}) = 2tn_{1} + k'(n_{1}-1)$. Therefore, $k' = k$.

In summary, one can represent the overlapping operator group defined for $P_{1}$ in the following way:
\begin{align}
\mathcal{O}^{P_{1}}\ =\ \left\langle
\begin{array}{ccccccccc}
   \alpha_{1} , &     \cdots , & \alpha_{t}     , &  \beta_{1} , &     \cdots , & \beta_{t}        ,& \ell_{1} ,& \cdots ,& \ell_{k}  \\
   \alpha_{t+1}  , &     \cdots , & \alpha_{2t}      , & \beta_{t+1}  , &     \cdots , & \beta_{2t} ,&    &  &   
\end{array}
 \right\rangle.
\end{align}
Here, we neglected stabilizers inside $\mathcal{S}_{P_{1}}$. Then, after applying some local unitary transformations, one can represent $S_{j}$ as in the forms presented in Section~\ref{sec:1D2}.

By extending the above discussion to the cases where $D >1$, one can obtain the following lemma which will be useful in discussions in~\ref{sec:translation}.

\begin{lemma}\label{lemma_last}
Consider a translation symmetric stabilizer Hamiltonian which is invariant under unit translations of composite particles:
\begin{align}
T_{m}(H) \ = \ H \qquad (m \ = \ 1,\cdots ,D)
\end{align}
and whose interaction terms are defined inside hypercubic regions with $2\times \cdots \times 2$ composite particles. Let $\vec{n}'$ be the size of the entire system which is defined with $n_{1}' \times \cdots \times n_{D}'$ composite particles where $n_{m}'$ are arbitrary fixed integers. The stabilizer group is denoted as $\mathcal{S}_{\vec{n}'}$ where $\vec{n}' = (n_{1}',\cdots,n_{D}')$. The number of logical qubits is denoted as $k'$. Let us assume that all the logical operators satisfy the translation equivalence of logical operators in the $\hat{1}$ direction:
\begin{align}
T_{1}(\ell)\ \sim\ \ell \qquad \mbox{for all}\ \ell\ \in\ \textbf{L}_{\vec{n}'}
\end{align} 
where $T_{1}$ represents unit translations of composite particles, and $\textbf{L}_{\vec{n}'}$ is a set of all the logical operators in the stabilizer Hamiltonian. Then, there exists the following canonical set of logical operators:
\begin{align}
\Pi(\mathcal{S}_{\vec{n}'})\ =\ \left\{
\begin{array}{ccc}
\ell_{1} ,& \cdots ,& \ell_{k} \\
r_{1}    ,& \cdots ,& r_{k}
\end{array} \right\}
\end{align}
where $\ell_{p}$ is defined inside a region with $1 \times n_{2}' \cdots \times n_{D}'$ composite particles, and $r_{p}$ is periodic in the $\hat{1}$ direction: $T_{1}(r_{p})=r_{p}$. 
\end{lemma}

The proof can be obtained through a similar discussion given in the above. For example, when $D = 2$, let us denote the region with $1 \times n_{2}$ composite particles as $Q(2)$. Then, we list all the independent stabilizers defined inside $Q(2) \cup T_{1}(Q(2))$ as
\begin{align}
S_{j}\ =\ \alpha_{j}T_{1}(\beta_{j})
\end{align}
where $\alpha_{j}$ and $\beta_{j}$ are Pauli operators defined inside $Q(2)$. By computing the overlapping operator group $\mathcal{O}^{Q(2)}$ through $\alpha_{j}$ and $\beta_{j}$, one can easily prove the above lemma. The generalizations to the cases where $D>2$ is immediate. Thus, we shall skip the proof of lemma~\ref{lemma_last}.

\section{Logical operators in two-dimensional STS models}\label{sec:proof}

In this appendix, we obtain a canonical set of logical operators for two-dimensional STS models which is presented in Section~\ref{sec:2D2}. We begin by introducing some lemmas which are useful in obtaining a canonical set of logical operators in~\ref{sec:proof1}. Then, we obtain a canonical set of logical operators for two-dimensional STS models in~\ref{sec:proof4}.

Here, we remind some notations again (Fig.~\ref{fig_notation}). An STS model has $n_{1}\times n_{2}$ composite particles and $k$ is the number of logical qubits. $Q(1)$ and $Q(2)$ represents regions with $n_{1}\times 1$ and $1 \times n_{2}$ composite particles. $U_{a,b}$ represents a region with $a \times b$ composite particles where $1 \leq a \leq n_{1}$ and $1 \leq b \leq n_{2}$. Note that $Q(1)=U_{n_{1},1}$ and $Q(2) = U_{1,n_{2}}$.

\begin{figure}[htb!]
\centering
\includegraphics[width=0.45\linewidth]{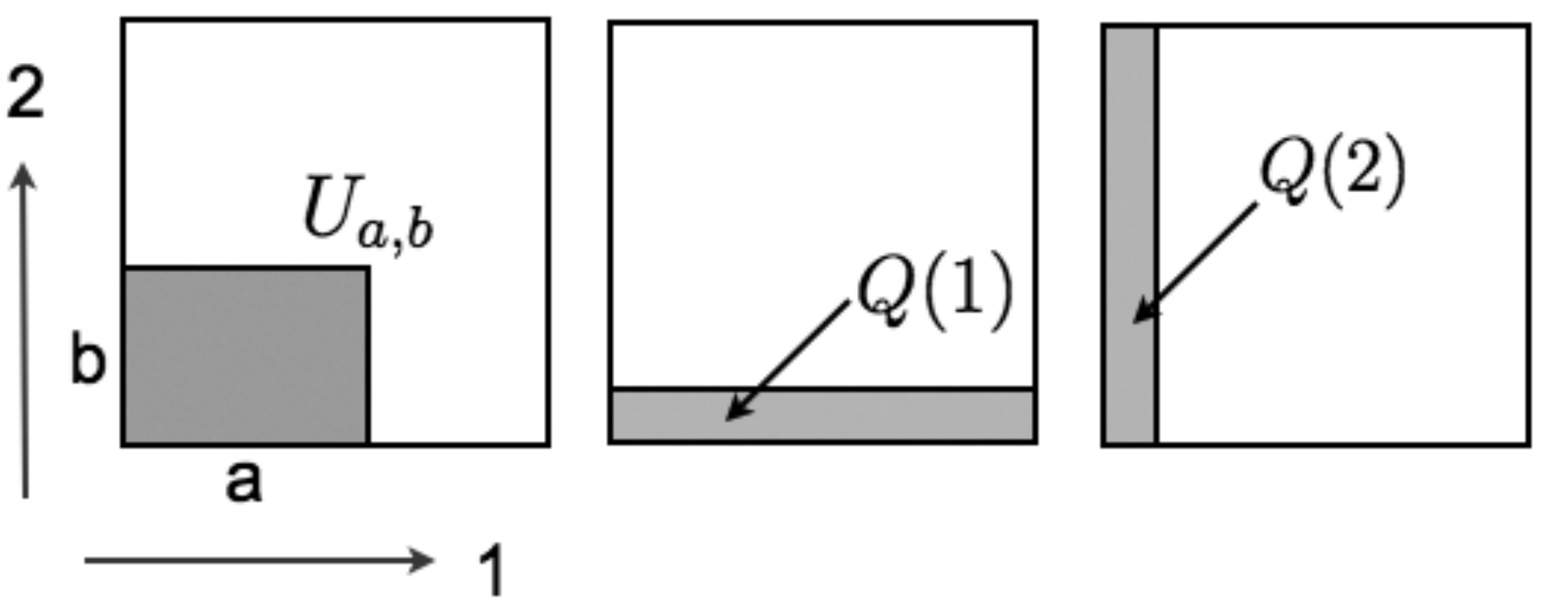}
\caption{Some regions of composite particles and their notations. 
} 
\label{fig_notation}
\end{figure}

\subsection{Some lemmas for obtaining a canonical set of logical operators}\label{sec:proof1}

In obtaining a canonical set of logical operators, we frequently use the following lemma.

\begin{lemma}\label{lemma_equiv}
For a stabilizer code $\mathcal{S}$ with a canonical set of logical operators
\begin{align}
\Pi(\mathcal{S})\ =\ \left\{
\begin{array}{ccc}
\ell_{1},& \cdots ,& \ell_{k}\\
r_{1},& \cdots ,& r_{k}
\end{array} \right\},
\end{align}
consider some logical operator $\ell$ where the commutation relations between $\ell$ and logical operators in $\Pi(\mathcal{S})$ are
\begin{align}
\ell \ell_{p}\ =\ (-1)^{e_{p}}\ell_{p} \ell, \qquad e_{p}\ =\ 0,1, \qquad \ell r_{p}\ =\ (-1)^{f_{p}} r_{p} \ell, \qquad f_{p}\ =\ 0,1.
\end{align}
Then, $\ell$ is equivalent to the following logical operator:
\begin{align}
\ell \ \sim\ \prod_{p=1}^{k} r_{p}^{e_{p}} \ell_{p}^{f_{p}}.
\end{align}
\end{lemma}

\begin{proof}
Given a logical operator $u$, $u$ is equivalent to some product of $\ell_{1}, \cdots , \ell_{k}$ and $r_{1}, \cdots, r_{k}$:
\begin{align}
\ell \ \sim\ \prod_{p=1}^{k} r_{p}^{e_{p}'} \ell_{p}^{f_{p}'}, \qquad e_{p}, f_{p}\ =\ 0,1.
\end{align}
By considering the commutation relations between $\ell$ and $\ell_{1}, \cdots , \ell_{k}$ and $r_{1}, \cdots, r_{k}$, we have $e_{p}= e_{p}'$ and $f_{p}=f_{p}'$.
\end{proof}

The following lemma becomes also useful in obtaining a canonical set of logical operators.

\begin{lemma}\label{lemma_periodicity}
Consider an STS model defined with the stabilizer group $\mathcal{S}_{n_{1},n_{2}}$ where $n_{1}$ is odd. Then, for each and every logical operator $\ell$, there exists a logical operator $\ell'$ which is equivalent to $\ell$ and is periodic in the $\hat{1}$ direction: $T_{1}(\ell)= \ell$.
\end{lemma}

\begin{proof}
Given a logical operator $\ell$, the following logical operator is equivalent to $\ell$ due to the translation equivalence of logical operators:
\begin{align}
\ell\ \sim \ \ell' \ \equiv\ \prod_{x=1}^{n_{1}} T_{1}^{x}(\ell)
\end{align}
since $n_{1}$ is odd. Note that $\ell'$ is periodic in the $\hat{1}$ direction: $T_{1}(\ell')=\ell'$. This completes the proof of the lemma.
\end{proof}

Now, let us begin finding a canonical set of logical operators. Our goal is to show that there exists a canonical set of logical operators in two-dimensional STS models which have the configurations of Pauli operators described in Section~\ref{sec:2D2}. Since the canonical set given in Section~\ref{sec:2D2} can be used as a canonical set for any $n_{1}$ and $n_{2}$ as long as $n_{1}$ and $n_{2}$ are larger than $2v$, it is sufficient to show the existence of such a canonical set for specific $n_{1}$ and $n_{2}$. Here, we consider the case with $n_{1}=2 \cdot 2^{2v}!$, where ``$!$'' represents a factorial: $2^{2v}! = 2^{2v}\times \cdots  \times 1$. The reason why we take such an artificial value as $2 \cdot 2^{2v}!$ will become clear in the course of computations. 

From lemma~\ref{lemma_1dim}, one can find a canonical set of logical operators
\begin{align}
\Pi(\mathcal{S}_{n_{1},n_{2}}) = \left\{
\begin{array}{ccc}
\ell_{1},& \cdots ,& \ell_{k}\\
r_{1},& \cdots ,& r_{k}
\end{array} \right\}\label{canonical_set}
\end{align}
where $\ell_{p}$ are defined inside $Q(1)$ while $r_{p}$ are periodic in the $\hat{2}$ direction: $T_{2}(r_{p})=r_{p}$ (Fig.~\ref{fig_periodicity}). We first analyze the properties of logical operators $\ell_{p}$ by proving the following lemma.

\begin{lemma}[Extraction of periodicity for $\ell_{p}$]\label{lemma_l}
Consider a two-dimensional STS model defined with the stabilizer group $\mathcal{S}_{n_{1},n_{2}}$ where $n_{1}=2\cdot 2^{2v}!$. If a logical operator $\ell$ is defined inside $Q(1)$, there exist centralizer operators $\ell^{\alpha}, \ell^{\beta} \in \mathcal{C}_{Q(1)}$ such that
\begin{align}
\ell \sim \ell^{\alpha}\ell^{\beta}
\end{align}
where
\begin{align}
T_{1}^{b}(\ell^{\beta})\ =\ \ell^{\beta} \qquad (b\ \leq\ 2^{2v}),
\end{align}
and $\ell^{\alpha}$ is defined inside $U_{2v,1}$. Here, $\ell^{\alpha}$ or $\ell^{\beta}$ may be an identity operator.
\end{lemma}

Here, $\mathcal{C}_{Q(1)}$ is a restriction of the centralizer group $\mathcal{C}$ into $Q(1)$. $\mathcal{C}_{Q(1)}$ includes stabilizers and logical operators defined inside $Q(1)$. Any logical operator $\ell$ defined inside $Q(1)$ can be decomposed as a product of two centralizer operators $\ell^{\alpha}$ and $\ell^{\beta}$ where $\ell^{\alpha}$ is defined inside a zero-dimensional region $U_{2v,1}$ and $\ell^{\beta}$ is periodic in the $\hat{1}$ direction with the periodicity $b$ (Fig.\ref{fig_periodicity}). In other words, a logical operator $\ell$ can be decomposed into a periodic part $\ell^{\beta}$ and a non-periodic part $\ell^{\alpha}$.

\begin{figure}[htb!]
\centering
\includegraphics[width=0.6\linewidth]{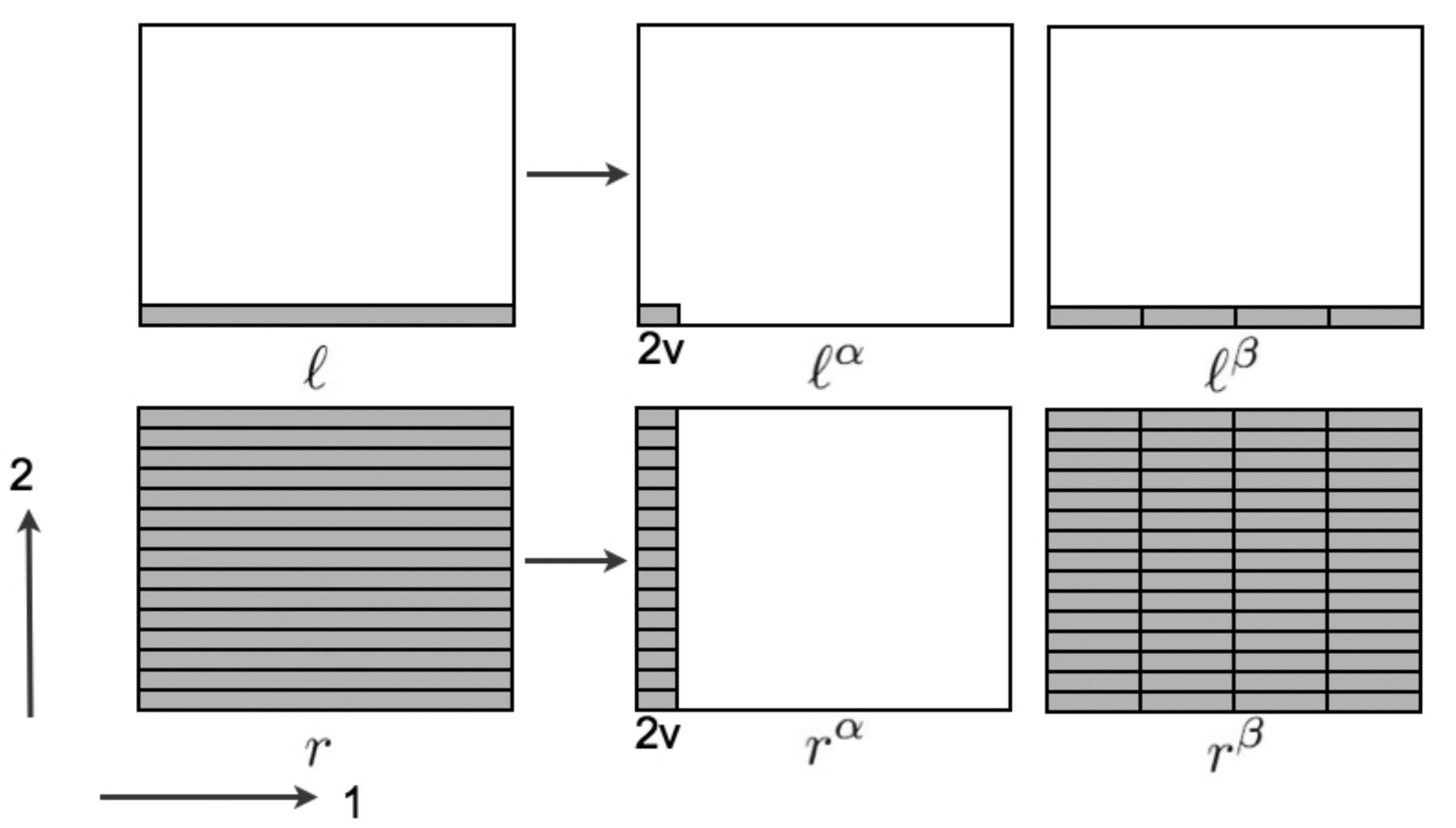}
\caption{Decompositions of logical operators into a periodic part and a non-periodic part. A logical operator $\ell$ defined inside $Q(1)$ is decomposed as a product of $\ell^{\alpha}$ and $\ell^{\beta}$ where $\ell^{\alpha}$ is defined inside $U_{2v,1}$ while $\ell^{\beta}$ is periodic in the $\hat{1}$ direction. A similar decomposition for $r$ is also shown where $r$, $r^{\alpha}$ and $r^{\beta}$ are periodic in the $\hat{2}$ direction.
} 
\label{fig_periodicity}
\end{figure}

The proof of this lemma is separated into three steps, summarized in three sublemmas respectively. The proof relies on the fact that there are $2v$ independent generators for the Pauli group acting on a single composite particle. We begin by proving the first sublemma. 

\begin{sublemma}\label{sublemma_1}
For any logical operator $\ell$ defined inside $Q(1)$ with $n_{1}= 2 \cdot 2^{2v}!$, there exist centralizer operators $\ell^{\alpha}, \ell^{\beta} \in \mathcal{C}_{Q(1)}$ with
\begin{align}
\ell\ \sim\ \ell^{\alpha}\ell^{\beta} \quad
\mbox{where} \quad T_{1}^{b}(\ell^{\beta})\ =\ \ell^{\beta} \qquad (b\ \leq\ 2^{2v}),
\end{align}
and $\ell^{\alpha}$ is defined inside $U_{n_{1}-1,1}$.
\end{sublemma}

The sublemma claims that $\ell$ can be decomposed into a periodic part $\ell^{\beta}$ and a non-periodic part $\ell^{\alpha}$ where $\ell^{\alpha}$ is defined inside a region $U_{n_{1}-1,1}$. In the proof, we shall frequently use the translation equivalence of logical operators. 

\begin{proof}
Let us represent $\ell$ as
\begin{align}
\ell\ =\ \prod_{x=1}^{n_{1}} T_{1}^{x-1}(u_{x}) \ =\ \left[ u_{1}, u_{2},\cdots, u_{n_{1}} \right].
\end{align}
Here, an $n_{1}\times 1$ matrix (vector) represents Pauli operators supporting $\ell$ inside $Q(1)$ graphically. We prove the sublemma for the case with $u_{x}\not=I$ for all $x$ since otherwise $\ell$ satisfies the condition described in the sublemma.

Since there are at most $2^{2v}$ different Pauli operators which act on a single composite particle, there must exist integers $j_{1},j_{2} \in \mathbb{Z}_{n_{2}}$ ($j_{1}<j_{2}$) with $j_{2}-j_{1} \equiv b \leq 2^{2v}$ such that $u_{j_{1}}=u_{j_{2}}$. If we represent $\ell$ explicitly with $u_{j_{1}}$ and $u_{j_{2}}$, we have 
\begin{align}
\ell \ =\ \left[ u_{1}, \cdots, u_{j_{1}}, \cdots ,u_{j_{2}-1}, u_{j_{2}}, \cdots, u_{n_{1}} \right]. 
\end{align}
Since $u_{j_{1}}=u_{j_{2}}$, one can construct a centralizer operator by using a sequence of Pauli operators $u_{j_{1}}, \cdots ,u_{j_{2}-1}$ appearing in $\ell$. In particular, one can construct the following centralizer operator:
\begin{equation}
\begin{split}
\ell^{\beta}\ &=\ \prod_{x=0}^{n_{1}/b -1} T_{1}^{xb}\left(\left[ u_{j_{1}}, u_{j_{1}+1} \cdots ,u_{j_{1}+b-1}, I, \cdots, I \right]\right) \\
&=\ \left[ u_{j_{1}}, u_{j_{1}+1} \cdots ,u_{j_{1}+b-1}, u_{j_{1}}, u_{j_{1}+1} \cdots ,u_{j_{1}+b-1}, \cdots , u_{j_{1}}, u_{j_{1}+1} \cdots ,u_{j_{1}+b-1} \right]
\end{split}
\end{equation}
where $j_{2} = j_{1}+b$. Here, $u_{j_{1}}, \cdots ,u_{j_{2}-1}$ appears periodically in $\ell^{\beta}$. Note that $n_{1}/b$ is an integer since $n_{1}= 2 \cdot 2^{2v}!$. $\ell^{\beta}$ is periodic in the $\hat{1}$ direction: $T_{1}^{b}(\ell^{\beta})\ =\ \ell^{\beta}$. One can show that $\ell^{\beta}$ is a centralizer operator by seeing that $\ell^{\beta}$ commutes with all the stabilizers. Here, we emphasize that \emph{such a construction of a centralizer operator is possible since stabilizers can be defined insider regions with $2 \times 2$ composite particles}.

Since both $\ell^{\beta}$ and $\ell$ consist of a sequence of Pauli operators $u_{j_{1}}, \cdots ,u_{j_{2}-1}$, one can cancel this sequence of Pauli operators in $\ell$ by applying some translation of $\ell^{\beta}$ to $\ell$. A resulting centralizer operator is defined inside $U_{n_{1}-b, 1}$, which is included by $U_{n_{1}-1, 1}$ since $b\geq 1$. Thus, due to the translation equivalence of logical operators, the decomposition described in the sublemma is possible. This completes the proof of the sublemma.
\end{proof}

Now, one can see the reason for setting $n_{1}=2 \cdot 2^{2v}!$. This choice of $n_{1}$ makes sure that one can extract a periodic centralizer operator $\ell^{\beta}$ from $\ell$ by requiring that $n_{1}/b$ are always integers for $b \leq 2^{2v}$. The reason for the extra factor ``two'' will become clear later. 

Next, we need to prove that a logical operator $\ell$ inside $U_{n_{1}-1,1}$ can be also defined inside $U_{2v,1}$. We prove the following sublemma.

\begin{sublemma}\label{sublemma_2}
When $n_{1} = 2 \cdot 2^{2v}!$, a logical operator $\ell$ defined inside $U_{n_{1}-1,1}$ can be also defined inside $U_{2^{2v},1}$:
\begin{align}
g_{U_{n_{1}-1,1}}\ =\ g_{U_{2^{2v},1}}.
\end{align}
\end{sublemma}

\begin{proof}
Let us consider the case where a logical operator $\ell$ satisfies 
\begin{align}
u_{1}\ =\ I \qquad \mbox{and} \qquad u_{x}\ \not=\ I \ (x\ \not=\ 1).
\end{align}
Then, there must exist integers $j_{1}, j_{2} \in \mathbb{Z}_{n_{2}}$ ($j_{1}< j_{2}$) with $j_{2}-j_{1} \equiv b \leq 2^{2v}$ such that $u_{j_{1}}=u_{j_{2}}$. 
If we represent $\ell$ explicitly, we have
\begin{align}
\ell \ = \ \left[ I, u_{2}, \cdots, u_{j_{1}}, u_{j_{1}+1} \cdots , u_{j_{2}},\cdots, u_{n_{1}} \right].
\end{align}
Then, one can construct the following centralizer operator:
\begin{align}
\ell'\ =\ \left[ I, u_{2}, \cdots, u_{j_{1}-1},u_{j_{1}}, u_{j_{2}+1}, u_{j_{2}+2}\cdots, u_{n_{1}}, I ,\cdots ,I \right].
\end{align}
We constructed $\ell'$ by discarding $u_{x}$ for $x = j_{1}+1,\cdots, j_{2}-1$ and attaching $u_{2}, \cdots, u_{j_{1}}$ to $u_{j_{2}},\cdots, u_{n_{1}}$. One can immediately see that $\ell'$ is a centralizer operator defined with $n_{1}-b-1$ composite particles. 

Now, both $\ell$ and $\ell'$ consist of a sequence of Pauli operators $u_{2}, \cdots, u_{j_{1}-1},u_{j_{1}}$. Then, if we apply $\ell'$ to $\ell$, a resulting centralizer operator is defined with $n_{1}-j_{1}$ composite particles. Then, due to the translation equivalence of logical operators, one may notice that $\ell$ can be also defined inside $U_{n_{1}-2,1}$ since $j_{1}\geq 2$. One can use the same discussion for a logical operator defined inside $U_{n_{1}-2,1}$ and shorten its length by one. By repeating this shrinkage until the length of $\ell$ becomes $2^{2v}$, we complete the proof of the sublemma.
\end{proof}

One can further shorten the length of $\ell$, as summarized in the following sublemma.

\begin{sublemma}\label{sublemma_3}
When $n_{1} = 2 \cdot 2^{2v}!$, a logical operator $\ell$ defined inside $U_{2^{2v},1}$ can be also defined inside $U_{2v,1}$:
\begin{align}
g_{U_{2^{2v},1}}\ =\ g_{U_{2v,1}}.
\end{align}\end{sublemma}

\begin{proof}
Let us suppose that $\ell$ is defined inside $U_{2^{2v},1}$ with $u_{x}\not=I$ for $x = 1 ,\cdots , 2^{2v}$:
\begin{align}
\ell \ = \ \left[ u_{1},u_{2},\cdots, u_{2^{2v}}, I ,\cdots,I \right].
\end{align}

If there exists an integer $j$ such that $u_{j}=u_{j+1}$ ($1 \leq j \leq 2^{2v}-1$), one can construct the following centralizer operator:
\begin{align}
\ell' \ = \ \left[ u_{1},u_{2},\cdots, u_{j}, u_{j+2}, \cdots, u_{2^{2v}}, I ,\cdots,I \right]
\end{align}
whose length is $2^{2v}-1$. Then, the length of $\ell \ell'$ is $2^{2v}-j$. Thus, $\ell$ can be defined inside a region $U_{2^{2v}-1,1}$.

Next, let us consider the case where there does not exist any integer $j$ such that $u_{j}=u_{j+1}$. Since there are $2v$ independent generators for the Pauli operator group defined for a single composite particle, there exist integers $1 \leq b \leq a \leq 2v$ such that $\{ u_{b}, \cdots , u_{a} \}$ consists only of independent Pauli operators, and 
\begin{align}
u_{a+1}\ \in\ \langle \{ u_{b}, \cdots , u_{a} \} \rangle \qquad \mbox{and} \qquad u_{a+1}\ \not\in\ \langle \{ u_{b+1}, \cdots , u_{a} \} \rangle. \label{eq:sublemma3}
\end{align}
Then, there exists some set $\textbf{t}$ of integers which is a subset of a set $\{b, b+1, \cdots, a-1,a\}$ where
\begin{align}
u_{a+1}\ =\ \prod_{x \in \textbf{t}} u_{x}
\end{align}
with $1 \in \textbf{t}$. Note that $\textbf{t}$ must include $b$ due to Eq.~(\ref{eq:sublemma3}). 

By constructing a certain centralizer operator based on $\textbf{t}$, one can decompose $\ell$ as a product of smaller centralizer operators. Here, in order to get some intuition, we consider an example where $b=1$, $a=4$ and $\textbf{t}= \{ 1, 3,4 \}$. Then, we have 
\begin{align}
u_{5} \ = \ u_{1}u_{3}u_{4}.
\end{align}
Then, consider the following centralizer operator
\begin{equation}
\begin{split}
\ell'\ &=\ T_{1}^{4}(\ell)T_{1}^{2}(\ell)T_{1}(\ell)\ell \\
       &=\ \left[ I, I , I , I, u_{1}, u_{2}, \cdots \right] \  \times \
          \left[ I, I , u_{1} , u_{2}, u_{3}, u_{4}, \cdots \right] \ \times \
          \left[ I, u_{1}, u_{2}, u_{3},u_{4}, u_{5}, \cdots \right] \\ 
         &\times \ \left[ u_{1}, u_{2}, u_{3},u_{4}, u_{5},u_{6},\cdots \right] \\
         &= \ \left[ u_{1}, u_{1}u_{2} , u_{1}u_{2}u_{3} , u_{2}u_{3}u_{4}, u_{1}u_{3}u_{4}u_{5}, u_{2}u_{4}u_{5}u_{6}, \cdots \right] \\
         &= \ \left[ u_{1}, u_{1}u_{2} , u_{1}u_{2}u_{3} , u_{2}u_{3}u_{4}, I, u_{2}u_{4}u_{5}u_{6}, \cdots \right] 
\end{split}
\end{equation}
since $u_{1}u_{3}u_{4}u_{5} = I$. Then, one may notice that 
\begin{align}
\ell'' \ = \ \left[ u_{1}, u_{1}u_{2} , u_{1}u_{2}u_{3} , u_{2}u_{3}u_{4}, I, I, \cdots \right]
\end{align}
is a centralizer operator since $\ell''$ commutes with all the stabilizers. By applying $\ell''$ to $\ell$, one can shorten the length of $\ell$ by one. 
Thus, it is possible to create an ``eraser'' $\ell''$ by taking a product of $\ell$ and its translations. 

Now, let us discuss more general cases. Let us consider the following centralizer operator as an eraser:
\begin{align}
\ell'\ =\ \prod_{x \in \textbf{t}} T_{1}^{a-(x-1)}(\ell). \label{eq:sublemma32}
\end{align}
Then, since $\ell'|_{P_{a+1,1}}=I$, one can construct a centralizer operator 
\begin{align}
\ell'' \ = \ \left[ u_{1}, u_{2}'', \cdots, u_{a+1}'' , I ,\cdots \right]
\end{align}
where $u_{2}'', \cdots, u_{a+1}''$ are some Pauli operators which can be computed from Eq.~(\ref{eq:sublemma32}). Thus by using $\ell''$, one can shrink the size of $\ell$.

One can use the same trick to shorten the length of a logical operator $\ell$ as long as the length of $\ell$ is larger than $2v$. Thus, $\ell$ can be defined inside $U_{2v,1}$.
\end{proof}

By integrating these three sublemmas, one can prove lemma \ref{lemma_l}. One can also obtain a similar lemma for logical operators $r_{p}$ in Eq.(\ref{canonical_set}) (see Fig.\ref{fig_periodicity}).

\begin{lemma}[Extraction of periodicity for $r_{p}$]\label{lemma_r}
Consider a two-dimensional STS model defined with the stabilizer group $\mathcal{S}_{(n_{1},n_{2})}$ with $n_{1}=2\cdot 2^{2v}!$. If a logical operator $r$ is periodic in the $\hat{2}$ direction: $T_{2}(r)=r$,
there exist periodic centralizer operators $r^{\alpha}, r^{\beta}$ 
\begin{align}
T_{2}(r^{\alpha})\ =\ r^{\alpha}, \qquad T_{2}(r^{\beta})\ =\ r^{\beta}
\end{align}
such that
\begin{align}
r\ \sim\ r^{\alpha}r^{\beta}  
\quad
\mbox{where} \quad
T_{1}^{b}(r^{\beta})\ =\ r^{\beta} \qquad (b\ \leq\ 2^{2v})
\end{align}
and $r^{\alpha}$ is defined inside $U_{2v,n_{2}}$. 
\end{lemma}

We skip the proof of lemma \ref{lemma_r} since the proof is essentially the same as the proof for $\ell_{p}$ due to the periodicity of $r_{p}$ in the $\hat{2}$ direction. As a result of these two lemmas on the properties of $\ell$ and $r$, both $\ell$ and $r$ can be separated into non-periodic parts $\ell^{\alpha}, r^{\alpha}$ and periodic parts $\ell^{\beta}, r^{\beta}$ (Fig.\ref{fig_periodicity}).

\subsection{Canonical set for a two-dimensional STS model}\label{sec:proof4}

Having introduced necessary tools in the previous section, let us now give a canonical set of logical operators for two-dimensional STS models. Recall that we have limited our consideration to the cases where $n_{1} = 2 \cdot 2^{2v}!$ in the previous subsection. In addition, we consider the case where $n_{2}$ is odd first. In order to indicate that we are considering STS models for some fixed $n_{1}$ and $n_{2}$, we use notations $n_{1}=\tilde{n}_{1} \equiv 2 \cdot 2^{2v}!$ and $n_{2} = \tilde{n}_{2}$ where $\tilde{n}_{2}$ is some fixed odd integer. Thus, we consider the case with $n_{1} = \tilde{n}_{1}$ and $n_{2} = \tilde{n}_{2}$ first.

From lemma~\ref{lemma_1dim_strong}, there are $k$ logical operators $\ell_{1}, \cdots , \ell_{k}$ defined inside $Q(1)$. Let the number of logical operators defined inside $U_{2v,1}$ be $g_{U_{2v,1}} = k_{0}$ (for $n_{1}=\tilde{n}_{1}$ and $n_{2}=\tilde{n}_{2}$). Then, we can choose $\ell_{1}, \cdots , \ell_{k}$ such that (Fig.~\ref{fig_proof_aid2}(b))
\begin{itemize}
\item $\ell_{1}, \cdots, \ell_{k_{0}}$ are defined inside $U_{2v,1}$,
\item $\ell_{k_{0}+1}, \cdots, \ell_{k}$ are defined inside $Q(1)$ with $T_{1}^{b_{p}}(\ell_{p})= \ell_{p}$ for $p = k_{0}+1, \cdots , k$ for some integers $b_{p} \leq 2^{2v}$.
\end{itemize}
Here, we attach the indexes ``$\alpha$'' and ``$\beta$'' to represent whether the corresponding logical operator is periodic in the $\hat{1}$ direction or not:
\begin{align}
\ell_{1}^{\alpha}, \cdots, \ell_{k_{0}}^{\alpha} \qquad \mbox{and} \qquad \ell_{k_{0}+1}^{\beta}, \cdots, \ell_{k}^{\beta} \qquad
\end{align}
where $\ell_{p}^{\alpha} \equiv \ell_{p}$ for $p = 1,\cdots, k_{0}$ and $\ell_{p}^{\beta} \equiv \ell_{p}$ for $p=k_{0}+1,\cdots, k$. Then, from lemma~\ref{lemma_1dim_strong}, there exists a canonical set of logical operators which includes $\ell_{1}^{\alpha}, \cdots, \ell_{k_{0}}^{\alpha}$ and $\ell_{k_{0}+1}^{\beta}, \cdots, \ell_{k}^{\beta}$:
\begin{align}
\Pi(\mathcal{S}_{\tilde{n}_{1}, \tilde{n}_{2}}) = \left\{
\begin{array}{cccccc}
\ell^{\alpha}_{1} ,& \cdots ,& \ell^{\alpha}_{k_{0}} ,& \ell^{\beta}_{k_{0}+1} ,& \cdots ,& \ell^{\beta}_{k} \\
r_{1}    ,& \cdots ,& r_{k_{0}}    ,& r_{k_{0}+1}    ,& \cdots ,& r_{k}
\end{array}\right\}
\end{align}
where $r_{p}$ are periodic in the $\hat{2}$ direction: $T_{2}(r_{p}) = r_{p}$ for $p = 1,\cdots , k$. Here, $\Pi(\mathcal{S}_{\tilde{n}_{1}, \tilde{n}_{2}})$ represents a canonical set of logical operators for $n_{1}=\tilde{n}_{1}$ and $n_{2}=\tilde{n}_{2}$. We shall see that $\ell_{1}^{\alpha}, \cdots, \ell_{k_{0}}^{\alpha}$ are independent logical operators and $g_{2v,1}=k_{0}$ for any $n_{1}$ and $n_{2}$ in the course of discussions.

\begin{figure}[htb!]
\centering
\includegraphics[width=0.80\linewidth]{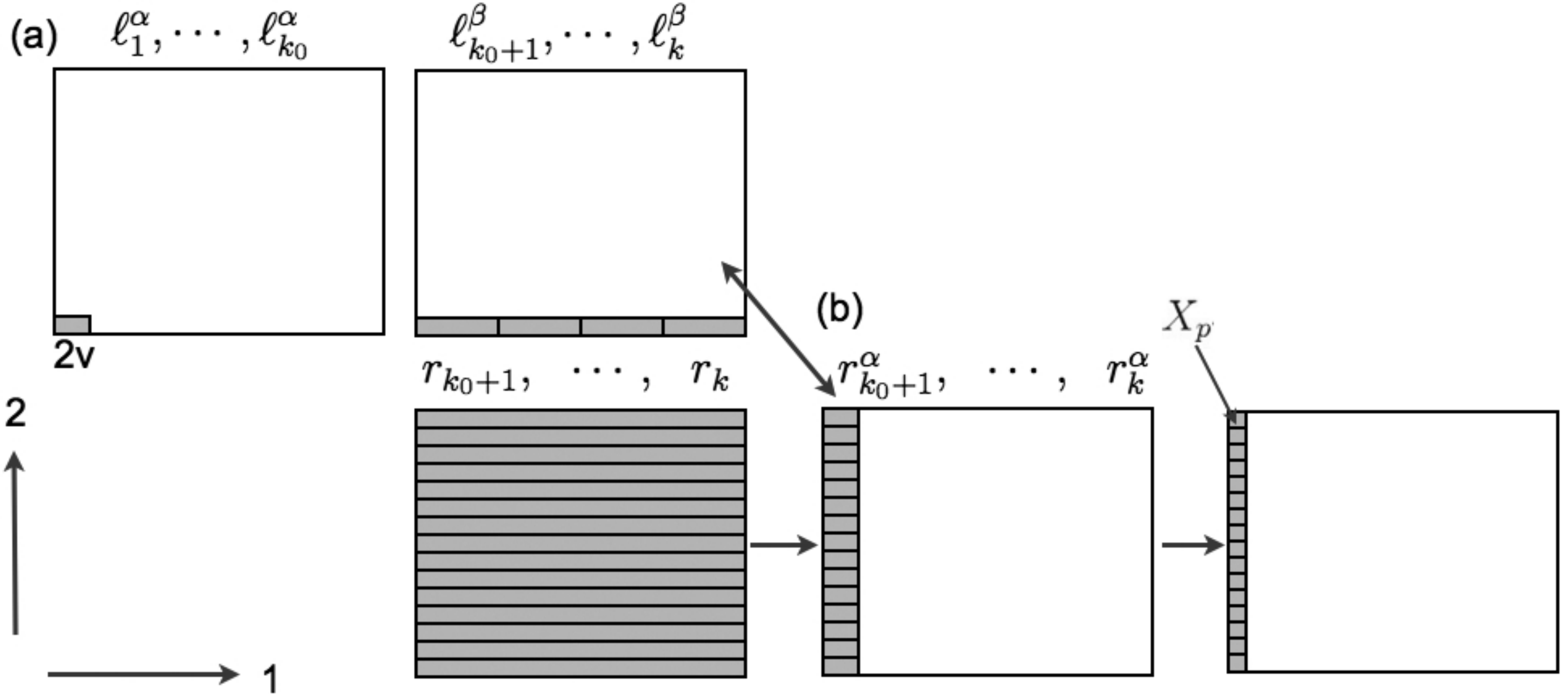}
\caption{(a) Logical operators $\ell_{p}^{\alpha}$ and $\ell_{p}^{\beta}$. $\ell_{p}^{\alpha}$ are defined inside $U_{2v,1}$. $\ell_{p}^{\beta}$ are defined inside $Q(1)$, and periodic in the $\hat{1}$ direction: $T_{1}^{b_{p}}(\ell^{\beta}_{p})=\ell^{\beta}_{p}$. Logical operators $r_{p}$ for $p = k_{0}+1, \cdots, k$ are anti-commuting pairs of $\ell^{\beta}_{p}$, and are periodic in the $\hat{2}$ direction. (b) $r_{p}^{\alpha}$ and their anti-commutation with $\ell_{p}^{\beta}$ for $p = k_{0}+1,\cdots,k$. A two-sided arrow represents the anti-commutation: $\{ \ell_{p}^{\beta}, r_{p}^{\alpha}\}=0$. $r_{p}^{\alpha}$ can be represented in a periodic way inside $Q(2)$: $r_{p}^{\alpha} = \prod_{y=1}^{n_{2}}T_{2}^{x}(X_{p})$.
} 
\label{fig_proof_aid2}
\end{figure}

For now, we concentrate on the analysis on $r_{k_{0}+1}, \cdots, r_{k}$. By using lemma~\ref{lemma_r}, we can decompose $r_{p}$ for $p = k_{0}+1, \cdots, k$ in the following way (see Fig.~\ref{fig_periodicity}):
\begin{align}
r_{p}\ \sim\ r_{p}^{\alpha}r_{p}^{\beta}  \quad \left(\ T_{2}(r_{p}^{\alpha})\ =\ r_{p}^{\alpha}\quad \mbox{and} \quad T_{2}(r_{p}^{\beta})\ =\ r_{p}^{\beta} \ \right) \qquad  \mbox{where} \quad T_{1}^{c_{p}}(r_{p}^{\beta})\ =\ r_{p}^{\beta} \quad (c_{p}\ \leq\ 2^{2v}).
\end{align}
Here, $r_{p}^{\alpha}$ are defined inside $U_{2v, \tilde{n}_{2}}$ (Fig.~\ref{fig_proof_aid2}(b)). Now, we show that $[\ell_{p}^{\beta}, r_{q}^{\beta}]=0$ for any $p$ and $q$ ($p,q = k_{0}+1, \cdots , k$). We have $T_{1}^{b_{p}}(\ell_{p}^{\beta}) = \ell_{p}^{\beta}$ and $T_{1}^{c_{q}}(r_{q}^{\beta}) = r_{q}^{\beta}$. Take the lowest common multiple of $b_{p}$ and $c_{q}$ as $d_{p,q}$. Then, we have 
\begin{align}
T_{1}^{d_{p,q}}(\ell_{p}^{\beta})\ =\ \ell_{p}^{\beta},  \qquad T_{1}^{d_{p,q}}(r_{q}^{\beta})\ =\ r_{q}^{\beta}.
\end{align}
In other words, $\ell_{p}^{\beta}$ and $r_{q}^{\beta}$ have a common periodicity $d_{p,q}$. Since $\tilde{n}_{1}/d_{p,q}$ is an even integer for $\tilde{n}_{1} = 2\cdot 2^{2v}!$, we must have $[\ell_{p}^{\beta},r_{q}^{\beta}]=0$. (This is the reason for the extra factor ``two'' in $2\cdot 2^{2v}!$). Now, commutation relations between $r_{p}^{\alpha}$ and $\ell_{q}^{\beta}$ are
\begin{align}
\left\{
\begin{array}{ccc}
\ell^{\beta}_{k_{0}+1} ,& \cdots ,& \ell^{\beta}_{k} \\
r^{\alpha}_{k_{0}+1}    ,& \cdots ,& r^{\alpha}_{k}
\end{array}
\right\} \ \subset \ \textbf{L}_{\tilde{n}_{1},\tilde{n}_{2}}.
\end{align}
Also, $r_{q}^{\alpha}$ commutes with $l^{\alpha}_{p}$ since $l^{\alpha}_{p}$ can be defined so that $l^{\alpha}_{p}$ has no overlap with $r_{p}^{\alpha}$ due to the translation equivalence of logical operators. Thus, we obtain the following set of independent logical operators (see Fig.~\ref{fig_proof_aid2}(a)(b)):
\begin{align}\label{eq:temp1}
\left\{
\begin{array}{cccccc}
\ell^{\alpha}_{1} ,& \cdots ,& \ell^{\alpha}_{k_{0}} ,& \ell^{\beta}_{k_{0}+1} ,& \cdots ,& \ell^{\beta}_{k} \\
              & &              & r_{k_{0}+1}^{\alpha}    ,& \cdots ,& r_{k}^{\alpha}
\end{array}\right\} \ \subset \ \textbf{L}_{\tilde{n}_{1},\tilde{n}_{2}}
\end{align}
where $\textbf{L}_{\tilde{n}_{1},\tilde{n}_{2}}$ is a set of all the logical operators for $n_{1}=\tilde{n}_{1}$ and $n_{2}=\tilde{n}_{2}$. Since all the logical operators defined inside $\overline{Q(2)}$ can be defined inside $Q(2)$ from lemma~\ref{lemma_1dim_strong}, logical operators $r_{k_{0}+1}^{\alpha}, \cdots, r_{k}^{\alpha}$ can be defined inside $Q(2)$. Let us recall that we have been considering the case with odd $\tilde{n}_{2}$. Given $r_{k_{0}+1}^{\alpha} , \cdots , r_{k}^{\alpha}$ defined inside $Q(2)$, we can represent $r_{p}^{\alpha}$ such that they are periodic in the $\hat{2}$ direction from lemma~\ref{lemma_periodicity} since
\begin{align}
r_{p}^{\alpha} \sim \prod_{y=1}^{\tilde{n}_{2}}T_{2}(r_{p}^{\alpha})
\end{align}
due to the translation equivalence of logical operators. Let us choose $r_{p}^{\alpha}$ in the following way (Fig.~\ref{fig_proof_aid2}(b)):
\begin{align}\label{eq:100}
r_{p}^{\alpha} \ =\ \prod_{y=1}^{\tilde{n}_{2}} T_{2}^{y}(X_{p})  \
                 = \ \begin{bmatrix}
                 X_{p}  ,&  I ,& \cdots ,& I \\
                 \vdots ,&  \vdots ,& \vdots ,& \vdots \\
                 X_{p}  ,&  I ,& \cdots ,& I \\
                 X_{p}  ,&  I ,& \cdots ,& I 
                 \end{bmatrix}\qquad (p = k_{0} + 1, \cdots,k)
\end{align}
where $X_{p}$ is some Pauli operator acting on a composite particle $P_{1,1}$. $X_{p}$ commutes with $X_{q}$ for $p,q = 1,\cdots, k_{0}$. Now, let us show that $r_{p}^{\alpha}$ for $p = k_{0}+1, \cdots, k$, defined in the forms in Eq.~(\ref{eq:100}), are independent logical operators for any $n_{2}$ and $n_{1} = \tilde{n}_{1}=2\cdot 2^{2v}!$. $r_{p}^{\alpha}$ are centralizer operators for any $n_{2}$. Since $r_{p}^{\alpha}$ anti-commutes with $\ell_{p}^{\beta}$, they are not stabilizers. Thus, the following logical operators are independent for any $n_{2}$ and $n_{1}= \tilde{n}_{1} = 2 \cdot 2^{2v}!$:
\begin{align}
\left\{
\begin{array}{ccc}
\ell^{\beta}_{k_{0}+1} ,& \cdots ,& \ell^{\beta}_{k} \\
r^{\alpha}_{k_{0}+1}    ,& \cdots ,& r^{\alpha}_{k}
\end{array}
\right\} \ \subset \ \textbf{L}_{\tilde{n}_{1}, \forall n_{2}}. \label{eq:canonical_summary1}
\end{align}
Here, $\textbf{L}_{\tilde{n}_{1}, \forall n_{2}}$ means that operators on the left hand side are logical operators for any $n_{2}$ and $n_{1}=\tilde{n}_{1}$.

We continue to consider the case where $n_{1}= \tilde{n}_{1} = 2\cdot 2^{2v}!$ and $n_{2} = \tilde{n}_{2}$. Since $\ell^{\alpha}_{1}, \cdots , \ell^{\alpha}_{k_{0}}$ and $r_{k_{0}+1}^{\alpha}, \cdots, r_{k}^{\alpha}$ are independent logical operators which can be defined inside $Q(2)$, by using lemma~\ref{lemma_1dim_strong}, we notice that there exists a following canonical set of logical operators:
\begin{align}
\Pi(\mathcal{S}_{\tilde{n}_{1}, \tilde{n}_{2}})'\ =\
\left\{
\begin{array}{cccccc}
\ell^{\alpha}_{1} ,& \cdots ,& \ell^{\alpha}_{k_{0}} ,& r_{k_{0}+1}^{\alpha}    ,& \cdots ,& r_{k}^{\alpha} \\
\bar{Z}_{1}       ,& \cdots ,&      \bar{Z}_{k_{0}}  ,& \bar{Z}_{k_{0}+1}    ,& \cdots ,& \bar{Z}_{k}
\end{array}\right\}
\end{align}
where $\bar{Z}_{p}$ are periodic in the $\hat{1}$ direction: $T_{1}(\bar{Z}_{p}) = \bar{Z}_{p}$. Here, we denote this new canonical set as $\Pi(\mathcal{S}_{\tilde{n}_{1}, \tilde{n}_{2}})'$ to make the difference from $\Pi(\mathcal{S}_{\tilde{n}_{1},\tilde{n}_{2}})$ clear. Recall that $\bar{Z}_{p}$ are periodic in the $\hat{1}$ direction. Since $n_{2}$ is odd, we can define $\bar{Z}_{p}$ such that they are periodic both in the $\hat{1}$ and $\hat{2}$ directions:
\begin{align}
\bar{Z}_{p}\ \sim\ \prod_{y=1}^{\tilde{n}_{2}} T_{2}^{y}(\bar{Z}_{p}).
\end{align}
Now, through some calculations of commutation relations, we notice that the following logical operator form a canonical set:
\begin{align}
\Pi(\mathcal{S}_{\tilde{n}_{1}, \tilde{n}_{2}})''\ =\
\left\{
\begin{array}{cccccc}
\ell^{\alpha}_{1} ,& \cdots ,& \ell^{\alpha}_{k_{0}} ,& r_{k_{0}+1}^{\alpha}    ,& \cdots ,& r_{k}^{\alpha} \\
\bar{Z}_{1}       ,& \cdots ,&      \bar{Z}_{k_{0}}  ,& \ell_{k_{0}+1}^{\beta}    ,& \cdots ,& \ell_{k}^{\beta}
\end{array}\right\}.
\end{align}

Here, we notice that $\ell_{p}^{\alpha}$ and $\bar{Z}_{p}$ ($p = 1, \cdots, k_{0}$) are independent logical operators for any $n_{1}$ and $n_{2}$ since they are centralizer operators and anti-commute with each other for any $n_{1}$ and $n_{2}$. Let us represent $\bar{Z}_{p}$ as
\begin{align}
\bar{Z}_{p} \ =\ \prod_{x=1}^{\tilde{n}_{1}}\prod_{y=1}^{\tilde{n}_{2}}  T_{1}^{x}T_{2}^{y}(U_{p})  \
              = \ \begin{bmatrix}
                 U_{p}  ,&  U_{p} ,& \cdots ,& U_{p} \\
                 \vdots ,&  \vdots ,& \vdots ,& \vdots \\
                 U_{p}  ,&  U_{p} ,& \cdots ,& U_{p} \\
                 U_{p}  ,&  U_{p} ,& \cdots ,& U_{p} 
                 \end{bmatrix} \qquad (p\ =\ 1,\cdots,k_{0})
\end{align}
where $U_{p}$ is a Pauli operator acting on a composite particle $P_{1,1}$. From geometric shapes, $\bar{Z}_{p}$ ($p = 1, \cdots, k_{0}$) are two-dimensional logical operators, and $\ell_{p}^{\alpha}$ ($p = 1, \cdots, k_{0}$) are zero-dimensional logical operators. 

While two-dimensional logical operators $\bar{Z}_{p}$ commute with each other: $[\bar{Z}_{p},\bar{Z}_{p'}]=0$ for $p, p' = 1, \cdots, k_{0}$ when $n_{1} = \tilde{n}_{1}$ and $n_{2} = \tilde{n}_{2}$, they may anti-commute when both $n_{1}$ and $n_{2}$ are odd integers. In order to find two-dimensional logical operators which commute with each other, we consider the case where $n_{1}$ and $n_{2}$ are some odd integers. To indicate the difference, we denote $n_{1}=\hat{n}_{1}$ and $n_{2}=\hat{n}_{2}$ where $\hat{n}_{1}$ and $\hat{n}_{2}$ are odd integers. Then, $\bar{Z}_{p}$ may not commute with each other since $U_{p}$ may not commute with each other for $p = 1, \cdots,k_{0}$. Let us represent the commutation relations between $\bar{Z}_{p}$ as follows:
\begin{align}
\bar{Z}_{p}\bar{Z}_{q}\ =\ (-1)^{e_{p,q}} \bar{Z}_{q}\bar{Z}_{p}, \quad e_{p,q}\ =\ \pm 1 \qquad \mbox{for}\ p,q\ =\ 1,\cdots, k_{0}.
\end{align}
Here, we define the following two-dimensional logical operators:
\begin{align}
r_{p}^{\beta} \ \equiv\ \bar{Z}_{p} \prod_{q=1}^{k_{0}} ( \hat{\ell}_{p}^{\alpha} )^{ e_{p,q} }, \qquad
\hat{ \ell }_{p}^{\alpha}\ \equiv\ \prod_{x=1}^{\hat{n}_{1}} \prod_{y=1}^{\hat{n}_{2}}  T_{1}^{x} T_{2}^{y} (\ell^{\alpha}_{p})\ \sim\ \ell_{p}^{\alpha}.
\end{align}
We note that $\hat{\ell}_{p}^{\alpha}$ is periodic both in the $\hat{1}$ and $\hat{2}$ directions. Then, we have a following set of independent logical operators (see lemma~\ref{lemma_equiv}):
\begin{align}\label{eq:111}
\left\{
\begin{array}{cccccc}
\ell^{\alpha}_{1} ,& \cdots ,& \ell^{\alpha}_{k_{0}} ,&                   & &  \\
r^{\beta}_{1}    ,& \cdots ,& r^{\beta}_{k_{0}}     ,& r_{k_{0}+1}^{\alpha}  ,& \cdots ,& r_{k}^{\alpha}
\end{array}\right\} \ \subset \ \textbf{L}_{\hat{n}_{1}, \hat{n}_{2}}. 
\end{align}
Here, $\ell^{\alpha}_{1}, \cdots , \ell^{\alpha}_{k_{0}}$ are zero-dimensional logical operators and $r^{\beta}_{1}, \cdots , r^{\beta}_{k_{0}}$ are two-dimensional logical operators. In particular, they are logical operators regardless of $n_{1}$ and $n_{2}$:
\begin{align}
\left\{
\begin{array}{ccc}
\ell^{\alpha}_{1} ,& \cdots ,& \ell^{\alpha}_{k_{0}}   \\
r^{\beta}_{1}    ,& \cdots ,& r^{\beta}_{k_{0}}  
\end{array}\right\} \ \subset \ \textbf{L}_{\forall n_{1}, \forall n_{2}}. \label{eq:canonical_summary2}
\end{align}
Here, $\textbf{L}_{\forall n_{1}, \forall n_{2}}$ means that $\ell^{\alpha}_{p}$ and $r^{\beta}_{p}$ are independent logical operators for all $n_{1}$ and $n_{2}$. 

Let us summarize our discussions so far. In particular, by combining Eq.~(\ref{eq:canonical_summary1}) and Eq.~(\ref{eq:canonical_summary2}), we have the following independent logical operators for any $n_{2} > 2v$ and $n_{1} = 2 \cdot 2^{2v}!$:
\begin{align}
\Pi(\mathcal{S}_{ \tilde{n}_{1}, \forall n_{2}})''\ =\
\left\{
\begin{array}{cccccc}
\ell^{\alpha}_{1} ,& \cdots ,& \ell^{\alpha}_{k_{0}} ,& r_{k_{0}+1}^{\alpha}    ,& \cdots ,& r_{k}^{\alpha} \\
r_{1}^{\beta}       ,& \cdots ,&      r_{k_{0}}^{\beta}  ,& \ell_{k_{0}+1}^{\beta}    ,& \cdots ,& \ell_{k}^{\beta}
\end{array}\right\}
\end{align}
where $\Pi(\mathcal{S}_{ \tilde{n}_{1}, \forall n_{2}})''$ means that this canonical set is valid for $n_{1}=\tilde{n}_{1}$ and any $n_{2}$. One can represent each logical operator in the following way.
\begin{itemize}
\item $\ell^{\alpha}_{p}$ are zero-dimensional logical operators defined inside $U_{2v,1}$ for $p=1,\cdots, k_{0}$.
\item $r_{p}^{\beta}$ are two-dimensional logical operators which are periodic in both directions: $T_{1}(r_{p}^{\beta})=T_{2}(r_{p}^{\beta})=r_{p}^{\beta}$ with 
\begin{align}
r_{p}^{\beta} \ = \ \prod_{x=1}^{n_{1}} \prod_{y=1}^{n_{2}} T_{1}^{x}T_{2}^{y}(X_{p}) \
= \ \begin{bmatrix}
                 X_{p}  ,&  X_{p} ,& \cdots ,& X_{p} \\
                 \vdots ,&  \vdots ,& \vdots ,& \vdots \\
                 X_{p}  ,&  X_{p} ,& \cdots ,& X_{p} \\
                 X_{p}  ,&  X_{p} ,& \cdots ,& X_{p} 
                 \end{bmatrix}
\end{align}
for $p = 1, \cdots , k_{0}$.
\item $r_{p}^{\alpha}$ are one-dimensional logical operators defined inside $Q(2)$, and periodic: $T_{2}(r_{p}^{\alpha})= r_{p}^{\alpha}$ with 
\begin{align}
r_{p}^{\alpha} \ = \ \prod_{y=1}^{n_{2}} T_{2}^{y}(X_{p}) \
                 = \ \begin{bmatrix}
                 X_{p}  ,&  I ,& \cdots ,& I \\
                 \vdots ,&  \vdots ,& \vdots ,& \vdots \\
                 X_{p}  ,&  I ,& \cdots ,& I \\
                 X_{p}  ,&  I ,& \cdots ,& I 
                 \end{bmatrix}
\end{align}
for $p = k_{0} + 1, \cdots , k$.
\item $\ell_{p}^{\beta}$ are one-dimensional logical operators defined inside $Q(1)$, and periodic with periodicities $b_{p}$: $T_{1}^{b_{p}}(\ell_{p}^{\beta})=\ell_{p}^{\beta}$ where $b_{p}\leq 2^{2v}$. 
\end{itemize}
Here, we notice that logical operators $\ell^{\alpha}_{p}$ ($p=1,\cdots, k_{0}$), $r_{p}^{\alpha}$ ($p=k_{0}+1,\cdots, k$) and $r_{p}^{\beta}$ ($p=1,\cdots, k_{0}$) have the forms described in a canonical set of logical operators described in Section~\ref{sec:2D2}.

To complete the proof of the theorem, we need to analyze $\ell_{p}^{\beta}$. In particular, we need to show that $\ell_{p}^{\beta}$ can be defined in a periodic way: $T_{1}(\ell_{p}^{\beta})=\ell_{p}^{\beta}$. In fact, it is immediate to show this by considering the case where $n_{2} = 2\cdot 2^{2v}!$ and $n_{1}$ is odd instead of the case where $n_{1} = 2\cdot 2^{2v}!$ and $n_{2}$ is odd. The proof basically follows the same discussion used in the analyses on the case $n_{1} = 2\cdot 2^{2v}!$ and $n_{2}$ is odd, and thus, we skip it. This completes the derivation of a canonical set of logical operators.


\section{The existence of loop-like stabilizers}\label{sec:loop}

In this appendix, we sketch the proof for the existence of stabilizers which are used to construct loop-like stabilizers appeared in Section~\ref{sec:2D2}. For simplicity of discussion, we discuss the cases without zero-dimensional logical operators ($k_{0}=0$ and $k_{1}=k$). Generalizations to the cases where $k_{0}\not=0$ are immediate.

Let us recall that one can represent all the one-dimensional logical operators as follows:
\begin{align}
\ell_{p}\ =\ \prod_{j} Z_{p}^{(1,j)} 
         \ = \ 
      \begin{bmatrix}
       Z_{p}      , & I      , & \cdots ,& I     \\
       \vdots  & \vdots  & \vdots & \vdots \\
       Z_{p}      , & I      , & \cdots ,& I     \\
       Z_{p}   , & I  , & \cdots ,& I       
      \end{bmatrix}, \quad
r_{p}\ =\ \prod_{i} X_{p}^{(i,1)} 
         \ = \ 
      \begin{bmatrix}
       I      , & I      , & \cdots ,& I     \\
       \vdots  & \vdots  & \vdots & \vdots \\
       I      , & I      , & \cdots ,& I     \\
       X_{p}   , & X_{p}  , & \cdots ,& X_{p}       
      \end{bmatrix} 
\end{align}
for $p =1, \cdots, k$. These $2k$ operators are independent logical operators regardless of $n_{1}$ and $n_{2}$.

We represent independent stabilizer generators which can be defined inside $2 \times 2$ composite particles in the following way:
\begin{align}
S_{q}^{(i,j)} \ = \ \left[
\begin{array}{cc}
C_{q} ,& D_{q} \\
A_{q} ,& B_{q} 
\end{array}\right]^{(i,j)}
\end{align}
for $q = 1, \cdots, v'$ with $v' \geq v$ where $v$ is the number of qubits inside each composite particle. Stabilizers $S_{q}^{(i,j)}$ are well defined for the cases where $n_{1}=1$ or $n_{2}=1$, as discussed in Section~\ref{sec:model}. For example, when $n_{2}=1$, $S_{q}^{(i,1)}$ may be represented as 
\begin{align}
S_{q}^{(i,1)} \ = \ [\alpha_{q} , \beta_{q}]^{(i,1)}
\end{align}
where $\alpha_{q} \equiv A_{q}C_{q}$ and $\beta_{q} \equiv B_{q}D_{q}$. $S_{q}^{(i,1)}$ is defined inside a region with $2 \times 1$ composite particles.

For convenience of discussion, we introduce the following notations:
\begin{align}
\alpha(\textbf{R}) \ = \ \prod_{q \in \textbf{R}} \alpha_{q}, \qquad \beta(\textbf{R}) \ = \ \prod_{q \in \textbf{R}} \alpha_{q}, \qquad S^{(i,1)}(\textbf{R}) \ = \ \prod_{q \in \textbf{R}} S^{(i,1)}_{q}
\end{align}
where $\textbf{R}$ represents a set of integers in $\mathbb{Z}_{v'}$. We also define the summation of sets of integers $\textbf{R}$ and $\textbf{R}'$ such that
\begin{align}
\alpha(\textbf{R} \oplus \textbf{R}') \ = \ \alpha(\textbf{R})\alpha(\textbf{R}').
\end{align}
Thus, ``$\oplus$'' satisfies the following properties:
\begin{equation}
\begin{split}
j \ \in \ \textbf{R} \quad j \ \in \textbf{R}' &\Rightarrow j \in \textbf{R} \oplus \textbf{R}'\\
j \ \in \ \textbf{R} \quad j \ \not\in \textbf{R}' &\Rightarrow j \not\in \textbf{R} \oplus \textbf{R}'\\
j \ \not\in \ \textbf{R} \quad j \ \in \textbf{R}' &\Rightarrow j \not\in \textbf{R} \oplus \textbf{R}'\\
j \ \not\in \ \textbf{R} \quad j \ \not\in \textbf{R}' &\Rightarrow j \in \textbf{R} \oplus \textbf{R}'.
\end{split}
\end{equation}

We begin by proving the following lemma.

\begin{lemma}\label{lemma:jo}
For a two-dimensional STS model, consider the cases with $n_{1}\times 1$ composite particles ($n_{1}\geq 2$ and $n_{2}=1$). Then, there exist sets of integers $\textbf{R}_{p}$ ($p = 1,\cdots, k$) such that
\begin{align}
\alpha(\textbf{R}_{p}) \ = \ \beta(\textbf{R}_{p}),
\end{align}
and 
\begin{align}
\left[ \alpha(\textbf{R}_{p}), I , \cdots ,I  \right] \ \sim \ \ell_{p} \ = \ \left[ Z_{p}, I, \cdots I \right].
\end{align}
\end{lemma}

\begin{proof}
Let us consider the case with $n_{1}=2$ and $n_{2}=1$ ($2 \times 1$ composite particles). In this case, logical operators are 
$\ell_{p} = [ Z_{p} , I ]$ and $r_{p} = [ X_{p} , X_{p}]$.
Here, we represent Pauli operators acting on the system of $2 \times 1$ composite particles through two component vectors. These representations should not be confused with commutators. 

Due to the translation equivalence of logical operators, $T_{1}(\ell_{p})\ell_{p}$ is a stabilizer:
\begin{align}
T_{1}(\ell_{p})\ell_{p} \ = \ [ Z_{p} , Z_{p} ] \ \in \mathcal{S}_{2,1}.
\end{align}
Here, $\mathcal{S}_{2,1}$ is the stabilizer group. Then, this stabilizer can be decomposed as a product of stabilizers $S_{q}^{(i,1)}$ and their translations. In particular, there exist some set of integers $\textbf{R}_{p}$ and $\textbf{R}_{p}'$ such that
\begin{align}
[Z_{p} , Z_{p} ] \ = \ S^{(i,1)}(\textbf{R}_{p})T_{1}(S^{(i,1)}(\textbf{R}_{p}')) = \ [\alpha(\textbf{R}_{p})\beta(\textbf{R}_{p}') , \alpha(\textbf{R}_{p}')\beta(\textbf{R}_{p})] 
\end{align}
where $S^{(i,1)}(\textbf{R}_{p}) = [ \alpha(\textbf{R}_{p}) , \beta(\textbf{R}_{p}) ]$ and $S^{(i,1)}(\textbf{R}_{p}') = [ \alpha(\textbf{R}_{p}') , \beta(\textbf{R}_{p}') ]$.
Here, we have 
\begin{align}
[Z_{p} , I ] \ = \ [\alpha(\textbf{R}_{p})\beta(\textbf{R}_{p}') , I]  \ = \ [\alpha(\textbf{R}_{p}')\beta(\textbf{R}_{p}), I] \ \in \ \mathcal{C}_{2,1}.
\end{align}
Also, let us consider the following stabilizer:
\begin{align}
S^{(i,1)}(\textbf{R}_{p} \oplus \textbf{R}_{p}') \ = \ [\alpha(\textbf{R}_{p} \oplus \textbf{R}_{p}') , \beta(\textbf{R}_{p} \oplus\textbf{R}_{p}')].
\end{align} 
Then, since 
\begin{align}
Z_{p} = \alpha(\textbf{R}_{p})\beta(\textbf{R}_{p}') = \alpha(\textbf{R}_{p}')\beta(\textbf{R}_{p}),
\end{align}
we have $\alpha(\textbf{R}_{p} \oplus \textbf{R}_{p}') = \beta(\textbf{R}_{p} \oplus \textbf{R}_{p}')$, and
\begin{align}
[ \alpha(\textbf{R}_{p} \oplus \textbf{R}_{p}') , I ]\ = \ [\beta(\textbf{R}_{p} \oplus \textbf{R}_{p}'), I] \ \in \ \mathcal{C}_{2,1}. 
\end{align}
Then, we notice that $[ \alpha(\textbf{R}_{p})\beta(\textbf{R}_{p}) , I ] \in  \mathcal{C}_{2,1}$.
Since $S^{(i,1)}(\textbf{R}_{p}) = [ \alpha(\textbf{R}_{p}) , \beta(\textbf{R}_{p}) ]$,
we also have $[ \alpha(\textbf{R}_{p}), I ],  [ \beta(\textbf{R}_{p}), I ]  \in  \mathcal{C}_{2,1}$.
By using a similar discussion, we eventually have 
\begin{align}
[ \alpha(\textbf{R}_{p}), I ], \ [ \alpha(\textbf{R}_{p}'), I ], \ [ \beta(\textbf{R}_{p}), I ], \ [ \beta(\textbf{R}_{p}'), I ]\ \in \ \mathcal{C}_{2,1}. 
\end{align}

Since $S^{(i,1)}(\textbf{R}_{p}) = [ \alpha(\textbf{R}_{p}) , \beta(\textbf{R}_{p}) ]  \in  \mathcal{S}_{2,1}$,
we have 
$[ \alpha(\textbf{R}_{p})\beta(\textbf{R}_{p}), I ] \in \mathcal{S}_{2,1}$
due to the translation equivalence of logical operators. Thus, we finally have 
\begin{align}
S^{(i,1)}(\textbf{R}_{p}\oplus \textbf{R}_{p}') \ = \ [\alpha(\textbf{R}_{p}\oplus \textbf{R}_{p}') , \alpha(\textbf{R}_{p}\oplus \textbf{R}_{p}') ]
\end{align}
and 
\begin{align}
[\alpha(\textbf{R}_{p}\oplus \textbf{R}_{p}') , I ] \ \sim \ [\beta(\textbf{R}_{p}) \alpha(\textbf{R}_{p}') , I ] \ = \ \ell_{p} \ = \ [Z_{p} , I ].
\end{align}
Though we have limited our considerations to the cases with $2\times 1$ composite particles, $S^{(i,1)}(\textbf{R}_{p}\oplus \textbf{R}_{p}')$ satisfy the properties described in the lemma for the cases with $n_{1}\times 1$ composite particles for arbitrary $n_{1}$. This can be shown by using the fact that $\ell_{p} = [Z_{p}, I,\cdots,I]$ are logical operators regardless of $n_{1}$. This completes the proof of the lemma. 
\end{proof}

At the beginning of the discussion, we have given a specific set of $S_{q}^{(i,j)}$. However, \emph{as long as they generate the same stabilizer group $\mathcal{S}_{n_{1},n_{2}}$, one can freely choose any stabilizers defined inside a region with $2\times 2$ composite particles as $S_{q}^{(i,j)}$}. Here, we choose $S_{p}^{(i,j)}$ so that the following operators are independent logical operators:
\begin{align}
\ell(\alpha)_{p} \ \equiv \ 
      \begin{bmatrix}
       \alpha_{p}      , & I      , & \cdots ,& I     \\
       \vdots  & \vdots  & \vdots & \vdots \\
       \alpha_{p}      , & I      , & \cdots ,& I     \\
       \alpha_{p}   , & I  , & \cdots ,& I       
      \end{bmatrix} 
\end{align}
for $p = 1, \cdots, k$ where $\alpha_{p} \equiv A_{p}C_{p}$.

Next, let us consider the case with $n_{1} \times 1$ composite particles.
When $n_{1}=1$, $S_{q}^{(1,j)}$ may be represented as 
\begin{align}
S_{q}^{(1,j)} \ = \ \left[
\begin{array}{c}
\delta_{q}  \\
\gamma_{q}
\end{array}\right]^{(i,1)}
\end{align}
where $\gamma_{q} \equiv A_{q}B_{q}$ and $\delta_{q} \equiv C_{q}D_{q}$. Here, we consider the following centralizer operators:
\begin{align}
r(\gamma)_{p} \ \equiv \
      \begin{bmatrix}
       I      , & I      , & \cdots ,& I     \\
       \vdots  & \vdots  & \vdots & \vdots \\
       I      , & I      , & \cdots ,& I     \\
       \gamma_{p}   , & \gamma_{p}  , & \cdots ,& \gamma_{p}       
      \end{bmatrix}.
\end{align}
While $\ell(\alpha)_{p}$ are independent logical operators for $p=1,\cdots ,k$, centralizer operators $r(\gamma)_{p}$ may not be independent logical operators. However, through some speculations, one may notice that one can choose $S_{p}^{(i,1)}$ such that both $\ell(\alpha)_{p}$ and $r(\gamma)_{p}$ are logical operators for $p = 1, \cdots, k$. Then, $\ell(\alpha)_{p}$ and $r(\gamma)_{p}$ are independent logical operators for $p = 1, \cdots, k$. 

Though we have obtained $2k$ independent logical operators, $\ell(\alpha)_{p}$ and $r(\gamma)_{p}$ do not form a canonical set of logical operators. In fact, one will soon see that $[\ell(\alpha)_{p}, r(\gamma)_{p}]=0$ since $[\alpha_{p}, \gamma_{p}]=0$. In order to represent a canonical set of logical operators in terms of $\ell(\alpha)_{p}$ and $r(\gamma)_{p}$, let us use the following commutation relations:
\begin{align}
[\alpha_{p}, \alpha_{p'}] \ = \ 0, \qquad [\gamma_{p}, \gamma_{p'}] \ = \ 0, \qquad [\alpha_{p}, \gamma_{p}] \ = \ 0.\label{eq:commutation_deform}
\end{align}
Since commutations $[\alpha_{p}, \alpha_{p'}] = 0$ and $[\gamma_{p}, \gamma_{p'}] = 0$ are obvious, let us show $[\ell(\alpha)_{p}, r(\gamma)_{p}]=0$. 
We use the fact that 
\begin{align}
S_{p}^{(i,j)} \ = \ \left[
\begin{array}{cc}
C_{p} ,& D_{p} \\
A_{p} ,& B_{p} 
\end{array}\right]^{(i,j)}
\end{align}
are stabilizers, and, they commute with their own translations. For $p = 1,\cdots,k$, we have $D_{p}=A_{p}B_{p}C_{p}$. Then, we have
\begin{align}
[A_{p}, D_{p}] \ = \ [A_{p}, A_{p}B_{p}C_{p}] \ = \ [A_{p}, B_{p}C_{p}] \ = \ 0, \qquad [B_{p}, C_{p}] \ = \ 0,
\end{align}
and, as a result, we have
\begin{align}
[\alpha_{p},\gamma_{p}] \ = \ [ A_{p}C_{p}, A_{p}B_{p} ] \ = \ 0. 
\end{align}
With small extension, one can also show that 
\begin{align}
 [\alpha(\textbf{R}), \gamma(\textbf{R})] \ = \ 0
\end{align}
for any $\textbf{R}$ where $\textbf{R}$ is a set of integers inside $\mathbb{Z}_{k}$. Here, $\gamma(\textbf{R})$ is defined in a way similar to $\alpha(\textbf{R})$.

Now, let us show that $k$ is even, and the existence of deformers described in Section~\ref{sec:2D2}. Below, we shall frequently change the choices of $S_{q}^{(i,j)}$ while keeping the group generated from $S_{q}^{(i,j)}$: $\langle \{ S_{q}^{(i,j)} \}_{q=1}^{k} \rangle$.

Let us begin by analyzing the case where $k=1$. Then, we have $[\alpha_{1}, \gamma_{1}]=0$. However, this contradicts with the fact that the following operators are independent logical operators:
\begin{align}
\begin{bmatrix}
       \alpha_{1}      , & I      , & \cdots ,& I     \\
       \vdots  & \vdots  & \vdots & \vdots \\
       \alpha_{1}      , & I      , & \cdots ,& I     \\
       \alpha_{1}   , & I  , & \cdots ,& I       
      \end{bmatrix}, \quad
      \begin{bmatrix}
       I      , & I      , & \cdots ,& I     \\
       \vdots  & \vdots  & \vdots & \vdots \\
       I      , & I      , & \cdots ,& I     \\
       \gamma_{1}   , & \gamma_{1}  , & \cdots ,& \gamma_{1}       
      \end{bmatrix}.
\end{align}
Thus, $k \not= 1$.

Next, let us consider the case where $k=2$. Then, the commutation relations between $\alpha_{p}$ and $\gamma_{p}$ must be
\begin{align}
\left\{
\begin{array}{cc}
\alpha_{1},& \alpha_{2} \\
\gamma_{2},& \gamma_{1}
\end{array}\right\}.
\end{align}
Here, we apply local unitary transformations to qubits inside each composite particle in the following way:
\begin{align}
U\left\{
\begin{array}{cc}
\alpha_{1},& \alpha_{2} \\
\gamma_{2},& \gamma_{1}
\end{array}\right\}U^{-1} \ = \ 
\left\{
\begin{array}{cc}
Z_{1},& Z_{2} \\
X_{1},& X_{2}
\end{array}\right\}.
\end{align}
Then, after unitary transformations, we have 
\begin{align}
A_{1}C_{1} \ = \ Z_{1}, \qquad  A_{2}C_{2} \ = \ Z_{2}, \qquad A_{1}B_{1} \ = \ X_{2}, \qquad  A_{2}B_{2} \ = \ X_{1}.
\end{align}
Thus, with some Pauli operators $P_{1}$ and $P_{2}$, one can represent $S^{i,j}_{1}$ and $S^{i,j}_{2}$ as follows:
\begin{align}
S_{1}^{(i,j)} \ = \ 
\left[
\begin{array}{cc}
P_{1}Z_{1}X_{2} ,& P_{1}Z_{1} \\
P_{1}X_{2}      ,&  P_{1} 
\end{array}\right]^{(i,j)}, \qquad
S_{2}^{(i,j)} \ = \ 
\left[
\begin{array}{cc}
P_{2}     ,& P_{2}X_{1} \\
P_{2}Z_{2} ,& P_{2}Z_{2}X_{1} 
\end{array}\right]^{(i,j)}
\end{align}
which are stabilizers described in Section~\ref{sec:2D2}.

Now, we consider the cases where $k \geq 3$. Since $\ell(\alpha)_{p}$ and $\ell(\gamma)_{p}$ are independent logical operators, there exists some integer $i \not=1$ such that $\{\alpha_{1}, \gamma_{i} \}=0$.  Without loss of generality, we can set $i = 2$. Also, one can choose $S_{p}^{(i,j)}$ so that $[\gamma_{3},\alpha_{1}]=0$. Then, from Eq.~(\ref{eq:commutation_deform}), we have 
\begin{align}
\left\{
\begin{array}{cc}
\alpha_{1},& \alpha_{2} \\
\gamma_{2},& \gamma_{1}
\end{array}\right\}.
\end{align}
Also, through some computations of commutation relations, one may notice that the following; if $\gamma_{3}$ commutes with $\alpha_{2}$, then $\alpha_{3}$ commutes with $\gamma_{2}$, and if $\gamma_{3}$ anti-commutes with $\alpha_{2}$, then $\alpha_{3}$ anti-commutes with $\gamma_{2}$. 

When $\gamma_{3}$ commutes with $\alpha_{2}$, we have 
\begin{align}
\left\{
\begin{array}{cccc}
\alpha_{1},& \alpha_{2} ,& \alpha_{3} ,& \gamma_{3} \\
\gamma_{2},& \gamma_{1}  & &
\end{array}\right\}. \label{eq:comm_loop}
\end{align}
When $\gamma_{3}$ anti-commutes with $\alpha_{2}$, we have 
\begin{align}
\left\{
\begin{array}{cccc}
\alpha_{1},& \alpha_{2} ,& \alpha_{3}\alpha_{1} ,& \gamma_{3}\gamma_{1} \\
\gamma_{2},& \gamma_{1}  & &
\end{array}\right\}.
\end{align}
Then, by changing the choice of $S_{p}^{(i,j)}$ so that
\begin{align}
S_{1}^{(i,j)} \ \rightarrow \  S_{1}^{(i,j)}, \qquad S_{2}^{(i,j)} \ \rightarrow \ S_{2}^{(i,j)}, \qquad S_{3}^{(i,j)} \ \rightarrow \ S_{1}^{(i,j)}S_{3}^{(i,j)},
\end{align}
we have
\begin{align}
\left\{
\begin{array}{cccc}
\alpha_{1},& \alpha_{2} ,& \alpha_{3} ,& \gamma_{3} \\
\gamma_{2},& \gamma_{1}  & &
\end{array}\right\}.
\end{align}
Thus, in any cases, by choosing $S_{p}^{(i,j)}$ appropriately, we have the above commutation relations described in Eq.~(\ref{eq:comm_loop}) when $k \geq 3$. 

Now, if $k = 3$, we have a contradiction in a way similar to the case with $k =1$. If $k = 4$, one can choose $S_{p}^{(i,j)}$ so that 
\begin{align}
\left\{
\begin{array}{cccc}
\alpha_{1},& \alpha_{2} ,& \alpha_{3} ,& \alpha_{4} \\
\gamma_{2},& \gamma_{1}  & \gamma_{4} ,& \gamma_{3}
\end{array}\right\}.
\end{align}
Then, we have the following stabilizers:
\begin{align}
&S_{1}^{(i,j)} \ = \ 
\left[
\begin{array}{cc}
P_{1}Z_{1}X_{2} ,& P_{1}Z_{1} \\
P_{1}X_{2}      ,& P_{1} 
\end{array}\right]^{(i,j)}, \ 
S_{3}^{(i,j)} \ = \ 
\left[
\begin{array}{cc}
P_{3}Z_{3}X_{4} ,& P_{3}Z_{3} \\
P_{3}X_{4}      ,& P_{3} 
\end{array}\right]^{(i,j)}\\
&S_{2}^{(i,j)} \ = \ 
\left[
\begin{array}{cc}
P_{2}     ,& P_{2}X_{1} \\
P_{2}Z_{2} ,& P_{2}Z_{2}X_{1} 
\end{array}\right]^{(i,j)}, \  
S_{4}^{(i,j)} \ = \ 
\left[
\begin{array}{cc}
P_{4}     ,& P_{4}X_{3} \\
P_{4}Z_{4} ,& P_{4}Z_{4}X_{3} 
\end{array}\right]^{(i,j)}.
\end{align}

One may use a similar discussion for an arbitrary $k$ and show that $k$ must be even, and there exist a choice of $S_{p}^{(i,j)}$ such that 
\begin{align}
\left\{
\begin{array}{ccccc}
\alpha_{1},& \alpha_{2} ,& \cdots ,& \alpha_{k-1},& \alpha_{k}\\
\gamma_{2},& \gamma_{1} ,& \cdots ,& \gamma_{k},& \gamma_{k-1}
\end{array}\right\}.
\end{align}
This leads to the existence of stabilizers described in Section~\ref{sec:2D2}. 

\section{Scale symmetries and weak breaking of translation symmetries}\label{sec:translation}

In this appendix, we discuss translation symmetries of ground states in STS models in the context of coarse-graining. 

A central idea behind RG transformations is coarse-graining of a system of spins by grouping several spins into a single particle with a larger Hilbert space. Now, in coarse-graining STS models, the main problem is to find the smallest suitable unit cell of qubits to be grouped in coarse-graining. Though one can freely coarse-grain the system by choosing an arbitrarily large unit cell of qubits, one may end up with a coarse-grained particle with an arbitrarily large Hilbert space, which are hard to analyze. Thus, we naturally hope to choose the smallest possible unit cell of qubits.

The main difficulty in grouping qubits comes from the fact that translation symmetries of ground states may be broken. One might hope that a translation symmetric Hamiltonian would give rise to translation symmetric ground states. However, this naive expectation is generally true only when there is a single ground state. In~\ref{sec:translation2}, we show an example of a stabilizer Hamiltonian whose degenerate ground states are invariant only under translations of several qubits, while the Hamiltonian is invariant under unit translations of qubits. Such a bizarre breaking of translation symmetries is sometimes called a weak breaking of translation symmetries, which is often seen in topologically ordered systems~\cite{Kitaev06b}. When a translation symmetry is weakly broken, there exists a ground state whose translation is orthogonal to itself. Thus, when coarse-graining the system of qubits, we hope to group qubits so that ground states are invariant under unit translations of newly defined particles in order to resolve a weak breaking of translation symmetries. 

Here, the whole procedure of coarse-graining may be summarized in the following three steps:
\begin{enumerate}
\item Coarse-grain so that the Hamiltonian is invariant under unit translations of composite particles.
\item Coarse-grain so that interaction terms are defined inside $2 \times 2$ composite particles.
\item Coarse-grain so that ground states are invariant under unit translations of composite particles. 
\end{enumerate}
We have already performed the first and second steps of coarse-graining when the STS model is introduced in Section~\ref{sec:model}. Here, we hope to further coarse-grain the system of composite particles so that all the ground states are invariant under unit translations of coarse-grained particles.

In this appendix, we show that the composite particle, after the first and second steps of coarse-graining, happens to be the smallest suitable unit cell by proving that all the ground states of STS models are invariant under unit translations of composite particles. In other words, we do not need to perform the third step of coarse-graining since \emph{STS models have been already coarse-grained appropriately through composite particles} ! 

In~\ref{sec:translation1}, we give an example of a stabilizer Hamiltonian which exhibits a weak breaking of translation symmetries and show that this system can be appropriately coarse-grained by introducing composite particles. In~\ref{sec:translation2}, we prove that all the ground states of STS models are invariant under unit translations of composite particles. 

Having seen that a system of composite particles has been already coarse-grained properly in STS models so that all the ground states are invariant under unit translations of composite particles due to scale symmetries, we consider the cases where scale symmetries are not satisfied. In~\ref{sec:translation3}, we extend our discussion to the cases without scale symmetries where the number of logical qubits $k_{\vec{n}}$ is not constant and depend on the system size $\vec{n}$. We show that one can always find some finite size of a unit cell of composite particles so that all the ground states are translation symmetric when the number of logical qubits $k_{\vec{n}}$ is upper bounded by some constant. Finally, in~\ref{sec:translation4}, we show that, in translation symmetric stabilizer Hamiltonians, translation symmetries of ground states are weakly broken if and only if scale symmetries are broken. 

\subsection{An example of weak breaking of translation symmetries}\label{sec:translation1}

A unique ground state of a translation symmetric Hamiltonian always has translation symmetries since the translation of the ground state must be the same as the original ground state. However, when ground states are degenerate, translation symmetries of ground states may be broken despite the translation symmetries of the supporting Hamiltonian. In particular, there are cases where degenerate ground states are invariant only under translations of several qubits, while the Hamiltonian is invariant under unit translations of qubits. This bizarre breaking of translation symmetries, sometimes called a weak breaking of translation symmetries, lies behind the main motivation for further coarse-graining the system.

Here, we give an example of a stabilizer Hamiltonian which exhibits a weak breaking of translation symmetries, and show that translation symmetries can be restored by introducing composite particles.

\textbf{Breaking and restoration of translation symmetries:}
Let us consider the following one-dimensional Hamiltonian with periodic boundary conditions:
\begin{align}
H = - \sum_{j} Z^{(j)}Z^{(j+1)}Z^{(j+2)}
\end{align}
where the total number of qubits is $N$, and $Z^{(j)}$ act on qubits labeled by $j$. This Hamiltonian is a stabilizer Hamiltonian, and has four ground states when $N$ is a multiple of three:
\begin{align}
|000 000 \cdots 000\rangle, \quad |011 011 \cdots 011\rangle, \quad |101 101 \cdots 101\rangle, \quad |110 110 \cdots 110\rangle.
\end{align}
Here, we notice that ground states are not invariant under unit translations of qubits despite that the original Hamiltonian is invariant under unit translations (Fig.~\ref{fig_restoration}). 

\begin{figure}[htb!]
\centering
\includegraphics[width=0.70\linewidth]{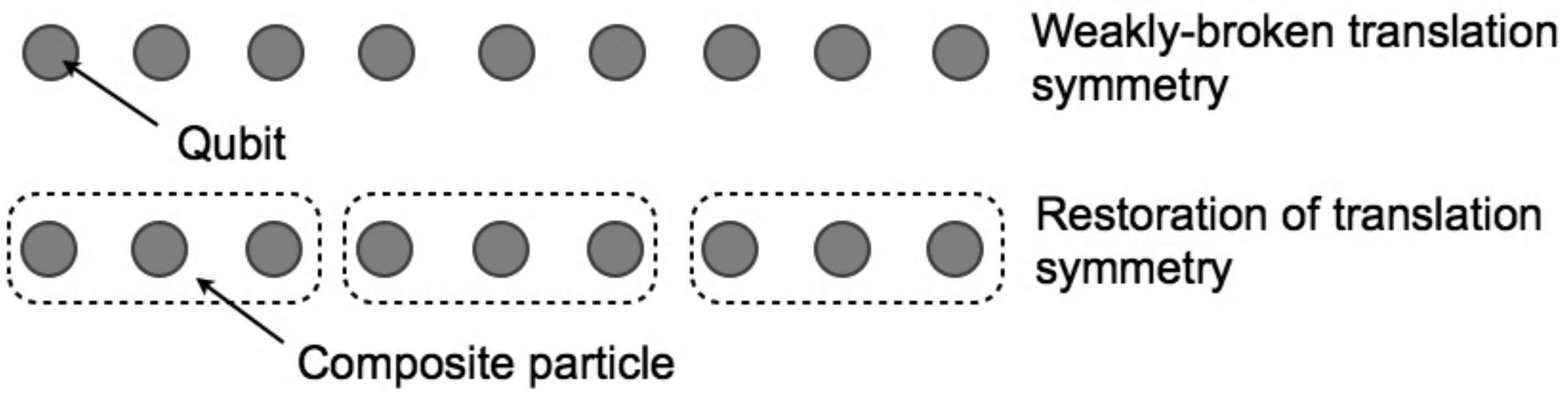}
\caption{Coarse-graining and restoration of translation symmetries.
} 
\label{fig_restoration}
\end{figure}

Now, let us coarse-grain the above system by introducing composite particles. Obviously, if one consider three consecutive qubits as a single composite particle, one can restore translation symmetries: $T^{3}(|\psi\rangle) = |\psi\rangle$
where $T$ shifts the position of qubits by one. Thus, while translation symmetries of ground states are weakly broken, one may restore translation symmetries through coarse-graining.

\textbf{Scale symmetries and composite particles:}
Next, let us see the relation between translation symmetries of ground states and scale symmetries. In a strict sense, this model is not an STS model since the number of degenerate ground states changes according to $N$. Here, let us return to the original picture through qubits before introducing composite particles. When $N$ is a multiple of three, there are four ground states with $k =2$ since $G(\mathcal{S}) = N - 2$. However, when $N$ is not a multiple of three, there is only a single ground state with $k =0$ since $G(\mathcal{S}) = N$. 

In order to treat the model as an STS model, one needs to limit considerations to the cases where $N$ is a multiple of three. Then, by considering three consecutive qubits as a composite particle ($v=3$), one can reduce the system to an STS model: $k_{n} = 2$ for all $n = N/3$.

Now that we have reduced the model to an STS model, one may readily notice that translation symmetries of ground states are not broken since all the ground states are invariant under unit translations of composite particles. Thus, for this STS model, we do not need to further coarse-grain the system. In particular, we notice that the smallest number of qubits to be grouped in order for this model to have scale symmetries is the same as the smallest number of qubits to be grouped for restoring translation symmetries in ground states. We shall discuss whether this relation between translation symmetries of ground states and scale symmetries holds for arbitrary STS models or not in the next subsection.

\subsection{Translation symmetries of ground states in STS models}\label{sec:translation2}

As seen in the above example, even if translation symmetries of ground states are broken, it may be possible to enlarge the size of a unit cell to restore translation symmetries of ground states. However, finding the smallest suitable unit cell of qubits is generally difficult since ground states are highly entangled in topologically ordered systems. In particular, unlike the above example where ground states can be represented simply as direct product states, topologically ordered ground states are not direct product states.

A valuable idea we can import from quantum coding theory is the study of physical symmetries of degenerate ground states through the analysis on physical symmetries of logical operators. In the language of quantum coding theory, the ground state space of the system Hamiltonian is the codeword space which is stabilized by the stabilizer group $\mathcal{S}$. Logical qubits are encoded in the codeword space which is exactly the same as the ground state space. Since the encoding of logical qubits can be analyzed by properties of logical operators, it happens that one can reveal physical symmetries of ground states by studying physical symmetries of logical operators. 

Here, we show that translation symmetries of ground states can be studied by analyzing how logical operators transform under translations. Through the analysis of logical operators, we prove that all the degenerate ground states of STS models are invariant under unit translations of composite particles.

\begin{theorem}[Translation symmetries of degenerate ground states]\label{theorem_main2}
Each and every degenerate ground state $|\psi\rangle$ of an STS model is invariant under unit translations of composite particles in any direction:
\begin{align}
T_{m}(|\psi\rangle)\ =\ |\psi\rangle, \qquad \forall |\psi\rangle\ \in\ V_{\mathcal{S}_{\vec{n}}} \qquad (m\ =\ 1, \cdots , D)
\end{align}
up to a trivial phase factor, where $V_{\mathcal{S}_{\vec{n}}}$ is the ground state space for an STS model defined with the stabilizer group $\mathcal{S}_{\vec{n}}$.
\end{theorem}

Below, we sketch the proof of the theorem. We begin by showing that translation symmetries of ground states can be studied by analyzing how logical operators transform under translations. In particular, we show that ground states are invariant under unit translations of composite particles if and only if the translation equivalence of logical operators holds. 

Consider a stabilizer code which is defined on a $D$-dimensional hypercubic lattice of $n_{1}\times \cdots \times n_{D}$ composite particles. Let us first analyze the necessary condition in order for all the degenerate ground states to possess translation symmetries. For this purpose, let us assume that all the degenerate ground states have translation symmetries and analyze how logical operators must transform under translations. We start by analyzing the case with two-fold degeneracy ($k=1$). Then, any degenerate ground state can be represented in the following way:
\begin{align}
|\psi\rangle\ =\ \alpha |\psi_{0}\rangle + \beta |\psi_{1} \rangle 
\end{align}
where
\begin{align}\label{eq:states_k2}
\ell |\psi_{0}\rangle\ =\ |\psi_{0}\rangle, \qquad \ell |\psi_{1}\rangle\ =\ - |\psi_{1}\rangle, \qquad 
r |\psi_{0}\rangle\ =\ |\psi_{1}\rangle, \qquad r |\psi_{1}\rangle\ =\ |\psi_{0}\rangle
\end{align}
with anti-commuting logical operators $\ell$ and $r$. Here, from an assumption, $|\psi_{0}\rangle$ and $|\psi_{1}\rangle$ have translation symmetries. Then one may notice that $T_{m}(\ell)$ acts in a way exactly the same as $\ell$ inside the ground state space, and one has $\ell \sim T_{m}(\ell)$. A similar discussion holds for $r$, and one has $r \sim T_{m}(r)$. Thus, the translation equivalence of logical operators holds. The generalization of the above discussion to the cases with $k \geq 2$ is immediate, and we notice that all the logical operators must remain equivalent in order for all the degenerate ground states to have translation symmetries:
\begin{align}
T_{m}(|\psi\rangle)\ =\ |\psi\rangle \qquad \mbox{for all}\ |\psi\rangle \ \in\ V_{\mathcal{S}} \qquad \Rightarrow \qquad
\ell\ \sim\ T_{m}(\ell)   \qquad \mbox{for all}\ \ell\ \in\ \textbf{L}
\end{align}
where $V_{\mathcal{S}}$ is the ground state space and $\textbf{L}$ is a set of all the logical operators. 

In fact, ``the translation equivalence of logical operators'' is the necessary and sufficient condition for ``the existence of translation symmetries in degenerate ground states'':
\begin{align}
T_{m}(|\psi\rangle)\ =\ |\psi\rangle \qquad \mbox{for all}\ |\psi\rangle \ \in\ V_{\mathcal{S}} \qquad \Leftrightarrow \qquad
\ell\ \sim\ T_{m}(\ell)   \qquad \mbox{for all}\ \ell\ \in\ \textbf{L}
\end{align}
In order to verify the ``$\Leftarrow$'' relation in the above equivalence between ground states and logical operators, we assume that all the logical operators remain equivalent under some finite translation: $T_{m}(\ell) \sim \ell$. For simplicity of discussion, we show the existence of translation symmetries in degenerate ground states for the cases with $k =1$. Let us represent ground states of the stabilizer Hamiltonian as eigenstates of $\ell$ as represented in Eq.~(\ref{eq:states_k2}). Then, $|\psi_{0}\rangle$ can be represented in terms of the stabilizer group $\mathcal{S}$ and logical operators in the following way~\cite{Beni10}:
\begin{align}
|\psi_{0}\rangle \langle \psi_{0}| \ =\ \frac{1}{2^{N}}(I + \ell)(I+S_{1})\cdots(I+ S_{N-1}) 
                                    \ =\ \frac{1}{2^{N}}\sum_{\forall U \in \langle \ell, \mathcal{S} \rangle} U
\end{align}
where $N$ is the total number of qubits and $S_{1},\cdots, S_{N-1}$ are independent generators for $\mathcal{S}$. Then, one can immediately notice that $T_{m}(|\psi_{0}\rangle)=|\psi_{0}\rangle$ since $T_{m}(\langle \ell, \mathcal{S} \rangle) = \langle \ell, \mathcal{S} \rangle$. One can also repeat the same argument for $|\psi_{1}\rangle$ and obtain $T_{m}(|\psi_{1}\rangle)=|\psi_{1}\rangle$. More generally, we have $T_{m}(|\psi\rangle)=|\psi\rangle$ for any ground state $|\psi\rangle$. Thus, all the degenerate ground states possess translation symmetries. A generalization to the cases where $k \not= 1$ is immediate. 

As a result of the above discussion, we notice that ``the translation equivalence of logical operators'' is equivalent to ``the existence of translation symmetries of ground states''. Thus, in an STS model, \emph{both the system Hamiltonian and degenerate ground states are invariant under unit translations of composite particles.}

\subsection{Reduction to STS models: coarse-graining with composite particles}\label{sec:translation3}

Originally, we have defined composite particles for convenience of discussion so that the system Hamiltonian is invariant under unit translations of composite particles. Fortunately, in the presence of scale symmetries, all the ground states are invariant under unit translations of composite particles, and thus, STS models are coarse-grained properly.

However, in~\ref{sec:translation1}, we have seen that the number of degenerate ground states may depend on the system size in general. In fact, the main challenge is to find a suitable unit cell which restores translations symmetries of ground states. From the discussion in the previous section, we know that it is sufficient to restore scale symmetries in order to translations symmetries of ground states. However, it is not clear whether it is possible to restore translation symmetries of ground states and scale symmetries or not.

Here, we give a sufficient condition for translation symmetric stabilizer Hamiltonians to be properly coarse-grained. We also give a procedure for coarse-graining. In particular, one can show that translation symmetric stabilizer Hamiltonians can be coarse-grained and discussed in the framework of STS models if the number of degenerate ground states are upper bounded by some finite integer, as summarized in the following theorem.

\begin{theorem}[Reduction to an STS model]\label{theorem_reduction}
Consider the case where the number of logical qubits is upper bounded by some finite integer $k$:
\begin{align}
k_{\vec{n}}\ \leq\ k \qquad \mbox{for all} \quad \vec{n}
\end{align}
where the bound is tight for some positive integers $n_{m}=n^{min}_{m}$ for $m= 1,\cdots ,D$. Then, there exist some finite positive integers $a_{m} \leq 2^{2k}$ such that
\begin{align}
k_{a_{1}n_{1},a_{2}n_{2},\cdots , a_{D}n_{D}}\ =\ k \qquad \mbox{for all} \quad n_{1},\cdots, n_{D}
\end{align}
where $n_{1}, \cdots, n_{D}$ are positive integers. 
\end{theorem}

Thus, if we consider $a_{1} \times \cdots \times a_{D}$ composite particles as a new composite particle, the stabilizer code can be treated as an STS model since the system possesses scale symmetries with newly defined composite particles. 

Here, we describe a procedure to find $a_{1},\cdots, a_{D}$ without giving a proof. It is possible to prove that this procedure appropriately coarse-grains stabilizer Hamiltonians through lemma~\ref{lemma_last}. 

\begin{itemize}
\item Find $\vec{n}^{(min)} = (n_{1}^{(min)},\cdots , n_{D}^{(min)})$ such that $k_{\vec{n}}^{(min)}=k$.
\item Consider a system with $\vec{n}^{(min)}$, and find the smallest integer $a_{1}$ such that 
\begin{align}
T_{1}^{a_{1}}(\ell) \ \sim \ell \quad \mbox{for all} \ \ell.
\end{align}
Note that $a_{1} \leq 2^{2k}$ and $n_{1}^{(min)}/a_{1}$ is an integer.
\item Consider a system with $\vec{n}=(a_{1}, n_{2}^{(min)},\cdots, n_{D}^{(min)})$, and find the smallest integer $a_{2}$ such that 
\begin{align}
T_{2}^{a_{2}}(\ell) \ \sim \ell \quad \mbox{for all} \ \ell.
\end{align}
\item Consider a system with $\vec{n}=(a_{1}, a_{2}, n_{3}^{(min)}, \cdots, n_{D}^{(min)})$, and obtain $a_{3}$.
\item Repeat the same procedure until we obtain $a_{1},\cdots, a_{D}$.
\end{itemize}

\subsection{Scale symmetries and weak-breaking of translation symmetries}\label{sec:translation4}

In summarizing the discussions so far, we have the following relation between scale symmetries and translation symmetries of ground states:
\begin{align}
k_{\vec{n}} \ = \ k \qquad \Rightarrow \qquad \ell\ \sim\ T_{m}(\ell) \qquad \Leftrightarrow \qquad  T_{m}(|\psi\rangle)\ =\ |\psi\rangle 
\end{align}
where the above equations represent ``scale symmetries'', ``the translation equivalence of logical operators'' and ``translation symmetries of ground states'' respectively. 

A naturally arising question is whether the existence of scale symmetries is the necessary and sufficient condition for the existence of translation symmetries of ground states or not. It turns out that scale symmetries and translation symmetries of ground states are the equivalent conditions for stabilizer Hamiltonians:
\begin{align}
k_{\vec{n}} \ = \ k \qquad \Leftrightarrow \qquad \ell\ \sim\ T_{m}(\ell) \qquad \Leftrightarrow \qquad  T_{m}(|\psi\rangle)\ =\ |\psi\rangle.
\end{align}
We summarize the above relation in the following theorem. 

\begin{theorem}[Translation symmetries of ground states and scale symmetries]
Consider translation symmetric stabilizer Hamiltonians which are invariant under unit translations of composite particles:
\begin{align}
T_{m}(H) \ = \ H \qquad ( m \ = \ 1,\cdots,D )
\end{align}
and whose interaction terms are defined inside hypercubic regions with $2 \times \cdots \times 2$ composite particles. Let $k_{\vec{n}}$ be the number of logical qubits for a system with the size $\vec{n}$. Then, the existence of scale symmetries is the necessary and sufficient condition for the existence of translation symmetries of ground states:
\begin{align}
k_{\vec{n}} \ = \ k \qquad \mbox{for all} \quad \vec{n} \qquad \Leftrightarrow \qquad T_{m}(|\psi\rangle) \ = \ |\psi\rangle \qquad \mbox{for all} \quad |\psi\rangle \ \in \ V_{\mathcal{S}_{\vec{n}}}
\end{align}
where $V_{\mathcal{S}_{\vec{n}}}$ represents the ground state space for a system with the size $\vec{n}$.
\end{theorem}

In order to prove this, one needs to show that the translation equivalence of logical operators is a sufficient condition for scale symmetries:
\begin{align}
k_{\vec{n}} \ = \ k \quad  \mbox{for all}\quad \vec{n}  \qquad \Leftarrow \qquad \ell\ \sim \ T_{m}(\ell) \quad \mbox{for all}\quad \vec{n}
\end{align}
Here, we give only a sketch of the proof since discussions used for the proof are similar to discussions used for the proofs of other theorems in this paper. The goal is to show that the number of logical qubits is the same for all $\vec{n}$. First, we consider the cases where $n_{m}$ are odd integers for all $m$. Due to the translation equivalence of logical operators, all the logical operators have translation symmetric representations. In particular, after some appropriate local unitary transformations on qubits inside each composite particles, we have
\begin{align}
\ell_{p} \ = \ \prod_{m=1}^{D}\prod_{x_{m}=1}^{n_{m}}T_{m}^{x_{m}}(X_{p}), \qquad
r_{p} \ = \ \prod_{m=1}^{D}\prod_{x_{m}=1}^{n_{m}}T_{m}^{x_{m}}(Z_{p}).
\end{align}
Then, when $n_{m}=1$ for all $m$, $X_{p}$ and $Z_{p}$ are independent logical operators. This implies that $k_{\vec{n}}$ must be the same for the cases where $n_{m}$ are odd integers for all $m$. We denote $k_{\vec{n}}$ for the cases where $n_{m}$ are odd as $k$. 

Next, we consider the case where only $n_{1}$ is an even integer and other $n_{m}$ are odd integers. Let us denote the system size as $\vec{n}'$ where $n_{1}'$ is an even integer and $n_{m}'$ are odd integers for $m >1$. We also denote the number of logical qubits as $k'$. Then, from lemma~\ref{lemma_last}, one can show that there exists a canonical set of logical operators as follows:
\begin{align}
\Pi(\mathcal{S}_{\vec{n}'}) \ = \ 
\left\{
\begin{array}{cccc}
\ell_{1}',& \cdots ,& \ell_{k'}' \\
r_{1}'   ,& \cdots ,& r_{k'}'
\end{array}\right\}
\end{align}
where $\ell_{p}'$ are defined inside a region with $1 \times n_{2}' \times \cdots \times n_{D}'$ composite particles, and $r_{p}'$ are defined in a periodic way in the $\hat{1}$ direction: $T_{1}(r_{p}')=r_{p}'$. Then, $\ell_{1}', \cdots , \ell_{k'}'$ and $r_{1}', \cdots , r_{k'}'$ are independent logical operators for the case where the system size is $n_{1} = 1$ and $n_{m} = n_{m}'$ for $m >1$ since $r_{p}'$ are defined in a periodic way. Then, we have $k' = k$. Thus, for any cases where $n_{1}$ is even and all the other $n_{m}$ are odd, the number of logical qubits is $k$. 

One can repeat the similar argument for the case where $n_{1}$ and $n_{2}$ are even, and other $n_{m}$ are odd. By using this argument, one can show that the number of logical qubits must be $k$ for any $\vec{n}$.

\section*{Acknowledgments}

I would like to thank Barbara Terhal and Sergey Bravyi for motivating some of the questions that led to the results in this paper. I also would like to thank Xie Chen for helpful discussions. In particular, at the early stage of this research, Xie Chen pointed me out the preliminary idea of the work~\cite{Bravyi10c} which led to the notion of scale symmetries in STS models. I would like to thank Sam Ocko and Isaac Chuang for proofreading some parts of the draft and giving valuable feedbacks. I am grateful to Alioscia Hamma for helpful discussions. Various comments and feedbacks received during the informal talk at the condensed matter theory group at MIT motivated me to seek for condensed matter theoretical interpretations of some results of this paper which were originally written in the language of quantum coding theory.

\end{document}